\documentclass[11pt,  reqno]{amsart}
\usepackage{graphicx,amsmath,amssymb,epsfig,epstopdf}
\usepackage{color}

\long\def\comment#1{}
\long\def\comment#1{}
\newtheorem{algorithm}{Algorithm}
\newtheorem{theorem}{Theorem}

\newtheorem{corollary}{Corollary}
\newtheorem{lemma}{Lemma}

\theoremstyle{definition}

\newtheorem{remark}{Comment}

\newcommand{\be}{\begin{eqnarray}}
\newcommand{\ee}{\end{eqnarray}}

\newcommand{\XX}{\mathcal{X}}

\newcommand{\ba}{\begin{array}}
\newcommand{\ea}{\end{array}}
\newcommand{\bs}{\begin{align}\begin{split}\nonumber}
\newcommand{\bsnumber}{\begin{align}\begin{split}}
\newcommand{\es}{\end{split}\end{align}}

\renewcommand{\(}{\left(}
\renewcommand{\)}{\right)}
\renewcommand{\[}{\left[}
\renewcommand{\]}{\right]}
\renewcommand{\hat}{\widehat}

\newcommand{\N}{\mathcal{N}}
\newcommand{\Z}{G}

\newcommand{\Gn}{\mathbb{G}_n}

\newcommand{\Pn}{\mathbb{P}_n}
\newcommand{\Pp}{P}
\newcommand{\En}{\mathbb{E}_n}
\newcommand{\Ep}{E}

\newcommand{\underf}{\underline{f}}
\newcommand{\Uniform}{{\text{Uniform}}}

\def\RR{ {\Bbb{R}}}

\def\hn{{h}}
\begin{document}

\title[]{Conditional Quantile Processes based on Series or Many Regressors}

\author[]{Alexandre Belloni \and Victor Chernozhukov \and Denis Chetverikov  \and  \\ Iv\'an Fern\'andez-Val}

\date{\tiny{\sc first version} May, 2006,  {\sc this version} of  \today.
The main results of this paper, particularly the pivotal method for
inference based on the entire quantile regression process, were
first presented at the NAWM Econometric Society, New Orleans, January, 2008
and also at the Stats in the Chateau in September 2009. We
are grateful to   Arun Chandraksekhar, Ye Luo, Denis Tkachenko, and  Sami Stouli for
careful readings of several versions of the paper. We thank the journal editor, two anonymous referees, Gary Chamberlain, Andrew Chesher, Holger Dette,
Roger Koenker, Tatiana Komarova, Arthur Lewbel, Oliver Linton,  Whitney Newey, Zhongjun Qu, and  seminar participants at the Econometric Society meeting, CEME Econometrics of Demand conference, CEMMAP master-class, CIREQ High-Dimensional Problems in Econometrics conference, ERCIM conference, ISI World Statistics Congress, Oberwolfach Frontiers in Quantile Regression workshop, Stats in the Chateau, Austin, BU, CEMFI, Columbia, Duke, Harvard/MIT, NUS, Rutgers, Sciences Po, SMU, Upenn, Virginia, and Yale for many
useful suggestions. We are grateful to Adonis Yatchew for giving us permission to use the data set in the empirical application. We gratefully acknowledge research support from
the NSF. The \texttt{R} package \texttt{quantreg.nonpar}  implements some of the methods of this paper \cite{RJ06}.}

\maketitle

\begin{abstract}
\footnotesize{Quantile regression (QR) is a principal regression
method for analyzing  the impact of covariates on outcomes. The
impact is described by the conditional quantile function and its
functionals. In this paper we develop the nonparametric QR-series
framework, covering many regressors as a special case, for
performing inference on the entire conditional quantile function and
its linear functionals.  In this framework, we approximate the
entire conditional quantile function by a linear combination of
series terms with quantile-specific coefficients and estimate the
function-valued coefficients from the data.  We develop large sample
theory for the QR-series coefficient process, namely we obtain
uniform strong approximations to the QR-series coefficient
process by conditionally pivotal and Gaussian processes. Based on these two strong approximations, or couplings, we develop four resampling methods (pivotal, gradient bootstrap, Gaussian, and weighted bootstrap) that can be used for inference on the entire QR-series coefficient function.

We apply these results to obtain estimation and  inference methods
for linear functionals of the conditional quantile function, such as
the conditional quantile function itself, its partial derivatives,
average partial derivatives, and conditional average partial
derivatives. Specifically, we obtain uniform rates of convergence and show how to use the four resampling methods mentioned above for inference on the functionals. All of the above results are for function-valued
parameters, holding uniformly in both the quantile index and the
covariate value, and covering the pointwise case as a by-product. We
demonstrate the practical utility of these results  with an
empirical example, where we estimate the price elasticity function
and test the Slutsky condition of the individual demand for gasoline, as indexed by the individual
unobserved propensity for gasoline consumption.  }



\end{abstract}

\section{Introduction}

Quantile regression (QR) is a principal method for
analyzing  the impact of covariates on outcomes, particularly when
the impact may be heterogeneous. This impact is characterized by
the quantile function of the conditional distribution of the outcome given covariates and its functionals
(Arias, Hallock and Sosa-Escudero \cite{AHS-E2001}, Buchinsky \cite{Buchinsky1994} and Koenker \cite{K2005}). For example, we can model
the log of the individual demand for some good, $Y$, as a function
of the price of the good, the income of the individual, and other
 observed individual characteristics, $X$, and an unobserved
preference  for consuming the good, $U$, as
$$
Y = Q(U,X),
$$
where the function $Q$ is strictly increasing in the  unobservable
$U$. With the normalization that $U\sim \Uniform(0,1)$ and the
assumption that $U$ and $X$ are independent, the function $Q(u,x)$
is the $u$-th quantile of the conditional distribution of $Y$ given $X = x$, i.e. $Q(u,x) =
Q_{Y | X}(u | x)$. This function can be used for policy analysis.
For example, we can determine how changes in taxes for the good
could impact demand heterogeneously across individuals.

In this paper we develop the nonparametric QR-series framework for
performing inference  on the entire conditional quantile function $Q(u,x)$
and its linear functionals.  In this framework, we approximate $Q(u,x)$ by a linear
combination of series terms, $Z(x)'\beta(u)$. The vector $Z(x)$
includes transformations of $x$ that have good approximation
properties such as powers, trigonometrics, local polynomials,
splines, and/or wavelets. The function $u \mapsto \beta(u)$ contains
quantile-specific coefficients that can be estimated from the data
using the QR estimator of Koenker and Bassett \cite{KB78}.  As the
number of series terms grows, the approximation error $Q(u,x)
- Z(x)'\beta(u)$ decreases, approaching zero in the limit. By
controlling the growth of the number of terms, we can obtain
consistent estimators and perform inference on the entire
conditional quantile function and its linear functionals. The QR-series framework also covers as a special case the so called many
regressors model, which is motivated by many new types of data that
emerge in the new information age, such as scanner and online
shopping data.

We describe now the main results in more detail. Let
$u \mapsto \hat\beta(u)$ denote the QR estimator of $u \mapsto \beta(u)$. The
first set of results provides large-sample theory for the normalized \textit{QR-series coefficient process} of increasing dimension $u \mapsto \sqrt{n}(\widehat
\beta(u) - \beta(u))$ that can be used to perform inference on the function $u \mapsto \beta(u)$. We note that inference on the function $u \mapsto \beta(u)$, in particular simultaneous inference on the parameters $\beta(u)$ that holds uniformly over all $u\in\mathcal U$, where $\mathcal U$ is a set of quantile indices of interest, is difficult because the standard asymptotic theory (van der Vaart and Wellner \cite{vdV-W}) based on limit distributions does not help here as the process $u \mapsto \sqrt n(\hat\beta(u) - \beta(u))$ in general does not have a limit distribution, even after an appropriate normalization. Instead, we develop high-quality inferential procedures based on the idea of coupling. The coupling is a construction of two processes on the same probability space that are  uniformly close to each other with high probability. Typically, one of the processes is the process of interest and the other one is a process whose distribution is known up-to a relatively small number of parameters that can be consistently estimated from the data. Thus, being able to construct an appropriate coupling means that we are able to approximate the distribution of the process of interest  by simulating the distribution of the coupling process from the data.

In this paper, we develop two couplings: pivotal and Gaussian, that is, for each sample size $n$, we construct a pivotal process and a Gaussian process on the same probability space as the data that are uniformly close to the process $u \mapsto \sqrt n(\hat\beta(u) - \beta(u))$ with high probability. In other words, these pivotal and Gaussian processes strongly approximate the process $u \mapsto \sqrt n(\hat\beta(u) - \beta(u))$. In addition, we develop four resampling methods (pivotal, gradient bootstrap, Gaussian, and weighted bootstrap) that allow us to approximately simulate the distribution of the pivotal (first two methods) and the Gaussian (last two methods) processes. These results provide an inference theory for the function $u \mapsto \beta(u)$ and also help to develop a unified feasible inference theory for linear functionals of the conditional quantile functions $x \mapsto Q(u,x)$, $u\in\mathcal U$, using any of the four proposed resampling methods. To the best of our knowledge, all of the above results are new.

The existence of the pivotal coupling emerges from the special
nature of QR, where a (sub) gradient of the sample objective
function evaluated at the true values of $\beta(u)$ is pivotal conditional on the
regressors, up-to an approximation error. This coupling allows us to perform high-quality inference based on pivotal and gradient bootstrap methods without
even resorting to Gaussian approximations.  We also show that the
gradient bootstrap method, originally introduced by Parzen, Wei and Ying
\cite{ParzenWeiYing1994} in the parametric context, is effectively a
means of carrying out the conditionally pivotal approximation
without explicitly estimating Jacobian matrices, which may be difficult in the quantile regression context. The conditions for
validity of the pivotal and gradient bootstrap methods require only a mild restriction on the
growth of the number of series terms in relation to the sample size.
To obtain the Gaussian coupling, we use chaining arguments and Yurinskii's
construction. This coupling implies that one can use the Gaussian method to perform inference on the function $u \mapsto \beta(u)$, and we also use this coupling to show that the weighted bootstrap method works to
approximate the distribution of the whole QR-series coefficient process for the same reason as the Gaussian
method. The conditions for validity of the Gaussian and
weighted bootstrap methods, however, may be stronger than those for the pivotal and gradient bootstrap methods.

As a corollary of our results on the process $u \mapsto \sqrt n(\hat\beta(u) - \beta(u))$, we also demonstrate that the QR-series estimator $x\mapsto Z(x)'\hat\beta(u)$ of the function $x\mapsto Q(u,x)$ based on either polynomials or splines has the fastest possible rate of convergence in the $L^2$ norm and that the QR-series estimator based on splines has the fastest possible rate of convergence in the sup norm in H\"{o}lder smoothness classes. These findings on optimality of series estimators in the quantile regression setting complement results in the literature on optimality of series estimators in the mean regression setting; see Newey \cite{W97}, Huang \cite{H98}, Cattaneo and Farrell \cite{CF13}, Belloni et al \cite{BelloniChenChernozhukov2009}, and Chen  and Christensen \cite{CC13}. In particular, our result on optimality in the sup norm is a major extension of Huang's work \cite{H03} on mean regression, as it requires us to establish some fine properties of the QR-series approximation; see the next section for details.

The second set of results provides estimation and inference methods
for linear functionals of the conditional quantile functions,
including
\begin{itemize}
\item[(i)] the conditional quantile function itself,  $(u,x) \mapsto Q(u, x)$,
\item[(ii)] the partial derivative function,  $(u,x) \mapsto \partial_{x_k} Q(u, x)$,
\item[(iii)] the average partial derivative function,  $u \mapsto \int \partial_{x_k} Q(u, x) d \mu(x)$, and
\item[(iv)] the conditional average partial derivative, $(u,x_{k}) \mapsto \int \partial_{x_k} Q(u, x) d \mu(x|x_k)$,
\end{itemize}where $\mu$ is a given measure and $x_k$ is the $k$-th component of $x$.
Specifically, we derive the pointwise rate of convergence and asymptotic normality of the QR-series estimators of the linear functionals. In addition, using our results on the QR-series coefficient process $u \mapsto \sqrt n(\widehat\beta(u) - \beta(u))$, we derive uniform rate of convergence and large sample inference procedures based on the pivotal, gradient bootstrap, Gaussian, and weighted bootstrap methods for the QR-series estimators of the linear functionals. These results provide solutions to a wide range of inference problems. For illustration purposes, we demonstrate how to use these results to construct uniform confidence bands for function-valued linear functionals and how to test shape constraints on the conditional quantile function $x\mapsto Q(u,x)$.
It is noteworthy that all of the above results apply to
function-valued parameters, holding uniformly in both the quantile
index $u$ and the covariate value $x$. We also emphasize that although we do not treat non-linear/non-smooth functionals in this paper, our results on couplings and resampling methods for the process $u \mapsto \sqrt n(\hat\beta(u) - \beta(u))$ are useful for treatment of such functionals.

The paper contributes and builds on the existing important
literature on  conditional quantile estimation.  First and foremost,
we build on the work of He and Shao \cite{HS00} that studied the
many regressors model and gave pointwise limit theorems for the QR
estimator in the case where only one quantile index $u$ is of interest.  We go beyond the
many regressors model to the series model and develop large sample
estimation and inference results for the entire QR process. We also
develop analogous estimation and inference results for the
conditional quantile function and its linear functionals, such as
derivatives, average derivatives, conditional average derivatives,
and others. None of these results were available in the previous
work. We also build on Lee \cite{Lee2003} that studied QR estimation
of partially linear models in the series framework for a single quantile index $u$, and on Horowitz and Lee \cite{HL2005} that studied
nonparametric QR estimation of additive quantile models for a single
quantile index $u$ in a series framework.  Our framework covers these
partially linear models and additive models as important special
cases, and allows us to perform inference on a considerably richer
set of functionals, uniformly across covariate values and a
continuum of quantile indices.   After the first version of this paper appeared in \cite{BCCF11}, Chao, Volgushev and Cheng  \cite{Chao-17} developed related results for QR-series coefficient processes based on weak approximations.\footnote{We refer the reader to  \cite{Chao-17} for a more detailed comparison with our results.}  In a more applied side, Koenker and Schorfheide \cite{KS1994}
 used a smoothing splines QR method to estimate the conditional quantile functions of global temperature over the last century.
 Other important work includes Stute \cite{S86},
Chaudhuri \cite{Chaudhuri-1991}, Chaudhuri, Doksum and Samarov
\cite{Chaudhuri-Doksum-Samarov-1997},  H\"{a}rdle, Ritov, and Song
\cite{hardle-ritov-song2009}, Horderlein and Mammen \cite{HoderleinMammen2009}, Cattaneo, Crump, and Jansson
\cite{cattaneo-crump-jansson2010}, Kong, Linton, and Xia
\cite{kong-linton-xia-2010}, Qu and Yoon \cite{QuYoon2011}, and Guerre and Sabbah \cite{GS12}, among others, but these papers focused on local, non-series, methods.

 Our work also relies on the series literature,
at least in a motivational and conceptual sense. In particular, we
rely on the work of Stone \cite{Stone1982}, Andrews
\cite{Andews1991}, Newey \cite{W97}, Chen and Shen
\cite{Chen-Shen1998}, Chen \cite{Chen2006} and others that
rigorously motivated the series framework as an approximation scheme
and gave pointwise normality results for least squares estimators,
and on Chen \cite{Chen2006} and van de Geer \cite{Geer2002} that
gave (non-uniform) consistency and rate results for general series
estimators, including quantile regression for the case of a single
quantile index $u$.  White \cite{White1992} established non-uniform
consistency of nonparametric estimators of the conditional quantile
function based on a nonlinear series approximation using artificial
neural networks. In contrast to the previous results, our rate
results are uniform in covariate values and quantile indices, and
cover both the quantile function and its functionals. Moreover, we
not only provide estimation rate results, but also derive a full set
of results on feasible inference based on the couplings for the process $u \mapsto \sqrt n(\hat\beta(u) - \beta(u))$.

While relying on previous work for motivation, our results require
to  develop both new proof techniques and new approaches to
inference.  In particular, our proof techniques rely on new maximal
inequalities for function classes with growing moments and uniform
entropy. In addition, as explained above, and in contrast to previous papers, our inference results heavily rely on the idea of the couplings. Yurinskii's coupling was previously used by Chernozhukov, Lee and Rosen \cite{CLR2009} to obtain a Gaussian coupling for the
least squares series estimator, but the use of this technique in our
context is new and much more involved. Thus, we approximate an
entire QR-series coefficient process of an increasing dimension, instead of a
vector of increasing dimension, by a Gaussian process. Finally, it
is noteworthy that our uniform inference results on functionals,
where uniformity is over covariate values, have not had analogs
even in the least squares series literature until recently (the extension of our results
to least squares has been recently published in Belloni et al \cite{BelloniChenChernozhukov2009}).

The results developed in this paper for series (global) estimation are of interest even though some of the results have analogs in the literature on kernel (local) estimation, because series estimators have several attractive features that are not shared by kernel estimators. First, series methods represent the estimate of the whole conditional quantile function $x\mapsto Q(u,x)$ and its linear functionals via a relatively small set of parameters, that is, the estimates $\widehat \beta(u)$ of the QR-series coefficients $\beta(u)$. As long as these parameters are reported, any researcher will be able to calculate the value of the estimate, for example, of $Q(u,x)$ for any $x\in\mathcal X$, which is very convenient. Second, series methods allow to easily impose shape constraints. For example, if it is known that the function $x\mapsto Q(u,x)$ is concave, one can simply impose this constraint on the QR-series optimization problem to obtain concave estimates (imposing shape constraints is especially convenient when B-splines are used; see DeVore and Lorentz \cite{DL93}). Third, as we demonstrate in the empirical example in Section \ref{sub: empirical example}, when the vector $X$ contains many controls in addition to one main covariate of interest, series methods are particularly convenient to impose a partial linear form on the functions $x\mapsto Q(u,x)$, which helps curb the curse of dimensionality.

This paper does not deal with sparse models, where there are some
key series terms and many ``non-key'' series terms which ideally
should be omitted from estimation. In these settings, the goal is to
find and indeed remove most of the ``non-key'' series terms before
proceeding with estimation. Belloni and Chernozhukov \cite{BC-SparseQR} obtained rate results
for quantile regression estimators in this case, but did not provide
inference results. Even though our paper does not explicitly deal
with inference in sparse models after model selection, the methods
and bounds provided herein are useful for analyzing this problem.
Investigating this matter rigorously is a challenging issue, since
it needs to take into account the model selection mistakes in
estimation, and is beyond the scope of the present paper; however,
it is a subject of our ongoing research, see Belloni, Chernozhukov and Kato \cite{belloni2013valid}.

{\bf Plan of the paper.}  The rest of the paper is organized as follows.
In Section \ref{Sec:Series}, we describe the nonparametric QR-series
model and estimators. In Section \ref{Sec:theory_coeff}, we derive
asymptotic theory for the QR-series processes. In Section
\ref{Sec:Functionals}, we give estimation and inference theory for
linear functionals of the conditional quantile function. In Section \ref{Sec:examples}, we present an empirical
application to the demand of gasoline and a computational experiment
calibrated to the application. The computational algorithms to
implement our inference methods are collected in the Appendix. The Supplemental Material \cite{BCCF16} contains some further results, and the proofs of the main results.


{\bf Notation.} In what follows, for all $x = (x_1,\dots,x_m)' \in\mathbb R^m$, we use $\|x\|$ to denote the Euclidean norm of $x$, that is, $\|x\| = (x_1^2 + \dots + x_m^2)^{1/2}$, and we use $\|x\|_{\infty}$ to denote the sup norm of $x$, that is, $\|x\|_{\infty} = \max_{1\leq j\leq m}|x_j|$. Also, we use
$S^{m-1}$ to denote the unit sphere  in $\RR^m$, that is, $S^{m-1} = \{x\in \mathbb R^m\colon \|x\| = 1\}$, and for $r>0$, we use $B_m(0,r)$ to denote the ball in $\mathbb R^m$ with center at $0$ and radius $r$, that is, $B_m(0,r) = \{x\in\mathbb R^m\colon \|x\|\leq r\}$. For all $m\times m$-dimensional matrices, we use $\|A\|$ to denote the operator norm of $A$ (also known as the spectral norm), that is, $\|A\| = \sup_{\alpha\in S^{m-1}}\|A\alpha\|$. For a set $I$, ${\rm diam}(I)=\sup_{v,\bar
v\in I}\|v-\bar v\|$ denotes the diameter of $I$, and
$\text{int}(I)$ denotes the interior of $I$. For any two real
numbers $a$ and $b$, $a\vee b = \max\{a,b\}$ and $a\wedge b =
\min\{a,b\}$. The relation $a_n \lesssim b_n $ means
that the inequality $a_n \leq C b_n$ holds for all $n$ with a
constant $C$ that is independent of $n$. We denote by $P^*$ the probability measure
induced by conditioning on realization of the data
$\mathcal{D}_n = (Y_i,X_i)_{i=1}^n$. We say that a
random variable $\Delta_n= o_{P^*}(1)$ in $P$-probability if for any $\epsilon>0$, we have
$P^*(|\Delta_n|>\epsilon) = o_P(1)$. We typically shall omit the
qualifier ``in $P$-probability." The operator $\Ep$ denotes the
expectation with respect to the probability measure $P$, $\En$
denotes the expectation with respect to the empirical measure, and
$\Gn$ denotes $\sqrt{n}(\En - \Ep)$. Finally, we use $\varepsilon'$ to denote a constant that depends only on the constant $\varepsilon$ but is such that its value can change at each appearance.

\section{Quantile Regression Series Framework}\label{Sec:Series}
\subsection{Model}
We consider the sequence of models indexed by the sample size $n$:
\begin{equation}\label{eq: model}
Y_{i,n} = Q_n(U_{i,n},X_{i,n}),\quad i=1,\dots,n,
\end{equation}
where $Y_{i,n}$ is a response variable, $X_{i,n}$ is a $d_n$-dimensional vector of covariates (elementary regressors) with the support $\XX_n\subset \RR^{d_n}$, $U_{i,n}$ is an unobservable individual ranking that is distributed uniformly on $(0,1)$ and is independent of $X_{i,n}$, and $Q_n\colon [0,1]\times\XX_n\to\RR$ is an unknown function such that for all $x\in\XX_n$, the function $u\mapsto Q_n(u,x)$ is strictly increasing. Since $U_{i,n} | X_{i,n} \sim \Uniform (0,1)$,  for each $u\in(0,1)$, $x \mapsto Q_n(u,x)$ is the $u$th quantile function of  $Y_{i,n}$ conditional on $X_{i,n}$, which we refer to as the conditional $u$-quantile function. For a compact set of quantile indices $\mathcal U\subset(0,1)$, we are interested in estimating the functions $x\mapsto Q_n(u,x)$ and their functionals. For given $n$, we assume that $(X_{i,n},U_{i,n},Y_{i,n})_{i=1}^n$ is a random sample from the distribution of the triple $(X^n,U^n,Y^n)$, where $n$ denotes the sample size.


The model \eqref{eq: model} covers two specifications of major interest:
\begin{itemize}
\item [1.] {\bf Nonparametric model:} In this model, the data generating process does not depend on $n$, that is, for all $n\geq 1$, we have $Y^n = Y$, $X^n = X$, $U^n = U$, and $Q_n = Q$ for a response variable $Y$, a $d$-dimensional vector of covariates $X$ with the support $\XX\subset\RR^d$, an unobservable individual ranking $U$ satisfying $U|X\sim \Uniform(0,1)$, and a function $Q\colon[0,1]\times\XX\to\RR^d$ with the property that for all $x\in\XX$, the function $u\mapsto Q(u,x)$ is strictly increasing. We assume that the functions $x\mapsto Q(u,x)$ are smooth but we do not impose any parametric structure on them. Throughout the paper, we refer to this specification as the NP model.
\item [2.] {\bf Many regressors model:} In this model, the dimension $d_n$ is allowed to grow with $n$ but the function $Q_n$ is linear in its second argument: $Q_n(u,x) = x'\beta_n(u)$ for some vector of coefficients $\beta_n(u)\in\RR^{d_n}$ and all $u\in\mathcal U$ and $x\in\XX_n$. Throughout the paper, we refer to this specification as the MR model.
\end{itemize}
Both models are of interest in econometrics. The NP model is important because it is very flexible as it does not impose any parametric structure on the functions $x\mapsto Q(u,x)$ and also does not require the function $(u,x)\mapsto Q(u,x)$ to be separately additive in $u$; see Matzkin \cite{M03} for extensive evidence on importance of this flexibility in economics. The MR model is also important, and different versions of this model have recently attracted much attention in the literature due to emergence of datasets with information on many variables; see Cattaneo, Jansson, and Newey \cite{CJN15} for some recent advances and also Mammen \cite{M93} for some classical results on the mean regression version of this model. As we demonstrate in this paper, both models can be treated in a unifying quantile regression series framework.

For brevity of notation, we shall omit the index $n$ whenever it does not lead to confusion, that is, we write $Y$, $X$, $U$, $Q$, $d$, and $\mathcal X$ instead of $Y^n$, $X^n$, $U^n$, $Q_n$, $d_n$, and $\XX_n$, respectively, even though we implicitly assume that all these quantities are allowed to depend on $n$. Also, we write $(X_i,U_i,Y_i)_{i=1}^n$ instead of $(X_{i,n},U_{i,n},Y_{i,n})_{i=1}^n$.

\subsection{QR-Series Approximation}
Next, we introduce the QR-series approximation to the function $x\mapsto Q(u,x)$. We start with preparing some notation. Fix $u\in\mathcal U$ and let $x\mapsto Z(x) = (Z_1(x),\dots,Z_m(x))'$ be a vector of series approximating functions of dimension $m=m_n$, where each function $x\mapsto Z_j(x)$ maps $\XX$ into $\RR$. Define the vector of coefficients $\beta(u) = (\beta_1(u),\dots,\beta_m(u))'$ as a solution to the \textit{QR-series approximation problem}:
\begin{equation}\label{define beta}
\min_{\beta \in \Bbb{R}^{m}}\Ep\Big[\rho_{u} (Y - Z(X)'\beta) - \rho_u(Y - Q(u,X))\Big],
\end{equation}
where $\rho_{u} (z) = (u - 1\{z<0\})z$ is the check function  (Koenker \cite{K2005}).\footnote{The optimization problem \eqref{define beta} has a finite solution if $\Ep[|Q(u,X)|]$ is finite. In addition, since the function $z\mapsto \rho_u(z)$ is strictly convex, the solution is unique if the matrix $E[Z(X) Z(X)']$ is non-singular, which is assumed in Condition S below. The  term $\rho_{u} (Y - Q(u,X))$ does not affect the optimization problem but guarantees the existence of the solution when $\Ep[|Y|]$ is not finite.} For the MR model, we assume that $Z(x) = x$ for all $x\in\XX$, so that $m=d$ and the vector $\beta(u)$ defined in \eqref{define beta} coincides with the vector $\beta_n(u)$ in the definition of the model, $Q_n(u,x) = x'\beta_n(u)$. For the NP model, we assume that the vector $Z$ consists of series functions with good approximation properties such as indicators, B-splines (or regression splines), polynomials, Fourier series, and/or compactly supported wavelets.\footnote{Interestingly, in the case of B-splines and compactly supported wavelets, the entire collection of series terms is dependent upon the sample size $n$.}$^,$\footnote{It is possible to combine these sets of approximating functions. For example, when we model gasoline consumption, we can simultaneously use Fourier series to capture seasonal effects and polynomials to capture long term growth.}  We refer the reader to Newey \cite{W97} and Chen \cite{Chen2006} for a careful and detailed description of these series functions; see also Belloni et al \cite{BelloniChenChernozhukov2009} for an overview of recent advances on series approximating functions.

We define the \textit{QR-series approximating function} $x\mapsto Z(x)'\beta(u)$ mapping $\XX$ into $\RR$, and, for all $x\in\XX$, the \textit{QR-series approximation error}
$$
R(u,x) := Q(u,x) - Z(x)'\beta(u).
$$
We will assume that the QR-series approximation error asymptotically vanishes, that is, $\sup_{x\in\XX, u \in \mathcal U}|R(u,x)|\to 0$ as $n\to\infty$. For the MR model, this assumption always holds because $R(u,x) = 0$ for all $x\in\XX$. For the NP model, we will demonstrate that this assumption holds under appropriate conditions as long as $m = m_n \to \infty$ as $n\to\infty$. In turn, given that the QR-series approximation error asymptotically vanishes, it follows that the QR-series approximating function $x\mapsto Z(x)'\beta(u)$ approximates  well the true conditional $u$-quantile function $x\mapsto Q(u,x)$. 

\subsection{QR-Series Estimator} The QR-series approximation motivates the {\em QR-series estimator} of the function $x\mapsto Q(u,x)$:
\begin{equation}\label{eq: quantile regression estimator}
\widehat Q(u,x) = Z(x)'\widehat\beta(u),\quad x\in\XX,
\end{equation}
where $\widehat\beta(u)$ is the Koenker and Bassett \cite{KB78} estimator of $\beta(u)$ that solves the empirical analog of the population problem (\ref{define beta}):
\begin{equation}\label{eq: empirical analog problem}
\min_{\beta\in\RR^m}\En[\rho_u(Y_i - Z_i'\beta)],
\end{equation}
where we denote $Z_i = Z(X_i)$ for all $i=1,\dots,n$. As $n$ gets large, both the estimation error $\widehat Q(u,x) - Z(x)' \beta(u)$ and the approximation error $R(u,x)$ asymptotically vanish.

Since we are interested in estimating the functions $x\mapsto Q(u,x)$ for a set of quantile indices $\mathcal U$, we solve the problem \eqref{eq: empirical analog problem} for all $u\in\mathcal U$ to obtain the \textit{QR-series coefficient process}
$$
\widehat \beta (\cdot) = \{ \hat \beta(u) \colon  u \in \mathcal{U}\}
$$
and the QR-series estimator \eqref{eq: quantile regression estimator} for all $u\in\mathcal U$.
We note that obtaining this estimator is computationally easy even if  $\mathcal U$ contains many quantile indices and the dimension $m$ of the vectors $Z_i$ is large. In particular, one can use the results of Portnoy and Koenker \cite{PortnoyKoenker97}, who developed interior points methods with preprocessing for the problem \eqref{eq: empirical analog problem} that are very efficient and give the solution for multiple quantile indices simultaneously.

\subsection{Main Regularity Conditions}
Let $\kappa\in(0,\infty]$ be some constant that is independent of $n$. Also, for $x\in\XX$, let $\mathcal Y_x$ denote the support of the conditional distribution of $Y$ given $X = x$. Moreover, let $\bar {\mathcal U}$ denote the convex hull of $\mathcal U$. Throughout the paper, we will use the following regularity condition:

\begin{samepage}
\noindent

\textbf{Condition S.}\text{ }\textit{\begin{itemize}
\item[S.1] The data form a triangular array of random variables so that for any given $n$, the data $\mathcal{D}_n = \{(X_i,Y_i) : 1 \leq i \leq n\}$ is an i.i.d. random sample from the distribution of the pair $(X,Y)$.
\item[S.2] (i) The conditional density $f_{Y|X}(y|x)$ is bounded from above uniformly over $y\in\mathcal Y_x$, $x \in \mathcal{X}$, and $n$; (ii) $f_{Y|X}(Q(u,x)|x)$ is bounded away from zero uniformly over $u\in\bar{\mathcal U}$, $x\in\XX$, and $n$; and (iii) the derivative of $y\mapsto f_{Y|X}(y|x)$ is continuous and bounded in absolute value from above uniformly over $y\in\mathcal Y_x$, $x \in \mathcal{X}$, and $n$.
\item[S.3]  The eigenvalues of the Gram matrix $\Sigma=\Ep[Z(X) Z(X)'] $ are bounded
from above and away from zero uniformly over $n$.
\item[S.4] The approximation error $R(u,x)$ satisfies $\sup_{x \in \mathcal{X}, u \in \mathcal{U}} |R(u,x)| \lesssim m^{-\kappa}$.
\end{itemize}
}
\end{samepage}
Condition S.1 requires that the data is i.i.d. but it can be extended to standard time series models at the expense of more technicalities. Condition S.2 imposes mild smoothness assumptions on the conditional density function $f_{Y|X}(y|x)$. Since it follows from simple algebra, see for example \eqref{eq: first derivative Q}, that
$$
\frac{1}{f_{Y|X}(y|x)} = \frac{\partial Q(Q^{-1}(y,x),x)}{\partial u},\quad y\in\mathcal Y_x, \ x\in\mathcal X,
$$
where $y\mapsto Q^{-1}(y,x)$ denotes the inverse of $u\mapsto Q(u,x)$, it is easy to provide a set of conditions in terms of the function $Q(u,x)$ that imply Condition S.2. Indeed, Condition S.2 follows if (i) $\partial Q(u,x)/\partial u$ is bounded away from zero uniformly over $u\in[0,1]$, $x\in\mathcal X$, and $n$; (ii) $\partial Q(u,x)/\partial u$ is bounded from above uniformly over $u\in \bar{\mathcal U}$, $x\in\XX$ and $n$; (iii) $\partial^2 Q(u,x)/\partial u^2$ is bounded in absolute value from above uniformly over $u\in [0,1]$, $x\in\XX$, and $n$.\footnote{Note that we assume that the conditional density $f_{Y|X}(y|x)$ is bounded away from zero only for $y = Q(u,x)$, where $u\in\bar{\mathcal U}$ and $x\in\mathcal X$. This allows us to avoid the stronger condition that assumes that $f_{Y|X}(y|x)$ is bounded away from zero for all $y\in\mathcal Y_x$ and $x\in\mathcal X$. The latter condition can simplify some arguments (see the proof of Lemma \ref{Lemma:AUX_L2sparse}) but it requires the conditional density of $Y$ given $X$ to have bounded support, thus excluding some important distributions such as the Gaussian.}

For the MR model, Condition S.3 implies that there is no perfect multicollinearity among covariates, and Condition S.4 is satisfied with $\kappa = \infty$ since $R(u,x) = 0$ for all $u\in\mathcal U$ and $x\in\XX$.

For the MR model, all conditions can be regarded as primitive. For the NP model, Conditions S.1 and S.2 are also primitive but Conditions S.3 and S.4 depend on the vector of approximating series functions $x \mapsto Z(x)$ used for the estimation. Therefore, below we provide some discussion of these conditions in the NP model. Suppose that $X$ is absolutely continuous with respect to the Lebesgue measure on $\mathcal X$ and let $f_X\colon \XX\to\RR$ denote its pdf. Then it is well-known that Condition S.3 holds if $f_X(x)$ is bounded from above and away from zero uniformly over $x\in\XX$, and the eigenvalues of the matrix
$$
\int_{x\in\XX}Z(x)Z(x)'d x
$$
are bounded from above and away from zero uniformly over $n$; see, for example, Proposition 2.1 in Belloni et al \cite{BelloniChenChernozhukov2009}. In turn, the latter condition holds if, for example, the vector $Z$ consists of functions that are orthonormal on $\XX$. In the case that the former condition is violated in the sense that the density $f_X(x)$ is not bounded away from zero uniformly over all $x\in \XX$, one can consider a subset $\widetilde \XX$ of $\XX$ such that $f_X(x)$ is bounded away from zero uniformly over $x\in\widetilde \XX$ and consider the estimation problem based on the subset of observations $i$ satisfying $X_i\in\widetilde \XX$. This will give the estimate of $Q(u,x)$ for all $x\in\widetilde \XX$ and $u\in\mathcal U$. As the sample size gets larger, one can increase the set $\widetilde \XX$ to extend the estimate of $Q(u,x)$ to a larger set of points. Developing a method how this truncation should be performed in practice, however, is beyond the scope of this paper.

To provide some primitive conditions for Condition S.4 in the NP model, we need to prepare some notation. For a $d$-tuple $\alpha = (\alpha_1,\dots,\alpha_d)$ of nonnegative integers, let $D^\alpha = \partial^{\alpha_1}_{x_1}\cdots \partial^{\alpha_d}_{x_d}$. Also, for $s>0$, let $[s]$ denote the largest integer strictly smaller than $s$. For the constant $C>0$, define the H\"{o}lder ball $\Omega(s,C,\XX)$ as the set of all functions $f\colon \XX\to\RR$ such that
\begin{equation}\label{eq: holder smoothness}
|D^\alpha f(x) - D^\alpha f(\widetilde x)|\leq C\Big(\textstyle{\sum_{j=1}^d} (x_j - \widetilde x_j)^2\Big)^{(s - [s])/2}\text{ and }|D^\beta f(x)|\leq C
\end{equation}
for all $x = (x_1,\dots,x_d)'$ and $\widetilde x = (\widetilde x_1,\dots,\widetilde x_d)'$ in $\mathcal X$ and all $d$-tuples $\alpha = (\alpha_1,\dots,\alpha_d)$ and $\beta = (\beta_1,\dots,\beta_d)$ of nonnegative integers satisfying $\alpha_1 +\dots+ \alpha_d = [s]$ and $\beta_1 + \dots + \beta_d \leq [s]$ (where the left-hand sides of the inequalities in \eqref{eq: holder smoothness} are set to be infinity if the derivatives do not exist). For example, any $s$-times continuously differentiable function belongs to the H\"{o}lder ball $\Omega(s,C,\XX)$ for some $C>0$ as long as $\mathcal X$ is compact. Also, we say that the vector of approximating functions $Z$ consists of tensor products of polynomials if $m = J^d$ for some integer $J>0$ and $Z$ consists of all functions of the form $x = (x_1,\dots,x_d)\mapsto \prod_{j=1}^d x_j^{\alpha_j}$ for some $d$-tuple $\alpha = (\alpha_1,\dots,\alpha_d)$ of nonnegative integers such that $\alpha_j\leq J-1$ for all $j=1,\dots,d$. Finally, we say that the vector of approximating functions $Z$ consists of tensor products of B-splines of order $s_0$ if $m = J^d$ for some integer $J>0$ and $Z$ consists of all functions of the form $x = (x_1,\dots,x_d)\mapsto \prod_{j=1}^d b_{\alpha_j}(x_j)$ for some $d$-tuple $\alpha = (\alpha_1,\dots,\alpha_d)$ of nonnegative integers such that $\alpha_j\leq J-1$ for all $j=1,\dots,d$ where $b_0,\dots,b_{J-1}$ is a sequence of $J$ B-splines of order $s_0$ on the interval $[0,1]$ with uniform knot sequence; see Chen \cite{Chen2006} for more explanations about H\"{o}lder balls and B-splines. The next lemma provides a set of primitive conditions for Condition S.4 in the NP model.

\begin{lemma}[Verification of Condition S.4 in the NP model for polynomials and B-splines]\label{lem: approximation error}
Consider the NP model. Suppose that Conditions S.2 and S.3 hold. In addition, suppose that $\XX = [0,1]^d$. Moreover, suppose that $Q(u,\cdot)\in\Omega(s,C,\XX)$ for all $u\in\mathcal U$ and some $s,C>0$. If the vector of approximating functions $Z$ consists of tensor products of polynomials and $s>d$, then
\begin{equation}\label{eq: approximation error 1}
(E[|R(u,X)|^2])^{1/2}\lesssim m^{-s/d}\text{ and }\sup_{x\in\XX}|R(u,x)|\lesssim m^{1-s/d},
\end{equation}
uniformly over $u\in\mathcal U$. Also, if the vector of approximating functions $Z$ consists of tensor products of B-splines of order $s_0$, $s\wedge s_0 > d$, and $X$ has the pdf $f_X(x)$ bounded from above and away from zero uniformly over $x\in\XX$, then
\begin{equation}\label{eq: approximation error 2}
(E[|R(u,X)|^2])^{1/2}\lesssim m^{-(s\wedge s_0)/d}\text{ and }\sup_{x\in\XX}|R(u,x)|\lesssim m^{-(s\wedge s_0)/d},
\end{equation}
uniformly over $u\in\mathcal U$. Thus, under the presented conditions, Condition S.4 is satisfied with $\kappa= s/d - 1$ in the case of polynomials and with $\kappa = (s\wedge s_0)/d$ in the case of B-splines.
\end{lemma}

\begin{remark}[Importance of Lemma \ref{lem: approximation error}]
Lemma \ref{lem: approximation error} makes precise the nature of the QR-series approximation in the NP model and plays a crucial role in our derivation of the convergence rate of the QR-series estimator for the NP model in the next section. Indeed, it is well-known from the approximation theory that under the assumptions of the lemma, in the case of polynomials, for example, there exists $\beta^m(u)$ such that $\sup_{x\in\XX}|Q(u,x) - Z(x)'\beta^m(u)|\lesssim m^{-s/d}$; see Chen \cite{Chen2006}. However, this result  does not help in our analysis because the QR-series estimator $\widehat\beta_n(u)$ converges in probability to $\beta(u)$, which may or may not be equal to $\beta^m(u)$. We therefore  need to derive  a bound on $\sup_{x\in\XX}|Q(u,x) - Z(x)'\beta(u)|$.

The part of the lemma concerning the B-splines case is particularly important because it allows us to prove in the next section that the QR-series estimator based on B-splines achieves the fastest possible rate of convergence in the sup norm. This part of the lemma is a major extension of a result in Huang \cite{H03}, who obtained similar inequalities with the QR-series approximation error replaced by the least-squares-series approximation error.
\qed
\end{remark}

\begin{remark}[Other series approximating functions]
Other popular choices of the series approximating functions include Fourier series and compactly supported wavelets. Although we do not provide formal results for these choices, we note that under conditions similar to those in Lemma \ref{lem: approximation error}, one can show that Condition S.4 holds with $\kappa = 1/2 - s/d$ in the case of Fourier series and with $\kappa = -(s\wedge s_0)/d$ in the case of compactly supported wavelets, where $s_0$ is the order of the wavelets.
\qed
\end{remark}

\subsection{Additional Notation}\label{lab: additional notation}
The properties of the QR-series coefficient process and of the QR-series estimator depend on the choice of the approximating functions and the dimension of $Z$. Like in the analysis of series estimators of conditional mean functions (see Newey \cite{W97}), the following quantity will play a crucial role in our analysis:
$$
\zeta_m = \sup_{x\in\mathcal X}\|Z(x)\|.
$$
Assuming that $\mathcal X = [0,1]^d$, it is well known that $\zeta_m \lesssim m$ if the vector $Z$ consists of tensor products of polynomials and $\zeta_m \lesssim m^{1/2}$ if the vector $Z$ consists of tensor products of B-splines.

As in the analysis of the parametric quantile regression, the following Jacobian matrix will also play a crucial role in the analysis:
\begin{equation}\label{eq: J matrix}
J(u) = \Ep\Big[f_{Y|X}(Q(u,X) | X) Z(X) Z(X)'\Big], \quad u\in\mathcal U.
\end{equation}
Implementing some of our inference methods will require an estimator of $J(u)$. For the purposes of this paper, we will use Powell's \cite{Powell1984} estimator defined by
\begin{equation}\label{inf-est-hatJ}
\hat J(u) = \frac{1}{2h} \En\Big[
1\{ |Y_i - Z_i'\hat\beta(u)| \leq h \} \cdot Z_i Z_i'\Big],
\end{equation}
where $h$ is some bandwidth value satisfying $h = h_n \to 0$. We will also use the estimator of $\Sigma$ defined by
\begin{equation}\label{inf-est-hatSigma0}
\hat \Sigma = \En[ Z_i Z_i'].
\end{equation}
The properties of $\hat J(u)$ and $\hat \Sigma$ in our high-dimensional setting are established in Lemma \ref{covariance} in Appendix \ref{App:EmpiricalProcess} of the Supplemental Material.


\section{Asymptotic Theory for QR-Series Coefficient Processes} \label{Sec:theory_coeff}
In this section, we study properties of the normalized QR-series coefficient process
$$
\sqrt n(\widehat \beta(\cdot) - \beta(\cdot)) = \Big\{\sqrt n(\hat\beta(u) - \beta(u))\colon u\in\mathcal U\Big\}.
$$
Specifically, we derive the rate of convergence and construct two couplings for this process. The couplings give two processes that are  uniformly close to $\sqrt n(\widehat\beta(\cdot) - \beta(\cdot))$ with high probability, but are such that their distribution can be simulated. In particular, we develop four resampling methods (pivotal, gradient bootstrap, Gaussian, and weighted bootstrap) to simulate the distribution of these processes. In the next section, these results allow us to develop a unified feasible inference theory for all the functionals of interest. We also derive rates of convergence for the QR-series estimator process $\hat Q(\cdot,\cdot) = \{\hat Q(u,x)\colon u\in\mathcal U, x\in\mathcal X\}$ in the NP model. In particular, we show that the QR-series estimator based on either polynomials or B-splines has the fastest possible rate of convergence in the $L^2$ norm and that the QR-series estimator based on B-splines has the fastest possible rate of convergence in the sup norm.

\subsection{Uniform-in-$u$ Rate of Convergence}
As explained in the previous section, given an i.i.d. sample $(X_i,Y_i)_{i=1}^n$ from the distribution of the pair $(X,Y)$, we estimate the coefficient function $\beta(\cdot) = \{\beta(u)\colon u\in\mathcal U\}$ using the QR-series coefficient process $\hat\beta(\cdot) = \{\hat\beta(u)\colon u\in\mathcal U\}$, namely, for each $u \in \mathcal{U}$, we define  $\hat \beta(u)$ as the Koenker and Bassett \cite{KB78} estimator that solves the
empirical analog \eqref{eq: empirical analog problem} of the population problem (\ref{define beta}).
Our first main result  is a uniform-in-$u$ rate of convergence for the
QR-series coefficient process.

%

\begin{theorem}[Uniform-in-$u$ rate of convergence for QR-Series coefficient process]\label{Thm:SeriesRates} Suppose that Condition S holds. In addition, suppose that $m \zeta_m^2\log^2 n =o(n)$ and $m^{-\kappa}\log n = o(1)$. Then
$$ \sup_{u \in
\mathcal{U}} \| \hat \beta(u) - \beta(u) \|
\lesssim_P  \sqrt{m/n}.$$
\end{theorem}
Theorem \ref{Thm:SeriesRates} establishes a rate of convergence of the estimator $\hat\beta(u)$ in our high-dimensional setting that holds uniformly over $u\in\mathcal U$. The theorem complements the rate of convergence results in the literature. Indeed, Koenker and Portnoy \cite{KoenkerPortnoy1987} established rate of convergence results that hold uniformly over $\mathcal{U}$ in the fixed-dimensional setting, and He and Shao \cite{HS00} established rates in the high-dimensional setting for the case when $\mathcal U$ is a singleton (pointwise-in-$u$ rate of convergence). Importantly, the uniform-in-$u$ rate of convergence in Theorem \ref{Thm:SeriesRates} is the same as the pointwise-in-$u$ rate of convergence. The proof of this theorem relies on new concentration inequalities that control the behavior of the eigenvalues of the design matrix $\widehat\Sigma$. Note also that our condition $m\zeta_m^2\log^2 n = o(n)$ is similar to the analogous condition in \cite{HS00}.

Theorem \ref{Thm:SeriesRates} has an implication for the uniform-in-$u$ rate of convergence in the $L^2$ norm of the QR-series estimator in the NP model. Indeed, define
$$
\|h\|_{L^2(X)} = \Big(E[|h(X)|^2]\Big)^{1/2},\quad \text{for }h\colon \XX\to \mathbb R.
$$
We then have the following corollary of Theorem \ref{Thm:SeriesRates}, which is the second main result together with Corollary \ref{cor: sup rate qr series estimator} below on the uniform-in-$u$ rate of convergence in the sup norm of the QR-series estimator in the NP model.
\begin{corollary}[Uniform-in-$u$ $L^2$ rate of convergence for QR-series estimator in the NP model]\label{cor: l2 rate qr series estimator}
Consider the NP model. Suppose that (i) Condition S.1-3 holds. In addition, suppose that (ii) $\XX = [0,1]^d$ and that (iii) $Q(u,\cdot)\in\Omega(s,C,\XX)$ for all $u\in\mathcal U$ and some $s,C>0$. If the vector of approximating functions $Z$ consists of tensor products of polynomials, $m^3\log^2 n = o(n)$, and $m^{1 - s/d}\log n = o(1)$, then
\begin{equation}\label{eq: L2 rate polynomials}
\sup_{u\in\mathcal U}\|\widehat Q(u,\cdot) - Q(u,\cdot)\|_{L^2(X)} \lesssim_P \sqrt{m/n} + m^{-s/d}.
\end{equation}
Also, if the vector of approximating functions $Z$ consists of tensor products of B-splines of order $s_0$, $s\wedge s_0 > d$, $m^2 \log^2 n = o(n)$, $m^{-(s\wedge s_0)/d} \log n = o(1)$, and $X$ has the pdf $f_X(x)$ bounded from above and away from zero uniformly over $x\in\XX$, then
\begin{equation}\label{eq: L2 rate bsplines}
\sup_{u\in\mathcal U}\|\widehat Q(u,\cdot) - Q(u,\cdot)\|_{L^2(X)} \lesssim_P \sqrt{m/n} + m^{-(s\wedge s_0)/d}.
\end{equation}
\end{corollary}

\begin{remark}[QR-series estimator achieves the fastest possible rate of convergence in the $L^2$ norm]
Consider the NP model and suppose that conditions (i)--(iii) of Corollary \ref{cor: l2 rate qr series estimator} hold. If $Z$ consists of a tensor product of polynomials and $s>d$, setting $m = C n^{d/(d+2 s)}$ for some constant $C>0$ satisfies conditions that $m^3\log^2 n = o(n)$ and $m^{1 - s/d}\log n = o(1)$, and so substituting this $m$ into the bound \eqref{eq: L2 rate polynomials} gives
\begin{equation}\label{eq: optimal L2 rate of convergence}
\sup_{u\in\mathcal U}\|\hat Q(u,\cdot) - Q(u,\cdot)\|_{L^2(X)} \lesssim_P n^{-s/(d + 2s)}.
\end{equation}
Similarly, if $Z$ consists of a tensor product of B-splines of order $s_0$, $s_0\geq s>d$, and $X$ has the pdf $f_X(x)$ bounded from above and away from zero uniformly over $x\in\mathcal X$, setting $m = C n^{d/(d+2 s)}$ for some constant $C>0$ satisfies conditions that $m^2\log^2 n = o(n)$ and $m^{-(s\wedge s_0)/d}\log n = o(1)$, and so substituting this $m$ into the bound \eqref{eq: L2 rate bsplines} again gives \eqref{eq: optimal L2 rate of convergence}. Note that the rate in \eqref{eq: optimal L2 rate of convergence} is the optimal $L^2$ rate of convergence for the estimators of nonparametric conditional quantile functions; see Chaudhuri \cite{C91}. Thus, the QR-series estimator based on either polynomials or B-splines has the fastest possible $L^2$ rate of convergence, and as we demonstrate, this rate is actually achieved uniformly in $u\in\mathcal U$. The same results can also be shown for the QR-series estimator based on Fourier series and compactly supported wavelets. This is one of the attractive properties of the QR-series estimator.\qed
\end{remark}


\subsection{Uniform Strong Approximations (Couplings) and Resampling Methods}
Here we derive two couplings yielding strong approximations to the process $\sqrt n(\hat\beta(\cdot) - \beta(\cdot))$ in the form of either a pivotal or a Gaussian process, and develop four resampling methods to approximate the distribution of these processes and thus approximate also the distribution of the original process $\sqrt n(\hat\beta(\cdot) - \beta(\cdot))$. We provide algorithms to implement the resampling methods in Appendix \ref{app:algorithms}.

\subsubsection{Pivotal Coupling:} Let
\begin{equation}\label{Def:U}
\mathbb{U}(u) = \frac{1}{\sqrt{n}} \sum_{i=1}^n Z_i (u-1\{U_i \leq u\}),\quad u\in\mathcal U.
\end{equation}
Note that the process $\mathbb U(\cdot) = \{\mathbb U(u)\colon u\in\mathcal U\}$ is (conditionally) pivotal since conditional on $(Z_i)_{i=1}^n$, the sequence $(U_i)_{i=1}^n$ consists of i.i.d. Uniform$(0,1)$ random variables.  The following theorem, which is the third main result together with the Gaussian coupling in Theorem \ref{theorem: strong} below, shows that the process $\sqrt n(\widehat\beta(\cdot) - \beta(\cdot))$ is strongly approximated by the (conditionally) pivotal process $J^{-1}(\cdot)\mathbb U(\cdot) = \{J^{-1}(u)\mathbb U(u)\colon u\in\mathcal U\}$.

\begin{theorem}[Pivotal Coupling]\label{Thm:MainULA}
Suppose that Condition S holds. In addition, suppose that $m^3\zeta_m^2 =o(n^{1-\varepsilon})$  and $m^{-\kappa+1} = o(n^{-\varepsilon})$ for some constant $\varepsilon>0$. Then
$$
\sqrt{n}\left(\hat \beta(u) -
\beta(u)\right) =  J^{-1}(u)\mathbb{U}(u) +
r(u),\quad u\in\mathcal U,
$$
where
$$
 \sup_{u \in \mathcal{U}}\| r(u) \| \lesssim_P  \frac{m^{3/4} \zeta_m^{1/2} \log^{1/2} n}{n^{1/4}} + \sqrt{m^{1-\kappa}\log n} = o(n^{-\varepsilon'})
 $$
 for some $\varepsilon'>0$.
\end{theorem}

This theorem is important because it has many useful implications. One of the implications is the following result for the uniform-in-$u$ rate of convergence in the sup norm of the QR-series estimator in the NP model.

\begin{corollary}[Uniform-in-$u$ sup-rate of convergence for QR-series estimator in the NP model]\label{cor: sup rate qr series estimator}
Consider the NP model. Suppose that (i) Condition S.1-3 holds. In addition, suppose that (ii) $\XX = [0,1]^d$ and that (iii) $Q(u,\cdot)\in\Omega(s,C,\XX)$ for all $u\in\mathcal U$ and some $s,C>0$. If the vector of approximating functions $Z$ consists of tensor products of polynomials and for some $\varepsilon > 0$, $m^5 = o(n^{1-\varepsilon})$ and $m^{2 - s/d} = o(n^{-\varepsilon})$, then
$$
\sup_{u\in\mathcal U} \sup_{x\in\mathcal X}|\widehat Q(u,x) - Q(u,x)| \lesssim_P \sqrt{m^2\log n/n} + m^{1-s/d}.
$$
Also, if the vector of approximating functions $Z$ consists of tensor products of B-splines of order $s_0$, $X$ has the pdf $f_X(x)$ bounded from above and away from zero uniformly over $x\in\XX$, and for some $\varepsilon > 0$, $m^4 = o(n^{1 - \varepsilon})$ and $m^{1 - (s\wedge s_0)/d} = o(n^{-\varepsilon})$, then
\begin{equation}\label{eq: sup rate bsplines}
\sup_{u\in\mathcal U} \sup_{x\in\mathcal X}|\widehat Q(u,x) - Q(u,x)| \lesssim_P \sqrt{m\log n/n} + m^{-(s\wedge s_0)/d}.
\end{equation}
\end{corollary}
\begin{remark}[B-splines version of QR-series estimator achieves the fastest possible rate of convergence in the sup norm]
Consider the NP model and suppose that conditions (i)--(iii) of Corollary \ref{cor: sup rate qr series estimator} hold. In addition, suppose that $Z$ consists of a tensor product of B-splines of order $s_0$, $s_0\geq s > 3d/2$, and $X$ has the pdf $f_X(x)$ bounded from above and away from zero uniformly over $x\in\mathcal X$. Then setting $m = C(n/\log n)^{d/(d + 2s)}$ for some constant $C>0$ satisfies conditions that $m^4 = o(n^{1-\varepsilon})$ and $m^{1 - (s\wedge s_0)/d} = o(n^{-\varepsilon})$ for some $\varepsilon>0$, and so substituting this $m$ into the bound \eqref{eq: sup rate bsplines} gives
$$
\sup_{u\in\mathcal U}\sup_{x\in\mathcal X}|\hat Q(u,x) - Q(u,x)|\lesssim_P \left(\frac{\log n}{n}\right)^{s/(d + 2s)},
$$
which is the optimal rate of convergence in the sup norm for an estimator of the nonparametric conditional quantile function; see Chaudhuri \cite{C91}. Thus, the QR-series estimator based on B-splines has the fastest possible rate of convergence in the sup norm, and as we demonstrate, this rate is actually achieved uniformly in $u\in\mathcal U$.\footnote{The same results can also be shown for the QR-series estimator based on compactly supported wavelets.} This is another attractive property of the QR-series estimator. \qed
\end{remark}
We also note that although the uniform convergence rate based on polynomials is not optimal, the rate derived in Corollary \ref{cor: sup rate qr series estimator} is faster than the (trivial) uniform rate implied by the $L_2$ rate and the relation between the $L_2$-norm and sup-norm.

\subsubsection{Resampling Methods Based on Pivotal Coupling:} Another implication of Theorem \ref{Thm:MainULA} is that it suggests the following high-quality method to approximate the distribution of the process $\sqrt n(\hat\beta(\cdot) - \beta(\cdot))$, which we refer to as the pivotal method. First, simulate an i.i.d. sequence $(U_i^*)_{i=1}^n$ of Uniform$(0,1)$ random variables that are independent of the data and define
\begin{equation}\label{Def:U*}
\mathbb{U}^*(u) =
\frac{1}{\sqrt{n}} \sum_{i=1}^n Z_i (u-1\{U^*_i \leq u\}),\quad u\in\mathcal U,
\end{equation}
so that conditional on $(Z_i)_{i=1}^n$, the process $\mathbb U^*(\cdot) = \{\mathbb U^*(u)\colon u\in\mathcal U\}$ is a copy of the process $\mathbb U(\cdot)$, and $J^{-1}(\cdot)\mathbb U^*(\cdot) = \{J^{-1}(u)\mathbb U^*(u)\colon u\in\mathcal U\}$ is a copy of $J^{-1}(\cdot)\mathbb U(\cdot)$. Second, calculate the estimators $\hat J(u)$ of the matrices $J(u)$ for all $u\in\mathcal U$ as in \eqref{inf-est-hatJ} of Section \ref{lab: additional notation} (recall that $h$ in the estimators $\widehat J(u)$ is some bandwidth value satisfying $h = h_n \to 0$). Then, as shown in the next theorem, one can use the conditional distribution of the process $\hat J^{-1}(\cdot)\mathbb U^*(\cdot) = \{\widehat J^{-1}(u)\mathbb U^*(u)\colon u\in\mathcal U\}$ given the data, which can be simulated, to approximate the distribution of the process $J^{-1}(\cdot)\mathbb U^*(\cdot)$, and, via Theorem \ref{Thm:MainULA}, also of the process $\sqrt n(\hat\beta(\cdot) - \beta(\cdot))$.

\begin{theorem}[Pivotal Method]\label{Thm:MainULAfeasible}
Suppose that Condition S holds. In addition, suppose that $h\sqrt m = o(n^{-\varepsilon})$, $m^2 \zeta_m^2 = o(n^{1-\varepsilon}\hn)$, and $m^{-\kappa+1/2} = o(n^{-\varepsilon})$ for some constant $\varepsilon > 0$. Then
$$
\widehat J^{-1}(u)\mathbb{U}^*(u) = J^{-1}(u)\mathbb{U}^*(u) +
r(u),\quad u\in\mathcal U,
$$
where
$$
 \sup_{u \in \mathcal{U}}\| r(u) \| \lesssim_P \sqrt{\frac{\zeta_m^2 m^2\log n}{n\hn}} + m^{-\kappa+1/2}   + \hn\sqrt{m}  = o(n^{-\varepsilon'})
$$
for some $\varepsilon'>0$. The stated bound continues to hold in $P$-probability if we replace
the unconditional probability $P$ by the conditional probability
$P^*$.
\end{theorem}
This theorem is the fourth main result together with Theorems \ref{Thm:MainULAstar}, \ref{thm: gaussian method}, and \ref{Thm:MainBootstrap} below on gradient bootstrap, Gaussian, and weighted bootstrap methods.
The pivotal method is closely related to another approach to inference,
which we refer to as the gradient bootstrap method.  This
approach was previously introduced by Parzen, Wei and Ying
\cite{ParzenWeiYing1994} for parametric models with fixed dimension.
We extend it to the considerably more general series
framework studied in this paper. The main idea  is to generate for all $u\in\mathcal U$ the gradient bootstrap estimator $\widehat \beta^*(u)$ as the solution to the perturbed QR problem
\be\label{define betastar}
\min_{\beta \in \Bbb{R}^m}\Big(\En
[\rho_{u} (Y_i - Z_i'\beta)] - \mathbb{U}^*(u)'\beta/\sqrt{n}\Big),
\ee
where  $\mathbb{U}^*(u)$ is defined in (\ref{Def:U*}). Then, as shown in the next theorem, one can use the conditional distribution of the process $\sqrt n(\hat \beta^*(\cdot) - \hat\beta(\cdot)) = \{\sqrt n(\hat \beta^*(u) - \hat\beta(u))\colon u\in\mathcal U\}$ given the data, which can be simulated, to approximate the distribution of the process $J^{-1}(\cdot)\mathbb U^*(\cdot)$, and, via Theorem \ref{Thm:MainULA}, also of the process $\sqrt n(\hat\beta(\cdot) - \beta(\cdot))$.

\begin{theorem}[Gradient Bootstrap Method]\label{Thm:MainULAstar}
Suppose that Condition S holds. In addition, suppose that $m^3\zeta_m^2=o(n^{1-\varepsilon})$ and
$m^{-\kappa+1/2} = o(n^{-\varepsilon})$ for some constant $\varepsilon>0$. Then
$$
\sqrt{n} \left(\hat \beta^*(u) -
\hat \beta(u)\right) = J^{-1}(u)\mathbb{U}^*(u) + r(u),
$$
where
$$
 \sup_{u \in \mathcal{U}} \| r(u) \| \lesssim_P  \frac{m^{3/4} \zeta_m^{1/2} \log^{1/2} n}{n^{1/4}} +  m^{-\kappa+1/2} = o(n^{-\varepsilon'})
 $$
  for some $\varepsilon'>0$.
The stated bound continues to hold in $P$-probability if we replace
the unconditional probability $P$ by the conditional probability
$P^*$.
\end{theorem}
\begin{remark}[Comparison of pivotal and gradient bootstrap methods]
Both the pivotal and gradient bootstrap methods have their own advantages. Perhaps the main advantage of the gradient bootstrap method relative to the pivotal method is that it does not require estimating the matrices $J(u)$, $u\in \mathcal{U}$, which is important because estimating these matrices requires a potentially subjective choice of the bandwidth $h$. In fact, implementing the gradient bootstrap method does not require any choice of smoothing parameters, making it particularly convenient for empirical researchers. On the other hand, an advantage of the pivotal method relative to the gradient bootstrap method is that it is computationally simple as it does not require solving the quantile optimization problem for each simulation of the process $\mathbb U^*(\cdot)$.  \qed
\end{remark}

\subsubsection{Gaussian Coupling:} Next, we turn to a strong approximation based on a sequence of
Gaussian processes. The following theorem shows that for each $n$, one can construct a Gaussian process $G(\cdot) = G_n(\cdot) = \{G_n(u)\colon u\in\mathcal U\}$ such that the process $J^{-1}(\cdot)G(\cdot) = \{J^{-1}(u)G(u)\colon u\in\mathcal U\}$ is with high probability uniformly close to the process $\sqrt n(\hat\beta(\cdot) - \beta(\cdot))$.

\begin{theorem}[Gaussian Coupling]\label{theorem: strong}
Suppose that Condition S holds. In addition, suppose that $m^{7} \zeta_m^6 = o(n^{1-\varepsilon})$ and $m^{-\kappa + 1} = o(n^{-\varepsilon})$ for some constant $\varepsilon>0$. Then
$$
\sqrt n\Big(\widehat\beta(u) - \beta(u)\Big) = J^{-1}(u) G(u) + r(u),\quad u\in\mathcal U,
$$
where $G(\cdot) = G_n(\cdot)$ is a process on $\mathcal U$ that,  conditionally on $(Z_i)_{i=1}^n$, is zero-mean Gaussian with a.s. continuous sample paths and the covariance function
\begin{equation}\label{eq: covariance matrix gaussian process thm 5}
E\Big[G(u_1)G(u_2)'\mid (Z_i)_{i=1}^n\Big] = \En[Z_i Z_i'](u_1\wedge u_2 - u_1 u_2), \ \text{ for all $u_1$ and $u_2$ in $\mathcal U$},
\end{equation}
and
$$
\sup_{u\in\mathcal U}\|r(u)\|  = o_P(n^{-\varepsilon'})
$$
for some $\varepsilon'> 0$.
\end{theorem}
\begin{remark}[Conditions of Theorem \ref{theorem: strong}]
Note that the strong approximation to the process $\sqrt n(\hat\beta(\cdot) - \beta(\cdot))$ by the Gaussian process $J^{-1}(\cdot)G(\cdot)$ constructed in Theorem \ref{theorem: strong} requires the condition that $m^7 \zeta_m^6 = o(n^{1 - \varepsilon})$, which is more restrictive than the corresponding condition in Theorem \ref{Thm:MainULA}, $m^3\zeta_m^2 = o(n^{1 - \varepsilon})$, required for the strong approximation to the process $\sqrt n(\hat\beta(\cdot) - \beta(\cdot))$ by the pivotal process $J^{-1}(\cdot)\mathbb U(\cdot)$. We note that this restrictive condition is sufficient but we do not know whether it is necessary. This condition is a consequence of a step in the proof of Theorem \ref{theorem: strong} that relies upon Yurinskii's coupling. Therefore, improving that step
through the use of another coupling could potentially lead to significant
improvements in the conditions of the theorem; see, in particular, Theorem \ref{thm: gaussian coupling for t process} in the next section. See also \cite{K94} and \cite{CNS15}, where a Hungarian coupling is derived that may give a result similar to that in Theorem \ref{theorem: strong} but under somewhat weaker conditions if $d$ is small and the vector of approximating functions $Z$ consists of a tensor products of B-splines or wavelets. \qed
\end{remark}

\subsubsection{Resampling Methods Based on Gaussian Coupling:} Although Theorem \ref{theorem: strong} requires strong conditions, it is important because it suggests that one can approximate the distribution of the process $\sqrt n(\hat\beta(\cdot) - \beta(\cdot))$ using Gaussian and weighted bootstrap methods, which are wide-spread in the literature in other contexts and which we now describe.

Let us start with the Gaussian method. Let $\hat\Sigma^{1/2}$ denote the square root of the matrix $\hat\Sigma$. Note that the covariance function of the process $G(\cdot)$ conditional on $(Z_i)_{i=1}^n$, given in  \eqref{eq: covariance matrix gaussian process thm 5}, is equal to that of the process $\hat\Sigma^{1/2} B_m(\cdot)$, where $B_m(\cdot) = \{B_m(u)\colon u\in\mathcal U\}$ is a standard $m$-dimensional Brownian bridge, that is, a vector consisting of $m$ independent scalar Brownian bridges. Since the sample path of the Brownian bridge is continuous a.s., it follows that the process $\hat\Sigma^{1/2} B_m(\cdot)$ is a copy of the process $G(\cdot)$, conditional on $(Z_i)_{i=1}^n$. Hence, one can simulate a standard $m$-dimensional Brownian bridge $B_m^*(\cdot) = \{B_m^*(u)\colon u\in\mathcal U\}$ that is independent of the data and define
\begin{equation}\label{eq: gaussian process to simulate}
G^*(u) = G_n^*(u) = \hat \Sigma^{1/2} B_m^*(u),\quad u\in\mathcal U,
\end{equation}
so that conditional on $(Z_i)_{i=1}^n$, the process $G^*(\cdot) = \{G^*(u)\colon u\in\mathcal U\}$ is a copy of the process $G(\cdot)$, and $J^{-1}(\cdot)G^*(\cdot) = \{J^{-1}(u)G^*(u)\colon u\in\mathcal U\}$ is a copy of $J^{-1}(\cdot)G(\cdot)$. Let  $\hat J(u)$ be the estimators of the matrices $J(u)$ for all $u\in\mathcal U$  in  \eqref{inf-est-hatJ}. Then, as shown in the next theorem, one can use the conditional distribution of the process $\hat J^{-1}(\cdot)G^*(\cdot) = \{\hat J^{-1}(u)G^*(u)\colon u\in\mathcal U\}$ given the data, which can be simulated, to approximate the distribution of the process $J^{-1}(\cdot)G^*(\cdot)$, and via Theorem \ref{theorem: strong} also of the process $\sqrt n(\hat \beta(\cdot) - \beta(\cdot))$.

\begin{theorem}[Gaussian Method]\label{thm: gaussian method}
Suppose that Condition S holds. In addition, suppose that $h\sqrt m = o(n^{-\varepsilon})$, $m^2\zeta_m^2 = o(n^{1 - \varepsilon} h)$, and $m^{-\kappa + 1/2} = o(n^{-\varepsilon})$ for some constant $\varepsilon > 0$. Then
$$
\widehat J^{-1}(u) G^*(u) = J^{-1}(u) G^*(u) + r(u),\quad u\in\mathcal U,
$$
where
$$
\sup_{u\in\mathcal U}\|r(u)\| \lesssim_P  \sqrt{\frac{m^2 \zeta_m^2\log n}{n\hn}} + m^{-\kappa+1/2}   + \hn\sqrt{m}  = o(n^{-\varepsilon'})
$$
 for some $\varepsilon' > 0$. The stated bound continues to hold in $P$-probability if we replace
the unconditional probability $P$ by the conditional probability
$P^*$.
\end{theorem}

Another related inference method is the weighted bootstrap method. Pr{\ae}stgaard and Wellner~\cite{Praestgaard-Wellner-93}, Hahn~\cite{H97}, Chamberlain and Imbens~\cite{Chamberlain-Imbens-03}, and Chen and Pouzo~\cite{CP09} previously used this method in the point-wise case, where the set $\mathcal U$ is a singleton. We extend this method to obtain the distributional approximation for the process $\sqrt n(\hat\beta(\cdot) - \beta(\cdot)) = \{\sqrt n(\hat\beta(u) - \beta(u))\colon u\in\mathcal U\}$ where $\mathcal U$ is not a singleton and in fact can be a continuum of quantile indices. To describe the method, consider a set of weights
$\pi_1,...,\pi_n$ that are i.i.d. draws from the distribution of a non-negative random variable $\pi$ with $\Ep[\pi] =1$ and $\Ep[\pi^2] =2$, such as the standard exponential distribution, and that are independent of the data. For all $u\in\mathcal U$, define the weighted
bootstrap estimator $\hat\beta^b(u)$ as the solution to the weighted QR problem
\begin{equation}\label{eq: weighted bootstrap problem}
 \hat \beta^b(u) \in \arg \min_{\beta \in \RR^m} \En[ \pi_i \rho_u(Y_i-Z_i'\beta)].
\end{equation}
Then, as shown in the next theorem, one can use the conditional distribution of the process $\sqrt n(\hat\beta^b(\cdot) - \hat\beta(\cdot)) = \{\sqrt n(\hat\beta^b(u) - \hat\beta(u))\colon u\in\mathcal U\}$ given the data, which can be simulated, to approximate the distribution of the process $J^{-1}(\cdot)G^*(\cdot)$, and via Theorem \ref{theorem: strong} also of the process $\sqrt n(\hat\beta(\cdot) - \beta(\cdot))$.

\begin{theorem}[Weighted Bootstrap Method]\label{Thm:MainBootstrap}
Suppose that Condition S holds. In addition, suppose that $m^{7} \zeta_m^6 = o(n^{1-\varepsilon})$ and $m^{-\kappa + 1} = o(n^{-\varepsilon})$ for some constant $\varepsilon>0$. Moreover, suppose that the random variable $\pi$ is non-negative and satisfies $E[\pi] = 1$, $E[\pi^2] = 2$, $E[\pi^4]\lesssim 1$. Finally, suppose that $\max_{1\leq i\leq n} \pi_i \lesssim \log n$. Then
\begin{equation}\label{eq: main implication thm 6}
\sqrt n\Big(\hat\beta^b(u) - \hat\beta(u)\Big) = J^{-1}(u)\Z^*(u) + r(u),
\end{equation}
where $G^*(\cdot) = G^*_n(\cdot)$ is a process on $\mathcal U$ that, conditionally on $(Z_i)_{i=1}^n$, is zero-mean Gaussian with a.s. continuous sample paths and the covariance function \eqref{eq: covariance matrix gaussian process thm 5}, and
$$
\sup_{u\in\mathcal U}\|r(u)\| \lesssim_P o(n^{-\varepsilon'})
$$
for some $\varepsilon' >0$. Moreover, the stated bound continues to hold in $P$-probability if we replace the
unconditional probability $P$ by the conditional probability $P^*$.
\end{theorem}

\begin{remark}[Comparison of Gaussian and weighted bootstrap methods]
The comparison of the Gaussian and weighted bootstrap methods is similar to that of the pivotal and gradient bootstrap methods. Again both methods have their own advantages. The main advantage of the weighted bootstrap method is arguably that it does not require estimating the matrices $J(u)$, $u\in \mathcal{U}$, which allows us to bypass the need to select a bandwidth $h$. An advantage of the Gaussian method is that it is computationally simple as it does not require solving the quantile optimization problem for each simulation of weights $(\pi_i)_{i=1}^n$.\qed
\end{remark}

\begin{remark}[Comparison of resampling methods based on the pivotal and Gaussian couplings]
Although it is difficult to compare the resampling methods based on the pivotal coupling (pivotal and gradient bootstrap methods) with those based on the Gaussian coupling (Gaussian and weighted bootstrap methods) from a theoretical point of view, our results suggest that the former methods might be more accurate than the latter ones. Indeed, the methods based on the pivotal coupling require weaker conditions (see, however, Theorem \ref{thm: gaussian coupling for t process} in the next section, where it is possible to substantially weaken conditions required for the Gaussian coupling in some examples) and, in addition, developing the Gaussian coupling requires a ``double approximation'': in order to construct a Gaussian process that strongly approximates the original process $\sqrt n(\hat\beta(\cdot) - \beta(\cdot))$, we first construct a coupling of the latter process with the pivotal process, and then we construct a coupling of the Gaussian process with the pivotal process, so that the Gaussian process is coupled with the original process $\sqrt n(\hat\beta(\cdot) - \beta(\cdot))$ via the pivotal process.  In the numerical examples of Section \ref{subsec:simulations} and the companion computational paper \cite{RJ06}, however, we find that the performance of the four methods is similar in finite samples.

In addition, the Gaussian coupling is important because of the existence of well-developed extreme value theory for Gaussian processes; see, for example, Leadbetter, Lindgren and Rootzen \cite{LLR83}. In combination with the Gaussian coupling, this theory can be used to develop inferential procedures for some linear functionals without relying upon resampling methods (that is, with non-bootstrap critical values) like in  Rio \cite{R94}. Moreover, the Gaussian coupling is important because of existence of anti-concentration inequalities for Gaussian processes  (Lemma \ref{lemma:Anti}), which are useful to construct uniform confidence bands for linear functionals in the next section.
\qed
\end{remark}

\section{Linear Functionals of the Conditional Quantile Function}\label{Sec:Functionals}

In addition to  the quantile functions $x \mapsto Q(u,x)$, $u\in\mathcal U$, we are also interested in various linear functionals of these functions. If $x$ is decomposed as $(w,v)$ and $x_k$ denotes the $k$-th component of $x$, examples of particularly useful linear functionals include
\begin{itemize}
\item[1.]  the derivative:  \quad $\theta(u,x) = \partial_{x_k} Q(u,x)$;
\item[2.]  the average derivative:  \quad $\theta(u) = \int \partial_{x_k} Q(u,x) d\mu(x)$;
\item[3.]  the conditional average derivative: \ $\theta(u,w) = \int \partial_{x_k} Q(u,w,v) d\mu(v|w)$.
\end{itemize}
The measures $\mu$ entering the definitions above are assumed to be known but our results can also be extended to include estimated measures. To cover all examples, we denote the linear functional of interest by $\theta(u,w)$, where $w\in\mathcal W\subset \mathbb R^{d_w}$.

Let $I\subset \mathcal U \times \mathcal W$ denote the set of values of $(u,w)$ of interest. For example, if we are interested in
\begin{itemize}
\item  the function $\theta(u,w)$ at a particular point $(u,w)$, then $I = \{(u,w)\}$,
\item the function $u\mapsto \theta(u,w)$ having fixed $w$, then $I =\mathcal{U}\times\{w\}$,
\item the function $w\mapsto \theta(u,w)$ having fixed $u$, then $I = \{u\}\times\mathcal{W}$,
\item the entire function $(u,w)\mapsto \theta(u,w)$, then $I = \mathcal{U}\times\mathcal{W}$.
\end{itemize}

\subsection{QR-Series Approximation} By the linearity of the series approximations, the function $\theta(u,w)$  can be seen as a linear functional of the quantile regression coefficients $\beta(u)$ up to an approximation error, that is,
\begin{equation}\label{eq: cool}
\theta(u,w) = \ell(w)' \beta(u) +  r(u,w), \quad  (u,w) \in I,
\end{equation}
where $\ell(w)' \beta(u)$ is the \textit{QR-series approximation}, with
$\ell(w)$ denoting the $m$-dimensional vector of loadings on the coefficients,
and $r(u,w)$ is the remainder term, which corresponds to the
\textit{QR-approximation error}.  Indeed, this
decomposition arises from the application of different linear
operators $\mathcal{A}$ to the decomposition $Q(u,\cdot) =
Z(\cdot)'\beta(u) + R(u,\cdot)$ and evaluating the resulting
functions at $w$:
\begin{equation}\label{eq: uncool}
\(\mathcal{A} Q(u,\cdot)\)[w] = \(\mathcal{A} Z(\cdot)\)[w]'\beta(u) +  \(\mathcal{A} R(u,\cdot)\)[w].
\end{equation}
In the three examples above the operator $\mathcal{A}$ is given by,
respectively,
\begin{itemize}
\item[1.]  a differential operator:  $(\mathcal{A}g) [x]= (\partial_{x_k}g) [x] $, so that
$$\ell(x) =\partial_{x_k} Z(x), \ \ \  r(u,x) = \partial_{x_k} R(u,x);$$
\item[2.]  an integro-differential operator:  $\mathcal{A} g=\int \partial_{x_k} g(x) d\mu(x)$, so that
$$\ell  = \int \partial_{x_k} Z(x) d \mu(x), \ \ \  r(u) = \int \partial_{x_k} R(u,x) d \mu(x);  $$
\item[3.]  a partial integro-differential operator: $(\mathcal{A}g) [w]=\int \partial_{x_k} g(w,v) d\mu(v|w)$, so that
$$\ell(w) =\int \partial_{x_k} Z(w,v)d\mu(v|w), \ \ \ r(u, w) =  \int \partial_{x_k} R(u,w,v) d\mu(v|w).$$
\end{itemize}
For notational convenience, we use the formulation (\ref{eq: cool}) in the analysis, instead of the motivational formulation (\ref{eq: uncool}).

\subsection{QR-Series Estimator} Given $\widehat Q(u,x) = Z(x)'\widehat\beta(u)$, we use the plug-in estimator
$$
\widehat \theta(u,w) = \ell(w)' \widehat \beta(u),\quad (u,w)\in I,
$$
to estimate $\theta(u,w)$. In cases where $\theta(u,w)$ is known to be monotone with respect to either  $w$ or $u$, we show in Appendix \ref{subset: mono} of the Supplemental Material how to impose this restriction after estimation to improve finite sample properties of  $\widehat \theta(u,w)$. 

In the rest of this section, we provide rates of convergence for $\widehat\theta(u,w)$ as well as the inference tools
that will be valid for inference on the QR-series approximation
$$
\ell(w)' \beta(u), \quad (u,w)  \in I,
$$
and, provided that the QR-approximation error $r(u,w)$ is small enough relative to the estimation noise,
will also be valid for inference on the  functional of interest:
$$
\theta(u,w), \ \ (u,w)  \in I.
$$
Thus,  the QR-series approximation $\ell(w)'\beta(u)$ is an important penultimate
target, whereas the functional $\theta(u,w)$ is the ultimate target. 


\subsection{Pointwise Asymptotic Theory}

We start with the rate of convergence of the estimator $\hat\theta(u,w)$ at a particular quantile index value $u$ and a particular covariate value $w$ (pointwise rate of convergence). In principle, the point $(u,w)$ can depend on $n$, but we suppress the dependence for simplicity of notation. We use the following assumption:

\noindent
\textbf{Condition P.} \textit{The QR-series decomposition  $\theta(u,w) = \ell(w)' \beta(u) +  r(u,w)$ satisfies
$$
\frac{\sqrt{n}|r(u,w)|}{\|\ell(w)\|} = o(1).
$$}
Condition P can be understood as an undersmoothing condition. Although undersmoothing conditions are widely spread in the literature, as Belloni et al \cite{BelloniChenChernozhukov2009} pointed out, there is no theoretically justified procedure in the literature that would lead to a desired level of undersmoothing for the estimators of the linear functionals even for least squares estimators. For example, under conditions of Lemma \ref{lem: approximation error}, when $w = x$, $\theta(u,w) = Q(u,x)$, so that $\ell(w) = Z(x)$ and $r(u,w) = R(u,x)$, and the vector $Z$ consists of a tensor product of B-splines of order $s_0$, Condition P holds as long as $n / m^{1+2(s\wedge s_0)/d} = o(1)$.   

Based on Condition P, we derive the following theorem for the pointwise rate of convergence of $\hat\theta(u,w)$, which is the fifth main result together with Theorem \ref{Thm:DistributionsInferentialpointwise} below on pointwise asymptotic normality of $\hat\theta(u,w)$.
\begin{theorem}[Pointwise Convergence Rate for Linear Functionals]
\label{theorem: pointwise rate}  Suppose that the conditions of
Theorem \ref{Thm:MainULA} hold. In addition, suppose that Condition P holds. Then
$$
| \widehat \theta(u,w) -  \theta(u,w)| \lesssim_P  \frac{\|\ell(w)\|}{\sqrt{n}}.
$$
\end{theorem}

\begin{remark}[Rates and norm of vector of loadings] The rate of convergence of $\widehat \theta(u,w)$ depends on the functional $\theta(u,w)$ through  the norm of the vector of loadings $\ell(w)$. For example, if we are interested in the coefficient $\beta_1(u)$, so that $\theta(u,w) = \beta_1(u)$, which might be a parameter of interest in the Many regressors (MR) model, then $\ell(w) = (1, 0, \ldots, 0)'$, and so $\|\ell(w) \| = 1$, yielding a $\sqrt{n}$-consistent estimator $\widehat\theta(u,w)$. See Comment \ref{comment: xi} below for additional examples of linear functionals with bounds on $\|\ell(w)\|$. \qed
\end{remark}

In order to perform inference, we consider the t-statistic
$$
t(u,w) = \frac{  \hat \theta(u,w) - \theta(u,w) }{ \widehat \sigma(u,w)},
$$
where
\begin{equation}\label{Def:hatsigma2n}
\widehat \sigma^2(u,w) =  u(1-u) \ell(w)' \widehat J^{-1}(u) \widehat \Sigma  \widehat J^{-1}(u)\ell(w)/n
\end{equation}
 is a consistent estimator of
\begin{equation}\label{Def:sigma2n}
 \sigma^2(u,w) = u(1-u) \ell(w)'J^{-1}(u)  \Sigma   J^{-1}(u)\ell(w)/n,
\end{equation}
the asymptotic variance of $\widehat \theta(u,w)$, obtained by the delta method. We can carry out standard inference based on this t-statistic because
$t(u,w) \to_d N(0,1)$, as we establish below.

\begin{theorem}[Pointwise Inference for Linear Functionals]\label{Thm:DistributionsInferentialpointwise}
Suppose that the conditions of Theorem \ref{Thm:MainULA} hold. In addition, suppose that Condition P holds, $h = o(1)$ and $m\zeta_m^2\log^2 n = o(n h)$. Then
$$
t(u,w) \to_d N(0,1).
$$
\end{theorem}
\begin{remark}[Using resampling methods for pointwise inference]
Although it is possible to establish validity of all the resampling methods from the previous section to perform pointwise inference on linear functionals, we do not show these results here because they will follow as a special case from our results below on uniform inference for linear functionals.  We provide an implementation algorithm to perform pointwise inference using the resampling methods in Appendix \ref{app:algorithms}.  \qed
\end{remark}

\subsection{Uniform Asymptotic Theory} Next, we derive the rate of convergence of the estimator $\hat\theta(u,w)$ that holds uniformly over $(u,w)\in I$. We use the following assumption:

\noindent
\textbf{Condition U.}\textit{
\begin{itemize}
\item[U.1] The set $I$ is such that its dimension $d_I$ is fixed and its diameter is bounded uniformly over $n$.
\item[U.2] For some $\varepsilon > 0$, the QR-approximation error $r(u,w)$ satisfies
$$
\sqrt{n}  \sup_{(u,w)\in I}\frac{ | r(u,w) |}{\|\ell(w)\|}  = o(n^{-\varepsilon}).
$$
\item[U.3]  The vector of loadings $\ell(w)$ satisfies
 $$
\|\ell(w)\| \leq \zeta_{m,\theta}\text{ and } \ \left\|\frac{\ell(w)}{\|\ell(w)\|} - \frac{\ell(w')}{\|\ell(w')\|}\right\| \leq \zeta_{m,\theta}^L\|w - w'\|
$$
for all $w,w'\in\mathcal W$, where $\log \zeta_{m,\theta}^L\lesssim  \log n$.
\end{itemize}
}

Condition U.1 on the dimension and the diameter of the set $I$ is mild and can be further relaxed at the expense of additional technicalities. As in the pointwise case, Condition U.2 can be understood as an undersmoothing condition. Condition U.3 requires that the vector of loadings $\ell(w)$ is bounded uniformly over $w\in\mathcal W$ in the Euclidean norm by $\zeta_{m,\theta}$ and the function $w\mapsto \ell(w)/\|\ell(w)\|$ is Lipschitz-continuous in the Euclidean norm with the Lipschitz constant $\zeta_{m,\theta}^L$. We note that the last condition is rather weak because the only requirement on the Lipschitz constant that we impose is that $\log\zeta_{m,\theta}^L \lesssim \log n$. We discuss some bounds on the constant $\zeta_{m,\theta}$ in a separate comment below.

\begin{remark}[Primitive bounds on $\zeta_{m,\theta}$]\label{comment: xi}
The uniform rate of convergence for the estimator $\hat\theta(u,w)$ derived below in Theorem \ref{theorem: uniform rate} will crucially depend on the constant $\zeta_{m,\theta}$ appearing in Condition U. Here we discuss some bounds on this constant. For brevity, we only discuss the case of B-splines and refer to Newey \cite{W97} and Chen \cite{Chen2006} for other choices of approximating functions. We assume that $\mathcal X = [0,1]^d$ and that the vector of approximating functions $Z$ consists of tensor products of B-splines of order $s_0$. As discussed above,  then $\zeta_m = \sup_{x\in\mathcal X} \|Z(x)\|\lesssim \sqrt m$ and it is also possible to verify that for all positive integers $\alpha \leq s_0$,  $\sup_{x\in\mathcal X}\|\partial^{\alpha}_{x_k} Z(x)\| \lesssim m^{1/2 + \alpha / d}$; see for example Chen and Christensen \cite{CC13}. Then \\

\begin{tabular}{ll}
$\bullet$  For &$\theta(u,w) = Q(u,x)$,  $\ell(w)=Z(x)$ and  $\zeta_{m,\theta} \lesssim m^{1/2}$;\\
$\bullet$  For &$\theta(u,w) = \partial_{x_k} Q(u,x)$,  $\ell(w)= \partial^{\alpha}_{x_k} Z(x)$ and $\zeta_{m,\theta} \lesssim m^{1/2 + \alpha/d}$;\\
$\bullet$  For  &$\theta(u) = \int \partial_{x_k} Q(u,x) d\mu(x)$ with ${\rm supp}(\mu) \subset {\rm int}(\mathcal{X})$ and $|\partial_{x_k} \mu(x)| \lesssim 1$,\\
  &$\ell = \int \partial_{x_k} Z(x)  \mu(x) \ dx = -\int Z(x) \partial_{x_k} \mu(x) \ dx$ and $\zeta_{m,\theta} \lesssim 1$;\\
\end{tabular}

\noindent
see Newey \cite{W97} for more explanations on the last bound. \qed
\end{remark}

\subsubsection{Uniform-in-$u$ Rate of Convergence:} The following theorem establishes the uniform  rate of convergence of the QR-series estimator $\widehat \theta(u,w)$, which is the sixth main result.

\begin{theorem}[Uniform Convergence Rate for Linear Functionals]
\label{theorem: uniform rate}  Suppose that the conditions of Theorem
\ref{Thm:MainULA} hold. In addition, suppose that Condition U hold. Then
$$
\sup_{(u,w)\in I} | \widehat \theta(u,w) -  \theta(u,w)| \lesssim_P  \sqrt{\frac{\zeta^2_{m,\theta}\log n}{n}}.
$$
\end{theorem}

The uniform rate of Theorem \ref{theorem: uniform rate}  is the same as the pointwise rate of Theorem \ref{theorem: pointwise rate} up to a small logarithmic factor. As in the pointwise rate result, the norm of the vector of loadings play a role which is controlled by $\zeta_{m,\theta}$ in Condition U. 

\begin{remark}[Comparison of Theorem \ref{theorem: uniform rate}  and Corollary \ref{cor: sup rate qr series estimator}] When $\theta(u,w)$ is the conditional quantile function, the convergence rate of Theorem \ref{theorem: uniform rate} is asymptotically equivalent to the rate of Corollary \ref{cor: sup rate qr series estimator} under the undersmoothing condition U.2. For example, in the case of B-splines, $\zeta^2_{m,\theta}\log n/ n = m \log n/ n$ and $m^{-(s\wedge s_0)/d} = o(\sqrt{m\log n/n})$ under U.2.
\end{remark}

\subsubsection{Gaussian and Pivotal Couplings for $t$-Statistic Processes:} Next, we consider inference on the function $(u,w)\mapsto \theta(u,w)$.  We base inference on  the t-statistic process $t(\cdot,\cdot) = \{t(u,w)\colon (u,w)\in I\}$ defined as follows:
\begin{equation}\label{tnprocess}
t(u,w) = \frac{  \hat \theta(u,w) - \theta(u,w) }{ \widehat \sigma(u,w)},
\end{equation}
where $\widehat \sigma^2(u,w)$, defined in  (\ref{Def:hatsigma2n}), is an estimator
of the asymptotic variance  $\sigma^2(u,w)$ of $\hat\theta(u,w)$ in (\ref{Def:sigma2n}). Using the results in the previous section, we construct pivotal and Gaussian couplings for this process in the following theorem, which is the seventh main result  together with Theorems \ref{thm: gaussian coupling for t process}, \ref{thm: resampling methods t process}, and \ref{thm: weighted bootstrap alternative conditions} below on couplings and resampling methods for the t-statistic process.

\begin{theorem}[Pivotal and Gaussian Couplings for t-statistic Process]\label{thm: couplings for t process}
Suppose that Conditions S and U hold. If $h = o(n^{-\varepsilon})$, $m\zeta_m^2 = o(n^{1 - \varepsilon} h)$, $m^3\zeta_m^2 = o(n^{1 - \varepsilon})$, and $m^{-\kappa + 1} = o(n^{-\varepsilon})$ for some constant $\varepsilon > 0$, then
\begin{equation}\label{eq: pivotal coupling t process}
\sup_{(u,w)\in I}\left|t(u,w) - \frac{\ell(w)'J^{-1}(u)\mathbb U(u)/\sqrt n}{\sigma(u,w)}\right| \lesssim_P o(n^{-\varepsilon'})
\end{equation}
for the process $\mathbb U(\cdot)$ defined in \eqref{Def:U} for some $\varepsilon' > 0$. Also, if $h = o(n^{-\varepsilon})$, $m\zeta_m^2 = o(n^{1 - \varepsilon} h)$, $m^7\zeta_m^6 = o(n^{1 - \varepsilon})$, and $m^{-\kappa + 1} = o(n^{-\varepsilon})$ for some constant $\varepsilon > 0$, then
\begin{equation}\label{eq: gaussian coupling t process}
\sup_{(u,w)\in I}\left|t(u,w) - \frac{\ell(w)'J^{-1}(u) G(u)/\sqrt n}{\sigma(u,w)}\right| \lesssim_P o(n^{-\varepsilon'})
\end{equation}
for the process $G(\cdot)$ defined in Theorem \ref{theorem: strong} for some $\varepsilon' > 0$.
\end{theorem}

The Gaussian coupling is derived in this theorem under rather strong condition $m^7\zeta_m^6 = o(n^{1 - \varepsilon})$. It turns out that it is possible to construct the same coupling under a different set of conditions:

\begin{theorem}[Gaussian Coupling for t-statistic Process under Alternative Conditions]\label{thm: gaussian coupling for t process}
Suppose that Conditions S and U hold. In addition, suppose that $h = o(n^{-\varepsilon})$, $m\zeta_m^2 = o(n^{1 - \varepsilon} h)$, $m^3\zeta_m^2 = o(n^{1 - \varepsilon})$, and $m^{-\kappa + 1} = o(n^{-\varepsilon})$ for some constant $\varepsilon > 0$. Moreover, suppose that $(1 + \zeta_{m,\theta}^L)^{2 d_I}\zeta_m^2 = o(n^{1 - \varepsilon})$. Then \eqref{eq: gaussian coupling t process} holds for the same process $G(\cdot)$ as that used in Theorem \ref{thm: couplings for t process}.
\end{theorem}
\begin{remark}[Comparison of conditions for the Gaussian coupling in Theorems \ref{thm: couplings for t process} and \ref{thm: gaussian coupling for t process}]
The conditions of Theorems \ref{thm: couplings for t process} and \ref{thm: gaussian coupling for t process} required for the Gaussian coupling are non-nested. In particular, Theorem \ref{thm: gaussian coupling for t process} requires the condition $m^3\zeta_m^2 = o(n^{1 - \varepsilon})$ that is weaker than the corresponding condition in Theorem \ref{thm: couplings for t process}, $m^7\zeta_m^6 = o(n^{1 - \varepsilon})$, but it also requires the condition $(1 + \zeta_{m,\theta}^L)^{2 d_I} \zeta_m^2 = o(n^{1 - \varepsilon})$ that is stronger than the corresponding condition in Theorem \ref{thm: couplings for t process}, $\log \zeta_{m,\theta}^L \lesssim \log n$. However, in most cases of practical importance, the conditions of Theorem \ref{thm: gaussian coupling for t process} are substantially weaker than those of Theorem \ref{thm: couplings for t process}. For example, consider the NP model and suppose that we are interested in the conditional quantile function $Q(u,x)$ itself, so that $\ell(\omega) = Z(x)$. Further, suppose that $\mathcal X = [0,1]^d$ and that the vector of approximating functions $z$ consists of tensor products of B-splines.  Then $\zeta_m \lesssim \sqrt m$, as discussed above, and it is also possible to show that $\zeta_{m,\theta}^L \lesssim m^{1/d}$. Hence, in this case Theorem \ref{thm: gaussian coupling for t process} requires that $m^{4\vee (2/d + 3)} = o(n^{1 - \varepsilon})$ since $d_I = 1 + d$ whereas Theorem \ref{thm: couplings for t process} requires $m^{10} = o(n^{1 - \varepsilon})$.
\end{remark}

\subsubsection{Resampling Methods:} As in Section \ref{Sec:theory_coeff}, we can use four resampling methods to approximately simulate the distribution of the pivotal and Gaussian processes. Specifically, define the processes $\mathbb U^*(\cdot)$ and $G^*(\cdot)$ as in \eqref{Def:U*} and \eqref{eq: gaussian process to simulate}, respectively. Recall that conditional on $(Z_i)_{i=1}^n$, these processes are copies of the processes $\mathbb U(\cdot)$ and $G(\cdot)$, respectively, and so the processes
$$
\left\{\frac{\ell(w)'J^{-1}(u)\mathbb U^*(u)/\sqrt n}{\sigma(u,w)}\colon (u,w)\in I\right\} \ \text{ and } \ \left\{\frac{\ell(w)'J^{-1}(u)\mathbb G^*(u)/\sqrt n}{\sigma(u,w)}\colon (u,w)\in I\right\}
$$
are copies of the the pivotal and Gaussian processes
$$
\left\{\frac{\ell(w)'J^{-1}(u)\mathbb U(u)/\sqrt n}{\sigma(u,w)}\colon (u,w)\in I\right\} \ \text{ and } \ \left\{\frac{\ell(w)'J^{-1}(u)\mathbb G(u)/\sqrt n}{\sigma(u,w)}\colon (u,w)\in I\right\},
$$
respectively. Also, define the t-statistic bootstrap process $t^*(\cdot,\cdot) = \{t^*(u,w)\colon (u,w)\in I\}$ for each method as
\begin{equation*}
\begin{array}{llll}
\\
\mbox{pivotal method:}  & \displaystyle t^*(u,w) = \frac{  \ell(w)' \widehat J^{-1}(u) \mathbb{U}^*(u)/\sqrt{n} }{ \widehat \sigma(u,w)};\\
\\
\mbox{gradient bootstrap method:}    & \displaystyle t^*(u,w) = \frac{  \ell(w)' (\hat \beta^*(u)- \hat\beta(u)) }{ \widehat \sigma(u,w)}; \\
\\
\mbox{Gaussian method:}  & \displaystyle  t^*(u,w) = \frac{  \ell(w)' \widehat J^{-1}(u) \Z^*(u)/\sqrt{n} }{ \widehat \sigma(u,w)}; \\
\\
\mbox{weighted bootstrap method:} & \displaystyle  t^*(u,w) =  \frac{  \ell(w)' (\hat \beta^b(u)- \hat \beta(u)) }{ \widehat \sigma(u,w)}.
\end{array}
\end{equation*}

The following theorem shows that the conditional distribution of the t-statistic bootstrap process $t^*(\cdot,\cdot)$ given the data, which can be simulated, approximates the distribution of the pivotal (in the case of pivotal and gradient bootstrap methods) and Gaussian (in the case of Gaussian and weighted bootstrap methods) processes, and via Theorems \ref{thm: couplings for t process} and \ref{thm: gaussian coupling for t process} also of the original t-statistic process $t(\cdot,\cdot)$.

\begin{theorem}[Validity of Resampling Methods for t-statistic Process]\label{thm: resampling methods t process}
Suppose that Conditions S and U hold. In addition, suppose that $h = o(n^{-\varepsilon})$ and $m\zeta_m^2 = o(n^{1 - \varepsilon} h)$ for some constant $\varepsilon > 0$. Moreover, suppose that (i) the conditions of Theorem \ref{Thm:MainULAfeasible} hold in the case of the pivotal method, (ii) the conditions of Theorem \ref{Thm:MainULAstar} hold in the case of gradient bootstrap method, (iii) the conditions of Theorem \ref{thm: gaussian method} hold in the case of Gaussian method, and (iv) the conditions of Theorems \ref{Thm:MainBootstrap} hold in the case of weighted bootstrap method. Then for the pivotal and gradient bootstrap methods,
$$
\sup_{(u,w)\in I}\left|t^*(u,w) - \frac{\ell(w)'J^{-1}(u)\mathbb U^*(u)/\sqrt n}{\sigma(u,w)}\right| \lesssim_P o(n^{-\varepsilon'})
$$
for some $\varepsilon' > 0$.
In addition, for the Gaussian and weighted bootstrap methods,
$$
\sup_{(u,w)\in I}\left|t^*(u,w) - \frac{\ell(w)'J^{-1}(u) G^*(u)/\sqrt n}{\sigma(u,w)}\right| \lesssim_P o(n^{-\varepsilon'})
$$
for some $\varepsilon' > 0$. Moreover, the stated bounds continue to hold in $P$-probability if we replace the unconditional probability $P$ by the conditional probability $P^*$.
\end{theorem}

Note that in the case of weighted bootstrap method, the theorem above imposes the rather strong condition $m^7\zeta_m^6 = o(n^{1 -  \varepsilon})$. Like in the case of Theorem \ref{thm: gaussian coupling for t process}, it turns out that it is possible to obtain the same approximation as in this theorem but under a different set of conditions:
\begin{theorem}[Weighted Bootstrap Method for t-statistic Process under Alternative Conditions]\label{thm: weighted bootstrap alternative conditions}
Suppose that Conditions S and U hold. In addition, suppose that $h = o(n^{-\varepsilon})$, $m\zeta_m^2 = o(n^{1 - \varepsilon} h)$, $m^3\zeta_m^2 = o(n^{1 - \varepsilon})$, and $m^{-\kappa + 1} = o(n^{-\varepsilon})$ for some constant $\varepsilon > 0$. Moreover, suppose that $(1 + \zeta_{m,\theta}^L)^{2 d_I}\zeta_m^2 = o(n^{1 - \varepsilon})$. Finally, suppose that the conditions of Theorem \ref{Thm:MainBootstrap} on the weights $\pi_i$ hold. Then for the weighted bootstrap method,
$$
\sup_{(u,w)\in I}\left|t^*(u,w) - \frac{\ell(w)'J^{-1}(u) G^*(u)/\sqrt n}{\sigma(u,w)}\right| \lesssim_P o(n^{-\varepsilon'})
$$
for some $\varepsilon' > 0$. Moreover, the stated bound continues to hold in $P$-probability if we replace the unconditional probability $P$ by the conditional probability $P^*$.
\end{theorem}

\subsection{Uniform Confidence Bands}
With the help of Theorems \ref{thm: couplings for t process} -- \ref{thm: weighted bootstrap alternative conditions}, we can solve a wide range of inference problems. For example, we can construct uniform confidence bands for linear functionals $(u,w) \mapsto \theta(u,w)$ on $I$, and test shape constraints for the conditional quantile functions $x\mapsto Q(u,x)$. For the former problem, let
$$
V = \sup_{(u,w)\in I} | t(u,w)|
$$
be the maximal t-statistic. Also, let $k(1 - \alpha)$ denote the $(1-\alpha)$ quantile of the distribution of $V$. If $k(1 - \alpha)$ were known, we would have the confidence band
\begin{equation}\label{eq: exact band}
\Big\{[\hat\theta(u,w) - k(1-\alpha)\hat\sigma(u,w), \hat\theta(u,w) + k(1 - \alpha)\hat\sigma(u,w)]\colon (u,w)\in I\Big\}
\end{equation}
covering the whole function $\{\theta(u,w)\colon (u,w)\in I\}$ with probability $1 - \alpha$ exactly. However, $k(1-\alpha)$ is typically unknown, and the confidence band \eqref{eq: exact band} is infeasible. Instead, we approximate $k(1 - \alpha)$ using the resampling methods developed in this paper. Specifically, let
$$
V^* = \sup_{(u,w) \in I} | t^*(u,w)|
$$
be the bootstrap maximal t-statistic, and let $k^*(1-\alpha)$ be the $(1 - \alpha)$ quantile of the conditional distribution of $V^*$ given the data. This quantity can be computed numerically by Monte Carlo methods, as we illustrate in the next section via empirical examples and give precise algorithms in Appendix \ref{app:algorithms}. We then form a two-sided $(1-\alpha)$ uniform confidence band as
$$
\Big\{[\dot{\iota} (u,w), \ddot{\iota}(u,w)] = [ \hat \theta(u,w) - k^*(1-\alpha) \hat \sigma(u,w), \ \hat \theta(u,w) + k^*(1-\alpha) \hat \sigma(u,w)]\colon (u,w) \in I\Big\}.
$$

The following theorem establishes that this confidence band covers the whole function $\{\theta(u,w)\colon (u,w)\in I\}$ with probability $(1-\alpha)$ in large samples. 

\begin{theorem}[Uniform Confidence Bands]\label{theorem: inference using couplings}
Suppose that Conditions S and U hold. In addition, suppose that $m^3\zeta_m^2 = o(n^{1 - \varepsilon})$ and $m^{-\kappa + 1} = o(n^{-\varepsilon})$ for some constant $\varepsilon > 0$. Moreover, suppose that (i) $h\sqrt m = o(n^{-\varepsilon})$ and $m^2\zeta_m^2 = o(n^{1 - \varepsilon} h)$ in the case of pivotal and Gaussian methods and (ii) $h = o(n^{-\varepsilon})$ and $m\zeta_m^2 = o(n^{1 - \varepsilon} h)$ in the case of gradient and weighted bootstrap methods. Finally, suppose that the conditions of Theorem \ref{Thm:MainBootstrap} on the weights $\pi_i$ hold in the case of the weighted bootstrap method.
(1) Then
\begin{equation}\label{eq: conservative}
P \Big( V \leq k^*(1-\alpha) \Big) = 1-\alpha + o(1).
\end{equation}
(2) As a consequence,
\begin{equation}\label{eq: conservative coverage}
P \Big( \theta(u,w) \in [\dot{\iota}(u,w), \ddot{\iota}(u,w)], \mbox{ for all } (u,w) \in I \Big) = 1-\alpha + o(1).
\end{equation}
(3) The width of the confidence band $2k^*(1-\alpha) \hat
\sigma(u,w)$ obeys
\begin{equation}\label{eq: nonconservative width}
2k^*(1-\alpha) \hat \sigma(u,w)  \lesssim_P \sqrt{\frac{\zeta_{m,\theta}^2\log n}{n}}
\end{equation}
uniformly over $(u,w)\in I$.
\end{theorem}


In addition to the validity of the uniform confidence band, Theorem \ref{theorem: inference using couplings} establishes that the width of the uniform confidence band is of the same order as the uniform rate of convergence of the estimator $\hat\theta(u,w)$. 


\begin{remark}[Related literature]
The construction of uniform confidence bands for nonparametric functions has been of large interest both in econometrics and statistics at least from the seventies. Early constructions can be traced back at least to the classic work \cite{BR1973} by Bickel and Rosenblatt. More recent contributions include Claeskens  and Keilegom \cite{CK03}, Horowitz  and Lee \cite{HL09}, Gin\'{e}  and Nickl \cite{GN10}, and Chernozhukov, Chetverikov and Kato \cite{CCK2013}, among many others. Most of the constructions in the literature rely on a two-step strategy. First, the distribution of an estimator of the function of interest is approximated by some Gaussian process uniformly over its domain. Second, extreme value theory is employed to obtain the limit distribution of the supremum of the absolute value of the Gaussian process and its appropriate quantile is used to choose the width of the confidence band. A widely understood problem of this construction, however, is that the limit distribution on the second step may not exist and even if it does, it is often difficult to derive its explicit form. This distribution depends both on the function of interest and on the estimator considered, so that treatment of any new estimation problem requires a separate theorem, and in fact considerable efforts have been devoted to derive this distribution even in relatively simple settings, like density estimation based on projection kernels; see Gin\'{e}  and Nickl \cite{GN10} and references therein. We avoid this problem: instead of deriving the limit distribution on the second step, we rely upon resampling methods developed in this paper. As a result, our construction yields asymptotically exact uniform confidence bands that work generically for all linear functionals of the conditional quantile functions. Our strategy is related to that used in Chernozhukov, Chetverikov and Kato \cite{CCK2013} for the problem of density estimation and is built on Chernozhukov, Lee and Rosen \cite{CLR2009}, who proposed a related strategy for inference on the minimum of a function. \qed
\end{remark}

\subsection{Test of Shape Constraints}
We consider the problem of testing shape constraints for the conditional quantile functions $x\mapsto Q(u,x)$. Let $x_k$ denote the $k$-th component of $x$. We assume that the functions $x\mapsto Q(u,x)$ are twice continuously differentiable and consider three types of shape constraints:
\begin{itemize}
\item[(i)] Monotonicity of $x\mapsto Q(u,x)$ with respect to $x_k$: $\partial_{x_k}Q(u,x) \leq 0$ for all $x\in\mathcal X$ and $u\in \mathcal U$;

\item[(ii)] Concavity of $x\mapsto Q(u,x)$ with respect to $x_k$: $\partial^2_{x_k}Q(u,x) \leq 0$ for all $x\in\mathcal X$ and $u\in \mathcal U$;

\item[(iii)] Concavity of $x\mapsto Q(u,x)$ with respect to $x$: $\alpha'\partial_x^2Q(u,x)\alpha \leq 0$ for all $\alpha\in S^{d-1}$, $x\in\mathcal X$, and $u\in \mathcal U$,
\end{itemize}
where in the third example $\partial_x^2 Q(u,x)$ denotes the $d\times d$-dimensional matrix whose $(j,k)$-th element is given by $\partial_{x_j}\partial_{x_k}Q(u,x)$ for all $j,k = 1,\dots,d$.\footnote{Note that the twice continuously differentiable function $f\colon \mathcal X\to\mathbb R$ is concave if and only if $\alpha'\partial_x^2 f(x)\alpha\leq 0$ for all $x\in\mathcal X$ and $\alpha\in S^{d-1}$. To prove this claim, note that $f$ is concave if and only if the function $t\mapsto f(x + t \alpha)$ mapping $\{t\in \mathbb R\colon x+t\alpha\in \mathcal X\}$ to $\mathbb R$ is concave for all $x\in\mathcal X$ and $\alpha\in S^{d-1}$, which in turns holds if and only if $\alpha'\partial_x^2 f(x)\alpha\leq 0$ for all $x\in\mathcal X$ and $\alpha\in S^{d-1}$.} Note that in the first example we focus on the case where $x\mapsto Q(u,x)$ is decreasing with respect to $x_k$ but we can also consider the case where $x\mapsto Q(u,x)$ is increasing simply by replacing $Y$ by $-Y$ and $u$ by $1 - u$. Similarly, we can consider the case of convexity in the second and third examples.

Now, observe that all three shape constraints discussed above can be expressed using the same notation:
$$
\theta(u,w)\leq 0,\quad\text{for all }(u,w)\in I,
$$
where $\theta(u,w)$ is a linear functional with the vector of loadings being $\ell(w) = \partial_{x_k}Z(x)$ with $w = x$ in the first example, $\ell(w) = \partial^2_{x_k}Z(x)$ with $w = x$ in the second example, and $\ell(w) = (\ell_1(w),\dots,\ell_m(w))'$ where $\ell_j(w)  = \alpha'\partial_x^2 Z_j(x)\alpha$ for all $j=1,\dots,m$ with $w = (x,\alpha)$ in the third example. Hence, we are interested in testing
$$
H_0\colon \sup_{(u,w)\in I}\theta(u,w)\leq 0 \ \text{ against } \ H_1\colon  \sup_{(u,w)\in I}\theta(u,w) > 0.
$$
To test $H_0$ against $H_1$, we consider the one-sided Kolmogorov-Smirnov statistic
$$
T = \sup_{(u,w)\in I} \frac{\hat\theta(u,w)}{\hat\sigma(u,w)}.
$$
Then under $H_0$,
$$
T = \sup_{(u,w)\in I} \frac{\hat\theta(u,w)}{\hat\sigma(u,w)} \leq \sup_{(u,w)\in I}\frac{\hat\theta(u,w) - \theta(u,w)}{\hat\sigma(u,w)} = \sup_{(u,w)\in I} t(u,w).
$$
Note also that too large values of $T$ suggest that $H_0$ is violated. Hence, letting $\tilde k(1 - \alpha)$ denote the $(1 - \alpha)$ quantile of $\sup_{(u,w)\in I}t(u,w)$, we would like to reject $H_0$ if $T > \tilde k(1 - \alpha)$. However, such a test is not feasible because $\tilde k(1 - \alpha)$ is unknown. Instead, we approximate $\tilde k(1 - \alpha)$ using the resampling methods developed in this paper. Specifically, let
$$
T^* = \sup_{(u,w)\in I}t^*(u,w)
$$
be the bootstrap statistic, and let $\tilde k^*(1 - \alpha)$ be the $(1 - \alpha)$ quantile of the conditional distribution of $T^*$ given the data. This quantity can be computed numerically by Monte Carlo methods. Then we reject $H_0$ in favor of $H_1$ if $T > \tilde k^*(1 - \alpha)$. The following theorem shows that this test controls size in large samples.
\begin{theorem}[Test of Shape Constraints]\label{thm: shape constraints}
Suppose that the conditions of Theorem \ref{theorem: inference using couplings} hold. Then under $H_0$,
$$
P\Big( T > \tilde k^*(1 - \alpha)\Big) \leq \alpha + o(1).
$$
Moreover, if $\mathcal M = \mathcal M_n$ is a set of data-generating processes that satisfy $H_0$ and is such that the conditions of Theorem \ref{theorem: inference using couplings} hold uniformly over  this set, then
$$
\sup_{M \in \mathcal M}P_M\Big( T > \tilde k^*(1 - \alpha) \Big) \leq \alpha + o(1),
$$
where $P_M$ denotes probability under the data-generating process $M$.
\end{theorem}

\section{Examples}\label{Sec:examples}

This section  illustrates the finite sample performance of the
estimation and inference methods with two examples. All the calculations were
carried out with the software \verb"R" (\cite{R08}), using the
package \verb"quantreg" for quantile regression (Koenker \cite{Koenker08}). We refer to Appendix \ref{app:algorithms} for implementation algorithms and to the companion computational paper  \cite{RJ06} for software and additional  examples.

\subsection{Empirical Example}\label{sub: empirical example}

To illustrate our methods with real data, we consider an empirical
application to nonparametric estimation of the demand for gasoline. Blundell, Horowitz  and Parey \cite{BHP2012}, Hausman and Newey
 \cite{HN1995}, Schmalensee and Stoker \cite{SS1999}, and Yatchew and No \cite{YN2001} estimated nonparametrically
the average demand function. We estimate nonparametrically the
quantile demand and price elasticity functions and apply our inference
methods to construct confidence bands for the average quantile
price elasticity function and to test the Slutsky condition of consumer demand. We use the same data set as in Yatchew and No \cite{YN2001}, which comes
from the National Private Vehicle Use Survey,
conducted by Statistics Canada between October 1994 and September
1996.\footnote{The data set can be downloaded from Adonis Yatchew's web site at \texttt{ www.economics.utoronto.ca/yatchew/}.} The main advantage of this data set, relative to similar data
sets for the U.S., is that it is based on fuel purchase diaries and
contains detailed household level information on prices, fuel
consumption patterns, vehicles and demographic characteristics. (See
Yatchew and No \cite{YN2001} for a more detailed description of the data.) Our sample selection and variable construction also follow
Yatchew and No \cite{YN2001}. We select into the sample households with non-zero
licensed drivers, vehicles, and distance driven. We focus on regular
grade gasoline consumption. This selection results in a sample
of 5,001 households. Fuel consumption and expenditure are recorded
by the households at the purchase level.


We consider a partially linear specification for the demand function:\footnote{This partially linear specification of the demand function arises from household preferences characterized by the indirect utility function $V(w,v,u) = v^{1-\beta(u)}/[1-\beta(u)]  - G(u,w)$, where $w$ is real gasoline price, $v$ is real income, and $g(u,w) = \partial_w G(u,w)$ (see Lewbel \cite{Lewbel1987}, Th. 1). We thank Arthur Lewbel for pointing this out.}
$$
Y= Q(U,X), \ \ \ Q(U,X)= g(U,W) + V'\beta(U), \ \ \ X =
(W,V),
$$
where $Y$ is the log of total gasoline consumption in liters per
month; $W$ is the log of price in Canadian dollars per liter; $U$ is the unobservable
preference of the household to consume gasoline; and $V$ is a vector
of 28 covariates. Following Yatchew and No \cite{YN2001}, the covariate vector includes the log of age, a dummy for the top coded value of age, the
log of income, a set of dummies for household size, a dummy for
urban dwellers, a dummy for young-single (age less than 36 and
household size of one), the number of drivers, a dummy for more than
4 drivers, 5 province dummies, and 12 monthly dummies. To estimate
the function $w \mapsto g(w, u)$ at each $u$, we consider three different vectors of series approximating functions $w\mapsto Z(w)$: linear, a power orthogonal polynomial of degree 6, and a cubic
B-spline with 5 knots at the $\{0, 1/4, 1/2, 3/4, 1\}$ quantiles of
the observed values of $W$. The series approximation to the function $(u,x) \mapsto Q(u,x)$ takes the following form:
$$
Q(u,x) = Z(w)'\delta(u) + v'\gamma(u) = Z(x)'\beta(u), \quad Z(x) = (Z(w),v), \quad \beta(u) = (\delta(u), \gamma(u)).
$$
The number of series terms in the power and B-spline specifications is selected
by undersmoothing over the specifications chosen by applying cross validation to the corresponding least squares estimators.\footnote{There is potentially a large set of methods that can be used to choose the number of series terms (cross-validation, penalization, the method of Lepski, among others). Indeed, the problem of selecting the number of series terms is a special case of the problem of model selection, and there are several textbooks/monographs in the literature on model selection in abstract settings; for example, Massart \cite{M07} and Koltchinskii \cite{K11}. However, to the best of our knowledge, there are no papers in the literature that apply to the problem of selecting the number of series terms in the nonparametric quantile regression problem studied here. Hence, we have opted to use an ad hoc method that consists of performing cross-validation as if we were to estimate the conditional mean function $x\mapsto E[Y | X = x]$, which is estimated by the series least squares method. Under the implicit assumption that the smoothness of the functions $x\mapsto Q(u,x)$ is similar to that of the function $x\mapsto E[Y | X = x]$, such a cross-validation would yield the number of series terms that approximately equalize variance and bias terms in estimating the functions $x\mapsto Q(u,x)$. We then slightly increase the number of series terms so that the bias term is of smaller order relative to the variance term (that is, to achieve undersmoothing, as stated in Condition U.2), so that valid inference can be performed.}
In the next section, we analyze the size of the
specification error of these series approximations in a numerical
experiment calibrated to mimic this example.

The empirical results for the B-spline specification are reported in Figures \ref{fig: surfaces} and \ref{fig: average elasticity cis}.\footnote{The results for the linear and power specifications are not reported for the sake of brevity. They are similar to the results for the B-spline specification.}
 The first two panels of fig. \ref{fig: surfaces} plot the initial and monotonized estimates of the quantile
demand surface for gasoline as a function of price and the quantile
index, that is
$$
(u,\exp(w)) \mapsto \theta(u,w) = \exp(g(w, u) + v'\beta(u)),
$$
where the value of $v$ is fixed at the sample median values of the
ordinal variables and one for the dummies corresponding to the
sample modal values of the rest of the variables.\footnote{\label{ft:cov_values}The
median values of the ordinal covariates are $\$40K$ for income, $46$
for age, and $2$ for the number of drivers. The modal values for the
rest of the covariates are $0$ for the top-coding of age, $2$ for
household size, $1$ for urban dwellers, $0$ for young-single, $0$
for the dummy of more than 4 drivers, $4$ (Prairie) for province,
and $11$ (November) for month.}
The monotonized estimates  are obtained using the average rearrangement over both
the price and quantile dimensions proposed in
Chernozhukov,  Fern\'andez-Val,  and Galichon \cite{CFG2010-Biometrika}; see Appendix \ref{subset: mono}. The demand surface show most noticeably non-monotone areas with respect
to price at high quantiles, which are removed by the rearrangement. The last panel of fig. \ref{fig: surfaces} shows the estimate of the
quantile price elasticity surface as a function of price and the quantile
index, that is:
$$
(u,\exp(w)) \mapsto \theta(u,w) = \partial_w g(u,w).
$$
The estimates  show substantial heterogeneity of the elasticity across quantiles and
prices, with individuals at the upper quantiles being less sensitive
to high prices.\footnote{These estimates are smoothed by local weighted
polynomial regression  across the price dimension (Cleveland \cite{C1979}),
because the unsmoothed elasticity estimates display very erratic
behavior.}

Fig. \ref{fig: average elasticity cis} shows 90\% uniform confidence
bands for the average quantile price elasticity function
$$
u \mapsto \theta(u) = \int \partial_{w} \ g(u,w) d\mu(w),
$$
over the quantile indices $\mathcal{I} = [0.1,0.9]$, where $\mu$ is the empirical distribution of $W$. The panels of the figure correspond to the pivotal, gradient bootstrap, Gaussian and weighted bootstrap
methods. For the pivotal and Gaussian methods the distribution of
the maximal t-statistic is obtained by 1,000 simulations. The gradient bootstrap uses 199 repetitions. The
weighted bootstrap uses standard exponential weights and 199
repetitions. The confidence bands show that the evidence of
heterogeneity in the elasticities across quantiles is not
statistically significant, because we can trace a horizontal line
within the bands. They also show that there is significant
evidence of negative price sensitivity at most quantiles as the
bands are bounded away from zero for most quantiles.

The Slutsky condition of consumer demand  states that the compensated price elasticity is negative for all the households. Dette, Hoderlein, and Neumeyer \cite{DHN2011} showed that this condition has testable implications for the quantile demand function and its derivatives in heterogeneous demand systems with multiple goods and possible infinite dimensional unobservables.  The one good version of their test is:
\begin{equation}\label{eq: slutsky}
H_0 : S(u,x) \leq 0,  \text{ for all } (u,x) \in \mathcal{U} \times \mathcal{X}, \text{ vs } H_1: S(u,x) > 0,  \text{ for some }(u,x) \in \mathcal{U}\times \mathcal{X},
\end{equation}
where $S(u,x)$ is the compensated quantile price elasticity that in our logarithmic specification takes the form
$$
S(u,x) = \exp(\ell) \partial_w Q(u,x) + \exp(Q(u,x) + w) \partial_\ell Q(u,x), \ x = (w,\ell,c),
$$
which is a smooth function of the quantile demand and derivatives. Here we have partitioned the covariate vector $X$ into $(W,L,C),$ where $W$ is log of price, $L$ is the log of income, and $C$ includes the rest of the covariates.

To test the  functional hypothesis ($\ref{eq: slutsky}$) we use the one-sided Kolmogorov-Smirnov statistic:
$$
K = \max_{(u,x) \in I} \frac{\hat S(u,x)}{\hat \sigma_S(u,x)},
$$
where
$$
\hat S(u,x) = \exp(\ell) \partial_w \hat Q(u,x) + \exp(\hat Q(u,x) + w) \partial_\ell \hat Q(u,x),
$$
is the plug-in series estimator of $S(u,x)$, $\hat Q(u,x)$ is the series estimator of $Q(u,x)$, $\hat \sigma_S(u,x)$ is a delta method estimator of the asymptotic standard deviation of $\hat S(u,x)$, and $I \subseteq \mathcal{U}\times \mathcal{X} $ denotes the set of values of interest. We reject $H_0$ if the p-value of $K$ under $H_0$ is less than $\alpha$, i.e. $\sup_{P \in H_0} P(K  > k) < \alpha$ where $k$ is the realized value of $K$. Dette, Hoderlein, and Neumeyer \cite{DHN2011}  proposed an alternative test based on kernel estimators of the quantile function and its derivatives.

We estimate the distribution of $K$ under $H_0$ by weighted bootstrap with moment selection  to reduce the asymptotic non-similarity on the boundary of composite one sided functional tests  (Linton, Song, and Whang \cite{LSW2010}).
The weighted bootstrap version of $K$ is
$$
K^*(c_n) = \max_{(u,x) \in I} \frac{\hat S^*(u,x) - \hat S(u,x)}{\hat \sigma_S(u,x)} 1[|\hat S(u,x) | < c_n \hat \sigma_S(u,x) ],
$$
where
$$
\hat S^*(u,x) = \exp(\ell) \partial_w \hat Q^*(u,x) + \exp(\hat Q^*(u,x) + w) \partial_\ell \hat Q^*(u,x),
$$
is the bootstrap version of $\hat S(u,x)$, $\hat Q^*(u,x)$ is the series estimator of $Q(u,x)$ in the weighted sample, $1[|\hat S(u,x) | < c_n \hat \sigma_S(u,x) ]$ is the moment selector (Chernozhukov, Hong, and Tamer \cite{CHT2007}, and Andrews and Soares \cite{AS2010}), and $c_n$ is a sequence of thresholds that can grow with $n$. The centering by $ \hat S(u,x)$ imposes the least favorable null hypothesis $S(u,x) = 0$ at all the points $(u,x) \in I$ in the bootstrap  to control the size of the test, whereas the moment selector discards points that are far from this hypothesis with very high probability to increase power. We consider three sequences for $c_n$: no moment selection with $c_n = 0,$ BIC moment selection with $c_n^2 = \log n,$ and LIL selection with $c_n^2 = 2 \log \log n.$ The estimator of the p-value for a realization of the statistic $k$  is the probability that $K^*(c_n)$ is greater than $k$ conditional on the data.


Table \ref{table: slutsky} presents the results of the test of the Slutsky condition in our data set. We set the region $I$ to the product of $\{0.01, 0.02, ..., 0.99\}$ and the observed support of $X$ in the data.
We obtain the p-values by weighted bootstraps with standard exponential weights and 199 replications. Here, we do not find sufficient evidence  to reject the Slutsky condition at standard significance levels in any of the specifications with or without the moment selection.


\begin{table}[h]
  \centering
  \caption{Test of Slutsky Condition}\label{table: slutsky}
  \begin{tabular}{ccccc}\hline\hline
\multicolumn{2}{c}{ } & \multicolumn{3}{c}{P-value$^{\dag}$ }\\ 
\cline{3-5}
Specification & $K$ stat & No selection & BIC selection & LIL selection  \\
\hline
Linear      & 0.47 & 0.95 & 0.76 &  0.58 \\
Power      & 3.63 & 0.30 & 0.30 &  0.28     \\
B-spline   & 2.30 & 0.96 & 0.96 &  0.96    \\
\hline\hline
\multicolumn{5}{l}{\footnotesize{$^{\dag}$P-values obtained by weighted bootstrap with standard exponential weights  }}\\
\multicolumn{5}{l}{\footnotesize{ and 199 replications.}}
\end{tabular}
\end{table}

\subsection{Numerical Example}\label{subsec:simulations}

To evaluate the performance of our estimation and  inference methods
in finite samples, we conduct a Monte Carlo experiment designed to
mimic the previous empirical example. We consider the following
design for the data generating process:
\begin{equation}\label{eq: dgp}
Y = g(W) + V'\beta +  \sigma \Phi^{-1}(U),
\end{equation}
where $g(w) = \alpha_0 + \alpha_1 w + \alpha_2 \sin(2\pi w) +
\alpha_3 \cos(2\pi w) + \alpha_4 \sin(4\pi w) + \alpha_5 \cos(4\pi
w),$  $V$ is the same covariate vector as in the empirical example,
$U \sim U(0, 1),$ and $\Phi^{-1}$ denotes the inverse of the CDF of
the standard normal distribution. The parameters of $g(w)$ and
$\beta$ are calibrated by applying least squares to the data set in
the empirical example and $\sigma$ is calibrated to the least
squares residual standard deviation. We consider linear, power and
B-spline series methods to approximate $g(w)$, with the same number
of series terms and other tuning parameters as in the empirical
example. In practice, we recommend to conduct this type of Monte Carlo experiment with a data generating process that mimics the application at hand to verify the  plausibility of the regularity conditions of the method. 

Figures \ref{fig: estimands demand} and \ref{fig: estimands
elasticity} in the Supplemental Material examine the quality of the series
approximations in the population. They compare the true quantile function
$$(u,\exp(w)) \mapsto \theta(u,w) = g(w) + v'\beta + \sigma \Phi^{-1}(u),$$ and
the quantile price elasticity function $$(u,\exp(w)) \mapsto \theta(u,w) =
\partial_w g(w),$$ to the estimands of the series approximations. In
the quantile demand function the value of  $v$ is fixed at the
sample median values of the ordinal variables and at one for the
dummies corresponding to the sample modal values of the rest of the
variables (see footnote \ref{ft:cov_values}). The estimands are obtained numerically from a mega-sample
(a proxy for infinite population) of $100 \times 5,001$ observations
with the values of $(W,V)$ as in the data set (repeated 100 times)
and with $Y$ generated from the DGP (\ref{eq: dgp}). Although the
derivative function does not depend on $u$ in our design, we do not
impose this restriction on the estimands. Both figures show that the
power and B-spline estimands are close to the true target functions,
whereas the more parsimonious linear approximation misses
important curvature features of the target functions, especially in
the elasticity function.

To analyze the properties of the inference methods in finite
samples, we draw 500 samples from the DGP (\ref{eq:
dgp}) with 4 sample sizes, $n$: $10,002$, $5,001$, $1,000,$ and $500$
observations. For $n=10,002$ we fix $W$ to the values in the data
set repeated twice, for $n=5,001$ we fix $W$ to the values in the data
set, whereas for the smaller sample sizes we draw $W$ with
replacement
 from the values in the data set and keep them fixed across samples. To
speed up computation, we drop the vector $V$ by fixing it at the
sample median values of the ordinal components and at one for the
dummies corresponding to the sample modal values for all the
individuals. We focus on the average quantile price elasticity function
$$
u \mapsto \theta(u) = \int \partial_w g(w) d\mu(w),
$$
over the region $I = [0.1, 0.9]$. We estimate this function using
linear, power and B-spline quantile regression with the same number
of terms and other tuning parameters as in the empirical example.
Although $\theta(u)$ does not change with $u$ in our design, again
we do not impose this restriction on the estimators.  For inference,
we compare the performance of 90\% confidence bands for the entire
elasticity function. These bands are constructed using the pivotal, gradient bootstrap, Gaussian, and weighted
bootstrap methods, all implemented in the same fashion as in the
empirical example. The interval $I$ is approximated by a finite grid
of 91 quantiles $\tilde{I} = \{0.10, 0.11, ..., 0.90\}$.

Table 2 reports estimation and inference results averaged across 500
simulations. The true value of the elasticity function is $\theta(u)
= -0.74$ for all $u \in \tilde I$. Bias and RMSE are the absolute
bias and root mean squared error integrated over $\tilde{I}$. SE/SD
reports the ratios of empirical average standard errors to empirical
standard deviations. SE/SD uses the analytical standard errors from
expression (\ref{Def:hatsigma2n}).
The bandwidth for $\hat J(u)$ is chosen using the Hall-Sheather
option  of the \verb"quantreg" \verb"R" package
(Hall and Sheather \cite{Hall-Sheather1988}). Length gives the empirical average of
the length of the confidence band. SE/SD
and length are integrated over the grid of quantiles $\tilde{I}$.
Cover reports empirical coverage of the confidence bands with
nominal level of 90\%. Stat is the empirical average of the 90\%
quantile of the maximal t-statistic used to construct the bands.
Table 2 shows that the linear estimator has higher absolute bias
than the more flexible power and B-spline estimators, but displays
lower rmse, especially for small sample sizes. The analytical
standard errors provide good approximations to the standard
deviations of the estimators. The confidence bands have empirical
coverage close to the nominal level of 90\% for all the estimators
and sample sizes considered; and both bootstrap bands tend to have
larger average length than the pivotal and Gaussian bands. The source of this difference in coverage might be that the bootstrap methods resample the distribution of the covariates, whereas the pivotal and Gaussian methods condition on the distribution in the sample.

All in all, these results strongly confirm the practical value of the
theoretical results and methods developed in the paper.  They also
support the empirical example by verifying that our estimation and
inference methods  work quite nicely in a very similar setting.

\appendix

\section{Implementation Algorithms}\label{app:algorithms}

Throughout this section we assume that we have a random sample
$\{(Y_i,Z_i): 1 \leq i \leq n\}$. We are interested in approximating
the distribution of the process $\sqrt{n}(\widehat
\beta(\cdot)-\beta(\cdot))$ or of the statistics associated with
functionals of it. Recall that for each quantile  $u \in \mathcal{U}
\subset (0,1)$, we estimate $\beta(u)$ by quantile regression
  $
  \hat\beta(u) = \arg \min_{\beta \in \mathbb{R}^{m}} \En[\rho_u (Y_i -
  Z_i'\beta)],$
the Gram matrix $\Sigma$ by $\widehat \Sigma = \En[Z_iZ_i']$,
and the Jacobian matrix $J(u)$ by
  Powell \cite{Powell1984} estimator
  $
\widehat J(u) = \En[1\{|Y_i-Z_i'\hat\beta(u)|\leq \hn\}\cdot
Z_iZ_i']/2\hn,
  $ where we recommend choosing the bandwidth $\hn$  as in the \verb"quantreg" \verb"R" package with the Hall-Sheather option  (Hall and Sheather \cite{Hall-Sheather1988}).

We begin describing the algorithms to implement the methods to approximate the distribution of the process
$\sqrt{n}(\hat\beta(\cdot)-\beta(\cdot))$ indexed  by $\mathcal{U}$. Let $B$ be a pre-specified number of bootstrap or simulation repetitions.


\begin{algorithm}[Pivotal method] \label{algorithm:pivotal}  (1)
 For $b = 1, \ldots, B$, draw $U_1^b,\ldots,U_n^b$ i.i.d. from $U \sim \Uniform(0,1)$ and
    compute
    $\mathbb{U}^b(u) = n^{-1/2}\sum_{i=1}^nZ_i(u-1\{U_i^b \leq u\}), \ \ u \in \mathcal{U}.$
(2) Approximate the distribution of $\{\sqrt{n}(\hat\beta(u)-\beta(u)) : u \in
    \mathcal{U}\}$ by the empirical distribution of $\{\widehat J^{-1}(u)\mathbb{U}^b(u) : u \in
    \mathcal{U}, 1 \leq b \leq B\}$.
\end{algorithm}

\begin{algorithm}[Gradient bootstrap method] \label{algorithm:perturbation}  (1)  For $b = 1, \ldots, B$, draw $U_1^b,\ldots,U_n^b$ i.i.d. from $U \sim \Uniform(0,1)$ and
    compute
    $\mathbb{U}^b(u) = n^{-1/2}\sum_{i=1}^nZ_i(u-1\{U_i^b \leq u\}), \ \ u \in \mathcal{U}.$
(2) For $b = 1, \ldots, B$, estimate the quantile regression process
    $
    \hat \beta^b(u) = \arg \min_{\beta \in \Bbb{R}^m}\sum_{i=1}^n\rho_{u} (Y_i - Z_i'\beta) + \rho_u(Y_{n+1} -
    X^b_{n+1}(u)'\beta), \ \ u \in \mathcal{U},$
    where $X^b_{n+1}(u) = -\sqrt{n}\ \mathbb{U}^b_n(u)/u,$ and
    $Y_{n+1}=n\max_{1\leq i\leq n}|Y_i|$ to ensure $Y_{n+1} > X^b_{n+1}(u)'\hat
    \beta^b(u)$, for all $u \in \mathcal{U}$. (3) Approximate the distribution of $\{\sqrt{n}(\hat\beta(u)-\beta(u)) : u \in
    \mathcal{U}\}$ by the empirical distribution of $\{ \sqrt{n}(\hat \beta^b(u)-\hat\beta(u)) : u \in
    \mathcal{U}, 1 \leq b \leq B\}$.
\end{algorithm}

\begin{algorithm}[Gaussian method] \label{algorithm:gaussian}  (1)
 For $b = 1, \ldots, B$, generate a $m$-dimensional standard Brownian bridge on $\mathcal{U}$, $B_m^b(\cdot)$. Define $\Z^b_n(u) = \widehat \Sigma^{1/2}B_m^b(u)$ for $u\in \mathcal{U}$. (2) Approximate the distribution of $\{\sqrt{n}(\hat\beta(u)-\beta(u)) : u \in
    \mathcal{U}\}$ by the empirical distribution of $\{\widehat J^{-1}(u)\Z_n^b(u) : u \in
    \mathcal{U}, 1 \leq b \leq B\}$.
\end{algorithm}

\begin{algorithm}[Weighted bootstrap method] \label{algorithm:bayesian}  (1)
 For $b = 1, \ldots, B$, draw $\pi_1^b,\ldots,\pi_n^b$ i.i.d. from the standard exponential  distribution and compute the weighted quantile regression process
    $ \hat \beta^b(u) = \arg \min_{\beta \in \mathbb{R}^{m}} \sum_{i=1}^n \pi_i^b \rho_u (Y_i -
  Z_i'\beta), \ \ u \in \mathcal{U}.$ (2) Approximate the distribution of $\{\sqrt{n}(\hat\beta(u)-\beta(u)) : u \in
    \mathcal{U}\}$ by the empirical distribution of $\{ \sqrt{n}(\hat \beta^b(u)-\hat\beta(u)) : u \in
    \mathcal{U}, 1 \leq b \leq B\}$.
\end{algorithm}

The previous algorithms provide approximations to the distribution
of $\sqrt{n}(\hat\beta(u)-\beta(u))$ that are uniformly valid in $u
\in \mathcal{U}$. We can use these  approximations  directly to make
inference on linear functionals of $Q(u,x)$ including the
conditional quantile functions itself, provided the approximation error is
small as stated in Conditions P and U.
Each linear functional is represented by $\{ \theta(u,w)=\ell(w)'\beta(u)+r_n(u,w)\colon (u,w) \in I\}$,
where $\ell(w)'\beta(u)$ is the series approximation,  $\ell(w) \in \RR^m$ is a loading vector, $r_n(u,w)$ is the remainder term, and $I$ is the set of pairs of quantile indices and covariates values of interest, see Section \ref{Sec:Functionals} for details and examples. Next we provide algorithms to conduct pointwise or uniform inference over linear functionals.

\begin{algorithm}[Pointwise Inference for Linear Functionals] \label{algorithm:pointwise} (1) Compute the variance estimate $ \widehat \sigma^2(u,w) =  u(1-u) \ell(w)' \widehat J^{-1}(u) \widehat \Sigma  \widehat J^{-1}(u)\ell(w)/n $. (2) Using any of the Algorithms 1-4, compute vectors $V_1(u),\ldots, V_B(u)$ whose empirical distribution approximates the distribution of $\sqrt{n}(\hat\beta(u)-\beta(u))$. (3) For $b=1,\ldots,B$, compute the $t$-statistic
$ t^{*b}(u,w) = \left|\frac{ \ell(w)'V_b(u) }{ \sqrt{n}\widehat \sigma(u,w)}\right|$. 
(4) Form a $(1-\alpha)$-confidence interval for $\theta(u,w)$ as $\ell(w)'\hat\beta(u) \pm k(1-\alpha) \hat\sigma(u,w)$, where $k(1-\alpha)$ is the $1-\alpha$ sample quantile of $\{t^{*b}(u,w): 1\leq b \leq B\}$.
\end{algorithm}

\begin{algorithm}[Uniform Inference for Linear Functionals] \label{algorithm:uniform}
 (1) Compute the variance estimates $ \widehat \sigma^2(u,w) = u(1-u) \ell(w)' \widehat J^{-1}(u) \widehat \Sigma  \widehat J^{-1}(u)\ell(w)/n$ for $(u,w) \in I$. (2) Using any of the Algorithms 1-4, compute the processes $V_1(\cdot),\ldots, V_B(\cdot)$ whose empirical distribution approximates the distribution of $\{\sqrt{n}(\hat\beta(u)-\beta(u))\colon u \in \mathcal{U}\}$. (3) For $b=1,\ldots,B$, compute the maximal $t$-statistic
$ \|t^{*b}\|_I = \sup_{(u,w)\in I}\left|\frac{ \ell(w)'V_b(u) }{ \sqrt{n}\widehat \sigma(u,w)}\right|$. 
(4) Form a $(1-\alpha)$-confidence band for $\{ \theta(u,w)\colon (u,w) \in I\}$ as $\{ \ell(w)'\hat\beta(u) \pm k(1-\alpha) \hat\sigma(u,w)\colon (u,w)\in I\}$, where $k(1-\alpha)$ is the $1-\alpha$ sample quantile of $\{\|t^{*b}\|_I: 1\leq b \leq B\}$.
\end{algorithm}


\begin{figure*}[!hp]
\centering
 \includegraphics[width=.95\textheight, height=.95\textwidth, angle=90]{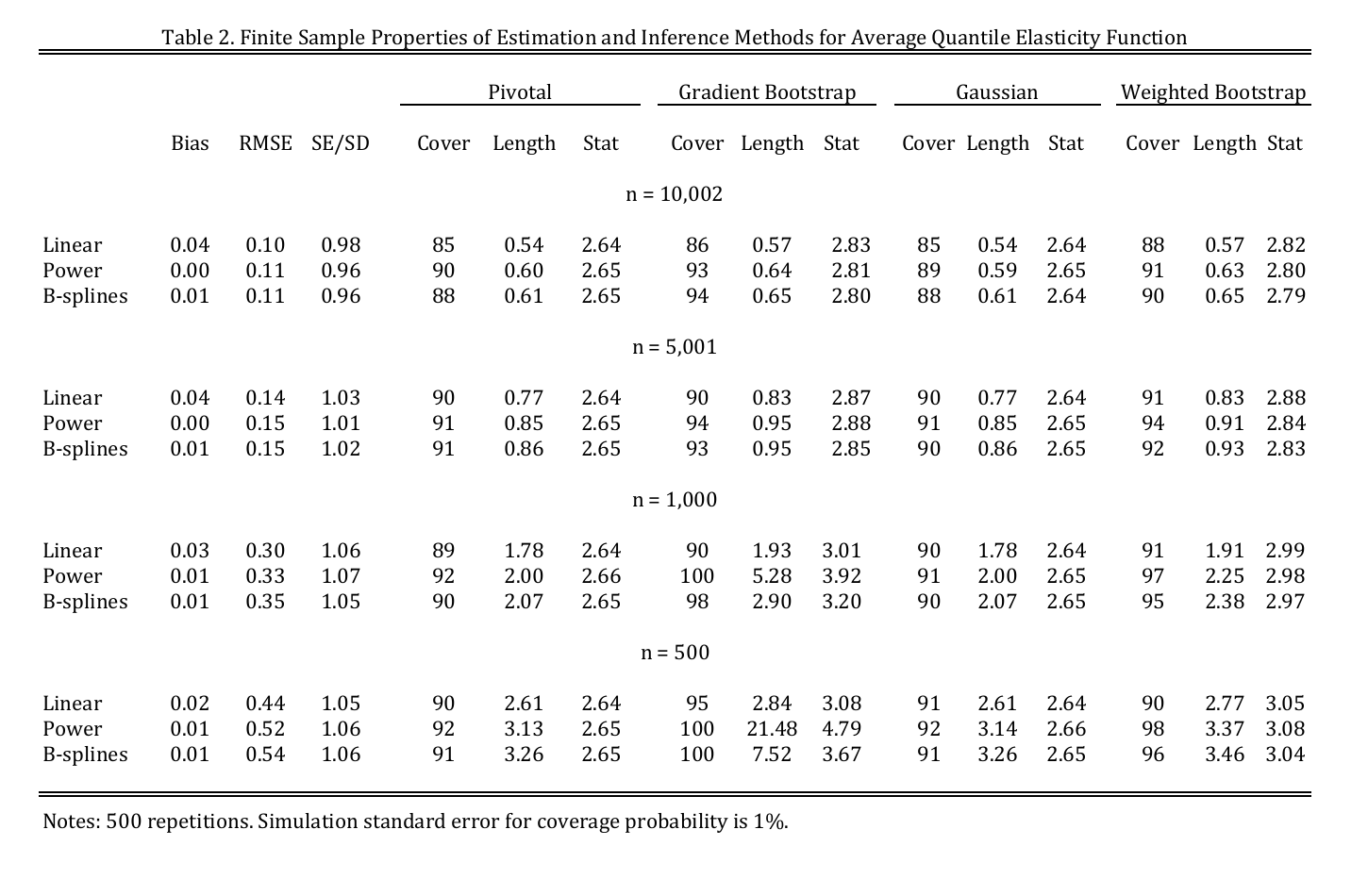}
\end{figure*}

\begin{figure}[!hp]
\centering
 \includegraphics[width=.9\textwidth, height=.9\textheight, angle=0]{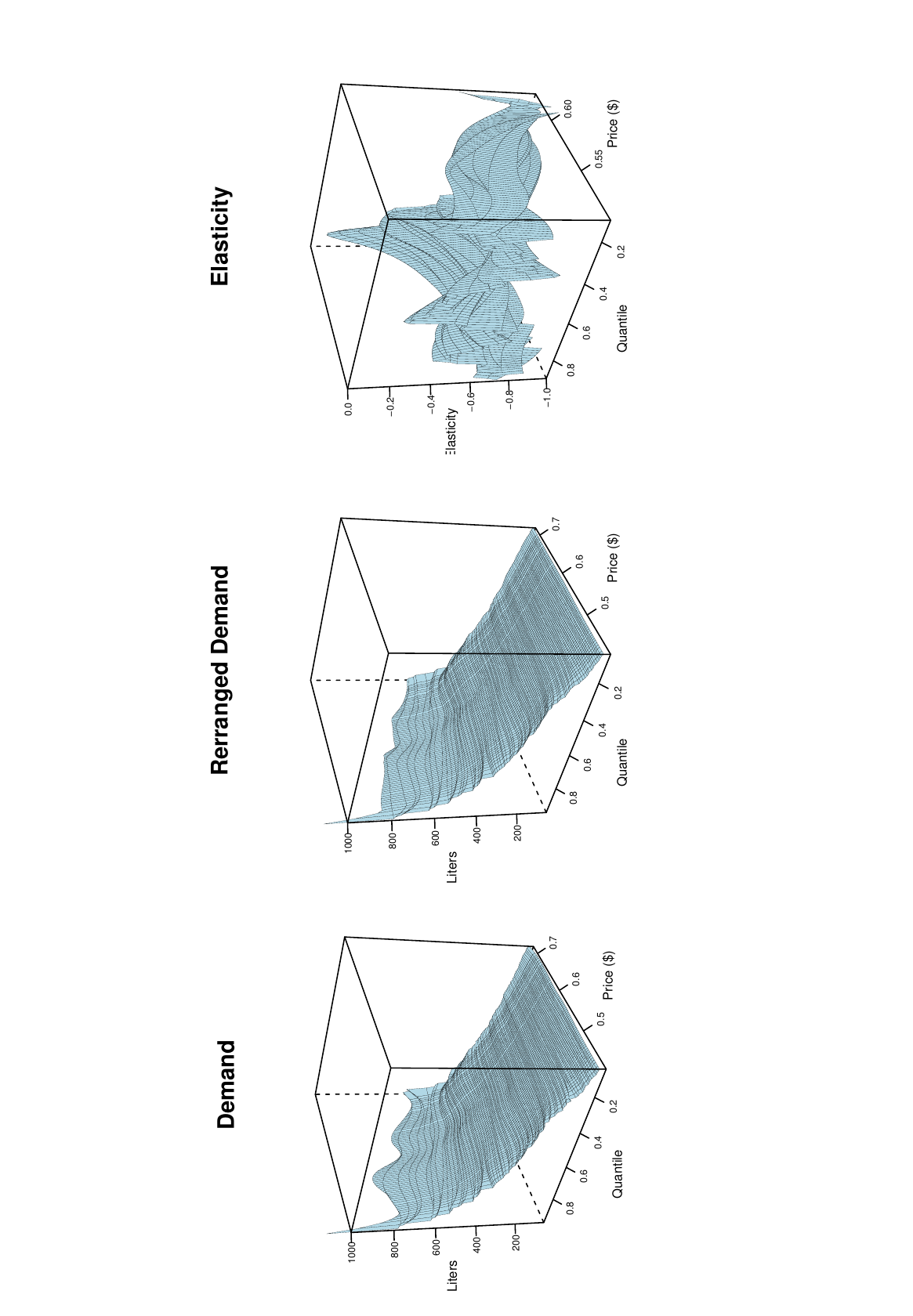}
\caption{\label{fig: surfaces} Demand and price elasticity
surfaces as a function of price and the quantile index using B-spline specification. The elasticity estimates are smoothed by local weighted
polynomial regression with bandwidth 0.5.}
\end{figure}

\begin{figure}[!hp]
\centering
 \includegraphics[width=.9\textwidth, height=.8\textheight]{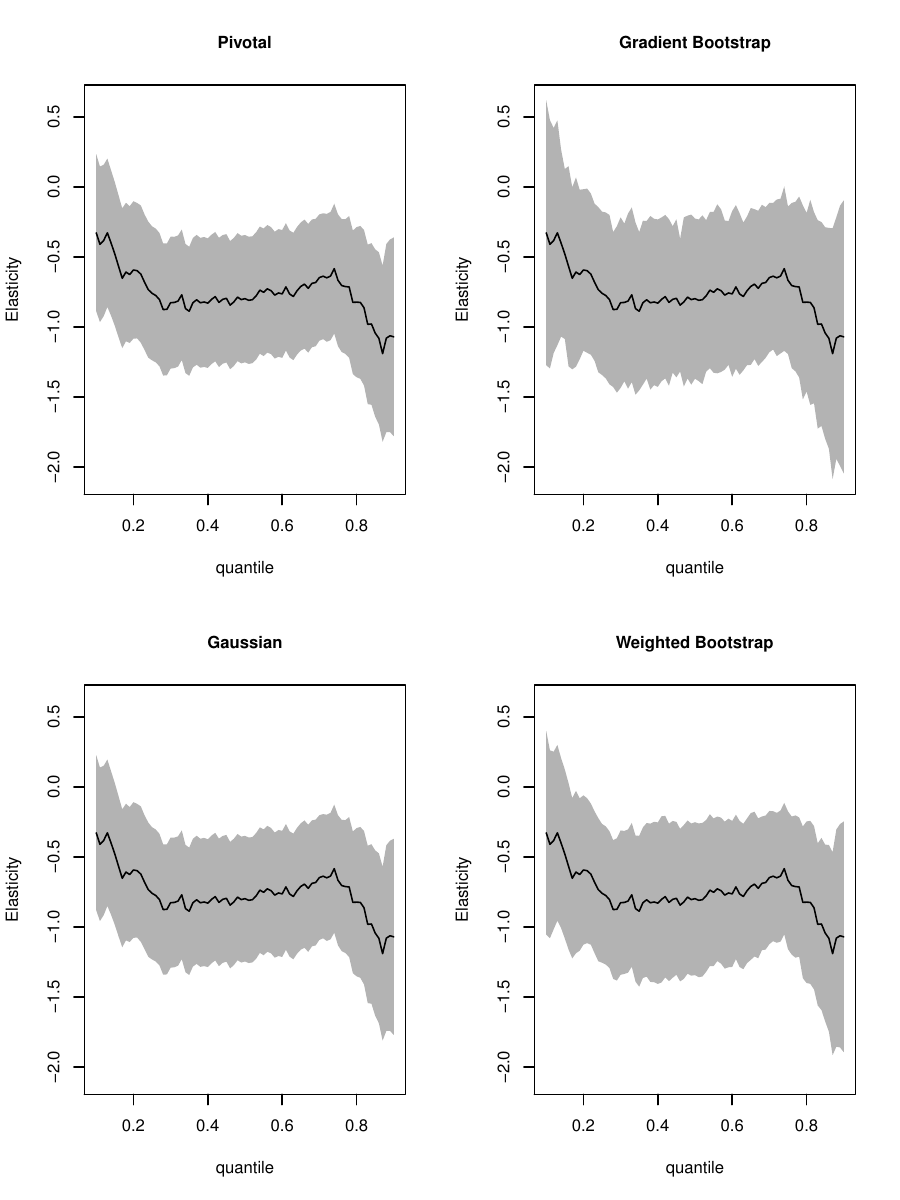}
\caption{\label{fig: average elasticity cis} 90\% Confidence bands
for the average quantile price elasticity function using B-spline specification. Pivotal and Gaussian
bands are obtained by 1,000 simulations. Gradient bootstrap bands are based on 199 bootstrap repetitions. Weighted bootstrap bands
are based on 199 bootstrap repetitions with standard exponential
weights.}
\end{figure}

\newpage
{\large {\bf Supplement to ``Conditional Quantile Processes Based on Series or Many Regressors"}}

By Alexandre Belloni, Victor Chernozhukov, Denis Chetverikov, and Iv\'{a}n Fern\'{a}ndez-Val

\section{Imposing Monotonicity on Linear Functionals}\label{subset: mono}
Consider the setting of Section \ref{Sec:Functionals}. The functionals of interest might be naturally monotone in some of their
arguments. For example, the conditional quantile function is
increasing in the quantile index and the conditional quantile demand
function is decreasing in price and increasing in the quantile
index. Therefore, it might be desirable to impose the same requirements on the estimators
of these functions.

Let $\theta(u,w),$ where $(u,w) \in I,$ be a weakly
increasing function in $(u,w)$, i.e. $\theta(w',u') \leq
\theta(u,w)$ whenever $(w',u') \leq (u,w)$
componentwise.\footnote{If $\theta(u,w)$ is decreasing in $w$, we
take the transformation $\tilde w = - w$ and $\tilde \theta(
\tilde w,u) = \theta(- \tilde w,u),$ where $\tilde \theta(\tilde
w,u)$ is increasing in $\tilde w$.} Let $\widehat \theta$ and
$[\dot{\iota}, \ddot{\iota}]$ be the point and band estimators of
$\theta,$ constructed using one of the methods described in the
previous sections. These estimators might not satisfy the
monotonicity requirement due to either estimation error or imperfect
approximation. However, we can monotonize these estimates and
perform inference  using the  method suggested in Chernozhukov, Fern\'andez-Val  and Galichon \cite{CFG2010-Biometrika}.

Let $q,f\colon  I \to K,$ where $K$ is a bounded subset of
$\mathbb{R}$, and consider any monotonization operator $\mathcal{M}$
that satisfies: (1) a monotone-neutrality condition
\begin{equation}\label{eq: monotone-neutrality}
\mathcal{M} q = q \text{ \ \ if $q$ monotone};
\end{equation}
(2) a distance-reducing condition
\begin{equation}\label{eq: distance-reducing}
\|\mathcal{M}q - \mathcal{M}f \|_I \leq \|q - f \|_I;
\end{equation}
and (3) an order-preserving condition
\begin{equation}\label{eq: order-preserving}
q \leq f \text{ \ \ implies \ \ } \mathcal{M}q \leq \mathcal{M}f.
\end{equation}
Examples of operators that satisfy these conditions include the multivariate
rearrangement (Chernozhukov, Fern\'andez-Val and Galichon \cite{CFG2010-Biometrika}), isotonic projection
(Barlow et al \cite{BarlowEtAl1972}), convex combinations of rearrangement and
isotonic regression (Chernozhukov, Fern\'andez-Val and Galichon \cite{CFG2010-Biometrika}), and
convex combinations of monotone minorants and monotone majorants.

Let $\mathcal{M} \widehat \theta$ and $[\mathcal{M}\dot{\iota}, \mathcal{M}\ddot{\iota}]$ be the monotonized QR-series process and confidence band.  We establish that $\mathcal{M} \widehat \theta$ has smaller estimation
error than the original QR-series process $\widehat \theta$, and that $[\mathcal{M}\dot{\iota}, \mathcal{M}\ddot{\iota}]$  has higher coverage
and smaller length than the original confidence band $[\dot{\iota}, \ddot{\iota}]$. The following result is a corollary  from
Theorems  \ref{theorem: uniform rate} and \ref{theorem: inference using couplings} using the same arguments
as in Propositions 2 and 3 in Chernozhukov, Fern\'andez-Val, and Galichon  \cite{CFG2010-Biometrika}.


\begin{corollary}[Inference for Monotone Linear Functionals]\label{Corollary:Monotonization}
Let $\theta: I \to K$ be weakly increasing over $I$ and
$\widehat \theta$ be the QR-series
process of Theorem \ref{theorem: uniform rate}. If $\mathcal{M}$
satisfies the conditions (\ref{eq: monotone-neutrality}) and
(\ref{eq: distance-reducing}), then the monotonized QR-series process is necessarily closer
to the true value:
$$
\| \mathcal{M}\widehat \theta - \theta \|_I \leq \| \widehat \theta -
\theta \|_I.
$$
Let $[\dot{\iota}, \ddot{\iota}]$ be a confidence band for $\theta$ of Theorem \ref{theorem:
inference using couplings}. If $\mathcal{M}$ satisfies the
conditions  (\ref{eq: monotone-neutrality}) and  (\ref{eq:
order-preserving}), the monotonized confidence bands maintain at least the asymptotic level
of the original bands:
$$
P \Big( \theta(u,w) \in [\mathcal{M}\dot{\iota}(u,w), \mathcal{M}\ddot{\iota}(u,w)]\colon
(u,w) \in I \Big) \geq 1-\alpha + o(1).
$$
If $\mathcal{M}$ satisfies the condition (\ref{eq:
distance-reducing}), the monotonized confidence bands are shorter in length than
the original bands:
$$
\| \mathcal{M}\ddot{\iota} - \mathcal{M}\dot{\iota} \|_I \leq \| \ddot{\iota} -
\dot{\iota} \|_I.
$$
\end{corollary}



\section{QR-series approximation error and proof of Lemma \ref{lem: approximation error}}\label{Sec:Identification}
In this section we study properties of the QR-series approximating functions and, in particular, we prove Lemma \ref{lem: approximation error}. In what follows, for any function $h\colon \mathcal{X}\to\RR$, we define
$$
Q_u(h)=E\Big[\rho_u(Y-h(X))-\rho_u(Y-Q(u,X))\Big],
$$
so that $\beta(u) \in \arg\min_{\beta\in\RR^m}Q_u(Z'\beta)$. For brevity and with some abuse of notation, depending on the context, we use $h$ to denote either a function or a random variable $h=h(X)$ (the same applies to the vector of functions $Z$). Also, let $\bar f=\sup_{x\in\XX,y\in\mathcal Y_x}f_{Y|X}(y|x)$ and $\underf= \inf_{x\in\mathcal{X},u\in\mathcal U} f_{Y|X}(Q(u,x)|x)$. In addition, let  $f'_{Y|X}(y|x)$ denote the derivative of the function $y\mapsto f_{Y|X}(y|x)$ and define $\bar f' = \sup_{x\in\XX,y\in\mathcal Y_x}|f'_{Y|X}(y|x)|$. Observe that $\bar f$ and $\bar f'$ are bounded from above and $\underf$ is bounded away from zero uniformly over $n$ by Condition S.2.  Moreover, let $\lambda_{\min}$ denote the minimal eigenvalue of the matrix $\Sigma = E[Z Z']$. Observe that $\lambda_{\min}$ is bounded away from zero uniformly over $n$ by Condition S.3.

For $u\in\mathcal U$, consider the best $L_2$-approximation to the conditional $u$-quantile function $x\mapsto Q(u,x)$ by a linear combination of functions in the vector $Z$, namely
\begin{equation}\label{eq: LS optimization problem}
\widetilde\beta(u) = \arg\min_{\beta \in \RR^m} E\Big[|Z(X)'\beta - Q(u,X)|^2\Big].
\end{equation}
Define
\begin{equation}\label{eq: c2 and cinf definitions}
c_{u,2}^2 = \Ep\Big[|Z(X)'\widetilde\beta(u) - Q(u,X) |^2\Big]\text{ and }c_{u,\infty} = \sup_{x\in \mathcal{X}} \Big|Z(x)'\widetilde\beta(u) - Q(u,x)\Big|.
\end{equation}
The next lemma relates the QR-series approximation error $R(u,x) = Q(u,x) - Z(x)'\beta(u)$ to $c_{u,2}$ and $c_{u,\infty}$. Define $\zeta_m = \sup_{x\in\XX}\|Z(x)\|$.
\begin{lemma}[Bounds on QR-series approximation error]\label{Lemma:AUX_L2sparse}
Assume that Conditions S.2-S.3 hold. In addition assume that
$$
c_{u,\infty}^2\leq \frac{\underf^3}{12 \bar f (\bar f')^2}\quad \text{ and }\quad \zeta_m^2 c_{u,2}^2 \leq \frac{\underf^3 \lambda_{\min}}{12\bar f (\bar f')^2}.
$$
Then we have for all $u\in\mathcal U$ that
\begin{align}
&\Ep\Big[ |Z(X)'\beta(u) - Q(u,X)|^2 \Big]  \leq 3 (\bar f/\underf ) c_{u,2}^2,\label{lem: 2 - ineq 1}\\
&\sup_{x\in\mathcal{X}}\Big|Z(x)'\beta(u)-Q(u,x)\Big| \leq \lambda_{\min}^{-1/2}\zeta_m \Big(1 + \sqrt{3\bar f/\underf}\Big)c_{u,2} + c_{u,\infty}.\label{lem: 2 - ineq 2}
\end{align}
\end{lemma}
\begin{proof}[Proof of Lemma \ref{Lemma:AUX_L2sparse}]
Fix $u\in\mathcal U$. Let $g_u\colon \XX\to\RR$ be the function defined by $g_u(x) = Q(u,x)$ for all $x\in\XX$.

\noindent
{\bf Step 1} (Main argument). Recall that $Z = Z(X)$. For notational convenience, let
\begin{equation}\label{eq: bar q definition}
\bar q = \frac{\underf^{3/2}\Big(\Ep\Big[ |Z'\beta(u) - g_u(X) |^2\Big]\Big)^{3/2}}{\bar f'\Ep\Big[ |Z'\beta(u) - g_u(X)|^3\Big]}.
\end{equation}
By Steps 2 and 3 below, we have respectively that
\begin{equation}\label{Def:UB2}
 Q_u(Z'\beta(u)) - Q_u(g_u)  \leq \bar f c_{u,2}^2,
\end{equation}
\begin{equation}\label{Def:LB2}
 Q_u(Z'\beta(u)) - Q_u(g_u) \geq  \frac{\underf\Ep[ |Z'\beta(u) - g_u(X)|^2]}{3} \wedge \left( \frac{\bar q}{3}\sqrt{\underf\Ep\[ |Z'\beta(u) - g_u(X)|^2\]} \right).
\end{equation}
Thus,
\begin{equation}\label{eq: implications steps 2 and 3}
 \frac{\underf\Ep[ |Z'\beta(u) - g_u(X)|^2]}{3} \wedge  \left(\frac{\bar q}{3}\sqrt{\underf\Ep\[ |Z'\beta(u) - g_u(X)|^2\]} \right) \leq \bar f c_{u,2}^2.
\end{equation}
Therefore, if the minimum in \eqref{eq: implications steps 2 and 3} is achieved by the first term, it immediately follows that $\Ep[ |Z'\beta(u) - g_u(X)|^2] \leq 3 (\bar f/\underf) c_{u,2}^2$, which is \eqref{lem: 2 - ineq 1}. On the other hand, if the minimum in \eqref{eq: implications steps 2 and 3} is achieved by the second term, it follows from the inequality $\bar f c_{u,2}^2 \leq \bar q^2/3$, which is established in Step 4 below, that
$$
\frac{\bar q}{3}\sqrt{\underf\Ep\[ |Z'\beta(u) - g_u(X)|^2\]} \leq \frac{\bar q}{\sqrt 3}\cdot (\bar f c_{u,2}^2)^{1/2},
$$
which again gives \eqref{lem: 2 - ineq 1}. Hence, \eqref{lem: 2 - ineq 1} follows.

To prove \eqref{lem: 2 - ineq 2}, note that it follows from the triangle inequality, \eqref{eq: c2 and cinf definitions}, and \eqref{lem: 2 - ineq 1} that
\begin{align*}
\sqrt{\Ep\Big[ |Z'\beta(u) - Z'\widetilde\beta(u)|^2\Big]}
&\leq \sqrt{\Ep\Big[ |Z'\beta(u) - g_u(X)|^2\Big]} + \sqrt{\Ep\Big[ |g_u(X) - Z'\widetilde\beta(u)|^2\Big]}\\
&\leq  \sqrt{3\bar f/\underf}c_{u,2} + c_{u,2} = \Big(1 + \sqrt{3\bar f/\underf}\Big)c_{u,2},
\end{align*}
and so
\begin{align*}\sup_{x\in \mathcal{X}}|Z(x)'\beta(u)-g_u(x)| & \leq \sup_{x\in \mathcal{X}}|Z(x)'\beta(u)-Z(x)'\widetilde\beta(u)|  + \sup_{x\in\XX}|Z(x)'\widetilde\beta(u) - g_u(x)|\\
& \leq \zeta_m\|\beta(u) - \widetilde\beta(u)\| + c_{u,\infty}\\
& \leq \lambda_{\min}^{-1/2}\zeta_m \sqrt{\Ep\[ |Z'\beta(u) - Z'\widetilde\beta(u)|^2\]} + c_{u,\infty}\\
& \leq \lambda_{\min}^{-1/2}\zeta_m \Big(1 + \sqrt{3\bar f/\underf}\Big)c_{u,2} + c_{u,\infty}.
\end{align*}
Hence, \eqref{lem: 2 - ineq 1} follows. This completes Step 1.

\noindent
{\bf Step 2} (Upper Bound). Here we establish \eqref{Def:UB2}. Observe that for any two scalars $w$ and $v$, we have
\begin{equation}\label{eq: knight identity}
\rho_u(w-v) - \rho_u(w) = -v (u - 1\{w\leq 0\}) + \int_0^v(
1\{w\leq t\} - 1\{w\leq 0\})dt.
\end{equation}
Therefore, for any function $h\colon \XX\to\RR$, using the law of iterated expectations and \eqref{eq: knight identity} with $w = Y - g_u(X)$ and $v = h(X) - g_u(X)$, we obtain
\begin{align*}
Q_u(h) - Q_u(g_u)
&=  \Ep\[ \int_0^{h-g_u} (F_{Y|X}(g_u + t|X) - F_{Y|X}(g_u|X) )d t \] \\
& =  \Ep\[ \int_0^{h-g_u} tf_{Y|X}(g_u+\tilde t_{X,t}|X) dt \]  \leq    (\bar f/2) \Ep[| h  - g_u |^2]
\end{align*}
for some $\tilde t_{X,t}$ between $0$ and $t$. Thus,
$$
Q_u(Z'\beta(u)) - Q_u(g_u) \leq Q_u(Z'\widetilde\beta(u))-Q_u(g_u) \leq \bar f c_{u,2}^2,
$$
which gives \eqref{Def:UB2}. This completes Step 2.

\noindent
{\bf Step 3} (Lower Bound). Here we establish \eqref{Def:LB2}. For any function $h\colon \XX\to\RR$, we have
\begin{align}
Q_u(h) - Q_u(g_u) &=  \Ep\[ \int_0^{h-g_u} (F_{Y|X}(g_u + t|X) - F_{Y|X}(g_u|X)) d t \] \label{eq: knight lower bound - 1}\\
& =  \Ep\[ \int_0^{h-g_u} tf_{Y|X}(g_u|X) + (t^2/2)f'_{Y|X}(g_u+\tilde t_{X,t}|X) d t \] \label{eq: knight lower bound - 2}\\
& \geq    (\underf/2) \Ep[| h - g_u|^2]  - (1/6)\bar f ' \Ep[|h - g_u|^3]\label{Eq:KnightLower2}
\end{align}
for some $\tilde t_{X,t}$ between $0$ and $t$. Consider the case $(\underf\Ep\[ |Z'\beta(u) - g_u(X)|^2\])^{1/2} \leq \bar{q}$ for $\bar q$ defined in \eqref{eq: bar q definition}. By definition of $\bar q$, we have
$$
\bar f' \Ep\[ |Z'\beta(u) - g_u(X)|^3\]\leq \underf\Ep\[ |Z'\beta(u) - g_u(X)|^2\],
$$
and so using the inequality (\ref{Eq:KnightLower2}) with $h=Z'\beta(u)$ yields
$$
Q_u(Z'\beta(u)) - Q_u(g_u) \geq \frac{\underf\Ep\[ |Z'\beta(u) - g_u(X)|^2\]}{3},
$$
which implies \eqref{Def:LB2}.

Now consider the case $(\underf\Ep\[ |Z'\beta(u) - g_u(X)|^2\])^{1/2} > \bar q$. Let
$$
h_u(x) = (1-\alpha) Z(x)'\beta(u) + \alpha g_u(x),\quad x\in\XX,
$$
where $\alpha \in (0,1)$ is picked so that $(\underf \Ep\[ |h_u - g_u|^2\])^{1/2} = \bar q$ (we can assume that $\bar q>0$, so that such $\alpha$ exists, since otherwise the claim of the lemma is trivial). Then
$$
1-\alpha = \frac{\bar q}{\sqrt{\underf E[|Z'\beta(u) - g_u(X)|^2]}},
$$
and by convexity of $Q_u$, we have
\begin{align*}
Q_u(Z'\beta(u)) - Q_u(g_u)
& \geq \frac{1}{1-\alpha}\cdot (Q_u(h_u) - Q_u(g_u))\\
& = \frac{\sqrt{\underf\Ep\[ |Z'\beta(u) - g_u(X)|^2\]}}{\bar q} \cdot (Q_u(h_u) - Q_u(g_u)).
\end{align*}
Also, note that
$$
\bar q = \frac{\underf^{3/2}}{\bar f'} \frac{(\Ep[ |Z'\beta(u) - g_u(X)|^2])^{3/2}}{\Ep[ |Z'\beta(u) - g_u(X)|^3]} = \frac{\underf^{3/2}}{\bar f'} \frac{(E[|h_u - g_u|^2])^{3/2}}{E[|h_u - g_u|^3]} = \frac{1}{\bar f'} \frac{\bar q^{3}}{ E[|h_u - g_u|^3]},
$$
so that
$$
\bar f' E[|h_u - g_u|^3] = \bar q^2.
$$
Hence, using the inequality \eqref{Eq:KnightLower2} with $h = h_u$ gives
$$
Q_u(h_u) - Q_u(g_u) \geq (\underf/2) \Ep[| h_u - g_u|^2]  - (1/6)\bar f ' \Ep[|h_u - g_u|^3] = \bar q^2/2 - \bar q^2/6 = \bar q^2/3.
$$
Combining the inequalities above yields
$$
Q_u(Z'\beta(u)) - Q_u(g_u)  \geq
\frac{\bar q}{3}\sqrt{\underf\Ep\[ |Z'\beta(u) - g_u(X)|^2\]},
$$
which implies \eqref{Def:LB2}. This completes Step 3.

\noindent
{\bf Step 4} (Auxiliary inequality). Here we show that $ \bar f c_{u,2}^2 \leq \bar q^2/3$.
Recall that $\zeta_m = \sup_{x\in\XX}\|Z(x)\|$. Also, note that since $\widetilde \beta(u)$ solves the problem \eqref{eq: LS optimization problem}, we have
$$
E\Big[(g_u(X) - Z'\widetilde\beta(u))\cdot(Z'(\beta(u) - \widetilde\beta(u)))\Big] = 0,
$$
and so
$$
E\Big[(g_u(X) - Z'\beta(u))^2\Big] = E\Big[(g_u(X) - Z'\widetilde\beta(u))^2\Big] + E\Big[(Z'(\beta(u) - \widetilde\beta(u)))^2\Big].
$$
Hence,
\begin{align*}
&E\Big[(g_u(X) - Z'\beta(u))^2\Big] \geq E\Big[(g_u(X) - Z'\widetilde\beta(u))^2\Big],\\
&E\Big[(g_u(X) - Z'\beta(u))^2\Big]\geq E\Big[(Z'(\beta(u) - \widetilde\beta(u)))^2\Big].
\end{align*}
Therefore, the ratio $(\Ep[|Z'\beta(u) - g_u(X)|^2])^{3/2}/\Ep[|Z'\beta(u) - g_u(X)|^3]$ is bounded from below by
\begin{align*}
& \frac{\Ep[|Z'\beta(u) - g_u(X)|^2]^{1/2}}{\sup_{x\in\mathcal{X}}|Z(x)'\beta(u) - g_u(x)|}\\
&\qquad \geq \frac{1}{2}\frac{\Ep[|Z'\beta(u) - Z'\widetilde\beta(u)|^2]^{1/2} + \Ep[|Z'\widetilde\beta(u) - g_u(X)|^2]^{1/2}}{\sup_{x\in\mathcal{X}}|Z(x)'\beta(u) - Z(x)'\widetilde\beta(u)| + \sup_{x\in\mathcal{X}}|Z(x)'\widetilde\beta(u) - g_u(x)|} \\
& \qquad \geq \frac{1}{2}\left(\frac{(\Ep[|Z'\beta(u) - Z'\widetilde\beta(u)|^2])^{1/2}}{\sup_{x\in \mathcal{X}}|Z(x)'\beta(u) - Z(x)'\widetilde\beta(u)|} \wedge \frac{(\Ep[|Z'\widetilde\beta(u)-g_u(X)|^2])^{1/2}}{\sup_{x\in\mathcal{X}}|Z(x)'\widetilde\beta(u)-g_u(x)|} \right)\\
&\qquad  \geq \frac{1}{2}\left(\frac{\lambda_{\min}^{1/2} \|\beta(u) - \widetilde\beta(u)\|}{\zeta_m\|\beta(u) -\widetilde\beta(u)\|} \wedge \frac{c_{u,2}}{c_{u,\infty}} \right) = \frac{1}{2}\left(\frac{\lambda_{\min}^{1/2}}{\zeta_m}\wedge\frac{c_{u,2}}{c_{u,\infty}}\right),
\end{align*}
where the inequality in the last line follows from \eqref{eq: c2 and cinf definitions}. So,
$$
\frac{\bar q^2}{3} \geq \frac{\underf^{3}}{12(\bar f')^2}\left(\frac{\lambda_{\min}}{\zeta_m^2}\wedge\frac{c_{u,2}^2}{c_{u,\infty}^2}\right)\geq \bar f c_{u,2}^2
$$
since $c_{u,\infty}^2\leq \underf^3/(12(\bar f')^2 \bar f)$ and $\zeta_m^2 c_{u,2}^2 \leq \underf^3 \lambda_{\min}/(12(\bar f')^2\bar f)$. This completes Step 4 and the proof of the lemma.
\end{proof}
\subsection*{Proof of Lemma \ref{lem: approximation error}}
Fix $u\in\mathcal U$. All inequalities and convergence statements in the proof apply uniformly over $u\in\mathcal U$ but we fix $u$ so that we do not have to repeat ``uniformly over $u\in\mathcal U$'' many times.

First, consider the case of polynomials. We have $\zeta_m\lesssim m$; see Newey \cite{W97}. Also, since $Q(u,\cdot)\in\Sigma(s,C,\XX)$ for some constant $C>0$, it follows that
$$
\inf_{\beta\in\RR^m}\sup_{x\in\XX}|Z(x)'\beta - Q(u,x)|\lesssim m^{-s/d};
$$
see Chen \cite{Chen2006}. Therefore, by \eqref{eq: LS optimization problem} and \eqref{eq: c2 and cinf definitions},
\begin{align*}
c_{u,2}
&= \Big(E[(Z'\widetilde\beta(u) - Q(u,X))^2]\Big)^{1/2}=\min_{\beta \in \RR^m} \Big(E\Big[|Z(X)'\beta - Q(u,X)|^2\Big]\Big)^{1/2}\\
& \lesssim \inf_{\beta\in\RR^m}\sup_{x\in\XX}|Z(x)'\beta - Q(u,x)|\lesssim m^{-s/d}.
\end{align*}
Also, it follows from Proposition 3.1 in Belloni et al \cite{BelloniChenChernozhukov2009} that
$$
c_{u,\infty} = \sup_{x\in\XX}|Z(x)'\widetilde\beta(u) - Q(u,x)| \lesssim (1+\zeta_m)m^{-s/d} \lesssim m^{1-s/d}
$$
since $\zeta_m\lesssim m$. Moreover, under the assumed condition $s>d$, since $c_{u,2}\lesssim m^{-s/d}$ and $c_{u,\infty} \lesssim m^{1-s/d}$, we have that $c_{u,\infty}\to 0$ and $\zeta_m c_{u,2}\to 0$ if $m\to\infty$. Hence, the inequalities \eqref{eq: approximation error 1} follow from Lemma \ref{Lemma:AUX_L2sparse}.

Next, consider the case of B-splines. We have $\zeta_m\lesssim \sqrt m$; see Newey \cite{W97}. Also, since $Q(u,\cdot)\in \Sigma(s,C,\XX)$ for some constant $C>0$, it follows that
$$
\inf_{\beta\in\RR^m}\sup_{x\in\XX}|Z(x)'\beta - Q(u,x)|\lesssim m^{-(s\wedge s_0)/d};
$$
see Chen \cite{Chen2006}. Therefore, by the same argument as above,
$$
c_{u,2} = \Big(E[(Z'\widetilde\beta(u) - Q(u,X))^2]\Big)^{1/2}\lesssim \inf_{\beta\in\RR^m}\sup_{x\in\XX}|Z(x)'\beta - Q(u,x)|\lesssim m^{-(s\wedge s_0)/d}.
$$
Also, it follows from Proposition 3.1 in Belloni et al \cite{BelloniChenChernozhukov2009} that
$$
c_{u,\infty} = \sup_{x\in\XX}|Z(x)'\widetilde\beta(u) - Q(u,x)| \lesssim (1+\zeta_m)m^{-(s\wedge s_0)/d} \lesssim m^{1/2-(s\wedge s_0)/d}
$$
since $\zeta_m\lesssim \sqrt m$. Moreover, under the assumed condition $s\wedge s_0 > d$, since $c_{u,2}\lesssim m^{-(s\wedge s_0)/d}$ and $c_{u,\infty}\lesssim m^{1/2 - (s\wedge s_0)/d}$, we have that $c_{u,\infty}\to 0$ and $\zeta_m c_{u,2}\to 0$ if $m\to \infty$. Hence, the first inequality in \eqref{eq: approximation error 2} follows from Lemma \ref{Lemma:AUX_L2sparse}.

To prove the second inequality in \eqref{eq: approximation error 2}, note that Lemma \ref{Lemma:AUX_L2sparse} also implies that
$$
\sup_{x\in\XX}|Z(x)'\beta(u) - Q(u,x)|\lesssim m^{1/2 - (s\wedge s_0)/d}.
$$
Further, let
\begin{equation}\label{eq: fu0}
f_{u 0} = \int_{x\in\mathcal X}f_X(x)f_{Y|X}(Q(u,x)|x)dx
\end{equation}
and let $X_u$ be a random variable with the support $\XX$ and the pdf
\begin{equation}\label{eq: auxiliary density}
f_u(x) = f_{u 0}^{-1}f_X(x)f_{Y|X}(Q(u,x)|x),\quad x\in\XX.
\end{equation}
Note that $f_u(x)$ is bounded from above and away from zero uniformly over $x\in\XX$. Define
\begin{equation}\label{eq: perturbed optimization problem 1}
\bar \beta(u) = \arg\min_{\beta\in\RR^m}E\Big[|Z(X_u)'\beta - Q(u,X_u)|^2\Big].
\end{equation}
Applying Theorem A.1 and Lemma 5.1 in Huang \cite{H03} shows that
\begin{equation}\label{eq: huang impl 2}
\sup_{x\in\XX}|Z(x)'\bar\beta(u) - Q(u,x)|\lesssim \inf_{\beta\in\RR^m}\sup_{x\in\XX}|Z(x)'\beta - Q(u,x)|\lesssim m^{-(s\wedge s_0)/d}.
\end{equation}
(Note that applying Theorem A.1 in Huang \cite{H03} requires verifying his Conditions A.1--A.3. Condition A.1 holds because $f_u(x)$ is bounded from above and away from zero uniformly over $x\in\XX$; Conditions A.2 and A.3 hold because the vector of approximating functions $Z$ consists of tensor products of B-splines with uniform knot sequence, see the discussion in Huang's paper for clarifications.) Hence, by the triangle inequality,
\begin{equation}\label{eq: huang impl 1}
\sup_{x\in\XX}|Z(x)'(\beta(u) - \bar\beta(u))|\to 0
\end{equation}
as $m\to \infty$. Moreover, since $\bar\beta(u)$ solves the optimization problem \eqref{eq: perturbed optimization problem 1}, we have
\begin{align}
& E\Big[|Z(X_u)'(\beta(u) - \bar\beta(u))|^2\Big]\notag\\
&\qquad = E\Big[|Z(X_u)'\beta(u) - Q(u,X_u)|^2\Big] - E\Big[|Z(X_u)'\bar\beta(u) - Q(u,X_u)|^2\Big].\label{eq: opt probl impl}
\end{align}
Also, let $\underf_u = \inf_{x\in\XX}f_u(x)$ and $\bar f_X = \sup_{x\in\mathcal X}f_X(x)$. Note that both $\underf_u$ and $f_{u 0}$ are strictly positive and both $\bar f_X$ and $\bar f'$ are finite (recall that $f_{u 0}$ is defined in \eqref{eq: fu0} and $\bar f'$ is defined in the beginning of this section).

Now, as in \eqref{eq: knight lower bound - 1}-\eqref{eq: knight lower bound - 2} of the proof of Lemma \ref{Lemma:AUX_L2sparse}, for $g_u(\cdot) = Q(u,\cdot)$, we have
\begin{align*}
Q_u(Z'\beta(u)) - Q_u(g_u)
&=\Ep\[ \int_0^{Z'\beta(u)-g_u} (tf_{Y|X}(g_u|X) + (t^2/2)f'_{Y|X}(g_u+\tilde t_{X,t}|X)) dt \]\\
&\geq (f_{u 0}/2)E\Big[|Z(X_u)'\beta(u) - Q_u(u,X_u)|^2\Big] \\
&\quad- (1/6)(\bar f_X \bar f'/\underf_u)E\Big[|Z(X_u)'\beta(u) - Q_u(u,X_u)|^3\Big]
\end{align*}
and
\begin{align*}
Q_u(Z'\bar\beta(u)) - Q_u(g_u)
&=\Ep\[ \int_0^{Z'\bar\beta(u)-g_u} (tf_{Y|X}(g_u|X) + (t^2/2)f'_{Y|X}(g_u+\tilde t_{X,t}|X)) dt \]\\
&\leq (f_{u 0}/2)E\Big[|Z(X_u)'\bar\beta(u) - Q_u(u,X_u)|^2\Big] \\
&\quad + (1/6)(\bar f_X \bar f'/\underf_u)E\Big[|Z(X_u)'\bar\beta(u) - Q_u(u,X_u)|^3\Big],
\end{align*}
where inequalities follow from the definition of the pdf of $X_u$ in \eqref{eq: auxiliary density}.
Combining these inequalities with \eqref{eq: opt probl impl} and using the fact that $Q_u(Z'\beta(u)) \leq Q_u(Z'\bar\beta(u))$ gives
\begin{align}
&E\Big[|Z(X_u)'(\beta(u) - \bar\beta(u))|^2\Big]\notag\\
&\quad \leq \frac{2\bar f_X \bar f'}{3 f_{u 0} \underf_u}\Big(E\Big[|Z(X_u)'\beta(u) - Q_u(u,X_u)|^3\Big] + E\Big[|Z(X_u)'\bar\beta(u) - Q_u(u,X_u)|^3\Big]\Big).\label{eq: almost there lemma 1}
\end{align}
In turn, $(\bar f_X \bar f')/(f_{u 0} \underf_u)$ is finite, and so the right-hand of the inequality \eqref{eq: almost there lemma 1} is bounded from above up-to a constant by
\begin{align*}
&E\Big[|Z(X_u)'(\beta(u) - \bar\beta(u))|^3\Big] + E\Big[|Z(X_u)'\bar\beta(u) - Q_u(u,X_u)|^3\Big]\\
&\qquad \leq c E\Big[|Z(X_u)'(\beta(u) - \bar\beta(u))|^2\Big] + \sup_{x\in\XX}|Z(x)'\bar\beta(u) - Q_u(u,x)|^3,
\end{align*}
where $c$ is a constant that is arbitrarily small (uniformly over $u\in\mathcal U$) if $m$ is large enough; see \eqref{eq: huang impl 1}. Hence,
\begin{equation}\label{eq: l1 der x}
E\Big[|Z(X_u)'(\beta(u) - \bar\beta(u))|^2\Big] \lesssim \sup_{x\in\XX}|Z(x)'\bar\beta(u) - Q_u(u,x)|^3\lesssim m^{-3(s\wedge s_0)/d},
\end{equation}
where the second inequality follows from \eqref{eq: huang impl 2}. Therefore,
$$
\|\beta(u) - \bar\beta(u)\|^2 \lesssim E\Big[|Z(X_u) '(\beta(u) - \bar\beta(u))|^2\Big]\lesssim m^{-3(s\wedge s_0)/d}.
$$
Conclude that
\begin{align*}
&\sup_{x\in\mathcal X}|Q(u,x) - Z(x)'\beta(u)|
\leq \sup_{x\in\XX}|Q(u,x) - Z(x)'\bar\beta(u)| + \sup_{x\in\XX}|Z(x)'(\bar\beta(u) - \beta(u))|\\
&\quad \lesssim m^{-(s\wedge s_0)/d} + \zeta_m \|\beta(u) - \bar\beta(u)\|\lesssim m^{-(s\wedge s_0)/d} + m^{1/2 - (3/2)(s\wedge s_0)/d}\lesssim m^{-(s\wedge s_0)/d},
\end{align*}
since $(s\wedge s_0) > d$.
This gives the second inequality in \eqref{eq: approximation error 2}. The last claim of the lemma follows from \eqref{eq: approximation error 1} and \eqref{eq: approximation error 2}. This completes the proof of the lemma.
\qed

\section{Proofs for Section \ref{Sec:theory_coeff}}\label{sec: proof of main results}


In this section we gather proofs of Theorems
\ref{Thm:SeriesRates}-\ref{Thm:MainBootstrap} and Corollaries \ref{cor: l2 rate qr series estimator} and \ref{cor: sup rate qr series estimator} from the main text. We
adopt the standard notation of the empirical process literature
(van~der Vaart and Wellner \cite{vdV-W}). In addition, for notational convenience,
we write $\psi_i(\beta,u) = Z_i(1\{Y_i \leq Z_i'\beta\}-u)$, where
$Z_i=Z(X_i)$ and $i=1,\dots,n$. Also, for a sequence of scalars $(r_n)_{n\geq 1}$ such that $r_n \to 0$, we define the set
\begin{equation}\label{eq: Rnm definition}
R_{n,m} := \Big\{ (u,\beta) \in \mathcal{U}\times \RR^m\colon
\| \beta - \beta(u) \| \leq r_n \Big\}
\end{equation}
and the
following error terms:
 \begin{equation}
\begin{array}{rccl}
 \epsilon_0(m,n)& := & \displaystyle \sup_{u \in \mathcal{U}} &\Big\|\Gn(\psi_i (\beta(u),u))\Big\|,\label{eq: three epsilons} \\
 & & \\
\epsilon_1(m,n)& := & \displaystyle \sup_{ (u, \beta) \in R_{n,m}} & \Big\| \Gn
(\psi_i(\beta,
u)) - \Gn (\psi_i(\beta(u),u)) \Big\|, \\
 & &\\
\epsilon_2(m,n)& := & \displaystyle \sup_{  (u, \beta) \in R_{n,m} } & n^{1/2} \Big\|
E[\psi_i(\beta, u)] - E[\psi_i(\beta(u),u)] - J(u)(\beta -
\beta(u))\Big\|, \\
\end{array}
\end{equation}
for the matrix $J(u)$ defined in \eqref{eq: J matrix}.
In what follows, we say that the data are in general position if the observations are independent and for any $\gamma \in \RR^m$, \begin{equation}\label{def:genpos}P\( Y_i = Z_i'\gamma,
\text{for at least one $i=1,\dots,n$}\mid Z_1,\dots,Z_n\) = 0,\end{equation} which holds under Condition S.  Finally, we emphasize that although we formally work with an i.i.d. setting in this paper, most of the results in this appendix do not rely on this assumption and can be extended to more general settings.

\subsection{Proof of Theorem  \ref{Thm:SeriesRates}}
The following technical lemma will be used in the proof of Theorem \ref{Thm:SeriesRates}.

\begin{lemma}[Rates in Euclidian Norm for Perturbed QR Process]\label{Lemma:Rates} Suppose that for all $u\in\mathcal U$, $\hat \beta(u)$ is a minimizer of
\begin{equation}\label{eq: perturbed optimization problem}
\En[\rho_{u}(Y_i - Z_i'\beta)] + \mathcal{A}_n(u)'\beta,
\end{equation}
where $\mathcal A_n(u)$ is the perturbation term, so that the QR-series coefficient $\widehat\beta(u)$ corresponds to the unperturbed case $\mathcal{A}_n (u) = 0$. Also, suppose that Condition S holds. Moreover, suppose that there exists a sequence of scalars $(\eta_n)_{n\geq 1}$ such that $\sup_{u\in\mathcal U}\|\mathcal A_n(u)\|\lesssim_P \eta_n$ and for any constant $B>0$, defining $R_{n,m}$, $\epsilon_0(m,n)$, $\epsilon_1(m,n)$, and $\epsilon_2(m,n)$ above with $r_n = B\eta_n$ gives
\begin{itemize}
\item[R1.] $\epsilon_0(m,n) \lesssim_P \sqrt{n}\eta_n$, %
\item[R2.] $\epsilon_1(m,n) = o_P(\sqrt{n}\eta_n)$, %
\item[R3.] $\epsilon_2(m,n) = o_P(\sqrt{n}\eta_n)$.
\end{itemize}
Then
\begin{equation}\label{Eq:URate} \sup_{u \in
\mathcal{U}}\left\| \hat \beta(u) - \beta(u) \right\| \lesssim_P \eta_n.
\end{equation}
\end{lemma}
\begin{proof}[Proof of Lemma \ref{Lemma:Rates}]
Due to convexity of the objective function \eqref{eq: perturbed optimization problem}, it suffices to
show that for any $\varepsilon>0$, there exists $B<
\infty$ such that
\begin{equation}\label{EventThm1}
P\left( \inf_{u \in \mathcal{U}}\inf_{\alpha\in S^{m-1}}
\alpha' \left[\En \[\psi_i\(\beta, u\)\]+ \mathcal{A}_n(u)\right]|_{\beta = \beta(u) + B \eta_n\alpha}>0
\right) \geq 1- \varepsilon
\end{equation}
for all sufficiently large $n$ since $\En \[\psi_i\(\beta, u\)\]+\mathcal{A}_n(u)$ is a
sub-gradient of the objective function at $\beta$.
To show \eqref{EventThm1}, let $B$ be some large constant to be chosen later and observe that uniformly in $u \in \mathcal{U}$ and $\alpha\in S^{m-1}$,
\begin{align*}
&\sqrt{ n } \alpha ' \En \[\psi_i\left(\beta(u) + B
\eta_n\alpha, u\right)\] \\
&\qquad \geq \Gn
(\alpha' \psi_i(\beta(u),u))  + \sqrt n \alpha' J(u) \alpha B\eta_n - \epsilon_1(m,n) - \epsilon_2(m,n),
\end{align*}
by setting $r_n = B\eta_n$ and observing that $\Ep\[\psi_i(\beta(u),u)\] =
0$ by definition of $\beta(u)$ (see the argument in the proof of Lemma \ref{Lemma:ULA} below).
Further, we have uniformly in $u\in\mathcal U$ and $\alpha \in S^{m-1}$ that
$$
| \Gn (\alpha '\psi_i(\beta(u), u))|  \leq    \sup_{u
\in \mathcal{U}} \| \Gn(\psi_i (\beta(u),u)) \| = \epsilon_0(m,n).
$$
Also, it follows from Condition S that all eigenvalues of the matrix $J(u)$ are bounded below from zero uniformly over $u\in\mathcal U$, so that
$$
\inf_{u\in\mathcal U}\inf_{\alpha\in S^{m-1}}\alpha' J(u)\alpha \geq c
$$
for some constant $c>0$.
Thus, the event of interest in (\ref{EventThm1}) is
implied by the event
 \begin{equation*}
\left\{ \sqrt n c B \eta_n - \epsilon_0(m,n) - \epsilon_1(m,n) -  \epsilon_2(m,n) - \sqrt{n}\sup_{u \in \mathcal{U}}\|\mathcal{A}_n(u)\|
> 0\right \},
 \end{equation*}
and the probability of this event can be made
arbitrarily close to one for all sufficiently large $n$ by setting $B$
sufficiently large since (i) $\epsilon_0(m,n) \lesssim_P \sqrt n\eta_n$ by R1 and $\epsilon_0(m,n)$ does not depend on $B$, (ii) $\sup_{u \in \mathcal{U}}\|\mathcal{A}_n(u)\|\lesssim_P \eta_n$ by assumptions of the lemma and $\sup_{u\in\mathcal U}\|\mathcal A_n(u)\|$ does not depend on $B$, and (iii) $\epsilon_1(m,n)+\epsilon_2(m,n) = o_P(\sqrt n \eta_n)$ by R2 and R3.
This completes the proof of the lemma.
\end{proof}

\begin{proof}[Proof of Theorem \ref{Thm:SeriesRates}.]
The result follows from Lemma \ref{Lemma:Rates} with $\mathcal{A}_n(u)=0$ and $\eta_n=\sqrt{m/n}$ provided we can show that if for any constant $B>0$, we define $\epsilon_0(m,n)$, $\epsilon_1(m,n)$, and $\epsilon_2(m,n)$ with $r_n = B\sqrt{m/n}$, then
\begin{equation}\label{eq: epsilon sum bound}
 \epsilon_0(m,n) \lesssim_P \sqrt m \ \text{ and } \ \epsilon_1(m,n) + \epsilon_2(m,n) = o_P (\sqrt{m}),
\end{equation}
In turn, to show \eqref{eq: epsilon sum bound}, note that by Lemma \ref{Lemma:error-rate-0},
$$\epsilon_0(m,n)\lesssim_P \sqrt{m}\Big(1 + \sqrt{m^{-\kappa}\log n}+\sqrt{m}\zeta_m\log n/\sqrt{n} \Big)\lesssim \sqrt{m} $$
since $m^{-\kappa}\log n = o(1)$ and $m\zeta_m^2\log^2 n = o(n)$. Moreover,  by Lemma \ref{bound e1 and e2},
$$
\epsilon_1(m,n)\lesssim_P \sqrt{m\zeta_m r_n \log n}+m\zeta_m\log n / \sqrt{n} = o(\sqrt{m}),
$$
where we again used $m\zeta_m^2\log^2 n = o(n)$, and
$$
\epsilon_2(m,n)\lesssim \sqrt{n}\zeta_mr_n^2 + \sqrt{n}m^{-\kappa} r_n = o(\sqrt m)
$$
since $m\zeta_m^2 = o(n)$ and $m\to \infty$. This completes the proof of the theorem.
\end{proof}

\subsection{Proof of Corollary \ref{cor: l2 rate qr series estimator}}
We have uniformly over $u\in\mathcal U$ that
\begin{align*}
\|\hat Q(u,\cdot) - Q(u,\cdot)\|_{L^2(X)}
&\leq \|Z(\cdot)'(\hat \beta(u) - \beta(u))\|_{L^2(X)} + \|R(u,\cdot)\|_{L^2(X)}\\
&\lesssim \|\hat \beta(u) - \beta(u)\| + \|R(u,\cdot)\|_{L^2(X)},
\end{align*}
where the first line follows from the triangle inequality and the second from Condition S.3. Further, it follows from Lemma \ref{lem: approximation error} that $\sup_{u\in\mathcal U}\|R(u,\cdot)\|_{L^2(X)}\lesssim m^{-s/d}$ in the case of polynomials and $\sup_{u\in\mathcal U}\|R(u,\cdot)\|_{L^2(X)}\lesssim m^{-(s\wedge s_0)/d}$ in the case of B-splines (note that in the case of polynomials, an application of Lemma \ref{lem: approximation error} requires the condition that $s > d$ but this condition follows from the assumption that $m^{1 - s/d}\log n = o(1)$ imposed in the corollary).

To bound $\sup_{u\in\mathcal U}\|\hat\beta(u) - \beta(u)\|$, we apply Theorem \ref{Thm:SeriesRates}. Lemma \ref{lem: approximation error} implies that Condition S.4 holds with $\kappa = s/d - 1$ in the case of polynomials and with $\kappa = (s\wedge s_0)/d$ in the case of B-splines. Also, we have $\zeta_m \lesssim m$ in the case of polynomials and $\zeta_m \lesssim m^{1/2}$ in the case of B-splines. Thus, in both cases, the conditions that $\zeta_m^2 m \log^2 n = o(n)$ and $m^{-\kappa} \log n = o(1)$ required in Theorem \ref{Thm:SeriesRates} follow from the assumptions of the corollary. Therefore, an application of Theorem \ref{Thm:SeriesRates} gives
$$
\sup_{u\in\mathcal U}\|\hat\beta(u) - \beta(u)\|\lesssim_P \sqrt{m/n}
$$
in both cases. Combining presented bounds gives both claims of the corollary.
\qed

\subsection{Proof of Theorem \ref{Thm:MainULA}}
For $u\in\mathcal U$, consider the function in \eqref{eq: perturbed optimization problem}, where $\mathcal A_n(u)$ is the perturbation term, and let $\widehat\beta(u)$ be a minimizer of this function, so that the QR-series coefficient $\widehat\beta(u)$ corresponds to the unperturbed case $\mathcal A_n(u) = 0$. Define the following approximation error:
\begin{equation}\label{eq: epsilon_3 definition}
\epsilon_3(m,n) :=  \sup_{ u \in \mathcal{U}}  n^{1/2}\Big\|
 \En [\psi_i(\hat\beta(u),u)] + \mathcal{A}_n(u)\Big\|.
\end{equation}
We have the following lemma.
\begin{lemma}[Uniform Linear Approximation]\label{Lemma:ULA}
Consider the setting specified above. Suppose that the data are in general position so that (\ref{def:genpos}) holds, and further that the conditions of Lemma \ref{Lemma:Rates} hold for some sequence of scalars $(\eta_n)_{n\geq 1}$. Then
\begin{equation}\label{Eq:ULA}
\sqrt{n} J(u)\left(\hat \beta(u) - \beta(u)\right) =
-\frac{1}{\sqrt{n}} \sum_{i=1}^n \psi_i(\beta(u),u) - \mathcal{A}_n(u)+
r_n(u),
\end{equation}
where $r_n(u)$ is such that for any $\varepsilon\in(0,1)$,
\begin{equation}\label{eq: sum of three epsilons lemma 4}
\sup_{u \in \mathcal{U}} \| r_n(u) \| \leq \epsilon_1(m,n)+\epsilon_2(m,n)+\epsilon_3(m,n),
\end{equation}
with probability at least $1 - \varepsilon$, where $\epsilon_1(m,n)$ and $\epsilon_2(m,n)$ are defined in \eqref{eq: three epsilons} using $R_{n,m}$ in \eqref{eq: Rnm definition} with $r_n = B\eta_n$ for some sufficiently large constant $B = B(\varepsilon)$.
 \end{lemma}
\begin{proof}[Proof of Lemma \ref{Lemma:ULA}]
First, note that
$\Ep\[\psi_i(\beta(u),u)\]=0$ by definition of $\beta(u)$. Indeed,
despite the possible approximation error, $\beta(u)$ minimizes
$E[\rho_u(Y-Z'\beta)]$ so that $\Ep\[\psi_i(\beta(u),u)\] = 0$
by the first order conditions. Therefore, equation (\ref{Eq:ULA}) can be recast as
$$r_n(u) = n^{1/2}J(u) (\hat \beta(u) -
\beta(u)) + \Gn \(\psi_i(\beta(u),u)\) + \mathcal{A}_n(u).
$$
Second, by Lemma \ref{Lemma:Rates}, for any $\varepsilon\in(0,1)$, there is a constant $B = B(\varepsilon)$ such that with probability $1-\varepsilon$, we have $(u,\hat \beta(u))\in R_{n,m}$ for all $u\in \mathcal{U}$ and all $n$, where $R_{n,m}$ is defined in \eqref{eq: Rnm definition} with $r_n = B\eta_n$. Therefore, we have by the triangle inequality that for all $u\in \mathcal{U}$,
 \begin{align*}
\| r_n(u) \|  & \leq
\Big\| \Gn
(\psi_i(\beta(u),u)) - \Gn (\psi_i(\hat\beta(u),u)) \Big\|  \\
& \quad +  n^{1/2} \Big\|  \Ep \[ \psi_i(\hat\beta(u), u)\]
- \Ep \[ \psi_i(\beta(u),u) \] - J(u) \( \hat\beta(u) -
\beta(u) \) \Big\| \\
& \quad +   n^{1/2} \Big\| \En \[ \psi_i(\hat\beta(u),u)\] + \mathcal{A}_n(u)\Big\|  \\
& \leq   \epsilon_1(m,n) + \epsilon_2(m,n) +
\epsilon_3(m,n)
 \end{align*}
 by the definitions of $\epsilon_1(m,n)$, $\epsilon_2(m,n)$, and $\epsilon_3(m,n)$. The asserted claim follows.
\end{proof}

\begin{proof}[Proof of Theorem \ref{Thm:MainULA}.]
To prove the asserted claim, we apply Lemma \ref{Lemma:ULA} with $\mathcal{A}_n (\cdot) = 0$.
Since conditions $m^3\zeta_m^2 = o(n^{1 - \varepsilon})$ and $m^{-\kappa + 1} = o(n^{-\varepsilon})$, which are assumed in Theorem \ref{Thm:MainULA}, imply conditions $m\zeta_m^2\log^2n = o(n)$ and $m^{-\kappa}\log n = o(1)$, which are assumed in Theorem \ref{Thm:SeriesRates}, it follows that conditions of Lemma \ref{Lemma:Rates} hold with $\eta_n = \sqrt{m/n}$ by the same argument as that used in the proof of Theorem \ref{Thm:SeriesRates}. Thus, since under Condition S, the data are in general position, the conclusion of Lemma \ref{Lemma:ULA} holds, and so for any $\varphi\in(0,1)$, there exists $B>0$ such that \eqref{eq: sum of three epsilons lemma 4} holds with probability at least $1 - \varphi$, where $\epsilon_1(m,n)$ and $\epsilon_2(m,n)$ are defined in \eqref{eq: three epsilons} using $R_{n,m}$ in \eqref{eq: Rnm definition} with $r_n = B\sqrt{m/n}$.

Next, we control $\epsilon_1(m,n)$, $\epsilon_2(m,n)$, and $\epsilon_3(m,n)$ for given $r_n$.  By Lemma \ref{bound e1 and e2},
$$
\epsilon_1(m,n) \lesssim_P \sqrt{m \zeta_m r_n \log n} + \frac{m\zeta_m \log n}{\sqrt{n}} \lesssim \frac{m^{3/4}\zeta_m^{1/2}\log^{1/2}n}{n^{1/4}},
$$
where the second inequality holds since $m^3\zeta_m^2 = o(n^{1-\varepsilon})$. Also by Lemma \ref{bound e1 and e2},
$$
\epsilon_2(m,n) \lesssim_P \sqrt{n}\zeta_m r^2+\sqrt{n}m^{-\kappa}r_n \lesssim \frac{m\zeta_m}{\sqrt n} + m^{1/2 - \kappa}.
$$
Further, by Lemma \ref{bound e3}, we have with probability one that
$$
\epsilon_3(m,n) \leq \frac{m \zeta_m}{\sqrt{n}}.
$$
Combining the inequalities above and using \eqref{Eq:ULA} and \eqref{eq: sum of three epsilons lemma 4} shows that,
$$
\sup_{u\in\mathcal U}\left\|\sqrt n J(u)\Big(\widehat \beta(u) - \beta(u)\Big) + \frac{1}{\sqrt n}\sum_{i=1}^n \psi_i(\beta(u),u)\right\|\lesssim_P \frac{m^{3/4}\zeta_m^{1/2}\log^{1/2}n}{n^{1/4}} + m^{1/2 - \kappa}.
$$
Finally, observe that
$$
\widetilde r_u := \frac{1}{\sqrt{n}} \sum_{i=1}^n Z_i\Big(1\{ Y_i \leq Q(u,X_i)\}-1\{ Y_i \leq Z_i'\beta(u)\}\Big),\quad u\in\mathcal U,
$$
satisfies
$$
\sup_{u\in\mathcal{U}}\|\widetilde r_u\| \lesssim_P \sqrt{m^{1-\kappa}\log n} + \frac{m\zeta_m\log n}{\sqrt{n}}
$$
by Lemma \ref{bound ApproxULA} and note that
$$
\frac{1}{\sqrt n}\sum_{i=1}^n \psi_i(\beta(u),u) = \frac{1}{\sqrt n}\sum_{i=1}^n Z_i (1\{Y_i \leq Z_i'\beta(u)\} - u),\quad u\in\mathcal U
$$
and
$$
\mathbb{U}(u) = \frac{1}{\sqrt{n}} \sum_{i=1}^n Z_i(u-1\{ Y_i \leq Q(u,X_i)\}),\quad u\in\mathcal U.
$$
The asserted claims now follow by noting that all eigenvalues of the matrix $J(u)$ are bounded below from zero uniformly over $u\in\mathcal U$.
\end{proof}

\subsection{Proof of Corollary \ref{cor: sup rate qr series estimator}}
By the triangle inequality, we have for all $u\in\mathcal U$ and $x\in\mathcal X$ that
\begin{equation}\label{eq: cor 2 bound 1}
|\hat Q(u,x) - Q(u,x)| \leq |Z(x)'(\hat\beta(u) - \beta(u))| + |R(u,x)|.
\end{equation}
By Lemma \ref{lem: approximation error},
\begin{equation}\label{eq: cor 2 bound 2}
\sup_{u\in\mathcal U}\sup_{x\in\mathcal X}|R(u,x)|\lesssim m^{1-s/d}
\end{equation}
in the case of polynomials and
\begin{equation}\label{eq: cor 2 bound 3}
\sup_{u\in\mathcal U}\sup_{x\in\mathcal X}|R(u,x)|\lesssim m^{-(s\wedge s_0)/d}
\end{equation}
in the case of B-splines.

To bound $\sup_{u\in\mathcal U}\sup_{x\in\mathcal X}|Z(x)'(\hat\beta(u) - \beta(u))|$, we apply Theorem \ref{Thm:MainULA}. Note that by Lemma \ref{lem: approximation error}, Condition S.4 is satisfied with $\kappa = s/d - 1$ in the case of polynomials and with $\kappa = (s\wedge s_0)/d$ in the case of B-splines. Also, we have $\zeta_m \lesssim m$ in the case of polynomials and $\zeta_m \lesssim m^{1/2}$ in the case of B-splines, so that the conditions $m^3\zeta_m^2 = o(n^{1-\varepsilon})$ and $m^{-\kappa + 1} = o(n^{-\varepsilon})$, required in Theorem \ref{Thm:MainULA}, are satisfied in both cases. Therefore, an application of Theorem \ref{Thm:MainULA} gives
$$
Z(x)'(\hat\beta(u) - \beta(u)) = \frac{1}{\sqrt n}Z(x)'J^{-1}(u)\mathbb U(u) + n^{-1/2}\zeta_m\|r(u)\|
$$
uniformly over $u\in\mathcal U$ and $x\in\mathcal X$ in both cases where we used $|Z(x)'r(u)|\leq \|Z(x)\| \ \|r(u)\| \leq \zeta_m\|r(u)\|$. Note that
$$
E[Z(x)'J^{-1}(u)\mathbb U(u)] = 0
$$
and
$$
E[(Z(x)'J^{-1}(u)\mathbb U(u))^2] \leq |Z(x)'J^{-1}(u)\Sigma J^{-1}(u)Z(x)|\lesssim \|Z(x)\|^2 \leq \zeta_m^2
$$
uniformly over $u\in\mathcal U$ and $x\in\mathcal X$ by Conditions S.2 and S.3. Therefore, the argument like that used in the proof of Theorem \ref{theorem: uniform rate} below shows that
$$
|Z(x)'J^{-1}(u) \mathbb U(u)|\lesssim_P \zeta_m\sqrt{\log n}
$$
uniformly over $u\in\mathcal U$ and $x\in\mathcal X$. Substituting the bound $\zeta_m\lesssim m$ in the case of polynomials and $\zeta_m\lesssim m^{1/2}$ in the case of B-splines shows that uniformly over $u\in\mathcal U$ and $x\in\mathcal X$, under the conditions corresponding to each case we have $\sup_{u\in \mathcal{U}}\|r(u)\|=o_P(1)$, so that
$$
|Z(x)'(\hat\beta(u) - \beta(u))| \lesssim_P \sqrt{m^2\log n/n}
$$
in the former case and
$$
|Z(x)'(\hat\beta(u) - \beta(u))| \lesssim_P \sqrt{m\log n/n}
$$
in the latter case. Combining these bounds with those in \eqref{eq: cor 2 bound 1}, \eqref{eq: cor 2 bound 2}, and \eqref{eq: cor 2 bound 3} gives the asserted claim.
\qed

\subsection{Proof of Theorems \ref{Thm:MainULAfeasible}, \ref{Thm:MainULAstar},  \ref{theorem: strong}, \ref{thm: gaussian method}, and \ref{Thm:MainBootstrap}}

\begin{proof}[Proof of Theorem \ref{Thm:MainULAfeasible}]
Note that
$$
\sup_{u\in \mathcal{U}} \| r(u) \| \leq \sup_{u\in\mathcal U}\Big\|(\widehat J^{-1}(u) - J^{-1}(u))\mathbb U^*(u)\Big\| \leq \sup_{u\in\mathcal U}\Big\|\widehat J^{-1}(u) - J^{-1}(u)\Big\|\cdot \sup_{u\in\mathcal U}\|\mathbb U^*(u)\|.
$$
In addition, by Lemma \ref{Lemma:error-rate-0}, we have $\sup_{u\in\mathcal U}\|\mathbb U(u)\| \lesssim_P \sqrt{m}$, and so
\begin{equation}\label{EqM01}
\sup_{u\in \mathcal{U}} \| \mathbb{U}^*(u) \| \lesssim_P \sqrt{m}
\end{equation}
since $\mathbb U^*(\cdot)$ is a copy of $\mathbb U(\cdot)$ conditional on $(Z_i)_{i=1}^n$. Next, we bound $\sup_{u\in\mathcal U}\|\widehat J^{-1}(u) - J^{-1}(u)\|$. Note that conditions of Theorem \ref{Thm:MainULAfeasible} imply conditions of Lemma \ref{covariance}, and so
$$
\sup_{u\in\mathcal{U}}\|\hat J(u) - J(u)\| \lesssim_P \sqrt{\frac{\zeta_m^2 m\log n}{n\hn}} + m^{-\kappa} + \hn = o(n^{-\varepsilon'}/\sqrt m)
$$
for some $\varepsilon' > 0$. Hence, with probability approaching one, all eigenvalues of $\hat J (u)$ are bounded below from zero uniformly over $u\in\mathcal U$, and so the matrix identity $A^{-1} - B^{-1} = B^{-1}(B-A)A^{-1}$ implies that
$$
\|J^{-1}(u)-\hat J^{-1}(u)\|  =  \|\hat J^{-1}(u)\|\cdot\|\hat J(u) - J(u)\|\cdot \|J^{-1}(u)\| = o_P(n^{-\varepsilon'}/\sqrt m)
$$
uniformly over $u\in \mathcal{U}$. Combining this bound with (\ref{EqM01}) gives the first asserted claim.

Note also that the results continue to hold in $P$-probability if we replace $P$ by $P^*$, which is the second asserted claim, since if a random variable $B_n  = O_P(1)$,
then $B_n = O_{P^*}(1)$. Indeed, the first relation means that $P(|B_n| > \ell_n)= o(1)$ for any $\ell_n \to \infty$, while the second means that $P^*(|B_n| > \ell_n)= o_{P}(1)$ for any $\ell_n \to \infty$. But the second follows from the first from the Markov inequality, observing that $E[P^*(|B_n| > \ell_n)] =  P(|B_n| > \ell_n)= o(1)$. This completes the proof of the theorem.
\end{proof}

\begin{proof}[Proof of Theorem \ref{Thm:MainULAstar}.] The proof is similar to the proof of Theorem \ref{Thm:MainULA} but it applies Lemma \ref{Lemma:ULA} twice, one to the unperturbed problem and one to the perturbed problem with $\mathcal{A}_n (u) = -\mathbb{U}^*(u)/\sqrt{n}$, for every $u\in \mathcal{U}$.

Since conditions of Theorem \ref{Thm:MainULAstar} imply conditions of Theorem \ref{Thm:SeriesRates}, we have by Theorem \ref{Thm:SeriesRates} that $\sup_{u\in\mathcal U}\|\widehat\beta(u) - \beta(u)\|\lesssim_P \eta_n = \sqrt{m/n}$. Similarly, using Lemma \ref{Lemma:Rates} like in the proof of Theorem \ref{Thm:SeriesRates}, we obtain $\sup_{u\in\mathcal U}\|\widehat\beta^*(u) - \beta(u)\|\lesssim_P \eta_n$, since
$$
\{\mathbb{U}^*(u)\}_{u\in\mathcal{U}}=_d \{\mathbb{U}(u)\}_{u\in\mathcal{U}}\lesssim_P \sqrt m
$$
by Lemma \ref{Lemma:error-rate-0}. Then, by applying Lemma \ref{Lemma:ULA} twice, we have for all $u\in\mathcal U$ that
\begin{align*}
\sqrt{n} J(u)\left(\hat \beta^*(u) - \hat \beta(u)\right) & = \sqrt{n} J(u)\left(\hat \beta^*(u) -  \beta(u)\right) - \sqrt{n} J(u)\left(\hat \beta(u) - \beta(u)\right) \\
& =  \mathbb{U}^*(u)
+ r_n^{pert}(u) - r_n^{unpert}(u),
\end{align*}
where $r_n^{pert}(u)$ and $r_n^{unpert}(u)$ are defined by \eqref{Eq:ULA} with $\mathcal A_n(u) = -\mathbb U^*(u)/\sqrt n$ and $\mathcal A_n(u) = 0$, respectively, and the same arguments as those used in the proof of Theorem \ref{Thm:MainULA} show that
$$
\sup_{u \in \mathcal{U}} \Big(\| r_n^{pert}(u)\|+\|r_n^{unpert}(u)\| \Big) \lesssim_P \sqrt{m \zeta_m \eta_n \log n} + \frac{m \zeta_m \log n}{\sqrt{n}} + m^{-\kappa}\sqrt{m}.
$$
The first asserted claim now follows by substituting $\eta_n = \sqrt{m/n}$ and using the growth conditions $m^3\zeta_m^2=o(n^{1-\varepsilon})$ and $m^{-\kappa+1/2}=o(n^{-\varepsilon})$. The second asserted claim follows from the same argument as that used in the proof of Theorem \ref{Thm:MainULAfeasible}. This completes the proof of the theorem.
\end{proof}

\begin{proof}[Proof of Theorem \ref{theorem: strong}]
The proof relies on the following lemma:
\begin{lemma}\label{lem: approximation of u by g}
Suppose that Condition S holds. In addition, suppose that $m^7 \zeta_m^6 = o(n^{1 - \varepsilon})$ for some constant $\varepsilon > 0$. Then there exists a process $G = G_n$ such that
\begin{equation}\label{eq: strong gaussian approximation of u process thm 5}
\sup_{u\in\mathcal U}\|\mathbb U(u) - G(u)\| \lesssim_P o(n^{-\varepsilon'})
\end{equation}
and the process $G$ is conditionally on $(Z_i)_{i=1}^n$ zero-mean Gaussian with a.s. continuous sample paths and the covariance function
$$
E\Big[ G(u_1)G(u_2)'\mid (Z_i)_{i=1}^n\Big] = \En[Z_i Z_i'](u_1\wedge u_2 - u_1 u_2), \ \text{for all $u_1$ and $u_2$ in $\mathcal U$.}
$$
\end{lemma}
\begin{proof}[Proof of Lemma \ref{lem: approximation of u by g}]
Since under our conditions, we have $\zeta_m^2 \log n / n =o(1)$ and $\lambda_{\max} = \sup_{\alpha\in S^{m-1}}E[(Z'\alpha)^2]\lesssim 1$, it follows from Corollary \ref{Corollary:phin} in Appendix \ref{App:EmpiricalProcess} that there exists a constant $C>0$ such that $\sup_{\alpha\in S^{m-1}} \En[(\alpha'Z_i)^2] \leq C$ with probability approaching one. Thus, the asserted claim of the lemma follows from applying Lemma \ref{lemma: strong} conditional on $(Z_i)_{i=1}^n$ on the event $\sup_{\alpha\in S^{m-1}} \En[(\alpha'Z_i)^2] \leq C$.
\end{proof}
Getting back to the proof of Theorem \ref{theorem: strong}, let $G$ be a process constructed in Lemma \ref{lem: approximation of u by g}. Then
\begin{align*}
\sup_{u \in \mathcal{U}} \| & \sqrt{n} (\hat \beta(u) - \beta(u))  - J^{-1}(u) \Z(u) \| \\
 & \leq  \sup_{u \in \mathcal{U}} \| J^{-1}(u) \mathbb{U}(u) - J^{-1}(u) \Z(u) \| + o_P(n^{-\varepsilon'}) \\
& \leq \sup_{u \in \mathcal{U}} \| J^{-1}(u) \| \cdot  \sup_{u \in \mathcal{U}}  \| \mathbb{U}(u) - \Z(u) \| + o_P(n^{-\varepsilon'}) = o_P(n^{-\varepsilon'}),
\end{align*}
where in the second line we invoked Theorem \ref{Thm:MainULA} and in the third line we used \eqref{eq: strong gaussian approximation of u process thm 5} and the fact that all eigenvalues of $J(u)$ are bounded below from zero uniformly over $u\in\mathcal U$. The asserted claim follows.
\end{proof}

\begin{proof}[Proof of Theorem \ref{thm: gaussian method}]
The proof relies on the following lemma:
\begin{lemma}\label{eq: upper bound for g* process}
Suppose that Condition S holds. In addition, suppose that $\zeta_m^2\log n = o(n)$. Then
$$
\sup_{u\in\mathcal U}\|G^*(u)\| \lesssim_P \sqrt m.
$$
\end{lemma}
\begin{proof}[Proof of Lemma \ref{eq: upper bound for g* process}]
Recall that the process $G^*(\cdot)$ is given by
$$
G^*(u) = \hat\Sigma^{1/2} B_m(u),\quad u\in\mathcal U,
$$
where $B_m(\cdot)$ is an $m$-dimensional vector of independent Brownian bridges $B_{m,j}(\cdot)$, $j=1,\dots,m$. In addition,
\begin{align*}
E\Big[\sup_{u\in \mathcal U}\|B_m(u)\|\Big]
& = E\Big[\sup_{u\in\mathcal U} \Big(\textstyle{\sum_{j=1}^m} B_{m,j}(u)^2\Big)^{1/2}\Big]\\
& \leq E\Big[\Big(\textstyle{\sum_{j=1}^m}\sup_{u\in\mathcal U}B_{m,j}(u)^2\Big)^{1/2}\Big]\\
& \leq \Big(\textstyle{\sum_{i=1}^m}E\Big[\sup_{u\in\mathcal U} B_{m,j}(u)^2\Big]\Big)^{1/2}\lesssim \sqrt m.
\end{align*}
Moreover,
$$
\|\hat\Sigma^{1/2}\| = \|\hat\Sigma\|^{1/2} \leq \Big(\|\Sigma\| + \|\hat\Sigma - \Sigma\|\Big)^{1/2} \lesssim_P 1,
$$
by Condition S and Lemma \ref{covariance}. Hence,
$$
\sup_{u\in\mathcal U}\|G^*(u)\| = \sup_{u\in\mathcal U}\left\|\hat \Sigma^{1/2} B_{m}(u)\right\| \lesssim_P \sup_{u\in\mathcal U}\|B_m(u)\| \lesssim_P \sqrt m.
$$
This completes the proof of the lemma.
\end{proof}
Getting back to the proof of Theorem \ref{thm: gaussian method}, note that by Lemma \ref{covariance},
$$
\sup_{u\in\mathcal U}\|J(u) - \widehat J(u)\|\lesssim_P \sqrt{\frac{\zeta_m^2 m \log n}{n h}} + m^{-\kappa} + h = o(n^{-\varepsilon'}/\sqrt m),
$$
and so
$$
\sup_{u\in\mathcal U}\|J^{-1}(u) - \widehat J^{-1}(u)\|\lesssim_P o(n^{-\varepsilon'}/\sqrt m)
$$
like in the proof of Theorem \ref{Thm:MainULAfeasible}. The first asserted claim follows from combining this inequality with the bound
\begin{align*}
\sup_{u\in\mathcal U}\|r(u)\|
&= \sup_{u\in\mathcal U}\|\widehat J^{-1}(u)G^*(u) - J^{-1}(u)G^*(u)\| \\
& \leq \sup_{u\in\mathcal U}\|\widehat J^{-1}(u)- J^{-1}(u)\| \cdot \sup_{u\in\mathcal U}\|G^*(u)\|
\end{align*}
and using Lemma \ref{eq: upper bound for g* process}.
The second asserted claim follows from the same argument as that used in the proof of Theorem \ref{Thm:MainULAfeasible}. This completes the proof of theorem.
\end{proof}

\begin{proof}[Proof of Theorem \ref{Thm:MainBootstrap}]
The proof relies on the following lemma:
\begin{lemma}\label{lem: gaussian approx weighted bootstrap process thm 7}
Suppose that Condition S holds. In addition, suppose that $m^3 \zeta_m^2=o(n^{1-\varepsilon})$ and $m^{-\kappa+1} = o(n^{-\varepsilon})$ for some constant $\varepsilon>0$. Moreover, suppose that the random variable $\pi$ is non-negative and satisfies $E[\pi] = 1$ and $E[\pi^4]\lesssim 1$. Finally, suppose that $\max_{1\leq i\leq n}\pi_i\lesssim_P \log n$. Then
$$
\sqrt{n} \left(\hat \beta^b(u) -
\hat \beta(u)\right) = \frac{J^{-1}(u)}{\sqrt{n}} \sum_{i=1}^n
(\pi_i-1)Z_i(u-1\{U_i \leq u\}) + r(u),
$$
where
$$
 \sup_{u \in \mathcal{U}} \| r(u) \| \lesssim_P  \frac{m^{3/4} \zeta_m^{1/2} \log n}{n^{1/4}} + \sqrt{m^{1-\kappa}\log n} = o(n^{-\varepsilon'})
 $$
 for some $\varepsilon'>0$.
\end{lemma}

\begin{proof}[Proof of Lemma \ref{lem: gaussian approx weighted bootstrap process thm 7}]
First, note that it follows from the first-order conditions of the optimization problem \eqref{eq: weighted bootstrap problem} that $\hat \beta^b(u)$ solves the quantile regression problem \eqref{eq: empirical analog problem} for the rescaled data $(\pi_iY_i,\pi_iZ_i)_{i=1}^n$. Second, recall that the non-negative random variable $\pi$ is such that $E[\pi] = 1$ and $E[\pi^4] \lesssim 1$ and the sequence $(\pi_i)_{i=1}^n$, which consists of i.i.d. random variables with the distribution of $\pi$, is independent of the data and satisfies $\max_{1\leq i\leq n}\pi_i \lesssim_P \log n$. Using these observations, we can follow the proof of Theorem \ref{Thm:MainULA} with $(Z_i,Y_i)_{i=1}^n$ replaced by $(\pi_i Z_i,\pi_iY_i)_{i=1}^n$, so that
$$
\psi_i(\beta,u) = \pi_i Z_i\Big(1\{\pi_i Y_i \leq \pi_i Z_i'\beta\} - u\Big) = \pi_i Z_i\Big(1\{Y_i \leq Z_i'\beta\} - u\Big),
$$
but keeping the same vectors $\beta(u)$ and matrices $J(u)$ to show that
$$
\sqrt n\Big(\hat\beta^b(u) - \beta(u)\Big) = \frac{J^{-1}(u)}{\sqrt n}\sum_{i=1}^n \pi_i Z_i \Big(u - 1\{Y_i \leq Z_i'\beta(u)\}\Big) + r^b(u),
$$
where
$$
\sup_{u\in\mathcal U}\|r^b(u)\| \lesssim_P  \frac{m^{3/4} \zeta_m^{1/2} \log n}{n^{1/4}} + \sqrt{m^{1-\kappa}\log n} = o(n^{-\varepsilon'})
$$
for some $\varepsilon'>0$, where in the proof we replace all applications of the third maximal inequality in Lemma \ref{Lemma:MaxIneq2} by applications of the second maximal inequality in the same lemma, and we also replace $\zeta_m$ by $\zeta_m \log n$, so that $\max_{1\leq i\leq n}\pi_i \|Z_i\| \lesssim_P \zeta_m\log n$. The asserted claim of the lemma follows by combining this result with Theorem \ref{Thm:MainULA}.
\end{proof}

Getting back to the proof of Theorem \ref{Thm:MainBootstrap}, we apply Lemma \ref{lemma: strong} with $v_i = \pi_i-1$, so that $E[v_i] = 0$, $E[v_i^2]=1$, $E[|v_i|^4]\lesssim 1$, and $\max_{1\leq i\leq n} |v_i| \lesssim_P \log n$. The lemma implies that there is a Gaussian process $\Z^*(\cdot) = \Z^*_n(\cdot)$ with the covariance structure $\En[Z_iZ_i'](u\wedge u'-uu')$ such that
$$ \sup_{u \in \mathcal{U}}\left\| \frac{1}{\sqrt{n}}\sum_{i=1}^n(\pi_i-1)Z_i\Big(u-1\{U_i\leq u\}\Big) - \Z^*(u)   \right\| \lesssim_P o(n^{-\varepsilon'})$$
for some $\varepsilon'>0$. Combining this result with Lemma \ref{lem: gaussian approx weighted bootstrap process thm 7} gives the first asserted claim. The second asserted claim follows from the same argument as that used in the proof of Theorem \ref{Thm:MainULAfeasible}. This completes the proof of theorem.
\end{proof}




\section{Proofs of Theorems \ref{theorem: pointwise rate} and \ref{Thm:DistributionsInferentialpointwise}}

\begin{proof}[Proof of Theorem \ref{theorem: pointwise rate}] By Theorem \ref{Thm:MainULA} and Condition P,
\begin{align*}
| \widehat \theta(u,w) -  \theta(u,w)| &\leq | \ell(w)'(\widehat \beta(u) -  \beta(u))| + |r(u,w)| \\
& \leq \frac{|\ell(w)'J^{-1}(u)\mathbb{U}(u)|}{\sqrt{n}} + o_P\Big(\frac{\|\ell(w)\|}{\sqrt{n}}\Big) + o\Big(\frac{\|\ell(w)\|}{\sqrt{n}}\Big).
\end{align*}
In addition,
$$
E\Big[|\ell(w)'J^{-1}(u)\mathbb{U}(u)|^2\Big]\lesssim \|\ell(w)\|^2\|J^{-1}(u)\|^2\sup_{\alpha\in S^{m-1}}E[(\alpha'Z)^2] \lesssim \|\ell(w)\|^2
$$
by Condition S. Combining these bounds gives the asserted claim.
\end{proof}

\begin{proof}[Proof of Theorem \ref{Thm:DistributionsInferentialpointwise}]
By  Theorem \ref{Thm:MainULA} and Condition P,
\begin{align}
t(u,w) &= \frac{\ell(w)'(\widehat \beta(u) - \beta(u))}{\hat \sigma(u,w)} - \frac{r(u,w)}{\hat\sigma(u,w)}\nonumber\\
& = \frac{\ell(w)'J^{-1}(u)\mathbb{U}(u)}{\sqrt{n}\hat\sigma(u,w)} + o_P\Big(\frac{\|\ell(w)\|}{\sqrt{n}\hat\sigma(u,w)}\Big).\label{eq: pointwise asymptotic normality proof}
\end{align}
In addition,
$$
\hat\sigma(u,w) = (1+o_P(1))\sigma(u,w)
$$
by Lemma \ref{covariance} and Condition S. Moreover,  $\sigma(u,w)\gtrsim \|\ell(w)\|/\sqrt n$ under Condition S, and so the second term in \eqref{eq: pointwise asymptotic normality proof} is $o_P(1)$. Hence, the asserted claim follows provided we can show that
$$
\frac{\ell(w)'J^{-1}(u)\mathbb{U}(u)}{\sqrt{n}\sigma(u,w)} \to_d N(0,1).
$$
However, this follows from the Lindeberg-Feller central limit theorem (see Theorem 9.6.1 in Dudley \cite{D04}) since for any $\epsilon > 0$, there exists a constant $C$ such that
\begin{align*}
&E\left[\frac{(\ell(w)'J^{-1}(u)Z(X))^2}{n \sigma^2(u,w)}1\left\{\frac{|\ell(w)'J^{-1}(u)Z(X)|}{n \sigma(u,w)} > \epsilon\right\}\right]\\
&\qquad \leq E\left[\frac{(\ell(w)'J^{-1}(u)Z(X))^2}{n \sigma^2(u,w)}1\left\{\zeta_m > C \sqrt n \epsilon\right\}\right] \to 0
\end{align*}
as $n\to\infty$ under our conditions. This completes the proof of the theorem.
\end{proof}

\section{Proofs of Theorems \ref{theorem: uniform rate}-\ref{theorem: inference using couplings}}

\begin{lemma}[Entropy Bound]\label{Lemma:Entropy}
Suppose that Conditions S and U hold. Consider the class of functions
$$
\mathcal{L}_m = \Big\{ Z\mapsto (\ell(w)/\|\ell(w)\|)'J^{-1}(u)Z  \colon  (u,w)\in I\Big\},
$$
mapping $B_m(0,\zeta_m)$ into $\mathbb R$, and let $L  = 1\vee \sup_{f \in \mathcal{L}_m} |f|$ denote its envelope. Then the uniform entropy numbers of $\mathcal L_m$ satisfy
$$
\sup_{Q} \log N(\epsilon\|L\|_{Q,2}, \mathcal{L}_m, L_2(Q)) \lesssim \log (n/\epsilon),\quad\text{uniformly over $0<\epsilon\leq 1$}.
$$
\end{lemma}
\begin{proof}
For $w\in\mathcal W$, denote $\xi(w) = \ell(w)/\|\ell(w)\|$. By Condition S,
$$
\sup_{u\in\mathcal U}\|J^{-1}(u)\| \lesssim 1
$$
and by Lemma \ref{Lemma:Auxiliary Matrix},
$$
\|J(u) - J(\tilde u)\|\lesssim |u - \tilde u|,\quad\text{uniformly over $u$ and $\tilde u$ in $I$.}
$$
Hence, uniformly over $(u,w)$ and $(\tilde u,\tilde w)$ in $I$ and $Z\in B(0,\zeta_m)$,
\begin{align*}
\Big|\xi(w)'J^{-1}(u)Z - \xi(\tilde w)'J^{-1}(\tilde u)Z\Big| & =   \Big| ( \xi(w) - \xi(\tilde w))'J^{-1}(u)Z + \xi(\tilde w)'(J^{-1}(u)- J^{-1}(\tilde u))Z\Big| \\
& =  \Big| ( \xi(w) - \xi(\tilde w))'J^{-1}(u)Z\Big|  \\
&  \quad  + \Big|\xi(\tilde w)'J^{-1}(\tilde u) (J(\tilde u)- J(u))J^{-1}(u)Z\Big| \\
& \lesssim  \zeta_{m,\theta}^L\|w -\tilde w\| \cdot \|J^{-1}(u)\|\cdot \|Z\|  \\
&  \quad + \|\xi(\tilde w)\|\cdot \|J^{-1}(\tilde u) \| \cdot |\tilde u- u|  \cdot  \|J^{-1}(u)\|\cdot \|Z\|\\
& \leq \zeta_m(1 + \zeta_{m,\theta}^L)  ( \|w -\tilde w\| +  |\tilde u- u| )
\end{align*}
since $\|\xi(\tilde w)\| = 1$. Hence, the asserted claim follows from a standard argument since $L \geq 1$ and the set $I$ is such that its dimension is independent of $n$ and its diameter is bounded from above uniformly over $n$.
\end{proof}

\begin{proof}[Proof of Theorem \ref{theorem: uniform rate}]
The proof relies on the following lemma:
\begin{lemma}\label{lem: ellju max inequality}
Suppose that Conditions S and U hold. Then
$$
\sup_{(u,w)\in I}\Big| (\ell(w)/\|\ell(w)\|)'J^{-1}(u)\mathbb U(u) \Big| \lesssim_P (\log n)^{1/2}.
$$
\end{lemma}
\begin{proof}[Proof of Lemma \ref{lem: ellju max inequality}]
Consider the class of functions
$$
\mathcal{L}_m = \Big\{ (Z,U)\mapsto (\ell(w)/\|\ell(w)\|)'J^{-1}(u)Z  \colon  (u,w)\in I\Big\},
$$
mapping $B_m(0,\zeta_m)\times[0,1]$ into $\mathbb R$. The function $L = C\zeta_m$ for sufficiently large constant $C$ is its envelope and by Lemma \ref{Lemma:Entropy}, its uniform entropy numbers satisfy
$$
\sup_{Q} \log N(\epsilon\|L\|_{Q,2}, \mathcal{L}_m, L_2(Q)) \lesssim \log (n/\epsilon),\quad\text{uniformly over $0<\epsilon\leq 1$}.
$$
 Also, consider the class of functions
$$
\mathcal{G} = \Big\{(Z,U)\mapsto 1\{ U \leq u \} - u \colon  u \in \mathcal{U} \Big\},
$$
mapping $B_m(0,\zeta_m)\times[0,1]$ into $\mathbb R$. The uniform entropy numbers of $\mathcal G$ satisfy
\begin{equation}\label{eq: entropy class G theorem 10}
\sup_{Q}\log N(\epsilon\|G\|_{Q,2},\mathcal G,L_2(Q)) \lesssim \log(1/\epsilon),\quad\text{uniformly over $0<\epsilon\leq 1$},
\end{equation}
where $G(Z,U) = 1$ is its envelope. Hence, by \ref{Lemma:ProductEntropy}, the uniform entropy numbers of the class of functions $\mathcal L_m \mathcal G$ satisfy
$$
\sup_Q \log N(\epsilon\|L G\|_{Q,2},\mathcal L_m\mathcal G, L_2(Q)) \lesssim \log(n/\epsilon),\quad\text{uniformly over $0<\epsilon\leq 1$}.
$$
In addition, by Condition S,
\begin{align*}
&E\Big[\Big((\ell(w)/\|\ell(w)\|)'J^{-1}(u)Z(X)(1\{U\leq u\} - u)\Big)^2\Big] \\
&\qquad = u(1-u)(\ell(w)/\|\ell(w)\|)'J^{-1}(u)\Sigma J^{-1}(u) (\ell(w)/\|\ell(w)\|) \lesssim 1
\end{align*}
uniformly over $(u,w)\in I$, so that
$$
\sup_{f\in \mathcal L_m\mathcal G} E[f^2] \lesssim 1.
$$
Hence, applying the third maximal inequality of Lemma \ref{Lemma:MaxIneq2} gives
\begin{equation}
\sup_{(u,w) \in I} \Big|(\ell(w)/\|\ell(w)\|)'J^{-1}(u)\mathbb{U}(u)\Big|
\lesssim_P \Big( 1 + \frac{ \zeta_m^2 \log n }{n} \Big)^{1/2}\log^{1/2}n \lesssim (\log n)^{1/2}.\label{LinearLinearApproxError}
\end{equation}
This completes the proof of the lemma.
\end{proof}
Getting back to the proof of the theorem, we have by the triangle inequality that
$$
\sup_{(u,w)\in I} | \hat \theta(u,w) - \theta(u,w) |  \leq \sup_{(u,w)\in I}| \ell(w)'(\hat \beta(u) - \beta(u))| + \sup_{(u,w)\in I}|r(u,w)|,
$$
where the second term satisfies
$$
\sup_{(u,w)\in I}|r(u,w)| \leq \zeta_{m,\theta} \sup_{(u,w)\in I}\frac{|  r(u,w)|}{\|\ell(w)\|} = o\left(\frac{\zeta_{m,\theta}}{\sqrt n\log n}\right)
$$
by Condition U. Also, by Theorem \ref{Thm:MainULA}, the first term can be bounded uniformly over $(u,w)\in I$ as
 \begin{align*}
 | \ell(w)'(\hat \beta(u) - \beta(u))|
 & \lesssim_P \frac{|\ell(w)'J^{-1}(u)\mathbb{U}(u)|}{\sqrt{n}} + o_P\Big(\frac{\zeta_{m,\theta}}{\sqrt{n}\log n}\Big) \\
 & \leq \Big|(\ell(w)/\|\ell(w)\|)'J^{-1}(u)\mathbb{U}(u)\Big|\cdot \sqrt{\frac{\zeta_{m,\theta}^2}{n}} + o_P\Big(\frac{\zeta_{m,\theta}}{\sqrt{n}\log n}\Big)\\
 &\lesssim_P \sqrt{\frac{\zeta_{m,\theta}^2\log n}{n}} + o_P\Big(\frac{\zeta_{m,\theta}}{\sqrt{n}\log n}\Big),
\end{align*}
where the second line follows from $\|\ell(w)\| \leq \zeta_{m,\theta}$, holding by Condition U, and the third line from Lemma \ref{lem: ellju max inequality}. The asserted claim follows.
\end{proof}

\begin{proof}[Proof of Theorem \ref{thm: couplings for t process}]
By Lemma \ref{covariance},
$$
\|\widehat\Sigma - \Sigma\| \lesssim_P \sqrt{\frac{\zeta_m^2\log n}{n}} = o(n^{-\varepsilon'})
$$
and
$$
\sup_{u\in\mathcal U}\|\widehat J(u) - J(u)\|
\lesssim_P \sqrt{\frac{m \zeta_m^2 \log n}{n h}} + m^{-\kappa} + h = o(n^{-\varepsilon'}).
$$
Combining these inequalities with Condition S gives
\begin{equation}\label{eq: sigma ratio bound thm 11 new}
\sup_{(u,w)\in I}\left|\frac{\hat\sigma(u,w)}{\sigma(u,w)} - 1\right| = o_P(n^{-\varepsilon'}).
\end{equation}
Further, by Condition S,
\begin{equation}\label{eq: sigma lower bound 0 thm 11 new}
\sigma(u,w)\gtrsim \|\ell(w)\|/\sqrt{n},\quad\text{uniformly over $(u,w)\in I$,}
\end{equation}
and so
\begin{equation}\label{eq: sigma lower bound thm 11 new}
\hat\sigma(u,w) \gtrsim_P \|\ell(w)\|/\sqrt{n},\quad\text{uniformly over $(u,w)\in I$.}
\end{equation}
Next, by Condition U and \eqref{eq: sigma lower bound thm 11 new}, we have uniformly in $(u,w)\in I$ that
 \begin{equation}\label{ZeroSup new}
 t(u,w) = \frac{\hat \theta(u,w)-\theta(u,w)}{\hat \sigma(u,w)} = \frac{\ell(w)'(\hat\beta(u)-\beta(u))}{\hat \sigma(u,w)} +o_P(n^{-\varepsilon'}).
 \end{equation}
Also, by Theorem \ref{Thm:MainULA} and \eqref{eq: sigma lower bound thm 11 new}, we have uniformly over $(u,w)\in I$ that
 \begin{equation}\label{FirstSup new}
 \frac{\ell(w)'(\hat\beta(u)-\beta(u))}{\hat \sigma(u,w)} = \frac{\ell(w)'J^{-1}(u)\mathbb{U}(u)/\sqrt{n}}{\hat \sigma(u,w)} + o_P(n^{-\varepsilon'}).
 \end{equation}
Moreover, by Lemma \ref{lem: ellju max inequality},
\begin{equation}\label{eq: process upper bound thm 11 new}
\sup_{(u,w)\in I} \Big|(\ell(w)/\|\ell(w)\|)'J^{-1}(u)\mathbb{U}(u)\Big| \lesssim_P (\log n)^{1/2}.
\end{equation}
Hence,
\begin{align}
&\sup_{(u,w)\in I}\Big|\frac{\ell(w)'J^{-1}\mathbb U(u)/\sqrt n}{\widehat \sigma(u,w)} - \frac{\ell(w)'J^{-1}\mathbb U(u)/\sqrt n}{\sigma(u,w)}\Big|\notag \\
&\qquad \leq \sup_{(u,w)\in I}\Big| \frac{\ell(w)'J^{-1}\mathbb U(u)/\sqrt n}{\widehat \sigma(u,w)} \Big|\cdot\sup_{(u,w)\in I}\Big| 1 - \frac{\widehat \sigma(u,w)}{\sigma(u,w)}  \Big| = o_P(n^{-\varepsilon'})\label{eq: sigma approximation thm 11 new}
\end{align}
by \eqref{eq: sigma ratio bound thm 11 new} and \eqref{eq: sigma lower bound thm 11 new}.
Combining \eqref{ZeroSup new}, \eqref{FirstSup new}, and \eqref{eq: sigma approximation thm 11 new}, we have uniformly over $(u,w)\in I$ that
\begin{equation}\label{eq: t process linearization main new}
t(u,w) = \frac{\ell(w)'J^{-1}(u)\mathbb U(u)/\sqrt n}{\sigma(u,w)} + o_P(n^{-\varepsilon'}).
\end{equation}
Hence, \eqref{eq: pivotal coupling t process} follows.

To prove \eqref{eq: gaussian coupling t process}, note that by Lemma \ref{lem: approximation of u by g},
$$
\sup_{u\in\mathcal U}\|\mathbb U(u) - G(u)\|\lesssim_P o(n^{-\varepsilon'});
$$
see \eqref{eq: strong gaussian approximation of u process thm 5}. Combining this inequality \eqref{eq: sigma lower bound 0 thm 11 new} and \eqref{eq: t process linearization main new}, we have uniformly over $(u,w)\in I$ that
$$
t(u,w) = \frac{\ell(w)'J^{-1}(u) \mathbb U(u)/\sqrt{n}}{\sigma(u,w)} + o_P(n^{-\varepsilon'}) = \frac{\ell(w)'J^{-1}(u) G(u)/\sqrt{n}}{\sigma(u,w)} + o_P(n^{-\varepsilon'}) \label{eq: final coupling thm 12}.
$$
Hence, \eqref{eq: gaussian coupling t process} follows. This complete the proof of the theorem.
\end{proof}

\begin{proof}[Proof of Theorem \ref{thm: gaussian coupling for t process}]
The proof relies on the following lemma:
\begin{lemma}\label{lem: lipshitz property ell functional}
Suppose that Conditions S and U hold. In addition, define
$$
\ell(u,w) = \left(\frac{\ell(w)'J^{-1}(u)/\sqrt n}{\sigma(u,w)}\right)',\quad (u,w)\in I.
$$
Then
$$
\sup_{(u,w)\in I}\|\ell(u,w)\|\lesssim 1
$$
and
$$
\|\ell(u,w) - \ell(\tilde u,\tilde w)\|\lesssim (1 + \zeta_{m,\theta}^L)\|(u,w) - (\tilde u,\tilde w)\|
$$
uniformly over $(u,w)$ and $(\tilde u,\tilde w)$ in $I$.
\end{lemma}
\begin{proof}[Proof of Lemma \ref{lem: lipshitz property ell functional}]
By Condition S,
\begin{equation}\label{eq: lower and upper bound for sigma}
\sigma(u,w)\gtrsim \|\ell(w)\|/\sqrt n,\quad\text{uniformly over $(u,w)\in I$}
\end{equation}
The first asserted claim follows from this inequality and the fact that all eigenvalues of the matrix $J(u)$ are bounded below from zero uniformly over $u\in\mathcal U$ by Condition S.

To prove the second asserted claim, note that by Lemma \ref{Lemma:Auxiliary Matrix},
$$
\|J(u) - J(\tilde u)\| \lesssim |u - \tilde u|,\quad\text{uniformly over $u$ and $\tilde u$ in $\mathcal U$.}
$$
Combining this inequality with Conditions S and U, it follows from the triangle inequality that
$$
\left|\frac{n\sigma^2(u,w)}{\|\ell(w)\|^2} - \frac{n\sigma^2(\tilde u,\tilde w)}{\|\ell(\tilde w)\|^2}\right|\lesssim (1 + \zeta_{m,\theta}^L)\|(u,w) - (\tilde u,\tilde w)\|
$$
uniformly over $(u,w)$ and $(\tilde u,\tilde w)$ in $I$. In turn, this inequality in combination with \eqref{eq: lower and upper bound for sigma} imply that
$$
\left|\frac{\sqrt n\sigma(u,w)}{\|\ell(w)\|} - \frac{\sqrt n\sigma(\tilde u,\tilde w)}{\|\ell(\tilde w)\|}\right|\lesssim (1 + \zeta_{m,\theta}^L)\|(u,w) - (\tilde u,\tilde w)\|
$$
uniformly over $(u,w)$ and $(\tilde u,\tilde w)$ in $I$ since for any $a,b > 0$, we have
$$
|a - b| = \frac{|a^2 - b^2|}{a + b}.
$$
Further, combining the inequality above and \eqref{eq: lower and upper bound for sigma} gives
\begin{equation}\label{eq: inverse lipshitz}
\left|\frac{\|\ell(w)\|}{\sqrt n\sigma(u,w)} - \frac{\|\ell(\tilde w)\|}{\sqrt n\sigma(\tilde u,\tilde w)}\right| \lesssim (1 + \zeta_{m,\theta}^L)\|(u,w) - (\tilde u,\tilde w)\|
\end{equation}
uniformly over $(u,w)$ and $(\tilde u,\tilde w)$ in $I$ since for any $a,b > 0$, we have
$$
\left|\frac{1}{a} - \frac{1}{b}\right| = \frac{|a - b|}{a b}.
$$
Similarly,
\begin{equation}\label{eq: numerator lipshitz}
\Big\| (\ell(w)/\|\ell(w)\|)'J^{-1}(u) - (\ell(\tilde w)/\|\ell(\tilde w)\|)'J^{-1}(\tilde u) \Big\| \lesssim (1 + \zeta_{m,\theta}^L)\|(u,w) - (\tilde u,\tilde w)\|
\end{equation}
uniformly over $(u,w)$ and $(\tilde u,\tilde w)$. The second asserted claim follows from combining inequalities \eqref{eq: lower and upper bound for sigma}, \eqref{eq: inverse lipshitz}, and \eqref{eq: numerator lipshitz}. This completes the proof of the lemma.
\end{proof}
Getting back to the proof of the theorem, like in the proof of Theorem \ref{thm: couplings for t process}, we have uniformly over $(u,w)\in I$ that
\begin{equation}\label{eq: t process linearization main new 1}
t(u,w) = \frac{\ell(w)'J^{-1}(u)\mathbb U(u)/\sqrt n}{\sigma(u,w)} + o_P(n^{-\varepsilon'});
\end{equation}
see \eqref{eq: t process linearization main new}. Next, it follows from Corollary \ref{Corollary:phin} in Appendix \ref{App:EmpiricalProcess} and Condition S that there exists a constant $C>0$ such that $\sup_{\alpha\in S^{m-1}} \En[(\alpha'Z_i)^2] \leq C$ with probability approaching one. Thus, the asserted claim follows from using \eqref{eq: t process linearization main new 1} and applying Lemma \ref{lemma:LinearStrong} conditional on $(Z_i)_{i=1}^n$ on the event $\sup_{\alpha\in S^{m-1}} \En[(\alpha'Z_i)^2] \leq C$ with
$$
\ell(u,w) = \left(\frac{\ell(w)'J^{-1}(u)/\sqrt n}{\sigma(u,w)}\right)',\quad (u,w)\in I;
$$
note that the conditions of Lemma \ref{lemma:LinearStrong} follow from Lemma \ref{lem: lipshitz property ell functional} and the assumption $(1 + \zeta_{m,\theta}^L)^{2 d_I}\zeta_m^2 = o(n^{1 - \varepsilon})$. This completes the proof of the theorem.
\end{proof}

\begin{proof}[Proof of Theorem \ref{thm: resampling methods t process}]
The proof relies on the following lemma:
\begin{lemma}\label{lem: maximal inequality gaussian process functional thm 13}
Suppose that Conditions S and U hold. Then there exists a constant $C_V>0$ such that with probability $1 - o(1)$,
\begin{equation}\label{eq: conditional maximal inequality gaussian process}
E\left[ \sup_{(u,w)\in I}\Big|(\ell(w)/\|\ell(w)\|)'J^{-1}(u)G(u)\Big| \mid (Z_i)_{i=1}^n\right] \leq C_V\sqrt{\log n}.
\end{equation}
In addition,
$$
\sup_{(u,w)\in I}\Big|(\ell(w)/\|\ell(w)\|)'J^{-1}(u)G(u)\Big| \lesssim_P \sqrt{\log n}.
$$
\end{lemma}
\begin{proof}[Proof of Lemma \ref{lem: maximal inequality gaussian process functional thm 13}]
The second asserted claim follows immediately from the first one. Hence, it suffices to prove the first asserted claim. To do so, note that by Corollary \ref{Corollary:phin} in Appendix \ref{App:EmpiricalProcess}, there exists a constant $C>0$ such that $\sup_{\alpha\in S^{m-1}} \En[(\alpha'Z_i)^2] \leq C$ with probability $1 - o(1)$. We will show that \eqref{eq: conditional maximal inequality gaussian process} holds on the event $\sup_{\alpha\in S^{m-1}} \En[(\alpha'Z_i)^2] \leq C$.

For all $w\in\mathcal W$, denote $\xi(w) = \ell(w)/\|\ell(w)\|$.
Observe that uniformly over $(u,w)\in I$, on the event $\sup_{\alpha\in S^{m-1}}\En[(\alpha' Z_i)^2] \leq C$,
\begin{align*}
E\Big[ |\xi(w)'J^{-1}(u)G(u)|^2 \mid(Z_i)_{i=1}^n\Big]
& = u(1-u)\xi(w)'J^{-1}(u)\hat\Sigma J^{-1}(u)\xi(w) \leq C_r
\end{align*}
for some constant $C_r$ by Condition S. Moreover, by Conditions S and U and Lemma \ref{Lemma:Auxiliary Matrix},
$$
\sup_{(u,w)\in I}\|J^{-1}(u) \xi(w)\| \lesssim 1
$$
and
\begin{align*}
\Big\| J^{-1}(u)\xi(w) - J^{-1}(\tilde u)\xi(\tilde w) \Big\|
&\leq \|J^{-1}(u) - J^{-1}(\tilde u)\|\cdot\|\xi(w)\| + \|J^{-1}(\tilde u)\|\cdot\|\xi(w) - \xi(\tilde w)\|\\
& \lesssim (1 + \zeta_{m,\theta}^L)\|(u,w) - (\tilde u,\tilde w)\|
\end{align*}
uniformly over $(u,w)$ and $(\tilde u,\tilde w)$ in $I$. Hence, by the calculation in the proof of Lemma \ref{Lemma:Fact4-b extension}, on the event $\sup_{\alpha\in S^{m-1}}\En[(\alpha' Z_i)^2] \leq C$,
$$
E\Big[ \Big( \xi(w)'J^{-1}(u)G(u) - \xi(\tilde w)'J^{-1}(\tilde u)G(\tilde u)\Big)^2 \mid (Z_i)_{i=1}^n\Big] \leq C_e(1 + \zeta_{m,\theta}^L)\|(u,w) - (\tilde u,\tilde w)\|,
$$
where $C_e$ is some constant. Combining these inequalities with $\log \zeta_{m,\theta}^L\lesssim \log n$ and using the fact that the dimension of the set $I$ is independent of $n$ and its diameter is bounded uniformly over $n$, \eqref{eq: conditional maximal inequality gaussian process} follows from Dudley's inequality applied conditional on $(Z_i)_{i=1}^n$ on the event $\sup_{\alpha\in S^{m-1}}\En[(\alpha' Z_i)^2] \leq C$; see Corollary 2.2.8 in van~der Vaart and Wellner \cite{vdV-W}).
\end{proof}
Getting back to the proof of the theorem, note that under the conditions $h = o(n^{-\varepsilon})$, $m\zeta_m^2 = o(n^{1 - \varepsilon}h)$, and $m^{-\kappa} = o(n^{-\varepsilon})$, like in the proof of Theorem \ref{thm: couplings for t process}, we have
$$
\sup_{(u,w)\in I}\left| \frac{\hat\sigma(u,w)}{\sigma(u,w)} - 1 \right| = o_P(n^{-\varepsilon'})
$$
and
\begin{equation}\label{eq: sigma lower bound again}
\hat\sigma(u,w) \gtrsim_P \|\ell(w)\|/\sqrt n,\quad\text{uniformly over $(u,w)\in I$};
\end{equation}
see \eqref{eq: sigma ratio bound thm 11 new} and \eqref{eq: sigma lower bound thm 11 new}. Also, by Lemma \ref{lem: ellju max inequality},
$$
\sup_{(u,w)\in I}\Big| (\ell(w)/\|\ell(w)\|)' J^{-1}(u)\mathbb U(u) \Big| \lesssim_P (\log n)^{1/2},
$$
and so
$$
\sup_{(u,w)\in I}\left| \frac{\ell(w)' J^{-1}(u)\mathbb U(u)/\sqrt n}{\hat\sigma(u,w)} \right| \lesssim_P (\log n)^{1/2}.
$$
Hence,
\begin{align*}
&\sup_{(u,w)\in I}\Big|\frac{\ell(w)'J^{-1}\mathbb U(u)/\sqrt n}{\widehat \sigma(u,w)} - \frac{\ell(w)'J^{-1}\mathbb U(u)/\sqrt n}{\sigma(u,w)}\Big|\\
&\qquad \leq \sup_{(u,w)\in I}\Big| \frac{\ell(w)'J^{-1}\mathbb U(u)/\sqrt n}{\widehat \sigma(u,w)} \Big|\cdot\sup_{(u,w)\in I}\Big| 1 - \frac{\widehat \sigma(u,w)}{\sigma(u,w)}  \Big| = o_P(n^{-\varepsilon'}).
\end{align*}
Further, since the process $\mathbb U^*(\cdot)$ is a copy of the process $\mathbb U(\cdot)$ conditional on $(Z_i)_{i=1}^n$, we also have
\begin{equation}\label{eq: useful approximation pivotal process thm 13 new}
\sup_{(u,w)\in I}\left|\frac{\ell(w)'J^{-1}\mathbb U^*(u)/\sqrt n}{\widehat \sigma(u,w)} - \frac{\ell(w)'J^{-1}\mathbb U^*(u)/\sqrt n}{\sigma(u,w)}\right| = o_P(n^{-\varepsilon'}).
\end{equation}
Now, consider the case of the pivotal method. We have uniformly over $(u,w)\in I$ that
\begin{align*}
t^*(u,w)
& = \frac{\ell(w)'\hat J^{-1}(u)\mathbb U^*(u)/\sqrt n}{\hat\sigma(u,w)}\\
& = \frac{\ell(w)'J^{-1}(u)\mathbb U^*(u)/\sqrt n}{\hat\sigma(u,w)} + o_P(n^{-\varepsilon'})
   = \frac{\ell(w)'J^{-1}(u)\mathbb U^*(u)/\sqrt n}{\sigma(u,w)} + o_P(n^{-\varepsilon'}),
\end{align*}
where the second equality follows from Theorem \ref{Thm:MainULAfeasible} and \eqref{eq: sigma lower bound again} and the third from \eqref{eq: useful approximation pivotal process thm 13 new}. The first asserted claim for this case follows.

Next, consider the case of the gradient bootstrap method. We have uniformly over $(u,w)\in I$ that
\begin{align*}
t^*(u,w)
& = \frac{\ell(w)'(\hat\beta^*(u) - \hat\beta(u))}{\hat\sigma(u,w)}\\
& = \frac{\ell(w)'J^{-1}(u)\mathbb U^*(u)/\sqrt n}{\hat\sigma(u,w)} + o_P(n^{-\varepsilon'})
   = \frac{\ell(w)'J^{-1}(u)\mathbb U^*(u)/\sqrt n}{\sigma(u,w)} + o_P(n^{-\varepsilon'}),
\end{align*}
where the second equality follows from Theorem \ref{Thm:MainULAstar} and \eqref{eq: sigma lower bound again} and the third from \eqref{eq: useful approximation pivotal process thm 13 new}. The first asserted claim for this case follows.

Next, consider the case of the Gaussian method. By Lemma \ref{lem: maximal inequality gaussian process functional thm 13}, we have
$$
\sup_{(u,w)\in I}\Big|(\ell(w)/\|\ell(w)\|)' J^{-1}(u)G(u)\Big| \lesssim_P (\log n)^{1/2}.
$$
Hence, using the same steps as above, we obtain
\begin{equation}\label{eq: useful approximation gaussian process thm 13 new}
\sup_{(u,w)\in I}\left|\frac{\ell(w)'J^{-1} G^*(u)/\sqrt n}{\widehat \sigma(u,w)} - \frac{\ell(w)'J^{-1} G^*(u)/\sqrt n}{\sigma(u,w)}\right| = o_P(n^{-\varepsilon'}).
\end{equation}
So, we have uniformly over $(u,w)\in I$ that
\begin{align*}
t^*(u,w)
& = \frac{\ell(w)'\hat J^{-1}(u) G^*(u)/\sqrt n}{\hat\sigma(u,w)}\\
& = \frac{\ell(w)'J^{-1}(u)\mathbb G^*(u)/\sqrt n}{\hat\sigma(u,w)} + o_P(n^{-\varepsilon'})
   = \frac{\ell(w)'J^{-1}(u)\mathbb G^*(u)/\sqrt n}{\sigma(u,w)} + o_P(n^{-\varepsilon'}),
\end{align*}
where the second equality follows from Theorem \ref{thm: gaussian method} and \eqref{eq: sigma lower bound again} and the third from \eqref{eq: useful approximation gaussian process thm 13 new}. The first asserted claim for this case follows.

Finally, consider the case of the weighted bootstrap method. We have uniformly over $(u,w)\in I$ that
\begin{align*}
t^*(u,w)
& = \frac{\ell(w)'(\hat\beta^b(u) - \hat\beta(u))}{\hat\sigma(u,w)}\\
& = \frac{\ell(w)'J^{-1}(u)\mathbb G^*(u)/\sqrt n}{\hat\sigma(u,w)} + o_P(n^{-\varepsilon'})
   = \frac{\ell(w)'J^{-1}(u)\mathbb G^*(u)/\sqrt n}{\sigma(u,w)} + o_P(n^{-\varepsilon'}),
\end{align*}
where the second equality follows from Theorem \ref{Thm:MainBootstrap} and \eqref{eq: sigma lower bound again} and the third from \eqref{eq: useful approximation gaussian process thm 13 new}. The first asserted claim for this case follows.

The second asserted claim follows from the same argument as that used in the proof of Theorem \ref{Thm:MainULAfeasible}. This completes the proof of theorem.
\end{proof}

\begin{proof}[Proof of Theorem \ref{thm: weighted bootstrap alternative conditions}]
Like in the proof of Theorem \ref{thm: couplings for t process},
\begin{align}
&\sup_{(u,w)\in I}\Big|\frac{\hat \sigma(u,w)}{\sigma(u,w)} - 1\Big| = o_P(n^{-\varepsilon'}),\label{eq: sigma ratio bound thm 13 - 3}\\
&\hat \sigma(u,w) \gtrsim_P \|\ell(w)\|/\sqrt n,\quad \text{uniformly over }(u,w)\in I;\label{eq: sigma lower bound thm 13 - 3}
\end{align}
see \eqref{eq: sigma ratio bound thm 11 new} and \eqref{eq: sigma lower bound thm 11 new}. Also, by the same argument as that used in Lemma \ref{lem: ellju max inequality},
$$
\sup_{(u,w)\in I}\left|\frac{(\ell(w)/\|\ell(w)\|)' J^{-1}(u)}{\sqrt n}\sum_{i=1}^n (\pi_i - 1)Z_i(u - 1\{U_i \leq u\})\right| \lesssim_P (\log n)^{1/2}.
$$
In addition, by Condition S, $\|\ell(w)\|/(\sqrt n\sigma(u,w))\lesssim 1$ uniformly over $(u,w)\in I$, and so
$$
\sup_{(u,w)\in I}\left|\frac{\ell(w)' J^{-1}(u)}{\sqrt n\sigma(u,w)}\cdot \frac{1}{\sqrt n}\sum_{i=1}^n (\pi_i - 1)Z_i(u - 1\{U_i \leq u\})\right| \lesssim_P (\log n)^{1/2}.
$$
Combining this inequality with \eqref{eq: sigma ratio bound thm 13 - 3} gives
$$
\sup_{(u,w)\in I}\left|\left(\frac{1}{\sigma(u,w)} - \frac{1}{\hat\sigma(u,w)}\right)\frac{\ell(w)' J^{-1}(u)}{\sqrt n} \frac{1}{\sqrt n}\sum_{i=1}^n (\pi_i - 1)Z_i(u - 1\{U_i \leq u\})\right| = o_P(n^{-\varepsilon'}).
$$
Hence, uniformly over $(u,w)\in I$,
\begin{align}
t^*(u,w)
& = \frac{\ell(w)'(\hat \beta^b(u) - \hat\beta(u))}{\hat\sigma(u,w)}\notag\\
& = \frac{\ell(w)'J^{-1}(u)}{\sqrt n \hat\sigma(u,w)}\cdot\frac{1}{\sqrt n}\sum_{i=1}^n(\pi_i - 1)Z_i(u - 1\{U_i \leq u\}) + o_P(n^{-\varepsilon'})\notag\\
& = \frac{\ell(w)'J^{-1}(u)}{\sqrt n \sigma(u,w)}\cdot\frac{1}{\sqrt n}\sum_{i=1}^n(\pi_i - 1)Z_i(u - 1\{U_i \leq u\}) + o_P(n^{-\varepsilon'}),\label{eq: weighted bootstrap t process final approximation thm 14}
\end{align}
where the second line follows from Lemma \ref{lem: gaussian approx weighted bootstrap process thm 7} and \eqref{eq: sigma lower bound thm 13 - 3}. Next, it follows from Corollary \ref{Corollary:phin} in Appendix \ref{App:EmpiricalProcess} and Condition S that there exists a constant $C>0$ such that $\sup_{\alpha\in S^{m-1}} \En[(\alpha'Z_i)^2] \leq C$ with probability approaching one. Thus, the first asserted claim follows from using \eqref{eq: weighted bootstrap t process final approximation thm 14} and applying Lemma \ref{lemma:LinearStrong} conditional on $(Z_i)_{i=1}^n$ on the event $\sup_{\alpha\in S^{m-1}} \En[(\alpha'Z_i)^2] \leq C$ with
$$
\ell(u,w) = \left(\frac{\ell(w)'J^{-1}(u)/\sqrt n}{\sigma(u,w)}\right)',\quad (u,w)\in I;
$$
note that the conditions of Lemma \ref{lemma:LinearStrong} follow from Lemma \ref{lem: lipshitz property ell functional} and the assumption $(1 + \zeta_{m,\theta}^L)^{2 d_I}\zeta_m^2 = o(n^{1 - \varepsilon})$. The second asserted claim follows from the same argument as that used in the proof of Theorem \ref{Thm:MainULAfeasible}. This completes the proof of theorem.
\end{proof}

\begin{proof}[Proof of Theorem \ref{theorem: inference using couplings}.]
We split the proof into five steps.

\noindent
{\bf Step 1} (Coupling for $V$). Here for the random variable
$$
V = \sup_{(u,w)\in I}\frac{|\hat \theta(u,w) - \theta(u,w)|}{\hat\sigma(u,w)} = \sup_{(u,w)\in I}|t(u,w)|
$$
and each sample size $n$, we show that there exists a random variable $\bar V = \bar V_n$ that is conditional on $(Z_i)_{i=1}^n$ equal in distribution to
\begin{equation}\label{eq: main random variable thm 13}
\sup_{(u,w)\in I} \frac{|\ell(w)'J^{-1}(u) G(u)/\sqrt n|}{\sigma(u,w)}
\end{equation}
and is such that
$$
|V - \bar V| = o_P(n^{-\varepsilon'}),
$$
where $G(\cdot) = G_n(\cdot)$ is a process on $\mathcal U$ that is conditionally on $(Z_i)_{i=1}^n$ zero-mean Gaussian with a.s. continuous sample paths and the covariance function
$$
E\Big[G(u_1) G(u_2)'\mid (Z_i)_{i=1}^n\Big] = \En[Z_i Z_i'](u_1\wedge u_2 - u_1 u_2),\quad \text{for all $u_1$ and $u_2$ in $\mathcal U$.}
$$
To prove this claim, first, note that by the first part of Theorem \ref{thm: couplings for t process}, we have
\begin{equation}\label{eq: first approximation thm 13}
|V - \widetilde V| = o_P(n^{-\varepsilon'}),
\end{equation}
where
$$
\widetilde V = \sup_{(u,w)\in I} \frac{|\ell(w)'J^{-1}(u)\mathbb U(u)/\sqrt n|}{\sigma(u,w)}.
$$

Second, consider the function classes
\begin{align*}
\mathcal W_1 &= \Big\{(Z,U)\mapsto u - 1\{U\leq u\}\colon (u,w)\in I\Big\},\\
\mathcal W_2 &= \Big\{(Z,U)\mapsto \|\ell(w)\|/(\sqrt n\sigma(u,w))\colon (u,w)\in I\Big\},\\
\mathcal W_3 &= \Big\{(Z,U)\mapsto (\ell(w)/\|\ell(w)\|)'J^{-1}(u)Z\colon (u,w)\in I\Big\},
\end{align*}
mapping $B_m(0,\zeta_m)\times[0,1]$ into $\mathbb R$. The function class $\mathcal W_1$ has an envelope $F_1(Z,U) = 1$ and its uniform entropy numbers satisfy
$$
\sup_Q \log N(\epsilon \|F_1\|_{Q,2},\mathcal W_1,L_2(Q))\lesssim \log(1/\epsilon),\quad\text{uniformly over }0<\epsilon\leq 1.
$$
Also, by Condition S, $\|\ell(w)\|/(\sqrt n\sigma(u,w))\lesssim 1$ uniformly over $(u,w)\in I$, and so the function class $\mathcal W_2$ has an envelope $F_2(Z,U) = C$ for sufficiently large constant $C$ and its uniform entropy numbers satisfy
$$
\sup_Q \log N(\epsilon \|F_2\|_{Q,2},\mathcal W_2,L_2(Q))\lesssim \log(1/\epsilon),\quad\text{uniformly over }0<\epsilon\leq 1.
$$
In addition, the function class $\mathcal W_3$ has an envelope $F_3(Z,U) = C\zeta_m$ for sufficiently large constant $C$ and by Lemma \ref{Lemma:Entropy}, the uniform entropy numbers of $\mathcal W_3$ satisfy
$$
\sup_Q \log N(\epsilon \|F_3\|_{Q,2},\mathcal W_3,L_2(Q)) \lesssim \log(n/\epsilon),\quad\text{uniformly over }0<\epsilon\leq 1.
$$
Hence, by Lemma \ref{Lemma:ProductEntropy}, the uniform entropy numbers of the function class
$$
\widetilde {\mathcal F} = \Big\{(Z,U)\mapsto \frac{(\ell(w)/\|\ell(w)\|)'J^{-1}(u) Z (u - 1\{U\leq u\})}{\sqrt n\sigma(u,w)/\|\ell(w)\|}\colon (u,w)\in I\Big\} \subset \mathcal W_1\mathcal W_2\mathcal W_3
$$
satisfy
$$
\sup_Q \log N(\epsilon\|F\|_{Q,2},\widetilde {\mathcal F},L_2(Q))\lesssim \log(n/\epsilon),\quad\text{uniformly over }0<\epsilon\leq 1,
$$
where $F(Z,U) = C\zeta_m$ for sufficiently large constant $C$ is its envelope. Thus, the uniform entropy numbers of the function class $\mathcal F = \widetilde{\mathcal F}\cup(-\widetilde{\mathcal F})$ satisfy
\begin{equation}\label{eq: entropy bound thm 13}
\sup_Q \log N(\epsilon\|F\|_{Q,2},\mathcal F,L_2(Q))\lesssim \log(n/\epsilon),\quad\text{uniformly over }0<\epsilon\leq 1,
\end{equation}
where the function $F$ is its envelope. Moreover,
\begin{equation}\label{eq: moment conditions thm 13}
\sup_{f\in\mathcal F}E[f(Z,U)^2] = 1, \ \sup_{f\in\mathcal F}E[|f(Z,U)|^3] \lesssim \zeta_m, \ \sup_{f\in\mathcal F}E[|f(Z,u)|^4] \lesssim \zeta_m^2
\end{equation}
Now, observe that
$$
\widetilde V = \sup_{f\in\mathcal F}\mathbb G_n f
$$
and that by \eqref{eq: entropy bound thm 13} and \eqref{eq: moment conditions thm 13}, Conditions (A), (B), and (C) of Chernozhukov, Chetverikov and Kato \cite{CCK2016} hold for the function class $\mathcal F$ with $A = n$, $v = C$, $\sigma = 1$, $b = C\zeta_m$, and arbitrarily large $q$. Hence, given that $\zeta_m^2 = o(n^{1 - \varepsilon})$, Theorem 2.1 of Chernozhukov, Chetverikov and Kato \cite{CCK2016} implies that there exists a random variable $\bar V$ with the required distribution and such that
$$
|\widetilde V - \bar V| = o_P(n^{-\varepsilon'}).
$$
Combining this bound with \eqref{eq: first approximation thm 13} gives the claim of this step.

\noindent
{\bf Step 2} (Coupling for $V^*$). Here for the random variable
$$
V^* = \sup_{(u,w)\in I}|t^*(u,w)|
$$
and for each sample size $n$, we show that there exists a random variable $\bar V^* = \bar V^*_n$ that (i) is conditional on $(Z_i)_{i=1}^n$ equal in distribution to the random variable in \eqref{eq: main random variable thm 13}, (ii) depends on the data only via $(Z_i)_{i=1}^n$, and (iii) is such that
\begin{equation}\label{eq: bootstrap bands main approximation}
| V^* - \bar V^*| = o_P(n^{-\varepsilon'}).
\end{equation}
To prove this claim, we consider four resampling methods separately. We start with the pivotal method case. By Theorem \ref{thm: resampling methods t process},   we have
\begin{equation}\label{eq: first approximation pivotal thm 13}
|V^* - \widetilde V^*| = o_P(n^{-\varepsilon'}),
\end{equation}
where
$$
\widetilde V^* = \sup_{(u,w)\in I}\frac{|\ell(w)'J^{-1}(u)\mathbb U^*(u)/\sqrt n|}{\sigma(u,w)}.
$$
Then given that the process $\mathbb U^*(\cdot)$ is a copy of $\mathbb U(\cdot)$ conditional on $(Z_i)_{i=1}^n$ and depends on the data only via $(Z_i)_{i=1}^n$, by the same argument as that used in Step 1, it follows that there exists a random variable $\bar V^*$ with the required distribution and such that
$$
|\widetilde V^* - \bar V^*| = o_P(n^{-\varepsilon'}).
$$
Combining this bound with \eqref{eq: first approximation pivotal thm 13}, gives the claim of Step 2 in this case.

Next, in the gradient bootstrap method case, construction of a random variable $\bar V^*$ with the required distribution follows from the same argument as that used in the pivotal process case.

Next, consider the Gaussian method case. By Theorem \ref{thm: resampling methods t process}, we have uniformly over $(u,w)\in I$ that
$$
t^*(u,w) = \frac{\ell(w)'J^{-1}(u) G^*(u)/\sqrt{n}}{\sigma(u,w)} + o_P(n^{-\varepsilon'}).
$$
Thus, it follows that \eqref{eq: bootstrap bands main approximation} holds for
$$
\bar V^* = \sup_{(u,w)\in I}\frac{|\ell(w)'J^{-1}(u)G^*(u)/\sqrt n|}{\sigma(u,w)},
$$
and since the random variable $\bar V^*$ has the required distribution, the claim of Step 2 in this case holds.

Finally, consider the weighted bootstrap method case.
Like in the proof of Theorem \ref{thm: couplings for t process},
\begin{align}
&\sup_{(u,w)\in I}\Big|\frac{\hat \sigma(u,w)}{\sigma(u,w)} - 1\Big| = o_P(n^{-\varepsilon'}),\label{eq: sigma ratio bound thm 13 - 2}\\
&\hat \sigma(u,w) \gtrsim_P \|\ell(w)\|/\sqrt n,\quad \text{uniformly over }(u,w)\in I;\label{eq: sigma lower bound thm 13 - 2}
\end{align}
see \eqref{eq: sigma ratio bound thm 11 new} and \eqref{eq: sigma lower bound thm 11 new}. Also, by the same argument as that used in Lemma \ref{lem: ellju max inequality},
$$
\sup_{(u,w)\in I}\left|\frac{(\ell(w)/\|\ell(w)\|)' J^{-1}(u)}{\sqrt n}\sum_{i=1}^n (\pi_i - 1)Z_i(u - 1\{U_i \leq u\})\right| \lesssim_P (\log n)^{1/2}.
$$
In addition, by Condition S, $\|\ell(w)\|/(\sqrt n\sigma(u,w))\lesssim 1$ uniformly over $(u,w)\in I$, and so
$$
\sup_{(u,w)\in I}\left|\frac{\ell(w)' J^{-1}(u)}{\sqrt n\sigma(u,w)}\cdot \frac{1}{\sqrt n}\sum_{i=1}^n (\pi_i - 1)Z_i(u - 1\{U_i \leq u\})\right| \lesssim_P (\log n)^{1/2}.
$$
Combining this inequality with \eqref{eq: sigma ratio bound thm 13 - 2} gives
$$
\sup_{(u,w)\in I}\left|\left(\frac{1}{\sigma(u,w)} - \frac{1}{\hat\sigma(u,w)}\right)\frac{\ell(w)' J^{-1}(u)}{\sqrt n} \frac{1}{\sqrt n}\sum_{i=1}^n (\pi_i - 1)Z_i(u - 1\{U_i \leq u\})\right| = o_P(n^{-\varepsilon'}).
$$
Hence, by Lemma \ref{lem: gaussian approx weighted bootstrap process thm 7}, we have uniformly over $(u,w)\in I$ that
\begin{align*}
t^*(u,w)
& = \frac{\ell(w)'(\hat \beta^b(u) - \hat\beta(u))}{\hat\sigma(u,w)}\\
& = \frac{\ell(w)'J^{-1}(u)}{\sqrt n \hat\sigma(u,w)}\cdot\frac{1}{\sqrt n}\sum_{i=1}^n(\pi_i - 1)Z_i(u - 1\{U_i \leq u\}) + o_P(n^{-\varepsilon'})\\
& = \frac{\ell(w)'J^{-1}(u)}{\sqrt n \sigma(u,w)}\cdot\frac{1}{\sqrt n}\sum_{i=1}^n(\pi_i - 1)Z_i(u - 1\{U_i \leq u\}) + o_P(n^{-\varepsilon'}),
\end{align*}
where we used \eqref{eq: sigma lower bound thm 13 - 2} in the second line. Thus, denoting
$$
\widetilde V^* = \sup_{(u,w)\in I}\left| \frac{\ell(w)'J^{-1}(u)}{\sqrt n \sigma(u,w)}\cdot\frac{1}{\sqrt n}\sum_{i=1}^n(\pi_i - 1)Z_i(u - 1\{U_i \leq u\}) \right|,
$$
we obtain
\begin{equation}\label{eq: weighted bootstrap intermediate approx thm 13}
|V^* - \widetilde V^*| = o_P(n^{-\varepsilon'}).
\end{equation}
Now we apply Lemma \ref{lemma:SupremumCoupling} with
\begin{equation}\label{eq: useful notation loading thm 13}
\ell(u,w) = \left(\frac{\ell(w)' J^{-1}(u)}{\sqrt n\sigma(u,w)}\right)'
\end{equation}
and $v_i = \pi_i - 1$ to couple $\widetilde V^*$ with a random variable $\bar V^*$ that has the required distribution. Observe that
\begin{equation}\label{eq: elluw upper bound thm 13}
\sup_{(u,w)\in I}\|\ell(u,w)\| \lesssim 1
\end{equation}
and
\begin{align}
\|\ell(u,w) - \ell(\tilde u,\tilde w)\| \lesssim (1 + \zeta_{m,\theta}^L)\|(u,w) - (\tilde u,\tilde w)\|,\label{eq: lipshitz property elluw thm 13}
\end{align}
uniformly over $(u,w)$ and $(\tilde u,\tilde w)$ in $I$ by Lemma \ref{lem: lipshitz property ell functional}. Note also that $\log \zeta_{m,\theta}^L\lesssim \log n$ by Condition U. Thus, conditions of Lemma \ref{lemma:SupremumCoupling} on $\ell(u,w)$ are satisfied. Moreover, by Corollary \ref{Corollary:phin} in Appendix \ref{App:EmpiricalProcess}, there exists a constant $C>0$ such that $\sup_{\alpha\in S^{m-1}} \En[(\alpha'Z_i)^2] \leq C$ with probability $1 - o(1)$. Thus, applying Lemma \ref{lemma:SupremumCoupling} conditional on $(Z_i)_{i=1}^n$ on the event $\sup_{\alpha\in S^{m-1}} \En[(\alpha'Z_i)^2] \leq C$ shows that there exists a random variable $\bar V^*$ with the required distribution such that
$$
|\widetilde V^* - \bar V^*| = o_P(n^{-\varepsilon}).
$$
Combining this bound with \eqref{eq: weighted bootstrap intermediate approx thm 13} gives \eqref{eq: bootstrap bands main approximation}, so that the claim of Step 2 in this case holds. This completes Step 2.

\noindent
{\bf Step 3} (Approximation of Critical Values). For all $\eta\in(0,1)$, let $\bar k^*(\eta)$ denote the $\eta$-th quantile of the conditional distribution of $\bar V^*$ given the data. Here we show that
\begin{align*}
&P\Big(k^*(1 - \alpha) \leq \bar k^*(1 - \alpha + \nu_{n,1}) + n^{-\varepsilon'}\Big) = 1-o(1),\\
&P\Big(k^*(1 - \alpha) \geq \bar k^*(1 - \alpha - \nu_{n,1}) - n^{-\varepsilon'}\Big) = 1 - o(1),
\end{align*}
for some sequence of positive numbers $(\nu_{n,1})_{n\geq 1}$ converging to zero. Indeed, recall that $k^*(\eta)$ denotes the $\eta$-th quantile of the conditional distribution of $V^*$ given the data and also that $|V^* - \bar V^*| = o_P(n^{-\varepsilon'})$ by Step 2. Hence, the claim of this step follows from Lemma \ref{lemma: quantiles are close}.

\noindent
{\bf Step 4} (Anti-Concentration). Here we show that
\begin{align*}
&P\Big(\bar k^*(1 - \alpha - \nu_{n,1}) - \bar k^*(1 - \alpha - \nu_{n,1} - \nu_{n,2}) > 2 n^{-\varepsilon'}\Big) = 1 - o(1),\\
&P\Big(\bar k^*(1 - \alpha + \nu_{n,1} + \nu_{n,2}) - \bar k^*(1 - \alpha + \nu_{n,1}) > 2 n^{-\varepsilon'}\Big) = 1 - o(1),
\end{align*}
for some sequence of positive numbers $(\nu_{n,2})_{n\geq 1}$ converging to zero. By Lemma \ref{lem: maximal inequality gaussian process functional thm 13}, there exists a constant $C_V > 0$ such that with probability $1 - o(1)$,
$$
E\left[ \sup_{(u,w)\in I}\Big|(\ell(w)/\|\ell(w)\|)'J^{-1}(u)G(u)\Big| \mid (Z_i)_{i=1}^n\right] \leq C_V\sqrt{\log n}.
$$
Also, by Condition S, $\|\ell(w)\|/(\sqrt n\sigma(u,w))\lesssim 1$ uniformly over $(u,w)\in I$, and so with probability $1 - o(1)$,
\begin{equation}\label{eq: maximal inequality for v*}
E\left[ \sup_{(u,w)\in I}\frac{|\ell(w)'J^{-1}(u)G(u)/\sqrt n}{\sigma(u,w)} \mid (Z_i)_{i=1}^n\right] \leq C_V\sqrt{\log n}.
\end{equation}
Hence, existence of the required sequence $(\nu_{n,2})_{n\geq 1}$ follows from recalling that $\bar V^*$ is equal in distribution conditional on $(Z_i)_{i=1}^n$ to the random variable under the expectation sign in \eqref{eq: maximal inequality for v*} and applying Lemma \ref{lemma:Anti} conditional $(Z_i)_{i=1}^n$ on the event in \eqref{eq: maximal inequality for v*}, since $n^{-\varepsilon'}\sqrt{\log n}\to 0$ as $n\to\infty$. This completes Step 4.

\noindent
{\bf Step 5} (Main Argument). Here we complete the proof of the theorem. Observe that
\begin{align*}
P\Big(V > k^*(1 - \alpha)\Big)
& \leq   P\Big( V > \bar k^*(1 - \alpha -\nu_{n,1})  - n^{-\varepsilon'}\Big) + o(1)   \\
& \leq   P\Big( \bar V > \bar k^*(1 - \alpha -\nu_{n,1})  - 2n^{-\varepsilon'}\Big) + o(1)    \\
& \leq   P\Big(\bar V > \bar k^*(1 - \alpha - \nu_{n,1} - \nu_{n,2}) \Big) + o(1)    \\
& \leq   E\Big[P\Big(\bar V > \bar k^*(1 - \alpha - \nu_{n,1} - \nu_{n,2}) |(Z_i)_{i=1}^n\Big)\Big]  + o(1) \\
& \leq   \alpha + \nu_{n,1} + \nu_{n,2} + o(1) = \alpha + o(1),
\end{align*}
where the first line follows from Step 3, the second from Step 1, the third from Step 4, the fourth from the law of iterated expectations, and the fifth from the definition of $\bar V$ and $\bar k^*(1 - \alpha - \nu_{n,1} - \nu_{n,2})$. Similarly,
\begin{align*}
P\Big(V > k^*(1 - \alpha)\Big)
& \geq   P\Big( V > \bar k^*(1 - \alpha + \nu_{n,1})  + n^{-\varepsilon'}\Big) + o(1)   \\
& \geq   P\Big( \bar V > \bar k^*(1 - \alpha + \nu_{n,1})  + 2n^{-\varepsilon'}\Big) + o(1)    \\
& \geq   P\Big(\bar V > \bar k^*(1 - \alpha + \nu_{n,1} + \nu_{n,2}) \Big) + o(1)    \\
& \geq   E\Big[P\Big(\bar V > \bar k^*(1 - \alpha + \nu_{n,1} + \nu_{n,2}) |(Z_i)_{i=1}^n\Big)\Big]  + o(1) \\
& \geq   \alpha - \nu_{n,1} - \nu_{n,2} + o(1) = \alpha + o(1),
\end{align*}
This gives the first asserted claim. The second asserted claim follows immediately from the first one. The last asserted claim follows since $\sup_{(u,w)\in I}\hat\sigma(u,w) \lesssim_P \sqrt{\zeta_{m,\theta}^2 / n}$ and $k^*(1 - \alpha)\lesssim_P  \sqrt{\log n}$ by Steps 3 and 4. This completes the proof of the theorem.
\end{proof}

We use the following lemmas in the proof of Theorem \ref{theorem: inference using couplings}.

\begin{lemma}[Anti-Concentration for Separable Gaussian Processes, Chernozhukov, Chetverikov and Kato \cite{CCK2013}]\label{lemma:Anti} Let $Y=(Y_t)_{t\in T}$ be a separable Gaussian process indexed by a semimetric space $T$ such that $E[Y_t]=0$ and $E[Y_t^2]=1$ for all $t\in T$. Assume that $\sup_{t\in T}|Y_t|<\infty$ a.s. Then for all $\varepsilon \geq 0$ and some absolute constant $A$ we have
$$ \sup_{x\in \RR} P\left(\left|\sup_{t\in T}|Y_t|-x\right| \leq \varepsilon\right) \leq A \varepsilon E\left[\sup_{t\in T}|Y_t|\right].$$
\end{lemma}

\begin{lemma}[Closeness in Probability Implies Closeness of Conditional Quantiles]\label{lemma: quantiles are close}
Let $X_n$ and $Y_n$ be random variables and $\mathcal{D}_n$ be a random vector. Let $F_{X_n}(x|\mathcal{D}_n)$ and $F_{Y_n}(x|\mathcal{D}_n)$
denote the conditional distribution functions, and  $F^{-1}_{X_n}(p|\mathcal{D}_n)$ and $F^{-1}_{Y_n}(p|\mathcal{D}_n)$
denote the corresponding conditional quantile functions. If $|X_n - Y_n| = o_P(\gamma_n)$, then there exists a sequence of positive numbers $(\nu_n)_{n\geq 1}$ converging to zero such that with probability $1 - o(1)$,
$$
F^{-1}_{X_n}(p|\mathcal{D}_n) \leq F^{-1}_{Y_n}(p+\nu_n|\mathcal{D}_n) + \gamma_n \text{ and } F^{-1}_{Y_n}(p|\mathcal{D}_n) \leq F^{-1}_{X_n}(p+\nu_n|\mathcal{D}_n) + \gamma_n
$$
for all $p\in (\nu_n,1 - \nu_n)$.
\end{lemma}
\begin{proof}
Since $|X_n - Y_n| = o_P(\gamma_n)$, there exists a sequence of positive numbers $(\nu_n)_{n\geq 1}$ converging to zero such that $P( |X_n - Y_n| >\gamma_n) = o(\nu_n)$.  Hence,
$$
P\Big(P(|X_n - Y_n| >\gamma_n|\mathcal{D}_n) \leq \nu_n\Big) \to 1,
$$
that is, there is a set $\Omega_n$ of values of $\mathcal D_n$ such that $P(\Omega_n) \to 1$
and $P(|X_n - Y_n| >\gamma_n|\mathcal{D}_n) \leq \nu_n$ for all $\mathcal{D}_n \in \Omega_n$. Now, for all $\mathcal{D}_n \in \Omega_n$,
$$
F_{X_n}(x|\mathcal{D}_n) \geq F_{Y_n+ \gamma_n}(x|\mathcal{D}_n) - \nu_n \text { and } F _{Y_n}(x|\mathcal{D}_n) \geq F_{X_n+ \gamma_n}(x|\mathcal{D}_n) - \nu_n,\quad\text{for all } x \in \Bbb{R},
$$
which implies the inequality stated in the lemma, by definition of the conditional quantile function and equivariance of quantiles to location shifts.
\end{proof}

\begin{proof}[Proof of Theorem \ref{thm: shape constraints}]
The first asserted claim follows from the same argument as that used in the proof of Theorem \ref{theorem: inference using couplings}: simply replace $\sup_{(u,w) \in I} |t(u,w)|$ and $\sup_{(u,w)\in I}|t^*(u,w)|$ by $\sup_{(u,w) \in I} t(u,w)$ and $\sup_{(u,w)\in I}t^*(u,w)$. To prove the second asserted claim, suppose to the contrary that there exists a set of data-generating processes $\mathcal M = \mathcal M_n$ satisfying $H_0$ and such that the conditions of Theorem \ref{theorem: inference using couplings} hold uniformly over this set but the inequality
$$
\sup_{M\in \mathcal M}P_M\Big(T > \tilde k^*(1 - \alpha)\Big) \leq \alpha + o(1)
$$
does not hold. Then there exists $\epsilon > 0$ and a sequence of integers $(n_k)_{k\geq 1}$ such that for each $k\geq 1$, we can find $M_{n_k} \in \mathcal M_{n_k}$ such that
$$
P_{M_{n_k}}\Big(T > \tilde k^*(1 - \alpha)\Big) > \alpha + \epsilon,\quad\text{for all }k\geq 1.
$$
However, this contradicts the first asserted claim because we allow the data-generating process to depend on $n$ and the sequence of data-generating processes $(M_{n_k})_{k\geq 1}$ satisfy conditions of Theorem \ref{theorem: inference using couplings}. The second asserted claim follows. This completes the proof of the theorem.
\end{proof}

\section{A Lemma on Strong Approximation of an Empirical
Process of an Increasing Dimension by a Gaussian Process}\label{App:StrongGaussianApprox}

\begin{lemma}\label{lemma: strong} (Approximation of a Sequence of Empirical Processes of Increasing Dimension by a Sequence of Gaussian Processes) Let $(Z_i)_{i=1}^n$ be a sequence of non-stochastic vectors in $\mathbb R^m$ and consider the empirical process $\mathbb{U}_{n}$ in $[\ell^{\infty}(\mathcal{U})]^m$, $\mathcal{U} \subseteq (0,1)$, defined by
$$
\mathbb{U}_{n}(u) = \mathbb{G}_n\(v_iZ_i \psi_i(u)\), \ \ \psi_i(u) = u-1\{U_i \leq u\},\quad u\in\mathcal U,
$$
where $(U_i,v_i)_{i=1}^n$ is an i.i.d. sequence of pairs of independent random variables where $U_i\sim\Uniform (0,1)$, $E[v_i^2] = 1$, $E[|v_i|^4]\lesssim 1$, and $\max_{1\leq i\leq n}|v_i| \lesssim_P \log n$.  Suppose that the vectors $Z_i$ are such that
$$
\sup_{\alpha\in S^{m-1}}\En\[(\alpha'Z_i)^2\] \lesssim 1, \  \ \max_{1\leq i\leq n}\|Z_i\| \lesssim \zeta_m, \ \ \mbox{and} \ \ m^{7}\zeta_m^6 = o(n^{1-\varepsilon})
$$ where $\zeta_m$ satisfies $1/\zeta_m\lesssim 1$ and $\varepsilon>0$ is some constant.
Then there exists a sequence of zero-mean Gaussian processes $(\Z_n)_{n\geq 1}$ with a.s. continuous paths such that (i) the covariance functions of $G_n$ coincide with those of $\mathbb{U}_{n}$, namely,
$$
E[\Z_n(u) \Z_n(\tilde u)'] = E[\mathbb{U}_{n}(u) \mathbb{U}_{n}(\tilde u)'] = \En[Z_i Z_i'] ( u \wedge \tilde u -  u  \tilde u), \text{ for all }  u \text{ and } \tilde u \in \mathcal{U},
$$
and (ii) $G_n$ approximates $\mathbb{U}_{n}$, namely,
$$
\sup_{u \in \mathcal{U}} \| \mathbb{U}_{n}(u) - \Z_n(u) \| \lesssim_P  o(n^{-\varepsilon'}),
$$
where $\varepsilon'>0$ is some constant.
\end{lemma}

\begin{proof} The proof is based on the use of maximal inequalities and Yurinskii's coupling. 
We define the sequence of projections $\pi_j\colon \mathcal{U} \to \mathcal{U}$, $j=0,1,2,\ldots,\infty$ by
 $\pi_j(u) = u_{k j} =  k/2^j$ if $u \in ( (k-1)/2^j, k/2^j], k = 1,\ldots,2^{j}$. In what follows, given a process $\Z$ in $[\ell^{\infty}(\mathcal{U})]^m$ and its projection $\Z \circ \pi_j$, whose paths are step functions with at most $2^{j}$ steps, we shall identify the process $\Z \circ \pi_j$ with a random vector $\Z \circ \pi_j$ in $\Bbb{R}^{m 2^{j}}$, when convenient. Analogously, given a random vector $W$ in $\Bbb{R}^{m 2^{j}}$, we identify it with a process $W$ in $[\ell^{\infty}(\mathcal{U})]^{m}$, whose paths are step functions with at most $2^{j}$ steps.

The following relations will be proven below for some $\varepsilon'>0$ and some $j = j_n\to\infty$:
\begin{enumerate}
\item (Finite-Dimensional Approximation)  $$r_1 =\sup_{u \in \mathcal{U}}\|\mathbb{U}_{n}(u) - \mathbb{U}_{n}\circ \pi_j(u)\| \lesssim_P o(n^{-\varepsilon'});$$
    \item (Coupling with a Normal Vector) there exists $\N_{nj} =_d N(0, \text{var}[\mathbb{U}_{n} \circ \pi_j])$ such that
$$r_2 = \| \N_{nj} - \mathbb{U}_{n} \circ \pi_j \| \lesssim_P o(n^{-\varepsilon'});$$
 \item (Embedding a Normal Vector into a Gaussian Process) there exists a Gaussian process $\Z_n$ with properties stated in the lemma such that $\N_{nj} = \Z_n \circ \pi_j \text{ a.s.};$
      \item (Infinite-Dimensional Approximation)
$$r_3=\sup_{u \in \mathcal{U}}\|\Z_n(u) - \Z_n\circ \pi_j(u)\| \lesssim_P o(n^{-\varepsilon'}).$$
\end{enumerate}
The result then follows from the triangle inequality:
$$
\sup_{u \in \mathcal{U}} \|\mathbb{U}_{n}(u) - \Z_n(u)\| \leq r_1 + r_2 + r_3.
$$

We now prove relations (1)-(4). Relation (1) follows from
\begin{align}\begin{split}\nonumber
r_1  = \sup_{u \in \mathcal{U}} \|\mathbb{U}_{n}(u) - \mathbb{U}_{n}\circ \pi_j(u)\|
  & \leq  \sup_{|u - \tilde u| \leq 2^{-j} } \|  \mathbb{U}_{n}(u) - \mathbb{U}_{n}(\tilde u) \|  \\
  & \lesssim_P   \sqrt{2^{-j} m \log n} + \sqrt{\frac{m^2\zeta_m^2\log^4n}{n}}\lesssim_P o(n^{-\varepsilon'}),
\end{split}\end{align}
where the first inequality in the second line follows from Lemma \ref{Lemma:Fact1} and second holds by setting $2^j = m n^{\tilde \varepsilon}$ for sufficiently small $\tilde\varepsilon$ and recalling that $m^7\zeta_m^6/n = o(n^{-\varepsilon})$ and $1/\zeta_m\lesssim 1$.

Relation (2) follows from the use of Yurinskii's coupling (Pollard \cite{PollardUsers}, Chapter 10, Theorem 10): Let $\xi_1,\ldots,\xi_n$ be independent zero-mean $p$-vectors such that $\kappa := \sum_{i=1}^n \Ep\[ \|\xi_i\|^3\]$ is finite. Let
$S = \xi_1 + \cdots + \xi_n$. Then for each $\delta>0$, there exists a random vector $T$ with a $N(0, \text{var}(S))$ distribution such that
$$
\Pp\{ \| S- T\| > 3 \delta \} \leq C_0 B \left( 1  + \frac{|\log (1/B)|}{p}   \right) \text{ where }
B:= \kappa p \delta^{-3},
$$
for some universal constant $C_0$.

In order to apply the coupling, we collapse $v_iZ_i \psi_i \circ \pi_j$ to a $p$-vector, and
let
$$
\xi_i = v_iZ_i \psi_i \circ \pi_j \in \Bbb{R}^p, \ \ p = 2^j m
$$
so that $\mathbb{U}_{n} \circ \pi_j = \sum_{i=1}^n \xi_i/\sqrt{n}$. Then
$$
\En\Ep [\|\xi_i\|^3]  =   \En\Ep\left[ \left( \sum_{k=1}^{2^j} \sum_{w=1}^m  \psi_i(u_{kj})^2   v_i^2Z_{iw}^2  \right)^{3/2} \right]  \leq   2^{3j/2}\Ep[|v_i|^3] \En [ \| Z_{i}\|^{3}] \lesssim 2^{3j/2}  \zeta_m^{3}.
$$
Therefore, by Yurinskii's coupling, since $\log n \lesssim 2^jm$, by the choice $2^j=m n^{\tilde \varepsilon}$,
$$
\Pp \left \{ \left \| \frac{\sum_{i=1}^n \xi_i}{\sqrt{n}} - \mathcal{N}_{nj} \right\|  \geq 3 \delta \right\} \lesssim   \frac{ n 2^{3j/2}  \zeta_m^{3} 2^jm}{(\delta \sqrt{n} )^3} =  \frac{  2^{5j/2} m \zeta_m^3}{\delta^3 n^{1/2} } \to 0
$$
by setting $\delta = (2^{5j}m^2\zeta_m^6 \log n / n)^{1/6}$. This verifies relation (2) with
$$r_2 \lesssim_P  \( \frac{2^{5j}m^2\zeta_m^6 \log n}{n} \)^{1/6} = \( \frac{n^{5\tilde\varepsilon}m^7\zeta_m^6\log n}{n} \)^{1/6} = o(n^{-\varepsilon'}),$$
provided that $5\tilde\varepsilon+6\varepsilon' < \varepsilon$.

Relation (3) follows from the a.s. embedding of a finite-dimensional random normal vector into a path of a continuous Gaussian process, which is possible by Lemma \ref{lemma:prescribed}. Here, we note that Lemma \ref{lemma:prescribed} gives an explicit construction of the process $G_n$. As pointed out by a referee, however, Lemma \ref{lemma:prescribed} can be avoided and the process $G_n$ can be constructed implicitly by referring to well-known results. Indeed, let $\tilde G_n$ be a zero-mean Gaussian process with a.s. continuous paths and the same covariance function as that of $\mathbb U_n$. Since $\mathcal N_{n j} =_d \tilde G_n\circ\pi_j$ and $\tilde G_n$ takes values in $[C(\mathcal U)]^m$, which is Polish, it follows from the Vorob'ev-Berkes-Philipp theorem (e.g., see Theorem 1.1.10 in \cite{Dudley2000}), with $\alpha$ and $\beta$ being the laws of $(\mathbb U_n\circ\pi_j,\mathcal N_{n j})$ and $(\tilde G_n\circ\pi_j,\tilde G_n)$, respectively, that one can construct a distribution law on $\mathbb R^{m 2^j}\times \mathbb R^{m 2^j}\times [C(\mathcal U)]^m$ such that its projection on $\mathbb R^{m 2^j}\times \mathbb R^{m 2^j}$ is equal to the law of $(\mathbb U_n\circ\pi_j, \mathcal N_{n j})$ and its projection on $\mathbb R^{m 2^j}\times [C(\mathcal U)]^m$ is equal to the law of $(\tilde G_n\circ\pi_j,\tilde G_n)$. Then the existence of the required Guassian process $G_n$ with $\mathcal N_{n j} = G_n\circ \pi_j$ follows from Lemma 2.7.3 in \cite{Dudley2000}, with $V = \mathbb U_n\circ\pi_j$.

Relation (4) follows from
\begin{align}\begin{split}\nonumber
r_3  = \sup_{u \in \mathcal{U}} \|\Z_n(u) - \Z_n\circ \pi_j(u)\|
  & \leq  \sup_{|u - \tilde u| \leq 2^{-j} } \|  \Z_n(u) - \Z_n(\tilde u) \|  \\
  & \lesssim_P   \sqrt{2^{-j} m \log n} \lesssim_P o(n^{-\varepsilon'}),
\end{split}\end{align}
where the first inequality in the second line follows from Lemma \ref{Lemma:Fact4} and the second holds by the choice of $j$. This completes the proof of the lemma.
\end{proof}

Next, we establish auxiliary lemmas that were used in the preceding proof.

\begin{lemma}[Finite-Dimensional Approximation]\label{Lemma:Fact1}
Consider the setting of Lemma \ref{lemma: strong} and denote $\varphi = \sup_{\alpha \in S^{m-1}} \En[(\alpha'Z_i)^2]$. Then for any $\gamma > 0$, the process $\mathbb U_n$ satisfies
$$
\sup_{|u - \tilde u|\leq \gamma} \| \mathbb{U}_{n}(u) - \mathbb{U}_{n}(\tilde u) \| \lesssim_P \sqrt{\gamma \varphi m\log n} + \sqrt{\frac{m^2 \zeta^2_m \log^4 n}{n}}.
$$
\end{lemma}
\begin{proof}
Consider the function class
$$
\mathcal G_{m,n} = \Big\{(Z,U,v)\mapsto g_{\alpha,u}(Z,U,v) = (\alpha' Z)\cdot (u - 1\{U\leq u\})\colon u\in\mathcal U,\alpha\in S^{m-1}\Big\},
$$
mapping $B_m(0,\zeta_m)\times [0,1]\times\mathbb R$ into $\mathbb R$. Note that $G_{m,n}(Z,U,v) = \zeta_m$ is its envelope. By Lemma \ref{Lemma:VC-G},  the uniform entropy numbers of $\mathcal G_{m,n}$ satisfy
\begin{equation}\label{eq: entropy numbers Gmn}
\sup_Q \log N(\epsilon\|G_{m,n}\|_{Q,2},\mathcal G_{m,n},L_2(Q))\lesssim O(m) \log(1/\epsilon),\quad \text{uniformly over }0<\epsilon\leq 1.
\end{equation}
Next, consider the function class
$$
\widetilde{\mathcal G}_{m,n} = \Big\{(Z,U,v)\mapsto v\cdot g(Z,U,v)\colon g\in\mathcal G_{m,n}\Big\}.
$$
The function $\widetilde G_{m,n}(Z,U,v) =  |v|\cdot G_{m,n}(Z,U,v)$ is an envelope of $\widetilde {\mathcal G}_{m,n}$. By \eqref{eq: entropy numbers Gmn} and Lemma \ref{Lemma:ProductEntropy}, the uniform entropy numbers of $\widetilde{\mathcal  G}_{m,n}$ satisfy
\begin{equation}\label{eq: entropy numbers tGmn}
\sup_Q \log N(\epsilon\|\widetilde G_{m,n}\|_{Q,2},\widetilde{\mathcal G}_{m,n},L_2(Q))\lesssim O(m)\log (1/\epsilon),\quad\text{uniformly over }0<\epsilon\leq 1.
\end{equation}
Further, consider the function class
$$
\mathcal G_{m,n,\gamma} = \Big\{(Z,U,v)\mapsto v\cdot(g_{\alpha,u}(Z,U,v) - g_{\alpha,\tilde u}(Z,U,v))\colon u,\tilde u\in\mathcal U,\alpha\in S^{m-1}, |u - \tilde u|\leq \gamma\Big\}.
$$
The function $2\widetilde G_{m,n}$ is its envelope. By \eqref{eq: entropy numbers tGmn} and Lemma \ref{Lemma:ProductEntropy}, the uniform entropy numbers of $\mathcal G_{m,n,\gamma}$ satisfy
\begin{equation}\label{eq: entropy numbers Gmng}
\sup_Q \log N(\epsilon\|2 \widetilde G_{m,n}\|_{Q,2},\mathcal G_{m,n,\gamma},L_2(Q))\lesssim O(m)\log(1/\epsilon),\ \text{uniformly over }0<\epsilon\leq 1.
\end{equation}
With this notation, we have
$$
\sup_{|u-\tilde u|\leq \gamma}\|\mathbb U_n(u) - \mathbb U_n(\tilde u)\| = \sup_{g\in\mathcal G_{m,n,\gamma}}|\mathbb G_n g|,
$$
and so to prove the asserted claim, we can apply the second part of Lemma \ref{Lemma:MaxIneq2} using the sequence of independent observations $(Z_i,U_i,v_i)_{i=1}^n$. Note that Lemma \ref{Lemma:MaxIneq2} does not require i.i.d. observations, and so it can be applied even though $Z_i$'s are non-stochastic.

By \eqref{eq: entropy numbers Gmng}, \eqref{Eq:J} is satisfied with $\omega = 1$, $J(m) = O(\sqrt m)$, and $F_m = 2\widetilde G_{m,n}$. Note that
$$
\max_{1\leq i\leq n} F_m(Z_i,U_i,v_i) \lesssim_P M_m  = \zeta_m \log n
$$
since $\max_{1\leq i\leq n}v_i \lesssim_P \log n$. Further,
\begin{align*}
\sup_{g\in\mathcal G_{m,n,\gamma}}\frac{1}{n}\sum_{i=1}^n \Ep[g(Z_i,U_i,v_i)^2]
&= \sup_{\alpha\in S^{m-1},|u - \tilde u|\leq \gamma}\frac{1}{n}\sum_{i=1}^n \Ep\Big[v_i^2(\alpha' Z_i)^2(\psi_i(u) - \psi_i(\tilde u))^2\Big]\\
& =\sup_{\alpha\in S^{m-1},|u - \tilde u|\leq \gamma}\frac{1}{n}\sum_{i=1}^n (\alpha' Z_i)^2 \Ep\Big[(\psi_i(u) - \psi_i(\tilde u))^2\Big]\\
&\leq \sup_{\alpha\in S^{m-1}}\frac{1}{n}\sum_{i=1}^n (\alpha' Z_i)^2 \gamma(1 - \gamma) \leq \varphi\gamma,
\end{align*}
where the second line holds because $Z_i$ is non-stochastic, $v_i$ is independent of $U_i$, and $\Ep[v_i^2] = 1$, and the third line holds because $(\psi_i(u) - \psi_i(\tilde u))^2 =_d (|u - \tilde u| - 1\{U_i \leq |u - \tilde u|\})^2$. Moreover,
\begin{align*}
\sup_{g\in\mathcal G_{m,n,\gamma}}\frac{1}{n}\sum_{i=1}^n \Ep[g(Z_i,U_i,v_i)^4]
&= \sup_{\alpha\in S^{m-1},|u - \tilde u|\leq \gamma}\frac{1}{n}\sum_{i=1}^n \Ep\Big[v_i^4(\alpha' Z_i)^4(\psi_i(u) - \psi_i(\tilde u))^4\Big]\\
& \lesssim \sup_{\alpha\in S^{m-1},|u - \tilde u|\leq \gamma}\frac{1}{n}\sum_{i=1}^n (\alpha' Z_i)^4 \Ep\Big[(\psi_i(u) - \psi_i(\tilde u))^4\Big]\\
& \lesssim \sup_{\alpha\in S^{m-1},|u - \tilde u|\leq \gamma}\frac{1}{n}\sum_{i=1}^n \zeta_m^2(\alpha' Z_i)^2 \Ep\Big[(\psi_i(u) - \psi_i(\tilde u))^2\Big]\\
&\leq \sup_{\alpha\in S^{m-1}}\frac{1}{n}\sum_{i=1}^n \zeta_m^2(\alpha' Z_i)^2 \gamma(1 - \gamma) \leq \zeta_m^2\varphi\gamma,
\end{align*}
where we used the same arguments as above in addition to the facts that $E[v_i^4]\lesssim 1$, $\max_{1\leq i\leq n}\|Z_i\|\leq \zeta_m$, and $|\psi_i(u) - \psi_i(\tilde u)|\leq 1$. Substituting these bounds into the second part of Lemma \ref{Lemma:MaxIneq2} gives
$$
\sup_{g\in\mathcal G_{m,n,\gamma}}|\mathbb G_n g| \lesssim_P \sqrt{\gamma \varphi m\log n} + \sqrt{\frac{m^2 \zeta^2_m \log^4 n}{n}}.
$$
This completes the proof of the lemma.
\end{proof}

\begin{lemma}[Infinite-Dimensional Approximation]\label{Lemma:Fact4}
Let $\Z_n\colon \mathcal{U} \to \RR^m$ be a zero-mean Gaussian process whose covariance structure is given by
$$
E\[ \Z_n(u) \Z_n(\tilde u)'\] = \En[Z_iZ_i'](u\wedge \tilde u - u \tilde u), \text{ for all }  u \text{ and } \tilde u\text{ in } \mathcal{U},
$$
where $(Z_i)_{i=1}^n$ is a non-stochastic sequence in $\RR^m$. Then for any $\gamma \in (0,1/2)$,
$$
\sup_{|u-\tilde u|\leq \gamma} \| \Z_n(u) - \Z_n(\tilde u) \| \lesssim_P \sqrt{\varphi \gamma m \log (m/\gamma)},
$$
where $\varphi = \sup_{\alpha\in S^{m-1}} \En[(\alpha'Z_i)^2]$.
\end{lemma}
\begin{proof}
We will use the following maximal inequality for Gaussian processes
(Proposition A.2.7 in van~der Vaart and Wellner \cite{vdV-W}). Let $X = (X_t)_{t\in T}$ be a separable zero-mean
Gaussian process indexed by a set $T$. Suppose that for some
$K>\sigma(X) = \sup_{t \in T}\sigma(X_t)$ and $0< \epsilon_0\leq
\sigma(X)$, we have
\begin{equation}\label{eq: entropy bound lem 11}
N(\epsilon, T,\rho)  \leq \left( \frac{K}{\epsilon}\right)^V, \ \mbox{for all} \ 0 <\epsilon < \epsilon_0,
\end{equation}
where $N(\epsilon,T,\rho)$ is the covering number of $T$ by $\epsilon$-balls with respect to the standard deviation metric $\rho(t,t') =\sigma(X_t - X_{t'})$. Then there exists a universal constant $D$ such that for every $\lambda \geq \sigma^2(X)(1+\sqrt{V})/\epsilon_0$,
\begin{equation}\label{eq: gaussian maximal inequality}
P\left( \sup_{t \in T} X_t > \lambda \right) \leq \left(\frac{DK\lambda}{\sqrt{V}\sigma^2(X)} \right)^V \bar{\Phi}(\lambda /\sigma(X)),
\end{equation}
where $\bar \Phi = 1 - \Phi$, and $\Phi$ is the cumulative distribution function of a standard Gaussian random variable.

We apply this result to the zero-mean Gaussian process $X_n\colon S^{m-1}\times \mathcal{U}\times \mathcal{U} \to \RR$ defined as
$$
X_{n,t} = \alpha'( \Z_n(u) - \Z_n(\tilde u)), \ \ t=(\alpha,u,\tilde u), \ \alpha \in S^{m-1}, \ |u-\tilde u|\leq \gamma,
$$
since this process is such that $\sup_{t \in T} X_{n,t} = \sup_{|u-\tilde u|\leq \gamma}\|\Z_n(u) - \Z_n(\tilde u)\|$. For this process, we have
$$
\sigma(X_n) \leq \sqrt{\gamma \sup_{\alpha\in S^{m-1}} \En[(\alpha'Z_i)^2]},
$$
and \eqref{eq: entropy bound lem 11} holds with
$$
\epsilon_0 = \sigma(X_n),\ \ K \lesssim \sqrt{\sup_{\alpha\in S^{m-1}} \En[(\alpha'Z_i)^2]},\ \ \mbox{and} \  V \lesssim m.
$$
Therefore, the result follows by setting
$$
\lambda = C\sqrt{\gamma m \log (m/\gamma) \sup_{\alpha\in S^{m-1}} \En[(\alpha'Z_i)^2]},
$$
where $C$ is a sufficiently large constant, and using \eqref{eq: gaussian maximal inequality}. This completes the proof of the lemma.
\end{proof}

In what follows, as before, given a process $\Z$ in $[\ell^{\infty}(\mathcal{U})]^m$ and its projection $\Z \circ \pi_j$, whose paths are step functions with at most $2^j$ steps, we shall identify the process $\Z \circ \pi_j$ with a random vector $\Z \circ \pi_j$ in $\Bbb{R}^{m 2^j}$, when convenient. Analogously, given a random vector $W$ in $\Bbb{R}^{m 2^j}$ we identify it with a process $W$ in $[\ell^{\infty}(\mathcal{U})]^{m}$, whose paths are step functions with at most $2^j$ steps.

\begin{lemma}\label{lemma:prescribed}(Construction of a Gaussian Process with a Pre-scribed Projection)
Let $\N_j$ be a given random vector such that $$\N_j =_d \tilde \Z \circ \pi_j =: N(0, \Sigma_j),$$ where $\Sigma_j:= {\rm Var}(\N_j)$ and $\tilde \Z$ is a zero-mean Gaussian process in $[\ell^{\infty}(\mathcal{U})]^m$ whose paths are a.s. uniformly continuous with respect to the Euclidian metric $|\cdot|$ on $\mathcal{U}$. There exists a zero-mean Gaussian process in $[\ell^{\infty}(\mathcal{U})]^m$, whose paths are a.s. uniformly continuous with respect to the Euclidian metric $|\cdot|$ on $\mathcal{U}$, such that
$$
\N_j = \Z \circ \pi_j \text{ and } \Z=_d \tilde \Z \text{ in } [\ell^{\infty}(\mathcal{U})]^m.
$$
 \end{lemma}
\begin{proof}  Consider a vector $\tilde \Z \circ \pi_\ell$ for $\ell  = j + 1$. Then $\tilde \N_j = \tilde \Z \circ \pi_j$ is a subvector of $\tilde \Z \circ \pi_\ell = \tilde \N_\ell$. Denote the remaining components of $\tilde \N_\ell$ as $\tilde \N_{\ell \setminus j}$. We can construct an identically distributed copy $\N_\ell$ of $\tilde \N_{\ell}$ such that $\N_j$ is a subvector of $\N_\ell$. Indeed, we
set $\N_{\ell}$ as a vector with components
$$
\N_j \text{ and } \N_{\ell\setminus j},
$$
arranged in appropriate order, namely that $\N_{\ell} \circ \pi_j = \N_j$,
where
$$
 \N_{\ell\setminus j} = \Sigma_{\ell \setminus j, j} \Sigma^{-1}_{j,j} \N_j + \eta_j,
$$
where $\eta_j \bot \N_j$ and  $\eta_j =_d N(0, \Sigma_{\ell \setminus j, \ell\setminus j} - \Sigma_{\ell\setminus j,j} \Sigma^{-1}_{j,j} \Sigma_{j, \ell\setminus j})$, where
$$
\left(
  \begin{array}{cc}
    \Sigma_{j,j} & \Sigma_{\ell\setminus j,j} \\
    \Sigma_{j, \ell\setminus j} &  \Sigma_{\ell \setminus j, \ell\setminus j}  \\
  \end{array}
\right)
 := \text{var}\left( \begin{array}{c}
            \tilde N_j \\
            \tilde N_{\ell \setminus j}
          \end{array} \right).
$$
Having constructed $\mathcal N_l = \mathcal N_{j + 1}$, we can proceed using the same procedure to construct $\mathcal N_{j + 2}$ from $\mathcal N_{j + 1}$. Repeating this procedure, we obtain the whole sequence $(\mathcal N_{l})_{l\geq j}$.

For each $l\geq j$, we then identify the vector $\N_\ell$ with a process $\N_\ell$  in $\ell^{\infty}(\mathcal{U})$, and define $G$ as the pointwise limit of this process:
$$
\Z(u):= \lim_{\ell \to \infty} \N_{\ell}(u) \text{ for each } u \in \mathcal{U}_0,
$$
where $\mathcal{U}_0 = (\cup_{j=1}^{\infty} \cup_{k=1}^{2^j} u_{kj})\cap \mathcal U$ is a countable dense subset of $\mathcal{U}$.  Note that the pointwise limit exists since by construction of $\{\pi_{\ell}\}$ and $\mathcal{U}_0$, for each $u \in \mathcal{U}_0$, we have that $\pi_{\ell}(u) = u$ for all $\ell \geq \ell(u)$, where $\ell(u)$ is a sufficiently large constant.

Next, we extend the process $G$ from $\mathcal U_0$ to $\mathcal U$ as follows.
By construction,  $\Z_{\ell} = \Z\circ \pi_{\ell} =_d \tilde \Z\circ \pi_{\ell}$.
Therefore, for each $\epsilon>0$, there exists $\eta(\epsilon)>0$ small enough such that
\begin{align*}
 \Pp\Big( \sup_{ u,\tilde u \in \mathcal{U}_0 : | u - \tilde u| \leq \eta(\epsilon) }  \| \Z (u) - \Z (\tilde u) \| \geq \epsilon \Big)
 & \leq \Pp\Big( \sup_{ | u - \tilde u| \leq \eta(\epsilon) } \sup_{k} \| \Z \circ \pi_k(u) - \Z \circ \pi_k(\tilde u) \| \geq \epsilon \Big) \\
 & \leq \Pp\Big( \sup_{ | u - \tilde u| \leq \eta(\epsilon) } \sup_{k} \| \tilde \Z \circ \pi_k(u) - \tilde \Z \circ \pi_k(\tilde u) \| \geq \epsilon \Big) \\
 & \leq \Pp\Big( \sup_{ | u - \tilde u| \leq \eta(\epsilon) }  \| \tilde \Z (u) - \tilde \Z (\tilde u) \| \geq \epsilon \Big) \leq \epsilon,
\end{align*}
where the last line holds because $\sup_{ | u - \tilde u| \leq \eta}  \| \tilde \Z (u) - \tilde \Z (\tilde u) \| \to 0$ as $\eta \to 0$ almost surely and thus also in probability,  by a.s. continuity of sample paths of $\tilde \Z$. Setting $\epsilon = 2^{-m}$ for each  $ m \in \mathbb{N}$ in the above display, and summing the resulting inequalities over $m$, we get a finite number on the right-hand side. Hence, by the Borel-Cantelli lemma,
$
|\Z(u) -\Z(\tilde u)| \leq 2^{-m} \text{ for all } | u- \tilde u| \leq \eta(2^{-m})
$
for all sufficiently large $m$ almost surely.  This implies that the sample path of the process $G$ is uniformly continuous on $\mathcal{U}_0$ almost surely, and so we can extend the process by continuity to a process $\{ \Z(u)\colon u \in \mathcal{U}\}$ such that its sample path is uniformly continuous almost surely.

Finally, in order to show that the law of $\Z$ is equal to the law of $\tilde \Z$ in $[\ell^{\infty}(\mathcal{U})]^m$, it suffices to demonstrate that
$$
\Ep[ g(\Z) ]= \Ep [g(\tilde \Z)] \text{ for all } g\colon [\ell^{\infty}(\mathcal{U})]^m \to \Bbb{R}\text{ with } |g(z) - g(\tilde z)| \leq \sup_{u \in \mathcal{U}} \| z(u) - \tilde z(u)\| \wedge 1.
$$
To do so, note that
\begin{align}\begin{split}\nonumber
\Big| \Ep [g(\Z)] - \Ep[ g(\tilde \Z)] \Big| & \leq \Big| \Ep[ g(\Z \circ \pi_{\ell}) ]  - \Ep[ g(\tilde \Z \circ \pi_{\ell})] \Big|  \\
& \quad +   \Ep \[ \sup_{u \in \mathcal{U}} \| \Z \circ \pi_\ell (u) - \Z (u)\|  \wedge 1 \]  \\
 & \quad + \Ep \[ \sup_{u \in \mathcal{U}} \| \tilde \Z \circ \pi_\ell (u) - \tilde \Z (u) \| \wedge 1 \]  \to 0 \text{ as } \ell \to \infty.
\end{split}\end{align}
Indeed, the first term on the right-hand side converges to zero by construction, and the second and
third terms converge to zero by the dominated convergence theorem since
$$
\Z \circ \pi_{\ell} \to \Z  \text{ and } \tilde \Z \circ \pi_{\ell} \to \tilde \Z  \text{ in } [\ell^{\infty}(\mathcal{U})]^m \text{ as } \ell \to \infty \text{ a.s.},
$$
holding due to a.s. uniform continuity of sample paths of $\Z$ and $\tilde \Z$. This completes the proof of the lemma.
\end{proof}

\section{Technical Lemmas on Bounding Empirical Errors}\label{App:EmpiricalProcess}


\comment{
\begin{lemma}[Control of Empirical Error]\label{Lemma:E3}
Under S.1-5, for any $t>0$ let
$$\epsilon(t):=\sup_{u \in \mathcal{U}, \|J^{1/2}_m(u) \delta \|\leq t} \left| \hat Q_u(\beta(u) + \delta) - Q_u(\beta(u) + \delta)-\left( \hat Q_u(\beta(u)) - Q_u(\beta(u))\right)\right|.$$
 Then, for $\kappa_0^2 = {\rm mineig}(\Sigma)$ we have
 $$\epsilon(t) \lesssim_P   \frac{ t }{\underline{f}^{1/2}\kappa_0}\sqrt{\frac{m \log( m\vee [L\underline{f}^{1/2}\kappa_0/t])}{n}}.$$
\end{lemma}
\begin{proof}
It follows from Lemma 5 in Belloni and Chernozhukov \cite{BC-SparseQR}  with $p=s=m$, $c_0 = 0$.
\end{proof}
}

\subsection{Some Preliminary Lemmas}

\begin{lemma}\label{Lemma:JJ}
Under Condition S, we have
$$
\|J(u) - \widetilde J(u)\| \lesssim  m^{-\kappa} = o(1),
$$
uniformly over $u\in\mathcal U$, where $\widetilde J(u) = \Ep[ f_{Y|X}(Z'\beta(u)|X)ZZ']$ and $J(u)$ is defined in \eqref{eq: J matrix}.
\end{lemma}
\begin{proof}
Note that $\|J(u) - \widetilde J(u)\| = \sup_{\alpha\in S^{m-1}}|\alpha'(J(u) - \widetilde J(u))\alpha|$. In addition,
\begin{align*}
|\alpha'(J(u) - \widetilde J(u))\alpha| & = \Ep\Big[ \Big|f_{Y|X}(Z'\beta(u)+R(u,X)|X) - f_{Y|X}(Z'\beta(u)|X)\Big| \cdot (Z'\alpha)^2 \Big]\\
& \lesssim E\Big[|R(u,X)| \cdot(Z'\alpha)^2\Big] \lesssim m^{-\kappa}
\end{align*}
uniformly over $u\in\mathcal U$ and $\alpha\in S^{m-1}$, where we used Condition S in the second line. The asserted claim follows.
\end{proof}

\begin{lemma}[Auxiliary Matrix]\label{Lemma:Auxiliary Matrix}
Suppose that Condition S holds. Then
$$
\|J(u_2) - J(u_1)\| \lesssim |u_2 - u_1|,\quad\text{uniformly over $u_1$ and $u_2$ in $\mathcal U$}.
$$
In addition,
$$
\left|\frac{z'( J^{-1}(u_2) - J^{-1}(u_1) ) \mathbb{U}(u_2)}{\sqrt{u_1(1-u_1)z' J^{-1}(u_1)\Sigma J^{-1}(u_1)z}}\right|\lesssim_P |u_2-u_1|\sqrt{m}
$$
uniformly over $z\in \{Z(x)\colon x\in\mathcal X\}$ and  $u_1, u_2 \in \mathcal{U}$ for $\mathbb U(u)$ defined in \eqref{Def:U}.
\end{lemma}
\begin{proof}
Recall that $J(u) = \Ep\[ f_{Y|X}(Q(u,X)|X)Z Z'\]$ for any $u
\in \mathcal{U}$. Moreover, we have $\sup_{u\in\mathcal{U}}\|J^{-1}(u)\mathbb{U}(u)\| \lesssim_P \sqrt{m}$
by Lemma \ref{Lemma:error-rate-0} since all eigenvalues of $J(u)$ are bounded below from zero uniformly over $u\in\mathcal U$. In addition, using the matrix identity $A^{-1} - B^{-1} = B^{-1}(B-A)A^{-1}$ with $A = J(u_2)$ and $B= J(u_1)$ gives
$$
J^{-1}(u_2) - J^{-1}(u_1)  =  J^{-1}(u_1) (J(u_1) - J(u_2) )J^{-1}(u_2).
$$
Also, since $|f_{Y|X}(Q(u_2,x)|x) - f_{Y|X}(Q(u_1,x)|x)| \lesssim |u_2-u_1|$ uniformly over $u_1,u_2\in\mathcal U$ and $x\in\mathcal X$ by Lemma
\ref{Lemma:PrimitiveD2}, it follows from the same argument as that used in the proof of Lemma \ref{Lemma:JJ} that
$$
\|J(u_2) - J(u_1)\| \lesssim |u_2-u_1|,\quad\text{uniformly over $u_1,u_2\in\mathcal U$,}
$$
which gives the first asserted claim. Further,
\begin{align*}
&\Big|\frac{z'( J^{-1}(u_2) - J^{-1}(u_1) ) \mathbb{U}(u_2)}{\sqrt{u_1(1-u_1)z' J^{-1}(u_1)\Sigma J^{-1}(u_1)z}}\Big|\\
&\qquad  = \Big| \frac{z' J^{-1}(u_1)}{\sqrt{u_1(1-u_1)z' J^{-1}(u_1)\Sigma J^{-1}(u_1)z}}( J(u_1) - J(u_2) ) J^{-1}(u_2)\mathbb{U}(u_2) \Big|\\
& \qquad \lesssim \frac{\|z' J^{-1}(u_1)\|}{\sqrt{z' J^{-1}(u_1)\Sigma J^{-1}(u_1)z}}\cdot\Big\| J(u_1) - J(u_2)\Big\|\cdot \Big\|J^{-1}(u_2) \mathbb U(u_2)\Big\| \lesssim_P |u_2 - u_1|\sqrt m,
\end{align*}
uniformly over $u_1,u_2\in\mathcal U$, where in the third line, we used the fact that $\mathcal U\subset (0,1)$ is compact to show that $u_1(1- u_1)$ is bounded away from zero uniformly over $u_1\in\mathcal U$ and also the fact that all eigenvalues of the matrix $\Sigma$ are bounded away from zero to show that
$$
\frac{\|z' J^{-1}(u_1)\|}{\sqrt{z' J^{-1}(u_1)\Sigma J^{-1}(u_1)z}} \lesssim 1
$$
uniformly over $u_1\in\mathcal U$. This gives the second asserted claim and completes the proof of the lemma.
\end{proof}

\begin{lemma}\label{Lemma:PrimitiveD2}
Suppose that Condition S holds. Then
\begin{align}
&|u_2-u_1| \lesssim |Q(u_2,x) - Q(u_1,x)| \lesssim |u_2-u_1|,\label{eq: Q lipschitz property}\\
&|f_{Y|X}(Q(u_2,x)|x) - f_{Y|X}(Q(u_1,x)|x)| \lesssim |u_2 - u_1|\label{eq: f lipschitz property},
\end{align}
uniformly over $u_1,u_2\in\mathcal U$ and $x\in\mathcal X$. In addition,
\begin{equation}\label{eq: second derivative of Q}
\left|\frac{\partial^2Q(u,x)}{\partial u^2}\right| \lesssim 1
\end{equation}
uniformly over $u\in\mathcal U$ and $x\in\mathcal X$.
\end{lemma}
\begin{proof}
Since $U | X\sim \Uniform(0,1)$, it follows that for all $u\in\mathcal U$ and $x\in\mathcal X$, we have
$$
u = \int_{-\infty}^{Q(u,x)} f_{Y|X}(y|x)d y,
$$
and so
\begin{equation}\label{eq: first derivative Q}
\frac{\partial Q(u,x)}{\partial u} = \frac{1}{f_{Y|X}(Q(u,x)|x)}.
\end{equation}
Hence,
$$
\frac{\partial^2 Q(u,x)}{\partial u^2} = - \frac{1}{f_{Y|X}^2(Q(u,x)|x)}\cdot f'_{Y|X}(Q(u,x)|x)\cdot \frac{\partial Q(u,x)}{\partial u} = - \frac{f_{Y|X}'(Q(u,x)|x)}{f_{Y|X}^3(Q(u,x)|x)},
$$
where $f_{Y|X}'(y|x)$ denotes the derivative of the function $y\mapsto f_{Y|X}(y|x)$. Hence, it follows from Condition S that \eqref{eq: second derivative of Q} holds uniformly over $u\in\mathcal U$.

Also, combining Condition S and \eqref{eq: first derivative Q} shows that $1\lesssim \partial Q(u,x)/\partial u \lesssim1$ uniformly over $u\in\mathcal U$ and $x\in\mathcal X$, and since $\mathcal U\subset (0,1)$ is a connected compact set (that is, closed interval), this implies that \eqref{eq: Q lipschitz property} holds uniformly over $u_1,u_2\in\mathcal U$ and $x\in\mathcal X$.

Finally, combining Condition S and \eqref{eq: Q lipschitz property} shows that \eqref{eq: f lipschitz property} hold uniformly over $u\in\mathcal U$ and $x\in\mathcal X$. This completes the proof of the lemma.
\end{proof}

\subsection{Maximal Inequalities}\label{Sec:MaxIneq}

In this section we derive some maximal inequalities that are useful to prove main results of the paper. Let $Z_1,\dots,Z_n$ be a sequence of independent random variables taking values in some set $\mathcal Z$.
Let $\mathcal{F}$ be a class of functions defined on $\mathcal Z$, and let
$$
F(z) \geq \sup_{f\in\mathcal F}|f(z)|,\quad z\in\mathcal Z,
$$
be an envelope of $\mathcal F$. Let $\mathbb P_n$ be the empirical measure corresponding to the sequence $Z_1,\dots,Z_n$, and let
$$
\mathbb G_n(f) = \frac{1}{\sqrt n}\sum_{i=1}^n\Big(f(Z_i) - E[f(Z_i)]\Big),\quad f\in\mathcal F,
$$
be the corresponding empirical process on $\mathcal F$. For a probability measure $Q$ and a constant $p>1$, such that $\|F\|_{Q,p} > 0$, we use $N(\varepsilon\|F\|_{Q,p}, \mathcal{F},L_p(Q))$ to denote the minimal number of $L_p(Q)$-balls of radius $\varepsilon\|F\|_{Q,p}$
needed to cover $\mathcal{F}$. Following literature, we refer to $\sup_Q \log N(\epsilon\|F\|_{Q,2},\mathcal F,L_2(Q))$, where the supremum is taken over all finitely-discrete probability measures $Q$, as a uniform entropy number of $\mathcal F$; see Dudley \cite{Dudley2000} for details of the definitions.

Below we derive some bounds on $\sup_{f\in\mathcal F}|\mathbb G_n(f)|$. Importantly, our bounds do not require that the random variables $Z_i$ are identically distributed, and it will be sufficient to assume only that these random variables are independent. In addition, the function class $\mathcal F$ will actually be allowed to depend on the sample size $n$ via the sequence $(m_n)_{n\geq 1}$, that is, we consider the case where $\mathcal F = \mathcal F_m$ and $m = m_n$.

Since $Z_i$'s are not necessarily identically distributed, we will use $E[f^2]$ and $E[f^4]$ to denote $n^{-1}\sum_{i=1}^n E[f(Z_i)^2]$ and $n^{-1}\sum_{i=1}^n E[f(Z_i)^4]$, respectively.
Further, we will say that the covering numbers $N(\epsilon\|F\|_{\mathbb P_{n,2}},\mathcal F, L_2(\mathbb P_n))$, corresponding to the empirical measure $\mathbb P_n$, satisfy the monotonicity hypotheses if
$$
N(\epsilon\|F\|_{\mathbb P_n,2}, \mathcal{F}, L_2(\mathbb P_n)) \leq
n(\epsilon, \mathcal{F},\mathbb P_n) \text{ for all } 0<\epsilon \leq 1,
$$
where
$n(\epsilon, \mathcal{F},\mathbb P_n)$ is such that (i) $\epsilon\mapsto n(\epsilon, \mathcal{F},\mathbb P_n)$ is decreasing, (ii) $\epsilon \mapsto \epsilon \sqrt{\log n(\epsilon, \mathcal{F},\mathbb P_n)}$ is increasing, and (iii) $\epsilon \sqrt{\log n(\epsilon, \mathcal{F},\mathbb P_n)}\to 0$ as $\epsilon \to 0$. Also, we define $\rho(\mathcal{F},\mathbb P_n) = \sup_{f \in \mathcal{F}} (\|f\|_{\mathbb P_n,2}/ \|F\|_{\mathbb P_n,2})$. In what follows, we refer to any function
$x\colon \mathcal Z^n \mapsto \RR$ as $k$-sub-exchangeable if for any $v, w \in
\mathcal Z^n$ and any vectors $\tilde v, \tilde w$ created by the
pairwise exchange of some components in $v$ with the corresponding components in $w$,
we have that $ x(\tilde v) \vee x(\tilde w) \geq [x(v) \vee
x(w)]/k$.

\begin{lemma}[Exponential inequality for separable empirical process]\label{expo3} Consider the setting specified above. Suppose that the empirical process $\{\mathbb G_n(f),f\in\mathcal F\}$ is separable. Also, suppose that the covering numbers $N(\epsilon\|F\|_{\mathbb P_{n,2}},\mathcal F, L_2(\mathbb P_n))$ satisfy the monotonicity hypotheses. Further, let $K>1$ and $\tau \in (0,1)$ be constants,  and  $e_n(\mathcal{F}, \mathbb{P}_n)=e_n(\mathcal{F},Z_1,\ldots,Z_n)$ be a $k$-sub-exchangeable random variable, such that
\begin{equation}\label{eq: exp ineq sym 1}
\|F\|_{\mathbb{P}_n,2} \int_{0}^{\rho(\mathcal{F},\mathbb{P}_n)/4}
\sqrt{\log  n (\epsilon, \mathcal{F}, \mathbb{P}_n)} d \epsilon  \leq
e_n(\mathcal{F}, \mathbb{P}_n)
\end{equation}
and
\begin{equation}\label{eq: exp ineq sym 2}
\sup_{f \in \mathcal{F}} \frac{1}{n}\sum_{i=1}^n \Big(E[f(Z_i)^2] - (E[f(Z_i)])^2\Big) \leq \frac{\tau}{2} (4kcK
e_n(\mathcal{F},\mathbb{P}_n) )^2
\end{equation}
for some universal constant $c>1$. Then
\begin{align*}
&\mathbb{P}\left\{\sup_{f \in \mathcal{F}} |\mathbb{G}_n(f)| \geq
4kcK e_n(\mathcal{F}, \mathbb{P}_n)\right\} \\
&\qquad  \leq  \frac{4}{\tau}
\Ep_{\mathbb{P}}\left(\left[ \int_{0}^{\rho(\mathcal{F},
\mathbb{P}_n)/2}  \epsilon^{-1} n (\epsilon, \mathcal{F}, \mathbb{P}_n)^{-\{ K^2
-1 \}} d\epsilon  \right]\wedge 1 \right) + \tau.
\end{align*}
\end{lemma}
\begin{proof}
See Lemma 18 in Belloni and Chernozhukov \cite{BC-SparseQR} and note that the proof there does not require identically distributed $Z_i$'s since only independence among $Z_i$'s is used.
\end{proof}

The next lemma establishes new maximal inequalities, which are most useful for the purposes of this paper.

\begin{lemma}\label{Lemma:MaxIneq2} Suppose that for all $n\geq 1$ and $0 < \epsilon \leq 1$, we have
\begin{equation}\label{Eq:J}
N(\epsilon\|F_m\|_{\mathbb P_n,2}, \mathcal{F}_m,L_2(\mathbb P_n)) \leq ( \omega / \epsilon )^{J(m)^2}
\end{equation}
\noindent
for some $\omega = \omega_n$ such that $\log \omega \lesssim \log n$ and some $J(m) = J(m_n)$ such that $J(m)\geq 1$, where $F_m \geq \sup_{f \in \mathcal{F}_m} |f|$ is an envelope of $\mathcal{F}_m$. Let $M_m$ be a random variable that satisfies $\max_{1\leq i\leq n}F_m(Z_i) \lesssim_P M_m$. Then we have the following results.

\noindent
1. A Maximal Inequality Based on Entropy and Moments:
{\small
\begin{align*}
&\sup_{f \in \mathcal{F}_m} | \Gn (f) |\\
&\quad  \lesssim_P  J(m) \sup_{f \in
 \mathcal{F}_m}\( \Ep [f^2]   +  J(m) \Big(\En [f^4] + \Ep[f^4] + n^{-1}\En[F_m^4] \Big)^{1/2}\Big(\frac{\log n}{n}\Big)^{1/2} \)^{1/2} \log^{1/2}n.
\end{align*}}
2. A Maximal Inequality Based on Entropy, Moments,
and Random Extremum:
{\small
$$
 \sup_{f \in \mathcal{F}_m} | \Gn (f ) | \lesssim_P J(m) \(\sup_{f \in \mathcal{F}_m}\(
 \Ep [f^2]   + J(m)\Big(\frac{E[f^4]\log n}{n}\Big)^{1/2}\) +  J(m)^2  \frac{M_m^2\log n}{n} \)^{1/2}\log^{1/2}n.
$$}
3. A Maximal Inequality Based on Entropy, Moments,
and Non-Random Extremum:
{\small
$$
\sup_{f \in \mathcal{F}_m} | \Gn (f ) | \lesssim_P J(m) \(\sup_{f \in \mathcal{F}_m}
 \Ep [f^2]  +  J(m)^2  \frac{\bar F_m^2\log n}{n} \)^{1/2}\log^{1/2}n,
$$
}
if there exists a non-stochastic constant $\bar F_m = \bar F_{m,n}\in\mathbb R$ such that $\sup_{z\in\mathcal Z}F_m(z)\leq \bar F_m$.

\end{lemma}
\begin{proof} Note that if the condition \eqref{Eq:J} holds for some $\omega$, it also holds with $\omega$ replaced by any $\omega' > \omega$. Therefore, it is without loss of generality to assume that $\omega\geq 1$, which we do. Also, since $\log \omega \lesssim \log n$, there exists a constant $C$ such that $\log \omega \leq C\log n$. This constant will be used later in the proof.

We divide the proof into three steps. Step 1 consists of the main argument, Step 2 is an application of Lemma \ref{expo3}, and Step 3 contains some auxiliary calculations.

\noindent
{\bf Step 1} (Main Argument). We first prove the first asserted claim. By Step 2 below, we have
\begin{equation}\label{Eq:First}
 \sup_{f \in \mathcal{F}_m} | \Gn (f ) | \lesssim_P J(m) \sup_{f\in \mathcal{F}_m} \Big(\En[f^2] + \Ep[f^2] + n^{-1}\En[F_m^2]\Big)^{1/2}\log^{1/2}n.
\end{equation}
Also by Step 2,
\begin{equation}\label{eq: exp ineq step 2 - 2}
\sup_{f \in \mathcal{F}_m} {|\Gn (f^2)|} \lesssim_P J(m) \sup_{f\in \mathcal{F}_m} \Big(\En[f^4] +\Ep[f^4] + n^{-1}\En[F_m^4]\Big)^{1/2}\log^{1/2}n.
\end{equation}
Now, by the triangle inequality and \eqref{eq: exp ineq step 2 - 2},
\begin{align}
&\sup_{f\in\mathcal F_m}\En[f^2]
\leq \sup_{f\in\mathcal F_m}\Ep[f^2] + n^{-1/2}\sup_{f\in\mathcal F_m}|\mathbb G_n(f^2)|\label{eq: fourth moment bound 1}\\
&\quad \lesssim_P \sup_{f\in\mathcal F_m}\Ep[f^2] + n^{-1/2} J(m) \sup_{f\in \mathcal{F}_m} \Big(\En[f^4] +\Ep[f^4] + n^{-1}\En[F_m^4]\Big)^{1/2}\log^{1/2}n.\label{eq: fourth moment bound 2}
\end{align}
Substituting this bound into \eqref{Eq:First} and using inequalities $\En[F_m^2] \leq (\En[F_m^4])^{1/2}$, $J(m)\geq 1$, and $\log^{1/2}n\geq 1$ gives the first asserted claim.

Next, we prove the second asserted claim. Using \eqref{eq: fourth moment bound 1}-\eqref{eq: fourth moment bound 2} and noting that $\mathbb E_n[f^4]\lesssim_P \mathbb E_n[f^2] M^2_m$ uniformly over $f\in\mathcal F_m$ and that $\En[F_m^4] \leq M_n^4$ gives
\begin{align*}
\sup_{f\in\mathcal F_m}\En[f^2]
&\lesssim_P \sup_{f\in\mathcal F_m}\Big(\Ep[f^2] + J(m) \Big(\frac{(E[f^4] + M_n^4/n)\log n}{n}\Big)^{1/2} \Big) \\
&\qquad + J(m) \Big(\frac{M_m^2 \log n}{n}\Big)^{1/2}\sup_{f \in \mathcal{F}_m}\Big(\En [f^2]\Big)^{1/2}.
\end{align*}
Since for positive numbers $a$, $c$, and $x$, $x\leq a + c|x|^{1/2}$ implies that $x \leq 4a + 4c^2$, this inequality gives
\begin{align*}
\sup_{f \in \mathcal{F}_m} \En [f^2]
&\lesssim_P  \sup_{f \in
\mathcal{F}_m}\left( \Ep [f^2] + J(m) \Big(\frac{(E[f^4] + M_n^4/n)\log n}{n}\Big)^{1/2}  \right) + J(m)^2 \frac{M_m^2\log n}{n}\\
&\lesssim_P  \sup_{f \in
\mathcal{F}_m}\left( \Ep [f^2] + J(m) \Big(\frac{E[f^4]\log n}{n}\Big)^{1/2}  \right) + J(m)^2 \frac{M_m^2\log n}{n}.
\end{align*}
Substituting this bound into (\ref{Eq:First}) and noting that $\En[F_m^2] \leq M_n^2$ gives the second asserted claim.

Finally, the third asserted claim follows from the second one by substituting $\bar F_m$ instead of $M_m$ and using the inequality $E[f^4]\leq \bar F_m E[f^2]$, which holds for any $f\in\mathcal F_m$.

\noindent
{\bf Step 2} (Applying Lemma \ref{expo3}). Here we prove \eqref{Eq:First} and \eqref{eq: exp ineq step 2 - 2}.
In fact, note that by Lemma \ref{Lemma:ProductEntropy}, \eqref{Eq:J} implies that
$$
N(\epsilon\|F^2_m\|_{\mathbb P_n,2}, \mathcal{F}^2_m,L_2(\mathbb P_n)) \leq ( 2\omega / \epsilon )^{2J(m)^2},
$$
so that  \eqref{eq: exp ineq step 2 - 2} follows from the same argument as that used for \eqref{Eq:First}, so we only prove \eqref{Eq:First}. To do so, we apply Lemma \ref{expo3} to $\mathcal F  = \mathcal{F}_m$ with
$$
\tau = \tau_m = \frac{4}{J(m)^2 (K^2-1)}
$$
for some large constant $K$ to be set later, and
$$
e_n(\mathcal{F}_m,\Pn)  =  \sqrt{C+2} J(m)\sup_{f\in \mathcal{F}_m}\left((\En[f^2])^{1/2} + (\Ep[f^2])^{1/2} + (n^{-1}\En[F_m^2])^{1/2}\right)\log^{1/2}n,
$$
where $C$ is the constant appearing in the beginning of the proof.

We first verify that conditions of Lemma \ref{expo3} hold for all $n$ large enough. Note that by (\ref{Eq:J}), the covering numbers $N(\epsilon\|F_m\|_{\mathbb P_{n,2}},\mathcal F_m, L_2(\mathbb P_n))$ satisfy the monotonicity hypotheses for $n(\epsilon,\mathcal F_m,\mathbb P_n) = (\omega/\epsilon)^{J(m)^2}$. Next, by Step 3 below, the function $(Z_1,\dots,Z_n)\mapsto \sup_{f \in \mathcal{F}_m}\|f\|_{\mathbb{P}_n,2}$ is $\sqrt{2}$-sub-exchangeable, and so is $e_n(\mathcal{F}_m,\mathbb{P}_n)$.
 Also, denoting $\tilde \rho = n^{-1/2} \vee \rho(\mathcal F_m,\mathbb P_n)/4$, we have
  \begin{align*}
   \|F_m\|_{\mathbb{P}_n,2} \int_{0}^{\rho(\mathcal{F}_m,\mathbb{P}_n)/4}  \sqrt{\log  n (\epsilon, \mathcal{F}_m,\mathbb P_n)} d \epsilon
 & \leq   \|F_m\|_{\mathbb{P}_n,2} \int_{0}^{\rho(\mathcal{F}_m,\mathbb{P}_n)/4} J(m)\sqrt{\log(\omega / \epsilon)} d\epsilon  \\
 &\leq \|F_m\|_{\mathbb P_n,2} \int_0^{\tilde \rho}J(m)\sqrt{\log(\omega/\epsilon)}d\epsilon\\
 &\leq \|F_m\|_{\mathbb P_n,2}J(m)\tilde\rho \sqrt{(C + 2)\log n}\\
& \leq  e_n(\mathcal{F}_m, \mathbb{P}_n),
\end{align*}
where the third line follows from
\begin{align*}
\int_0^{\tilde \rho}
\sqrt{\log(\omega/\epsilon)} d\epsilon
&\leq \Big( \int_0^{\tilde \rho} 1
d\epsilon \Big)^{1/2}\Big(\int_0^{\tilde \rho} \log(\omega/\epsilon)
d\epsilon \Big)^{1/2}\\
&=\tilde \rho^{1/2} \Big(\int_0^{\tilde \rho} (\log \omega + \log (1/\epsilon)) d\epsilon\Big)^{1/2}\\
&=\tilde \rho^{1/2}\Big(\tilde \rho\log \omega + \tilde \rho\log(1/\tilde \rho) + \tilde \rho\Big)^{1/2}
\leq \tilde\rho \sqrt{(C+ 2)\log n },
\end{align*}
since $\tilde \rho \geq n^{-1/2}$ and $\log\omega \leq C\log n$. Hence, condition \eqref{eq: exp ineq sym 1} is satisfied.
Moreover, the condition \eqref{eq: exp ineq sym 2} is satisfies by our choice of $e_n(\mathcal F_m,\mathbb P_n)$ for sufficiently large $n$.


Therefore, all conditions of Lemma \ref{expo3} are satisfied and its application gives
\begin{align}
\mathbb{P} \Big( \sup_{f \in \mathcal{F}_m} |
\mathbb{G}_n(f)|  >   4 \sqrt{2}c K e_n(\mathcal{F}_m,
\mathbb{P}_n)\Big)
&  \leq  \frac{4}{\tau_m}\int_0^{1/2} \frac{\omega^{-J(m)^2(K^2 - 1)}}{\epsilon^{1 - J(m)^2(K^2-1)}}  d \epsilon + \tau_m \notag\\
&  \leq  \frac{4}{\tau_m} \frac{(1/(2\omega))^{J(m)^2[K^2-1]}}{J(m)^2[K^2-1]} + \tau_m\notag\\
& = \Big(\frac{1}{2\omega}\Big)^{J(m)^2[K^2-1]} + \frac{4}{J(m)^2(K^2 - 1)},\label{eq: exp ineq step 2 last}
\end{align}
where we used the inequality $\rho(\mathcal{F}_m,\Pn)\leq 1$. Since $\omega\geq 1$ and $J(m)\geq 1$, the expression in \eqref{eq: exp ineq step 2 last} can be made arbitrarily small by setting $K$ sufficiently large. Hence, \eqref{Eq:First} follows.

\noindent
{\bf Step 3} (Auxiliary calculations). Here we establish that the function $(Z_1,\dots,Z_n)\mapsto \sup_{f\in\mathcal F_m}(\En[f(Z_i)^2 + F_m(Z_i)^2/n])^{1/2}$ mapping $\mathcal Z^n$ into $\mathbb R$ is $\sqrt 2$-sub-exchangeable. Indeed, let $Z = (Z_1,\dots,Z_n)$ and $Y = (Y_1,\dots,Y_n)$ be two elements of $\mathcal Z^n$ and define $\tilde Z$ and $\tilde Y$
by exchanging some components in $Z$ with corresponding components in $Y$. Then
\begin{align*}
&2\Big(\sup_{f\in\mathcal F_m}(\En[f(\tilde Z_i)^2 + F_m(\tilde Z_i)^2/n]) \vee \sup_{f\in\mathcal F_m}(\En[f(\tilde Y_i)^2 + F_m(\tilde Y_i)^2/n])\Big)\\
&\quad \geq \sup_{f\in\mathcal F_m}(\En[f(\tilde Z_i)^2 + F_m(\tilde Z_i)^2/n]) + \sup_{f\in\mathcal F_m}(\En[f(\tilde Y_i)^2 + F_m(\tilde Y_i)^2/n]) \\
&\quad \geq \sup_{f\in\mathcal F_m}\Big((\En[f(\tilde Z_i)^2 + F_m(\tilde Z_i)^2/n]) + (\En[f(\tilde Y_i)^2 + F_m(\tilde Y_i)^2/n])\Big)\\
&\quad \geq \sup_{f\in\mathcal F_m}\Big((\En[f(Z_i)^2 + F_m(Z_i)^2/n]) + (\En[f(Y_i)^2 + F_m(Y_i)^2/n])\Big)\\
&\quad \geq \sup_{f\in\mathcal F_m}(\En[f(Z_i)^2 + F_m(Z_i)^2/n]) \vee \sup_{f\in\mathcal F_m} (\En[f(Y_i)^2 + F_m(Y_i)^2/n]).
\end{align*}
This gives the asserted claim and completes the proof of the lemma.
\end{proof}

\subsection{Uniform Entropy Numbers}

\begin{lemma}[Uniform Entropy of VC classes]\label{VCclass} Suppose that the class of functions $\mathcal{F}$ has VC index $V$ and an envelope $F$. Then for some absolute constants $c$ and $C$,
$$
\sup_Q N(\epsilon\|F\|_{Q,2},
\mathcal{F},L_2(Q)) \leq ( C/\epsilon )^{c V},\quad\text{for all }0<\epsilon\leq 1,
$$
where $Q$ ranges over all finitely-discrete probabilities measures.
\end{lemma}
\begin{proof}
The bound follows from Theorem 2.6.7 in van~der Vaart and Wellner \cite{vdV-W}.
\end{proof}


\begin{lemma}[Uniform entropy for products and sums]\label{Lemma:ProductEntropy}
Let $\mathcal{F}$ and $\mathcal{G}$ be two classes of functions with envelopes $F$ and $G$ respectively. Then the uniform entropy numbers of
$ \mathcal{F} \mathcal{G} = \{ f g\colon f \in \mathcal{F}, \ g \in \mathcal{G} \}$
satisfy
\begin{align*}
&\sup_Q \log N(\epsilon\|FG\|_{Q,2},
\mathcal{F}\mathcal{G},L_2(Q))\\
&\qquad  \leq \sup_Q \log N\Big(\frac{\epsilon \|F\|_{Q,2}}{2},
\mathcal{F},L_2(Q)\Big) + \sup_Q \log N\Big(\frac{\epsilon\|G\|_{Q,2}}{2},
\mathcal{G},L_2(Q)\Big)
\end{align*}
for all $\epsilon > 0$. Also, the uniform entropy numbers of $\mathcal F + \mathcal G = \{f + g\colon f\in\mathcal F, \ g\in\mathcal G\}$ satisfy
\begin{align*}
&\sup_{Q} \log N(\epsilon\|F + G\|_{Q,2},\mathcal F + \mathcal G,L_2(Q))\\
&\qquad  \leq \sup_Q \log N\Big(\frac{\epsilon\|F\|_{Q,2}}{2},
\mathcal{F},L_2(Q)\Big) + \sup_Q \log N\Big(\frac{\epsilon\|G\|_{Q,2}}{2},
\mathcal{G},L_2(Q)\Big)
\end{align*}
for all $\epsilon > 0$. In both cases, $Q$ ranges over all finitely-discrete probability measures.
\end{lemma}
\begin{proof}
The result is proven in the proof of Theorem 3 of Andrews \cite{Andrews1994}.
\end{proof}


\begin{lemma}\label{Lemma:VC-F}
For any $r>0$, define the class of functions
\begin{multline*}
\mathcal{F}_{m,n} = \Big\{ (Z,Y)\mapsto (\alpha'Z)\cdot \Big(1\{Y\leq Z'\beta\} - 1\{Y \leq Z'\beta(u)\}\Big) \colon \\
 u \in \mathcal{U}, \  \alpha\in S^{m-1}, \  \| \beta - \beta(u) \| \leq
r  \Big\},
\end{multline*}
mapping $\mathbb R^m\times \mathbb R$ into $\mathbb R$, and let $F_{m,n}(Z,Y) = \|Z\|$ be its envelope. The uniform entropy numbers of $\mathcal  F_{m,n}$ satisfy
$$
\sup_Q \log N(\epsilon\|F_{m,n}\|_{Q,2},\mathcal F_{m,n},L_2(Q))\lesssim m\log (1/\epsilon),\quad\text{uniformly over }0<\epsilon\leq 1.
$$
\end{lemma}
\begin{proof}
Consider function classes
$$
\mathcal{W}_1= \Big\{(Z,Y)\mapsto \alpha'Z\colon
\alpha \in \RR^m\Big\}
\text{ and }\mathcal{V}_1= \Big\{(Z,Y)\mapsto 1\{Y \leq
Z'\beta \}\colon \beta \in \RR^m\Big\}.
$$
Their VC indices are bounded by $m+2$ by Lemmas 2.6.15 and 2.6.18 in van~der Vaart and Wellner \cite{vdV-W}. Hence, since any $f\in\mathcal F_{m,n}$ can be written as $f = g\cdot(v - p)$ for $g\in\mathcal W_1$, $v\in\mathcal V_1$, and $p\in\mathcal V_1$, the asserted claim follows from Lemmas \ref{VCclass} and \ref{Lemma:ProductEntropy}.
\end{proof}





\begin{lemma}\label{Lemma:VC-H} 
For any $r,h >0$, define the class of
functions
$$
\mathcal{H}_{m,n} =\Big\{(Z,Y)\mapsto (\alpha'Z)^2\cdot 1\{|Y - Z'\beta|\leq h\}
\colon u\in\mathcal U, \ \alpha\in S^{m-1}, \ \| \beta - \beta(u)\| \leq r \Big\},
$$
mapping $\mathbb R^m\times\mathbb R$ into $\mathbb R$, and let $H_{m,n}(Z,Y) = \|Z\|$ be its envelope. The uniform entropy numbers of $\mathcal  H_{m,n}$ satisfy
$$
\sup_Q \log N(\epsilon\|H_{m,n}\|_{Q,2},\mathcal H_{m,n},L_2(Q))\lesssim m\log (1/\epsilon),\quad\text{uniformly over }0<\epsilon\leq 1.
$$
\end{lemma}

\begin{proof}
Consider function classess
\begin{align*}
\mathcal{W}_2&= \Big\{(Z,Y)\mapsto 1\{Y - Z'\beta - a\leq
0 \}\colon \beta \in \RR^m, \ a\in\mathbb R\Big\},\\
\mathcal{V}_2&= \Big\{(Z,Y)\mapsto 1\{Y - Z'\beta - a<
0 \}\colon \beta \in \RR^m, \ a\in\mathbb R\Big\}.
\end{align*}
Their VC indices are $m+3$ by Lemmas 2.6.15 and 2.6.18 in van~der Vaart and Wellner \cite{vdV-W}. Hence, since any $f\in\mathcal H_{m,n}$ can be written as $f = g^2\cdot(v - p)$ for $g\in\mathcal W_1$, $v\in\mathcal W_2$, and $p\in \mathcal V_2$, the asserted claim follows from Lemmas \ref{VCclass} and \ref{Lemma:ProductEntropy}.
 \end{proof}

\begin{lemma}\label{Lemma:VC-G}
Define the class of functions
$$
\mathcal{G}_{m,n} = \Big\{(Z,Y)\mapsto (\alpha'Z)\cdot \Big(1\{Y\leq Z'\beta(u)\} - u\Big)\colon  u \in
\mathcal{U}, \  \alpha \in S^{m-1} \Big\},
$$
mapping $\mathbb R^m\times\mathbb R$ into $\mathbb R$, and let $G_{m,n}(Z,Y) = \|Z\|$ be its envelope. The uniform entropy numbers of $\mathcal  G_{m,n}$ satisfy
$$
\sup_Q \log N(\epsilon\|G_{m,n}\|_{Q,2},\mathcal G_{m,n},L_2(Q))\lesssim m\log (1/\epsilon),\quad\text{uniformly over }0<\epsilon\leq 1.
$$
\end{lemma}
\begin{proof}
Consider the function class
$$
\mathcal W_3 = \Big\{(Z,Y)\mapsto u\colon u\in\mathcal U\Big\}.
$$
Its VC index is $O(1)$. Hence, since any function $f\in\mathcal G_{m,n}$ can be written as $f = g\cdot(v - p)$ for $g\in\mathcal W_1$, $v\in\mathcal V_1$, and $p\in\mathcal W_3$, the asserted claim follows from Lemmas \ref{VCclass} and \ref{Lemma:ProductEntropy}.
\end{proof}

\begin{lemma}\label{Lemma:VC-A}
Define the class of functions
$$
\mathcal{A}_{m,n} = \Big\{(X,Y)\mapsto (\alpha'Z(X))\cdot \Big( 1\{Y\leq Q(u,X) \} -  1\{Y\leq Z(X)'\beta(u) \}\Big) \colon u \in \mathcal{U}, \  \alpha \in S^{m-1} \Big\},
$$
mapping $\mathcal X\times \mathbb R$ into $\mathbb R$, and let $A_{m,n}(X,Y) = \|Z(X)\|$ be its envelope. The uniform entropy numbers of $\mathcal  A_{m,n}$ satisfy
$$
\sup_Q \log N(\epsilon\|A_{m,n}\|_{Q,2},\mathcal A_{m,n},L_2(Q))\lesssim m\log (1/\epsilon),\quad\text{uniformly over }0<\epsilon\leq 1.
$$
\end{lemma}
\begin{proof}
Consider the function class
$$
\mathcal V_3 = \Big\{(X,Y)\mapsto 1\{Y\leq Q(u,X)\} \colon u\in\mathcal U\Big\}.
$$
Note that since $u\mapsto Q(u,x)$ is increases for all $x\in\mathcal X$, it follows that
$$
\Big\{(X,Y)\colon Y \leq Q(u_1,X)\Big\}\subset \Big\{(X,Y)\colon Y \leq Q(u_2,X)\Big\}
$$
for all $u_1,u_2\in\mathcal U$ with $u_1<u_2$. Therefore, VC index of $\mathcal V_3$ is $O(1)$. Hence, the asserted claim follows from the same argument as that used in the proof of Lemma \ref{Lemma:VC-F}.
\end{proof}


\subsection{Eigenvalues of Gram Matrices}
Consider the maximum between the maximum eigenvalues associated with the empirical Gram matrix and the population Gram matrix:
\begin{equation}\label{Def:phi}
\phi_n=\max_{\alpha \in S^{m-1}} \En\[ (\alpha'Z_i)^2 \] \vee \Ep\[ (\alpha'Z_i)^2 \].
\end{equation}
The factor $\phi_n$ will be used to bound the quantities $\epsilon_0(m,n)$ and $\epsilon_1(m,n)$ in the analysis for the rates of convergence. To bound $\phi_n$, we use the following result due to Gu\'edon and Rudelson \cite{GuedonRudelson} specialized to our framework:

\begin{lemma}[Gu\'edon and Rudelson \cite{GuedonRudelson}]\label{Lemma:Rudelson} Let $(Z_i)_{i=1}^n$ be i.i.d. random vectors in $\mathbb R^m$. Suppose that
$$
\delta^2 := \frac{\log n}{n} \cdot \frac{\Ep[\max_{1\leq i\leq n} \|Z_i\|^2]}{\max_{\alpha \in S^{m-1}} \Ep[(Z_1'\alpha)^2]} < 1.
$$
Then we have
$$\Ep\[ \max_{\alpha \in S^{m-1}} \left| \frac{1}{n}\sum_{i=1}^n (Z_i'\alpha)^2 - \Ep[ (Z_i'\alpha)^2] \right| \] \leq 2\delta  \cdot \max_{\alpha \in S^{m-1}} \Ep[(Z_1'\alpha)^2].$$
\end{lemma}

\begin{corollary}\label{Corollary:phin}
Denote $\lambda_{max} = \max_{\alpha \in S^{m-1}} \Ep[(Z_1'\alpha)^2]$. Suppose that Condition S holds. In addition, suppose that $\zeta_m^2\log n = o(n)$. Then for all sufficiently large $n$, we have for $\phi_n$ defined in (\ref{Def:phi}) that
$$
\Ep\[\phi_n\] \leq \( 1 +  2\sqrt{\frac{\zeta_m^2\log n}{n\lambda_{max}}} \) \lambda_{max}  \ \ \mbox{and}\ \ P( \phi_n > 2 \lambda_{max} ) \leq 2\sqrt{\frac{\zeta_m^2\log n}{n\lambda_{max}}}.
$$
\end{corollary}
\begin{proof}
Let $\delta$ be defined as in Lemma \ref{Lemma:Rudelson}. Note that $\|Z_i\|^2 \leq \zeta_m^2$ for all $i=1,\dots,n$ and $1\lesssim \lambda_{\max}\lesssim 1$ by Condition S. Hence, $\delta^2 \lesssim \zeta_m^2\log n/n$, and so $\delta^2<1$ for all $n$ large enough under our assumption that $\zeta_m^2 \log n = o(n)$. Therefore, the first result follows by applying Lemma \ref{Lemma:Rudelson} and the triangle inequality.

To show the second result, note that the event $\{\phi_n>2\lambda_{max}\}$ cannot occur if $\phi_n = \max_{\alpha \in S^{m-1}} \Ep[(Z_1'\alpha)^2] = \lambda_{max}$. Thus,
\begin{align*}
P( \phi_n > 2 \lambda_{max} ) & = P\Big( \max_{\alpha \in S^{m-1}} \En[(Z_i'\alpha)^2] > 2 \lambda_{max} \Big) \\
&  \leq P\Big( \max_{\alpha \in S^{m-1}} \left| \En[(Z_i'\alpha)^2] - \Ep\[ (Z_i'\alpha)^2\] \right| > \lambda_{max}\Big) \\
&  \leq  \Ep\Big[ \max_{\alpha \in S^{m-1}} \left| \En[(Z_i'\alpha)^2] - \Ep\[ (Z_i'\alpha)^2\] \right| \Big]/ \lambda_{max} \leq 2\delta,
\end{align*}
 by the triangle inequality, Markov's inequality, and Lemma \ref{Lemma:Rudelson}. This completes the proof of the corollary.
\end{proof}


\subsection{Estimation of Matrices}

Some of the procedures proposed in this paper rely on estimators of the Gram and Jacobian matrices defined in (\ref{inf-est-hatSigma0}) and \eqref{inf-est-hatJ}, respectively. The following result
establishes rates of convergence of these estimators.

\begin{lemma}[Estimation of Gram and Jacobian Matrices]\label{covariance}
Suppose that Condition S holds. Then
\begin{equation}\label{eq: sigma estimator bound}
\|\hat \Sigma - \Sigma\| \lesssim_P \sqrt{\frac{\zeta_m^2\log n}{n}}
\end{equation}
for $\hat\Sigma$ defined in \eqref{inf-est-hatSigma0} as long as $\zeta_m^2 \log n = o(n)$. In addition,\begin{equation}\label{eq: j estimator bound}
 \sup_{u\in \mathcal{U}}  \|\hat J(u) - J(u)\| \lesssim_P \sqrt{\frac{m \zeta_m^2 \log n}{n\hn}} + m^{-\kappa} + \hn
\end{equation}
for $\hat J(u)$ defined in \eqref{inf-est-hatJ} as long as $\hn = o(1)$, $m\zeta^2_m\log^2 n = o(n h)$, and $m^{-\kappa} \log n = o(1)$.
\end{lemma}
\begin{proof}
Note that \eqref{eq: sigma estimator bound} follows from Lemma \ref{Lemma:Rudelson}. Hence, it suffices to show \eqref{eq: j estimator bound}. To do so, observe that conditions of the lemma imply Conditions of Theorem \ref{Thm:SeriesRates}, and so for any $\varphi\in(0,1)$, there exists $B > 0$ such that with probability at least $1 - \varphi$, we have $(u,\widehat \beta(u))\in R_{n,m}$ for all $u\in\mathcal U$ for $R_{n,m}$ defined in \eqref{eq: Rnm definition} with $r_n = B\sqrt{m/n} = o(1)$. For this $R_{n,m}$, define
\begin{eqnarray}\label{EQ: est Q and J appendix}
\label{EQ: est J} \epsilon_4(m,n) &=&
\frac{1}{2\sqrt n\hn} \sup_{ { \alpha \in S^{m-1}, \atop (u, \beta) \in
R_{m,n} } } \Big|  \Gn\Big(1\{ |Y_i - Z_i'\beta| \leq \hn \}
(\alpha' Z_i)^2\Big) \Big|, \\
\label{eq: epsilon 6 definition}
\epsilon_5(m,n) &=&  \sup_{ { \alpha \in S^{m-1}, \atop (u,
\beta) \in R_{n,m} } }\left| \frac{1}{2\hn}\ E\Big[1\{ |Y_i - Z_i'\beta| \leq
\hn \} (\alpha' Z_i)^2\Big] - \alpha' J(u) \alpha  \right|.
\end{eqnarray}
Then on the event that $(u,\hat\beta(u))\in R_{n,m}$ for all $u\in\mathcal U$,
$$
\sup_{u\in\mathcal{U}}\|\hat J(u) - J(u)\| \leq \epsilon_4(m,n) + \epsilon_5(m,n).
$$
But by Lemma \ref{bound e6 and e7},
\begin{align*}
\epsilon_4(m,n) + \epsilon_5(m,n)
& \lesssim_P \sqrt{\frac{m \zeta_m^2\log n}{n\hn}} + \frac{m\zeta_m^2\log n}{n\hn} + m^{-\kappa} + \zeta_m \sqrt{m/n}+ \hn \notag\\
& \lesssim \sqrt{\frac{m \zeta_m^2 \log n}{n\hn}} + m^{-\kappa} + \hn.
\end{align*}
Hence, \eqref{eq: epsilon 6 definition} follows. This completes the proof of the lemma.
\end{proof}

\subsection{Bounds on Various Empirical Errors}\label{Sec:VerifyE-Q}
Here we provide probabilistic bounds for the error terms $\epsilon_0(m,n)$--$\epsilon_5(m,n)$ along with some auxiliary bounds. Our results rely on empirical processes techniques; in particular, they rely on the maximal inequalities derived in Section \ref{Sec:MaxIneq}.

\begin{lemma}[Bounds on Approximation Error for Uniform Linear Approximation]\label{bound ApproxULA} Under Condition S,
\begin{equation}\label{eq: ru definition}
\tilde r_u := \frac{1}{\sqrt{n}} \sum_{i=1}^n Z_i\Big(1\{ Y_i \leq Q(u,X_i)\}-1\{ Y_i \leq Z_i'\beta(u)\}\Big), \ \ u \in \mathcal{U},
\end{equation}
satisfies
\begin{equation}\label{eq: ru bound uniform}
\sup_{u \in \mathcal{U}} \|\tilde r_u\| \lesssim_P \sqrt{m^{1-\kappa} \log n }
 +  \frac{\zeta_m m \log n}{\sqrt{n}}.
\end{equation}
In addition, $E[\tilde r_u] = 0$ for all $u\in\mathcal U$.
\end{lemma}

\begin{proof} We first prove the second asserted claim. Fix $u\in\mathcal U$. Note that
$$
E\Big[Z \cdot1\{Y \leq Q(u,X)\}\Big] = E\Big[Z \cdot1\{Q(U,X) \leq Q(u,X)\}\Big] = E\Big[Z\cdot 1\{U \leq u\}]\Big] = u\Ep[Z].
$$
Also, since $\beta(u)$ is the solution of the optimization problem \eqref{define beta}, the first-order conditions imply that
$$
E\Big[Z\cdot 1\{Y\leq Z'\beta(u)\}\Big] = uE[Z].
$$
Thus,
$$
E\Big[Z\cdot (1\{Y\leq Q(u,X)\} - 1\{Y\leq Z'\beta(u)\})\Big] = 0,
$$
and so $E[\tilde r_u] = 0$, which is the second asserted claim.

To prove the first asserted claim, recall the class of functions $\mathcal A_{m,n}$ defined in Lemma \ref{Lemma:VC-A}. Since $E[\widetilde r_u]=0$ for all $u\in\mathcal{U}$, it follows that
\begin{align*}
\sup_{u \in \mathcal{U}} \|\tilde r_u\| & = \sup_{f \in \mathcal{A}_{m,n}} |
\Gn(f)| \\
& \lesssim_P  J(m)  \Big(\sup_{f \in \mathcal{A}_{m,n}}
 \Ep\[ f^2\]   +  n^{-1}   J(m)^2 \bar F_m^2  \log n  \Big)^{1/2}\log^{1/2} n
\end{align*}
by the third maximal inequality in Lemma \ref{Lemma:MaxIneq2}, where $\bar F_m = \zeta_m$ by Condition $S$ and $J(m)\lesssim \sqrt m$ by Lemma \ref{Lemma:VC-A}. Hence, \eqref{eq: ru bound uniform} holds provided we can show that
\begin{equation} \label{Eq:upper bound on f-squaredSeries}
\sup_{f \in \mathcal{A}_{m,n}}
 \Ep \[f^2\]  \lesssim m^{-\kappa}. 
\end{equation}
In turn, to show \eqref{Eq:upper bound on f-squaredSeries}, note that for any $f \in \mathcal{A}_{m,n}$, there exist $\alpha\in S^{m-1}$ and $u\in\mathcal U$ such that
\begin{align*}
| f(X,Y)| &= |\alpha' Z|\cdot \Big| 1\{Y\leq Q(u,X) \} -  1\{Y\leq Z'\beta(u) \}\Big|\\
& \leq   |\alpha' Z| \cdot 1\Big\{ |Y - Z'\beta(u)| \leq |R(u,X)| \Big\},
\end{align*}
where the second line holds because $R(u,X) = Q(u,X) - Z'\beta(u)$. Thus,
$$
E[|f(X,Y)|^2] \lesssim m^{-\kappa} E[|\alpha' Z|^2] \lesssim m^{-\kappa},
$$
since $Z = Z(X)$ and we have by Condition S that (i) the conditional pdf of $Y$ given $X$ is bounded from above, (ii) $|R(u,X)|\lesssim m^{-\kappa}$, and (iii) all eigenvalues of the matrix $E[Z Z']$ are bounded from above. This gives \eqref{Eq:upper bound on f-squaredSeries} and completes the proof of the lemma.
\end{proof}

\begin{lemma}[Bounds on $\epsilon_0(m,n)$ and $\sup_{u\in\mathcal U}\|U(u)\|$]\label{Lemma:error-rate-0}
Under Condition S, we have
$$
 \epsilon_0(m,n) \lesssim_P \displaystyle  \sqrt{m}\left(1+   \sqrt{m^{-\kappa}\log n} + \frac{\sqrt{m}\zeta_m \log n}{\sqrt{n}}\right) \ \ \mbox{and} \ \ \sup_{u\in\mathcal{U}}\| \mathbb{U}(u)\| \lesssim_P \sqrt{m},  \\
$$
for $\epsilon_0(m,n)$ and $\mathbb U(u)$ defined in \eqref{eq: three epsilons} and \eqref{Def:U}, respectively.
\end{lemma}
\begin{proof}
First, we establish the bound on $\sup_{u\in\mathcal U}\|\mathbb U(u)\|$. We have
$$
\sup_{u\in\mathcal{U}}\| \mathbb{U}(u)\|^2 =  \sup_{u\in\mathcal{U}}\frac{1}{n}\sum_{j=1}^m\left(\sum_{i=1}^n Z_{ij} (u-1\{U_i \leq u\})\right)^2.
$$
Therefore, by the triangle inequality,
$$
E\left[\sup_{u\in\mathcal{U}}\| \mathbb{U}(u)\|^2\right] \leq \frac{1}{n} \sum_{j=1}^m E\left[\sup_{u\in\mathcal{U}}\Big(\sum_{i=1}^n Z_{ij} (u-1\{U_i \leq u\})\Big)^2\right].
$$
In addition, since $E[Z_{i j}(u-1\{U_i \leq u\})]=0$, by Theorem 2.14.1 in van~der Vaart and Wellner \cite{vdV-W},
$$
E\left[\sup_{u\in\mathcal U}\Big(\sum_{i=1}^n Z_{i j}(u - 1\{U_i \leq u\})\Big)^2\right] \lesssim n E[|Z_{1 j}|^2]
$$
uniformly over $j=1,\dots,m$ since the function class
$$
\mathcal G_n = \Big\{(Z,U)\mapsto Z\cdot (u - 1\{U \leq u\}\colon u\in\mathcal U\Big\},
$$
mapping $\mathbb R\times[0,1]$ into $\mathbb R$, has an envelope $F(Z,U) = |Z|$, and its uniform entropy numbers satisfy
$$
\sup_Q N(\epsilon\|G_{n}\|_{Q,2},\mathcal G_{n},L_2(Q))\lesssim (1/\epsilon)^{O(1)},\quad \text{uniformly over }0<\epsilon\leq 1,
$$
by Lemmas \ref{VCclass} and \ref{Lemma:ProductEntropy} (here we used the fact that the function class  $\{U \mapsto 1\{U\leq u\}\colon u\in\mathcal U\}$, mapping $[0,1]$ into $\mathbb R$, has VC index $O(1)$).
In addition, $E[|Z_{1 j}|^2] \leq \|E[Z_1 Z_1']\| \lesssim 1$ uniformly over $j=1,\dots,m$ by Condition S. Combining these inequalities shows that
$E[\sup_{u\in\mathcal{U}}\| \mathbb{U}(u)\|^2] \lesssim  m$,
and so $\sup_{u\in\mathcal{U}}\| \mathbb{U}(u)\| \lesssim_P \sqrt m$ by Markov's inequality, which is the asserted claim for $\sup_{u\in\mathcal{U}}\| \mathbb{U}(u)\|$.

Second, to bound $\epsilon_0(m,n)$, observe that $\epsilon_0(m,n)$ is equal to
\begin{align*}
 & \sup_{u\in\mathcal U}\Big\|\frac{1}{\sqrt n}\sum_{i=1}^n\Big(Z_i\cdot(1\{Y_i\leq Z_i\beta(u)\} - u) -  E\Big[Z_i\cdot(1\{Y_i\leq Z_i\beta(u)\} - u)\Big]\Big)\Big\|\\
&\quad  = \sup_{u\in\mathcal U}\Big\|\frac{1}{\sqrt n}\sum_{i=1}^n\Big(Z_i\cdot(1\{Y_i \leq Q(u,X_i)\} - u) - E\Big[Z_i\cdot(1\{Y_i \leq Q(u,X_i)\} - u)\Big]\Big) - \tilde r_u\Big\|\\
&\quad = \sup_{u\in\mathcal U}\Big\|\frac{1}{\sqrt n}\sum_{i=1}^n Z_i\cdot (1\{U_i \leq u\} - u) -\tilde r_u\Big\| = \sup_{u\in\mathcal U}\Big\|\mathbb U(u) + \tilde r_u\Big\| \leq \sup_{u\in\mathcal U}\|\mathbb U(u)\| + \sup_{u\in\mathcal U}\|\tilde r_u\|,
\end{align*}
where the second line follows from the definition of $\tilde r_u$ in \eqref{eq: ru definition} and the fact that $E[\tilde r_u] = 0$ for all $u\in\mathcal U$, which is proven in Lemma \ref{bound ApproxULA}, and the third line follows from the definition of $\mathbb U(u)$ and the fact that $Y_i = Q(U_i,X_i)$ with the function $u\mapsto Q(u,X_i)$ being increasing for all $i=1,\dots,n$. Hence,
\begin{align*}
\epsilon_0(m,n) &\leq \sup_{u\in\mathcal U}\|\mathbb U(u)\| +\sup_{u\in\mathcal{U}}\| \tilde r_u\|  \\
& \lesssim_P  \sqrt m + \sqrt{m^{1-\kappa}\log n} + \frac{\zeta_m m \log n}{\sqrt{n}}
\end{align*}
by the bound on $\sup_{u\in\mathcal U}\|\mathbb U(u)\|$ derived above and the bound on $\sup_{u\in\mathcal U}\|\tilde r_u\|$ derived in Lemma \ref{bound ApproxULA}. This completes the proof of the lemma.
\end{proof}

\begin{lemma}[Bounds on $\epsilon_1(m,n)$ and $\epsilon_2(m,n)$]\label{bound e1 and e2} Under Condition S, we have
\begin{align}
&\epsilon_1(m,n) \lesssim_P \Big(m  r_n \zeta_m \log n\Big)^{1/2}
 + \frac{m  \zeta_m}{\sqrt{n}} \log n,\label{eq: epsilon1 bound}\\
&\epsilon_2(m,n) \lesssim_P \sqrt{n}r_n^2\zeta_m + \sqrt{n}m^{-\kappa}r_n\label{eq: epsilon2 bound}
\end{align}
for $\epsilon_1(m,n)$ and $\epsilon_2(m,n)$ defined in \eqref{eq: three epsilons}.
\end{lemma}
\begin{proof}
We first prove \eqref{eq: epsilon1 bound}. Recall the class of functions $\mathcal F_{m,n}$ defined in Lemma \ref{Lemma:VC-F}. It follows that
\begin{align*}
\epsilon_1(m,n)
&=\sup_{f\in\mathcal F_{m,n}}|\mathbb G_n(f)|\\
&\lesssim_P J(m)\Big(\sup_{f\in\mathcal F_{m,n}} E[f^2] + n^{-1}J(m)^2\bar F_m^2\log n\Big)^{1/2}\log^{1/2}n
\end{align*}
by the third maximal inequality in Lemma \ref{Lemma:MaxIneq2}, where $\bar F_m = \zeta_m$ by Condition S and $J(m)\lesssim \sqrt m$ by Lemma \ref{Lemma:VC-F}. Hence, \eqref{eq: epsilon1 bound} holds provided we can show that
\begin{equation} \label{Eq:upper bound on f-squared}
\sup_{f \in \mathcal{F}_{m,n}} \Ep [f^2]  \lesssim r_n\zeta_m.
\end{equation}
In turn, to show \eqref{Eq:upper bound on f-squared}, note that for any $f\in\mathcal F_{m,n}$, there exist $\alpha\in S^{m-1}$, $u\in\mathcal U$, and $\beta\in\mathbb R^m$ such that $\|\beta - \beta(u)\|\leq r_n$ and
\begin{align*}
|f(Z,Y)|
& = |\alpha'Z|\cdot \Big|1\{Y\leq Z'\beta\} - 1\{Y\leq Z'\beta(u)\}\Big|\\
&\leq |\alpha'Z|\cdot \Big|1\{|Y - Z'\beta(u)|\leq |Z'(\beta - \beta(u))|\}\Big|\\
&\leq |\alpha'Z|\cdot 1\{|Y - Z'\beta(u)|\leq r_n\zeta_m\},
\end{align*}
where the third line follows from the Cauchy-Schwarz inequality and the observation that $\|Z\| = \|Z(X)\|\leq \zeta_m$, almost surely. Thus,
$$
E[|f(Z,Y)|^2] \lesssim r_n\zeta_m E[|\alpha'Z|^2] \lesssim r_n \zeta_m,
$$
like in the proof of Lemma \ref{bound ApproxULA}. This gives \eqref{Eq:upper bound on f-squared} and completes the proof of \eqref{eq: epsilon1 bound}.

Second, we prove \eqref{eq: epsilon2 bound}. By Lemma \ref{Lemma:JJ}, the matrix  $\widetilde J(u) = E[f_{Y|X}(Z'\beta(u)| X)Z Z']$ satisfies
$$
 \sqrt{n}\Big|\alpha'(J(u)-\widetilde J(u))(\beta -
\beta(u))\Big| \leq \sqrt n\|\alpha\|\cdot\|J(u) - \widetilde J(u)\|\cdot\|\beta - \beta(u)\| \lesssim \sqrt{n}m^{-\kappa}r_n
$$
uniformly over $\alpha\in S^{m-1}$ and $(u,\beta)\in R_{n,m}$. Thus, if we define
$$
\epsilon_2(m,n,\alpha,u,\beta) = \sqrt n\Big|
\alpha'\Big(E[\psi_i(\beta, u)] -
E[\psi_i(\beta(u),u)]\Big) - \alpha'\widetilde J(u)(\beta -
\beta(u))\Big|,
$$
then
$$
 \epsilon_2(m,n) \lesssim \sup_{\alpha \in S^{m-1}, (u,\beta) \in R_{n,m}} \epsilon_2(m,n, \alpha,u,\beta) + \sqrt{n}m^{-\kappa}r_n.
$$
Next, note that by the law of iterated expectations and the mean-value theorem,
$$
\alpha'\Big(E[\psi_i(\beta, u)] - E[\psi_i(\beta(u),u)]\Big) = E\Big[f_{Y|X}(Z'\widetilde\beta|X)\cdot (\alpha' Z)\cdot (Z'(\beta - \beta(u)))\Big]
$$
for some $\widetilde \beta$ on the line segment between $\beta(u)$ and $\beta$. In addition, it follows from Condition S that
$$
\Big|f_{Y|X}(Z'\widetilde\beta|X) - f_{Y|X}(Z'\beta(u)|X)\Big|\lesssim |Z'(\widetilde \beta - \beta(u))| \leq |Z'(\beta - \beta(u))| \lesssim \zeta_m r_n
$$
uniformly over $X\in\mathcal X$, $\alpha\in S^{m-1}$, and $(u,\beta)\in R_{n,m}$. Hence,
\begin{align*}
\epsilon_2(m,n,\alpha,u,\beta) &= \sqrt n \Big| E\Big[ (f_{Y|X}(Z'\widetilde\beta|X) - f_{Y|X}(Z'\beta(u)|X)) \cdot(\alpha'Z) \cdot(Z' (\beta - \beta(u))) \Big]   \Big|\\
&  \lesssim   \sqrt n \zeta_m r_n \cdot E\Big[ \alpha'Z Z'(\beta - \beta(u)) \Big] \lesssim \sqrt n \zeta_m r_n^2
\end{align*}
uniformly over $\alpha\in S^{m-1}$ and $(u,\beta)\in R_{n,m}$ by Condition S. Combining presented inequalities gives \eqref{eq: epsilon2 bound} and completes the proof of the lemma.
\end{proof}


\begin{lemma}[Bound on $\epsilon_3(m,n)$]\label{bound e3} Let $\hat \beta(u)$ be a solution to the perturbed QR problem \eqref{eq: perturbed optimization problem}. If the data are in general position so that (\ref{def:genpos}) holds, we have
$$
\epsilon_3(m,n)  \leq \min\left( \frac{\zeta_m m }{\sqrt{n}}, \phi_n \sqrt m\right)
$$
holds with probability $1$ for $\epsilon_3(m,n)$ and $\phi_n$ defined in \eqref{eq: epsilon_3 definition} and \eqref{Def:phi}, respectively.
\end{lemma}
\begin{proof}
Note that the dual problem associated with the perturbed QR problem \eqref{eq: perturbed optimization problem} is
$$ \max_{(u-1)\leq a_i \leq  u} \En[Y_ia_i] : \En[Z_ia_i] = - \mathcal{A}_n(u).$$
Letting $\widehat a(u)$ denote the solution for the dual problem above, and letting $a_i(\widehat\beta(u)) := ( u - 1\{Y_i \leq Z_i'\hat\beta(u)\} )$, we have by the triangle inequality that
\begin{align*}
\epsilon_3(m,n)
&\leq  \sup_{\alpha\in S^{m-1}, u\in \mathcal{U}} \sqrt{n}\left|\En\[ (Z_i'\alpha) ( a_i(\widehat\beta(u)) - \hat a_i(u))\]\right| \\
& \qquad  + \sup_{u\in \mathcal{U}}\sqrt{n}\| \En\[ Z_i\hat a_i(u)\] + \mathcal{A}_n(u)\|.
\end{align*}
By dual feasibility $\En\[ Z_i\hat a_i(u) \]= -\mathcal{A}_n(u)$, and the second term is identically equal to zero.

Also, note that  $ a_i(\widehat\beta(u))  \neq \hat a_i(u)$ only if the $i$th point is interpolated. Since the data are in general
position, with probability one the quantile regression interpolates $m$ points ($Z_i'\hat\beta(u) = Y_i$ for $m$ points for every $u\in \mathcal{U}$).

Therefore, noting that $| a_i(\widehat\beta(u)) -\hat a_i(u)|\leq 1$, we have
$$
\epsilon_3(m,n) \leq \sup_{\|\alpha\|\leq 1, u\in \mathcal{U}} \sqrt{n}\sqrt{\En\[ (Z_i'\alpha)^2 \]}\sqrt{\En\[ \{  a_i(\widehat\beta(u)) - \hat a_i(u)\}^2\]} \leq \phi_n \sqrt{m}
$$
and, with probability $1$,
$$
\epsilon_3(m,n) \leq \sup_{\|\alpha\|\leq 1, u\in \mathcal{U}} \sqrt{n} \En\[ 1\{ a_i(\widehat\beta(u))  \neq \hat a_i(u)\} \max_{1\leq i \leq n} \|Z_i\|\]
\leq \frac{ m }{\sqrt{n}} \max_{1\leq i \leq n} \|Z_i\|.
$$
This completes the proof of the lemma.
\end{proof}

\begin{lemma}[Bounds on $\epsilon_4(m,n)$ and $\epsilon_5(m,n)$]\label{bound e6 and e7} Suppose that Condition S holds. In addition, suppose that $\hn = h_n = o(1)$ and $r_n=o(1)$. Then
$$
\epsilon_4(m,n) \lesssim_P \sqrt{\frac{\zeta_m^2m\log n}{n \hn
}}  + \frac{m \zeta_m^2}{n \hn}\log n \ \ \mbox{
and} \ \
\epsilon_5(m,n) \lesssim m^{-\kappa} + r_n \zeta_m + \hn
$$
for $\epsilon_4(m,n)$ and $\epsilon_5(m,n)$ defined in \eqref{EQ: est J} and \eqref{eq: epsilon 6 definition}, respectively.
\end{lemma}
\begin{proof} To bound $\epsilon_4(m,n)$, consider the function class
$$
\mathcal{H}_{m,n} =\Big\{(Z,Y)\mapsto (\alpha'Z)^2\cdot 1\{|Y - Z'\beta|\leq h\}
\colon  \alpha\in S^{m-1}, \ (u,\beta)\in R_{n,m}\Big\}
$$
for $R_{n,m}$ defined in \eqref{eq: Rnm definition}.  Then
\begin{align*}
\epsilon_4(m,n) & =  \frac{1}{2\sqrt n \hn}\sup_{f \in
\mathcal{H}_{n,m}} |\Gn(f)|  \\
& \lesssim_P  \frac{1}{\sqrt n \hn} J(m) \(\sup_{f \in
\mathcal{H}_{m,n}}
 \Ep [f^2]   +  n^{-1}   J(m)^2 \bar F_,^2  \log n \)^{1/2}  \log^{1/2} n
\end{align*}
by the third maximal inequality in Lemma \ref{Lemma:MaxIneq2},
where $\bar F_m = \zeta_m^2$ by Condition S and $J(m) \lesssim \sqrt{m}$ by Lemma \ref{Lemma:VC-H}. In addition,
\begin{align*}
\sup_{f \in \mathcal{H}_{m,n}} \Ep [f^2]
& = \sup_{\alpha\in S^{m-1}}\sup_{(u,\beta)\in R_{n,m}} E\Big[(\alpha'Z)^4\cdot 1^2\{|Y - Z'\beta|\leq h\}\Big]  \\
&=\sup_{\alpha\in S^{m-1}}\sup_{(u,\beta)\in R_{n,m}} E\Big[(\alpha'Z)^4\cdot 1\{|Y - Z'\beta|\leq h\}\Big]\\
& \lesssim h\sup_{\alpha\in S^{m-1}}E\Big[(\alpha' Z)^4\Big] \lesssim h\zeta_m^2 \sup_{\alpha\in S^{m-1}} E\Big[(\alpha' Z)^2\Big]\lesssim h\zeta_m^2
\end{align*}
by Condition S. The bound for $\epsilon_4(m,n)$ follows from combining the inequalities above.

To show the bound on $\epsilon_5(m,n)$, let $f'_{Y|X}(y|x)$ denote the derivative of the function $y\mapsto f_{Y|X}(y|x)$. By Condition S, we have for some finite constant $\bar f'$ that $|f_{Y|X}'(y|x)|\leq \bar{f'}$ for all $y\in \mathcal Y_x$ and $x\in \mathcal X$. Therefore,
\begin{align*}
E\Big[ 1\{ |Y-Z'\beta| \leq \hn\} (\alpha'Z)^2\Big] & =  E\Big[ (\alpha'Z)^2 \int_{- \hn}^{\hn} f_{Y|X}(Z'\beta+t|X)dt\Big] \\
& =  E\Big[ (\alpha'Z)^2 \int_{- \hn}^{\hn} \Big(f_{Y|X}(Z'\beta|X) + t f_{Y|X}'(Z'\beta+\tilde t|X)\Big) dt\Big] \\
& =  2\hn E\Big[ f_{Y|X}(Z'\beta|X)(\alpha'Z)^2 \Big] + O\Big( \hn^2\bar f' E[ (Z'\alpha)^2]\Big)\\
& =  2\hn E\Big[ f_{Y|X}(Z'\beta|X)(\alpha'Z)^2 \Big] + O( \hn^2)
\end{align*}
for some $\tilde t$ between 0 and $t$ by the mean-value theorem.
Moreover, for any $(u,\beta) \in R_{m,n}$,
\begin{align*}
E\Big[ f_{Y|X}(Z'\beta|X)(\alpha'Z)^2 \Big] & =  E\Big[ f_{Y|X}(Z'\beta(u)|X)(\alpha'Z)^2 \Big]  \\
& \quad +   E\Big[ ( f_{Y|X}(Z'\beta|X)- f_{Y|X}(Z'\beta(u)|X))(\alpha'Z)^2 \Big] \\
& =  E\Big[ f_{Y|X}(Z'\beta(u)|X)(\alpha'Z)^2 \Big] +  O\Big(E[ \bar{f'}Z'(\beta-\beta(u))(\alpha'Z)^2 ]\Big) \\
& =  \alpha'\widetilde J(u)\alpha + O\Big(r_n \zeta_m E[(\alpha'Z)^2]\Big)\\
& =  \alpha'J(u)\alpha + O(m^{-\kappa}) + O(r_n\zeta_m),
\end{align*}
by Condition S, where $\widetilde J(u) = \Ep[ f_{Y|X}(Z'\beta(u)|X)ZZ']$ and where the last line follows from Lemma \ref{Lemma:Auxiliary Matrix}. Thus, $\epsilon_5(m,n) \lesssim m^{-\kappa} + r_n\zeta_m + h$. This completes the proof of the lemma.
\end{proof}

\section{A Lemma on Strong Approximation for a Process in the Sup-Norm}\label{YSup}
In this section, we develop a novel technique to construct a Gaussian coupling for certain empirical processes in the sup-norm. The technique is based on an extension of the Yurinskii's coupling in the Euclidean norm to arbitrary norm (including the sup-norm). Although the technique can be applied to general empirical processes and may be of independent interest, we apply it to a particular empirical process, which is needed to establish Theorems \ref{thm: gaussian coupling for t process} and \ref{thm: weighted bootstrap alternative conditions}, for brevity of the paper.

\begin{lemma}\label{lemma:LinearStrong} (Approximation of Linear Functionals of a Sequence of Empirical Processes of Increasing Dimension by a Sequence of Gaussian Processes)
Let $(Z_i)_{i=1}^n$ be a sequence of non-stochastic vectors in $\mathbb R^m$ and consider the empirical process $\mathbb{U}_{n}$ in $[\ell^{\infty}(\mathcal{U})]^m$, $\mathcal{U} \subseteq (0,1)$, defined by
$$
\mathbb{U}_{n}(u) = \mathbb{G}_n\(v_iZ_i \psi_i(u)\), \ \ \psi_i(u) = u-1\{U_i \leq u\},\quad u\in\mathcal U,
$$
where $(U_i,v_i)_{i=1}^n$ is an i.i.d. sequence of pairs of independent random variables where $U_i\sim\Uniform (0,1)$, $E[v_i^2] = 1$, $E[|v_i|^4]\lesssim 1$, and $\max_{1\leq i\leq n}|v_i| \lesssim_P \log n$.  Suppose that the vectors $Z_i$ are such that
$$
\sup_{\alpha\in S^{m-1}}\En\[(\alpha'Z_i)^2\] \lesssim 1 \ \text{and} \ \max_{1\leq i\leq n}\|Z_i\| \lesssim \zeta_m,
$$
where $\zeta_m$ satisfies $1/\zeta_m\lesssim 1$. Also, let $\mathcal W$ be a set in $\mathbb R^d$ and let $I$ be a subset of $\mathcal U\times\mathcal W$ whose dimension $d_I$ is independent of $n$ and whose diameter is bounded uniformly over $n$. Moreover, let $\ell\colon I \to \mathbb R^m$ be a function such that
$$
\|\ell(u,w)\|\leq \xi_{\ell} \ \text{and} \ \|\ell(u,w) - \ell(\tilde u,\tilde w)\| \leq L_{\ell}\|(u, w) - (\tilde u, \tilde w)\|
$$
for all $(u,w)$ and $(\tilde u,\tilde w)$ in $I$, where $\xi_{\ell}$ and $L_{\ell}$ satisfy $\xi_{\ell} \lesssim 1$ and $L_{\ell}\gtrsim 1$, and consider the process $\mathbb L_n$ in $\ell^{\infty}(I)$ defined by
$$
\mathbb L_n(u,w) = \ell(u,w)'\mathbb U_n(u),\quad (u,w)\in I.
$$
Finally, suppose that
\begin{equation}\label{eq: growth condition lem 31}
L_{\ell}^{2 d_I}\zeta_m^2 = o(n^{1 - \varepsilon})
\end{equation}
for some constant $\varepsilon > 0$. Then there exists a sequence of zero-mean Gaussian processes $(\Z_n)_{n\geq 1}$ with a.s. continuous paths such that (i) the covariance functions of $\Z_n$ coincide with those of $\mathbb L_n$, namely,
$$
E[\Z_n(u,w) \Z_n(\tilde u,\tilde w)] = \En[\ell(u,w)'Z_i Z_i'\ell(\tilde u,\tilde w)] ( u \wedge \tilde u -  u  \tilde u),
$$
for all   $(u,w)$ and $(\tilde u,\tilde w)$ in $I$,
and (ii) $\Z_n$ approximates $\mathbb{L}_{n}$, namely,
$$
\sup_{(u,w) \in I} \Big| \mathbb L_n(u,w) - \Z_n(u,w) \Big| \lesssim_P  o(n^{-\varepsilon'}),
$$
where $\varepsilon'$ is some constant.
\end{lemma}

\begin{proof} The proof is based on the use of maximal inequalities and our extension of Yurinskii's coupling to the sup norm; see Lemma \ref{Lemma:YCdmetric} below. Although the proof is closely related to that of Lemma \ref{lemma: strong}, the details of calculations are rather different.

For each $j\geq 1$, we define a sequence of projections $\pi_j\colon I \to I$ associated with a covering of the set $I$ by balls of radius $2^{-j}$, where each element of the set $I$ is mapped to the center of the ball containing this element (if an element is contained by several balls, choose one according to some predetermined rule). We assume that all balls used to construct the projection $\pi_j$ are also used to construct projections $\pi_{j'}$ for all $j' > j$. Since the set $I$ is such that its dimension $d_I$ is independent of $n$ and its diameter is bounded uniformly over $n$, it is standard to show that the projection $\pi_j$ can be constructed using $k_j$ balls with $k_j\lesssim 2^{j d_I}$.

In what follows, given a process $\Z$ in $\ell^{\infty}(I)$ and its projection $\Z \circ \pi_j$, we shall identify the process $\Z \circ \pi_j$ with a random vector $\Z \circ \pi_j$ in $\Bbb{R}^{k_j}$, when convenient. Analogously, given a random vector $W$ in $\Bbb{R}^{k_j}$, we shall identify it with a process $W$ in $\ell^{\infty}(I)$ that is piece-wise constant.

The following relations will be proven below for some $\varepsilon'>0$ and $j = j_n \to \infty$:
\begin{enumerate}
\item (Finite-Dimensional Approximation)
$$
r_1 =\sup_{(u,w)  \in I}|\mathbb L_n(u,w) - \mathbb L_n\circ \pi_j(u,w)| \lesssim_P o(n^{-\varepsilon'});
$$
    \item (Coupling with a Normal Vector) there exists $\N_{nj} =_d N(0, \text{var}[\mathbb L_n \circ \pi_j])$ such that
$$
r_2 = \| \N_{nj} - \mathbb L_n\circ \pi_j \|_{\infty} \lesssim_P o(n^{-\varepsilon'});
$$
 \item (Embedding a Normal Vector into a Gaussian Process) there exists a Gaussian process $\Z_n$ with properties stated in the lemma such that $\N_{nj} = \Z_n \circ \pi_j \text{ a.s.};$
      \item (Infinite-Dimensional Approximation)
$$r_3=\sup_{(u,w) \in I}|\Z_n(u,w) - \Z_n\circ \pi_j(u,w)| \lesssim_P o(n^{-\varepsilon'}).$$
\end{enumerate}
The result then follows from the triangle inequality:
$$
\sup_{(u,w) \in I} |\mathbb L_n(u,w) - \Z_n(u,w)| \leq r_1 + r_2 + r_3.
$$

We now prove relations (1)-(4). Relation (1) follows from
\begin{align}\begin{split}\nonumber
r_1  & = \sup_{(u,w) \in I} |\mathbb L_n(u,w) - \mathbb L_n\circ \pi_j(u,w)|\\
  & \leq  \sup_{\|(u,w) - (\tilde u,\tilde w)\| \leq 2^{-j} } |  \mathbb L_n(u,w)- \mathbb L_n(\tilde u,\tilde w) |  \\
  & \lesssim_P  \sqrt{(L_{\ell}^2 2^{-2j} + 2^{-j})\log n} + \sqrt{\frac{\zeta_m^2\log^4 n}{n}} \leq o(n^{-\varepsilon'}),
\end{split}\end{align}
where the first inequality in the third line holds by Lemma \ref{Lemma:Fact1linear} and the second by substituting $2^j = L_{\ell} n^{\tilde \varepsilon}$ for an appropriate $\tilde \varepsilon>0$ and using \eqref{eq: growth condition lem 31}.

Relation (2) follows from the use of Yurinskii's coupling extended to the sup norm; see Lemma \ref{Lemma:YCdmetric} with $\|\cdot\|_d = \|\cdot\|_{\infty}$. In order to apply the coupling, set $\widehat I_j = \{\pi_j(u,w)\colon (u,w)\in I\}$ and for all $i=1,\dots,n$, let
$$
\xi_i = \Big(v_i\cdot(\ell(u,w)'Z_i)\cdot \psi_i(u)/\sqrt{n}\Big)_{(u,w)\in\widehat I_j},
$$
so that $\xi_i$ is a zero-mean random vector in $\mathbb R^{k_j}$, and we have $\mathbb L_n\circ\pi_j =_d \sum_{i=1}^n \xi_i$. Then
\begin{align*}
&\sum_{i=1}^n E\Big[ \|\xi_i\|^2\|\xi_i\|_{\infty} \Big]\\
& \quad = \frac{1}{n^{3/2}}\sum_{i=1}^n E\Big[ |v_i|^3\Big(\sum_{(u,w)\in\widehat I_j} |\ell(u,w)'Z_i)|^2|\psi_i(u)|^2\Big)\Big(\max_{(u,w)\in\widehat I_j}|\ell(u,w)'Z_i|\cdot |\psi_i(u)|\Big)\Big]\\
&\quad \lesssim \frac{1}{n^{3/2}}\sum_{i=1}^n E\Big[ \Big(\sum_{(u,w)\in \widehat I_j}|\ell(u,w)'Z_i|^2\Big)\Big(\max_{(u,w)\in\widehat I_j}|\ell(u,w)'Z_i|\Big) \Big]\\
&\quad \leq \frac{1}{n^{3/2}}\sum_{i=1}^n \xi_{\ell}\zeta_m \sum_{(u,w)\in \widehat I_j}\ell(u,w)'Z_i Z_i'\ell(u,w) = \frac{k_j \xi_{\ell}^3\zeta_m}{\sqrt n}.
\end{align*}
Further, for $i=1,\dots,n$, let $g_i\sim N(0,\Sigma_i)$, where $\Sigma_i = \text{var}(\xi_i)$. Note that the elements of the matrix $\Sigma_j$ are given by
$$
\ell(u,w)'Z_i Z_i'\ell(\tilde u,\tilde w)(u\wedge \tilde u - u\tilde u)/n,\quad \text{for $(u,w)$ and $(\tilde u,\tilde w)$ in $\widehat I_j$.}
$$
Hence,
$$
\Big(E[\|g_i\|_{\infty}^2]\Big)^{1/2} \lesssim \frac{\xi_{\ell}\zeta_m \sqrt{\log k_j}}{\sqrt n}.
$$
Also,
\begin{align*}
\sum_{i=1}^n \Big(E[\|g_i\|^4]\Big)^{1/2}
& = \sum_{i=1}^n \Big(E\Big[\Big(\sum_{l = 1}^{k_j}g_{i l}^2\Big)^2\Big]\Big)^{1/2} \leq \sum_{i=1}^n \sum_{l=1}^{k_j}\Big(E[g_{i l}^4]\Big)^{1/2}\\
&\lesssim \sum_{i=1}^n\sum_{l=1}^{k_j}E[g_{i l}^2] = \sum_{i=1}^n E[\|g_i\|^2] = \sum_{i=1}^n E[\|\xi_i\|^2]\\
& \lesssim \frac{1}{n}\sum_{i=1}^n\sum_{(u,w)\in\widehat I_j}|\ell(u,w)'Z_i|^2 \lesssim k_j \xi_{\ell}^2.
\end{align*}
Therefore,
$$
\sum_{i=1}^n E\Big[\|g_i\|^2\|g_i\|_{\infty}\Big] \leq \sum_{i=1}^n \Big(E[\|g_i\|^4]\Big)^{1/2}\Big(E[\|g_i\|_{\infty}^2]\Big)^{1/2} \lesssim \frac{k_j \xi_{\ell}^3 \zeta_m\sqrt{\log k_j}}{\sqrt n}.
$$
So, applying Lemma \ref{Lemma:YCdmetric} with
$$
\beta = \sum_{i=1}^n E\Big[\|\xi_i\|^2\|\xi_i\|_{\infty}\Big] + \sum_{i=1}^n E\Big[\|g_i\|^2\|g_i\|_{\infty}\Big] \lesssim \frac{k_j \xi_{\ell}^3\zeta_m\sqrt{\log k_j}}{\sqrt n}
$$
shows that for any $\delta > 0$, there exists a vector $\mathcal N_{n j} =_d N(0,\text{var}[\mathbb L_n\circ \pi_j])$ such that
\begin{align}
P\Big(\|\mathbb L_n\circ\pi_j - \mathcal N_{n j}\|_{\infty} > 3\delta\Big)
&\leq \min_{t\geq 0}\left(2P(\|Z\|_{\infty} > t) + \frac{\beta}{\delta^3}t^2\right)\notag \\
&= o(1) + O\left(\frac{k_j\xi_{\ell}^3\zeta_m(\log k_j)^{3/2}}{\delta^3 \sqrt n}\right)\label{eq: yur coupling rhs},
\end{align}
where $Z$ is a standard Gaussian $k_j$-dimensional random vector and where in the second line we set $t = C(\log k_j)^{1/2}$ for a sufficiently large constant $C$. Since $\xi_{\ell} \lesssim 1$ and $k_j \lesssim 2^{j d_I}$, substituting $2^j = L_{\ell} n^{\tilde \varepsilon}$ and using \eqref{eq: growth condition lem 31} shows that there exists $\delta = o(n^{-\varepsilon'})$ such that the expression in \eqref{eq: yur coupling rhs} is $o(1)$ leading to relation (2).

Relation (3) follows from the a.s. embedding of a finite-dimensional random normal vector into a path of a continuous Gaussian process, which is possible by Lemma \ref{lemma:prescribed} applied with $m = 1$ and $\mathcal U$ replaced by $I$.

Finally, relation (4) follows from
\begin{align*}
r_3  & = \sup_{(u,w) \in I} |\Z_n(u,w) - \Z_n\circ \pi_j(u,w)| \\
& \leq  \sup_{\|(u,w) - (\tilde u,\tilde w)\| \leq 2^{-j} } |  \Z_n(u,w) - \Z_n(\tilde u,\tilde w) |  \\
& \lesssim_P  \sqrt{ L_{\ell} 2^{-j} \log(2^j)} \leq o (n^{-\varepsilon'}),
\end{align*}
where the first inequality in the third line holds by Lemma \ref{Lemma:Fact4-b extension} and the second by substituting $2^j = L_{\ell} n^{\tilde \varepsilon}$ and using \eqref{eq: growth condition lem 31}. This completes the proof of the lemma.
\end{proof}

\begin{lemma}[Finite-Dimensional Approximation]\label{Lemma:Fact1linear}
Consider the setting of Lemma \ref{lemma:LinearStrong}. Then for any $\gamma > 0$, the process $\mathbb L_n$ satisfies
$$
 \sup \Big| \mathbb L_n(u,w) - \mathbb L_n(\tilde u,\tilde w) \Big| \lesssim_P \sqrt{(L_{\ell}^2\gamma^2 + \gamma)\log n} + \sqrt{\frac{\zeta_m^2\log^4 n}{n}},
$$
where the supremum is over all $(u,w)$ and $(\tilde u,\tilde w)$ in $I$ such that $\|(u,w) - (\tilde u,\tilde w)\|\leq \gamma$.
\end{lemma}
\begin{proof}
Consider the function class
$$
\mathcal L_{m,n} = \Big\{(Z,U,v)\mapsto f_{u,w}(Z,U,v) = \ell(u,w)'Z\colon (u,w)\in I\Big\},
$$
mapping $B_m(0,\zeta_m)\times[0,1]\times\mathbb R$. Note that $L_{m,n}(Z,U,v) = (\xi_{\ell} \vee 1)\zeta_m$ is its envelope. Also, since (i) for all $(u,w)$ and $(\tilde u,\tilde w)$ in $I$ we have
$$
\Big|\ell(u,w)'Z - \ell(\tilde u,\tilde w)'Z\Big| \leq \|\ell(u,w) - \ell(\tilde u,\tilde w)\|\cdot \|Z\|\leq L_{\ell}\|(u,w) - (\tilde u,\tilde w)\| \zeta_m
$$
and (ii) the set $I$ is such that its dimension is independent of $n$ and its diameter is bounded uniformly over $n$, it follows that the uniform entropy numbers of $\mathcal L_{m,n}$ satisfy
\begin{equation}\label{eq: FC1 - lem 32}
\sup_Q \log N(\epsilon\|L_{m,n}\|_{Q,2},\mathcal L_{m,n},L_2(Q)) \lesssim \log(L_{\ell}/\epsilon),\quad\text{uniformly over }0<\epsilon\leq 1.
\end{equation}
Next, consider the function class
$$
\widetilde {\mathcal L}_{m,n} = \Big\{(Z,U,v)\mapsto v\cdot(\ell(u,w)'Z)\cdot (u - 1\{U\leq u\})\colon (u,w)\in I\Big\}.
$$
The function $\widetilde L_{m,n}(Z,U,v) = |v|\cdot L_{m,n}(Z,U,v)$ is an envelope of $\widetilde {\mathcal L}_{m,n}$. By \eqref{eq: FC1 - lem 32} and Lemma \ref{Lemma:ProductEntropy}, the uniform entropy numbers of $\widetilde{\mathcal L}_{m,n}$ satisfy
\begin{equation}\label{eq: FC2 lem 32}
\sup_Q \log N(\epsilon\|\widetilde L_{m,n}\|_{Q,2},\widetilde{\mathcal L}_{m,n},L_2(Q))\lesssim \log(L_{\ell}/\epsilon),\quad\text{uniformly over }0<\epsilon\leq 1.
\end{equation}
Further, consider the function class
\begin{multline*}
\overline{\mathcal L}_{m,n} = \Big\{(Z,U,v)\mapsto v\cdot(\ell(u,w)'Z)\cdot (u - 1\{U \leq u\}) \\
- v\cdot(\ell(\tilde u,\tilde w)'Z)\cdot (\tilde u - 1\{U \leq \tilde u\}\colon (u,w)\in I, \ (\tilde u,\tilde w)\in I, \  \|(u,w) - (\tilde u,\tilde w)\|\leq \gamma \Big\}.
\end{multline*}
The function $2\widetilde L_{m,n}$ is its envelope. By \eqref{eq: FC2 lem 32} and Lemma \ref{Lemma:ProductEntropy}, the uniform entropy numbers of $\overline{\mathcal L}_{m,n}$ satisfy
\begin{equation}\label{eq: FC3 lem 32}
\sup_Q \log N(\epsilon\|2\widetilde L_{m,n}\|_{Q,2},\overline {\mathcal L}_{m,n},L_2(Q))\lesssim \log(L_{\ell}/\epsilon),\quad\text{uniformly over }0<\epsilon\leq 1.
\end{equation}
With this notation, we have
$$
\sup_{\|(u,w) - (\tilde u,\tilde w)\|\leq \gamma}\Big|\mathbb L_n(u,w) - \mathbb L_n(\tilde u,\tilde w)\Big| = \sup_{f\in\overline {\mathcal L}_{m,n}}|\mathbb G_n f|,
$$
and so to prove the asserted claim, we can apply the second part of Lemma \ref{Lemma:MaxIneq2} using the sequence of independent observations $(Z_i,U_i,v_i)_{i=1}^n$.

By \eqref{eq: FC3 lem 32}, \eqref{Eq:J} is satisfied with $\omega = L_{\ell}$, $J(m) = O(1)$, and $F_m = 2\widetilde {\mathcal L}_{m,n}$. Also, note that
$$
\max_{1\leq i\leq n}F_m(Z_i,U_i,v_i) \lesssim_P M = \xi_{\ell}\zeta_m \log n
$$
since $\max_{1\leq i\leq n}v_i \lesssim_P \log n$. Moreover, $\log \omega \lesssim \log n$ since $\log L_{\ell} \lesssim \log n$. Further,
\begin{align*}
&\sup_{f\in\overline{\mathcal L}_{m,n}}\frac{1}{n}\sum_{i=1}^n \Ep[f(Z_i,U_i,v_i)^2]\\
&\qquad = \sup_{\|(u,w) - (\tilde u,\tilde w)\|\leq \gamma}\frac{1}{n}\sum_{i=1}^n \Ep\Big[v_i^2\Big((\ell(u,w)'Z_i)\cdot\psi_i(u) - (\ell(\tilde u,\tilde w)' Z_i)\cdot \psi_i(\tilde u)\Big)^2\Big]\\
&\qquad = \sup_{\|(u,w) - (\tilde u,\tilde w)\|\leq \gamma}\frac{1}{n}\sum_{i=1}^n \Ep\Big[\Big((\ell(u,w)'Z_i)\cdot\psi_i(u) - (\ell(\tilde u,\tilde w)' Z_i)\cdot \psi_i(\tilde u)\Big)^2\Big]\\
&\qquad \lesssim \sup_{\|(u,w) - (\tilde u,\tilde w)\|\leq \gamma}\frac{1}{n}\sum_{i=1}^n \Big((\ell(u,w) - \ell(\tilde u,\tilde w))' Z_i\Big)^2\\
&\qquad\quad + \sup_{\|(u,w) - (\tilde u,\tilde w)\|\leq \gamma}\frac{1}{n}\sum_{i=1}^n \Ep\Big[ (\ell(\tilde u,\tilde w)' Z_i)^2 \cdot (\psi_i(u) - \psi_i(\tilde u))^2 \Big]\\
&\qquad \lesssim L_{\ell}^2 \gamma^2 \varphi + \gamma(1 - \gamma)\xi_{\ell}^2\varphi
\lesssim L_{\ell}^2 \gamma^2 + \xi_{\ell}^2 \gamma,
\end{align*}
where $\varphi = \sup_{\alpha\in S^{m-1}}\En[(\alpha' Z_i)^2]\lesssim 1$ and where we used arguments similar to those used in the proof of Lemma \ref{Lemma:Fact1}. Moreover,
\begin{align*}
&\sup_{f\in\overline{\mathcal L}_{m,n}}\frac{1}{n}\sum_{i=1}^n \Ep[f(Z_i,U_i,v_i)^4]\\
&\qquad = \sup_{\|(u,w) - (\tilde u,\tilde w)\|\leq \gamma}\frac{1}{n}\sum_{i=1}^n \Ep\Big[v_i^4\Big((\ell(u,w)'Z_i)\cdot\psi_i(u) - (\ell(\tilde u,\tilde w)' Z_i)\cdot \psi_i(\tilde u)\Big)^4\Big]\\
&\qquad = \sup_{\|(u,w) - (\tilde u,\tilde w)\|\leq \gamma}\frac{1}{n}\sum_{i=1}^n \Ep\Big[\Big((\ell(u,w)'Z_i)\cdot\psi_i(u) - (\ell(\tilde u,\tilde w)' Z_i)\cdot \psi_i(\tilde u)\Big)^4\Big]\\
&\qquad \lesssim \sup_{\|(u,w) - (\tilde u,\tilde w)\|\leq \gamma}\frac{1}{n}\sum_{i=1}^n \xi_{\ell}^2\zeta_m^2\Ep\Big[\Big((\ell(u,w)'Z_i)\cdot\psi_i(u) - (\ell(\tilde u,\tilde w)' Z_i)\cdot \psi_i(\tilde u)\Big)^2\Big]\\
&\qquad \lesssim \xi_{\ell}^2\zeta_m^2 (L_{\ell}^2\gamma^2 + \xi_{\ell}^2 \gamma),
\end{align*}
where we used the bound above and the inequalities $\|\ell(u,w)\| \lesssim \xi_{\ell}$, $\max_{1\leq i\leq n}\Ep[v_i^4] \lesssim 1$ and $\max_{1\leq i\leq n}\|Z_i\| \lesssim \zeta_m$. Substituting these bounds into the second part of Lemma \ref{Lemma:MaxIneq2} gives
\begin{align*}
\sup_{f\in\overline{\mathcal L}_{m,n}}|\mathbb G_n f|
& \lesssim_P \sqrt{(L_{\ell}^2\gamma^2 + \xi_{\ell}^2\gamma)\log n} + \sqrt{\frac{\xi_{\ell}^2\zeta_m^2\log^4 n}{n}}\\
& \lesssim \sqrt{(L_{\ell}^2\gamma^2 + \gamma)\log n} + \sqrt{\frac{\zeta_m^2\log^4 n}{n}}.
\end{align*}
This completes the proof of the lemma.
\end{proof}

\begin{lemma}[Infinite-Dimensional Approximation]\label{Lemma:Fact4-b extension}
Consider the setting of Lemma \ref{lemma:LinearStrong}, and let $\Z_n\colon I \to \mathbb R$ be a zero-mean Gaussian process whose covariance structure is given by
$$
E[G_n(u,w) G_n(\tilde u,\tilde w)] = \En[\ell(u,w)'Z_i Z_i'\ell(\tilde u,\tilde w)](u\wedge \tilde u - u \tilde u),
$$
for all $(u,w)$ and $(\tilde u,\tilde w)$ in $I$. Then for any $\gamma\in(0,1/2)$,
$$
\sup_{\|(u,w) - (\tilde u,\tilde w)\|\leq \gamma} |G_n(u,w) - G_n(\tilde u,\tilde w)| \lesssim_P \sqrt{ L_{\ell}\gamma\log(1 / \gamma)}.
$$
\end{lemma}
\begin{proof}
To prove the asserted claim, we apply the maximal inequality \eqref{eq: gaussian maximal inequality} from Lemma \ref{Lemma:Fact4} to the zero-mean Gaussian process $X_n\colon I\times I \to \mathbb R$ defined by
$$
X_{n,t} = G_n(u,w) - G_n(\tilde u,\tilde w), \ t = (u,w,\tilde u,\tilde w), \ \|(u,w) - (\tilde u,\tilde w)\|\leq \gamma.
$$
Note that for any $(u,w)$ and $(\tilde u,\tilde w)$ in $I$ with $\|(u,w) - (\tilde u,\tilde w)\|\leq \gamma$ and $u\leq \tilde u$, we have
\begin{align*}
E\Big[\Big(G_n(u,w) - G_n(\tilde u,\tilde w)\Big)^2\Big]
& = \En\Big[\ell(u,w)'Z_i Z_i'\ell(u,w)\Big] (u - u^2)\\
&\quad + \En\Big[\ell(\tilde u,\tilde w)'Z_i Z_i'\ell(\tilde u,\tilde w)\Big](\tilde u - \tilde u^2)\\
&\quad - 2\En\Big[\ell(u,w)'Z_i Z_i'\ell(\tilde u,\tilde w)\Big](u - u\tilde u)\\
&\leq I_1 + I_2 \lesssim \xi_{\ell}(\xi_{\ell} + L_{\ell})\gamma,
\end{align*}
where
\begin{align*}
I_1
& = \Big|\En\Big[\ell(u,w)'Z_i Z_i'\ell(u,w)\Big](u - u^2) - \En\Big[\ell(u,w)'Z_i Z_i'\ell(\tilde u,\tilde w)\Big](u - u\tilde u)\Big|\\
& \leq 2\Big|\En\Big[\ell(u,w)'Z_i Z_i'(\ell(u,w) - \ell(\tilde u,\tilde w))\Big]\Big|
 + \Big|\En\Big[\ell(u,w)'Z_i Z_i'\ell(\tilde u,\tilde w)\Big] (u - \tilde u)\Big|\\
&\lesssim \xi_{\ell}L_{\ell}\gamma + \xi_{\ell}^2\gamma
\end{align*}
and
\begin{align*}
I_2
& = \Big|\En\Big[\ell(\tilde u,\tilde w)'Z_i Z_i'\ell(\tilde u,\tilde w)\Big](\tilde u - \tilde u^2) - \En\Big[\ell(u,w)'Z_i Z_i'\ell(\tilde u,\tilde w)\Big](u - u\tilde u)\Big|\\
& \leq 2\Big|\En\Big[(\ell(\tilde u,\tilde w) - \ell(u,w))'Z_i Z_i'\ell(\tilde u,\tilde w)\Big]\Big|
 + 2\Big|\En\Big[\ell(u,w)'Z_i Z_i'\ell(\tilde u,\tilde w)\Big] (u - \tilde u)\Big|\\
&\lesssim \xi_{\ell}L_{\ell}\gamma + \xi_{\ell}^2\gamma.
\end{align*}
Therefore,
$$
\sigma(X_n)= \left(\sup_{t\in I\times I}E[X_{n,t}^2]\right)^{1/2} \lesssim \Big(\xi_{\ell}(\xi_{\ell} + L_{\ell})\gamma\Big)^{1/2}.
$$
The calculation above also implies that \eqref{eq: entropy bound lem 11} holds for the process $X_n$ with
$$
\epsilon_0 =\sigma(X_n), \ K\lesssim \sqrt{\xi_{\ell}(\xi_{\ell} + L_{\ell})},\  \text{ and }V = d_I,
$$
since the set $I$ has dimension $d_I$, which is independent of $n$, and its diameter is bounded from above uniformly over $n$. Therefore, the result follows by setting
$$
\lambda = C\sqrt{\xi_{\ell}(\xi_{\ell} + L_{\ell})\gamma\log(1 / \gamma)} \lesssim \sqrt{L_{\ell}\gamma\log(1/\gamma)},
$$
where $C$ is a sufficiently large constant, and using \eqref{eq: gaussian maximal inequality}. This completes the proof of the lemma.
\end{proof}

\subsection{Yurinskii's coupling for $d$-norm}
In this section, we extend the original Yurinskii's coupling, used in Appendix \ref{App:StrongGaussianApprox}, from the Euclidean norm $\|\cdot\|$ to arbitrary norm $\|\cdot\|_d$ on $\mathbb R^k$. The new coupling is used in Section \ref{Sec:Functionals} with the sup-norm $\|\cdot\|_{\infty}$ to construct Gaussian couplings for functionals. Replacing the Euclidean norm by the sup-norm is important because it allows us to construct couplings under substantially weakened side conditions.

\begin{lemma}[Yurinskii's coupling, $d$-norm]\label{Lemma:YCdmetric}
Let $\xi_1,\ldots,\xi_n$ be independent random zero-mean $k$-vectors, and let
$$
\beta = \sum_{i=1}^n E[\|\xi_i\|^2 \|\xi_i\|_d]+ \sum_{i=1}^n E[\|g_i\|^2\|g_i\|_d]
$$
be finite, where $g_1,\dots,g_n$ is a sequence of independent random $k$-vectors such that $g_i\sim N(0,{\rm var}(\xi_i))$ for all $i=1,\dots,n$. Let $S = \xi_1 + \cdots + \xi_n$. Then for each $\delta>0$, there exists a random vector $T$ with a $N(0,{\rm var}(S))$ distribution such that
$$P( \| S - T \|_d > 3\delta ) \leq \min_{t\geq 0} \left(2P(\|Z\|_d > t ) + \frac{\beta}{\delta^3}t^2\right),$$
where $Z\sim N(0,I_k)$.
\end{lemma}

\begin{proof}[Proof of Lemma \ref{Lemma:YCdmetric}]
Fix $\delta>0$. By Strassen's theorem (see, for example, Theorem 8 in Section 10 of Pollard \cite{PollardUsers}), one can construct $T$ such that
$P(\|S - T\| > 3\delta ) \leq \epsilon'$ if and only if $P(S \in A) \leq P(T\in A^{3\delta}) + \epsilon'$ for all Borel subsets $A$ of $\RR^k$, where $A^{\delta}=\{ x \in \RR^k\colon  d(x,A)\leq \delta\}$ and the metric $d(\cdot,\cdot)$ is induced by the norm $\|\cdot\|$. Below we will apply this theorem.

Fix a Borel subset $A$ of $\mathbb R^k$. Using Lemma \ref{Lemma:ApproxLinfty}, construct a smooth function $f\colon \mathbb R^k\to \mathbb R$ that approximates the indicator function of $A$, namely, a function $f$ such that for all $x,y\in\mathbb R^k$,
\begin{align*}
&\Big|f(x+y) - f(x) - y'\nabla f(x) - \frac{1}{2}y'\nabla^2f(x)y\Big| \leq \frac{\|y\|^2\|y\|_d}{\sigma^2\delta},\\
&(1-\epsilon)1\{x\in A\} \leq f(x) \leq \epsilon + (1-\epsilon)1\{x\in A^{3\delta}\},
\end{align*}
where $\epsilon = P(\|Z\|_d > \delta/\sigma)$ and $\sigma$ is to be chosen below. 
Then we have
\begin{align*}
P(S \in A) & = E[ 1\{ S\in A\} - f(S) ] + E[ f(S) - f(T) ] + E[f(T)] \\
& \leq \epsilon E[ 1\{ S\in A\}] + (\sigma^2\delta)^{-1}\sum_{i=1}^n E\Big[\|\xi_i\|^2\|\xi_i\|_d +\|g_i\|^2\|g_i\|_d\Big] \\
& \quad + \epsilon + (1-\epsilon)E[1\{T \in A^{3\delta}\}]\\
& \leq P(T \in A^{3\delta}) + 2\epsilon + (\sigma^2\delta)^{-1} \sum_{i=1}^n E\Big[\|\xi_i\|^2\|\xi_i\|_d + \|g_i\|^2\|g_i\|_d\Big]\\
& = P(T \in A^{3\delta}) + 2\epsilon + (\sigma^2\delta)^{-1}\beta,
\end{align*}
where in the first inequality we used
\begin{align*}
E[f(S) - f(T)]
&= \sum_{i=1}^n E[f(X_i + Y_i) - f(X_i + W_i)]\\
&\leq \sum_{i=1}^n E\Big[f(X_i) + Y_i'\nabla f(X_i) + \frac{1}{2}Y_i'\nabla^2 f(X_i) Y_i + \frac{\|Y_i\|^2\|Y_i\|_d}{\sigma^2\delta}\Big]\\
&\quad - \sum_{i=1}^n E\Big[f(X_i) + W_i'\nabla f(X_i) + \frac{1}{2}W_i'\nabla^2 f(X_i) W_i - \frac{\|W_i\|^2\|W_i\|_d}{\sigma^2\delta}\Big]\\
& = \sum_{i=1}^n \frac{\|Y_i\|^2\|Y_i\|_d + \|W_i\|^2\|W_i\|_d}{\sigma^2 \delta}
\end{align*}
for $X_i = \xi_1 + \dots + \xi_{i - 1} + g_{i + 1} + \dots + g_n$, $Y_i = \xi_i$, $W_i = g_i$, and we assumed (without loss of generality) that the sequences $(\xi_i)_{i=1}^n$ and $(g_i)_{i=1}^n$ are independent. Therefore, setting $\sigma = \delta / t$ for $t > 0$, we obtain
$$
P(S\in A) \leq P(T\in A^{3\delta}) + 2P(\|Z\|_d > t) + \frac{\beta}{\delta^3}t^2.
$$
The asserted claim now follows by minimizing the right-hand of this inequality with respect to $t>0$ and applying Strassen's theorem.
\end{proof}

In the next lemma, we construct a function $f$ used in the previous lemma.
\begin{lemma}[Smooth Approximation, $d$-norm]\label{Lemma:ApproxLinfty} Let $A$ be a Borel subset of $\RR^k$, $Z$ be a $N(0,I_k)$ random vector, and let $d(\cdot,\cdot)$ be the metric on $\mathbb R^k$ induced by a norm $\|\cdot\|_d$. Also, for positive constants $\sigma$ and $\delta$, define
$$ g(x) = \left(1 - \frac{d(x,A^\delta)}{\delta} \right)_+ \ \ \mbox{and} \ \  f(x) = E[ g(x+\sigma Z)].$$
Then for all $x,y\in\mathbb R^k$, the function $f$ satisfies
\begin{align*}
&\Big|f(x+y) - f(x) - y'\nabla f(x) - \frac{1}{2}y'\nabla^2f(x)y\Big| \leq \frac{\|y\|^2\|y\|_d}{\sigma^2\delta},\\
&(1-\epsilon)1\{x\in A\} \leq f(x) \leq \epsilon + (1-\epsilon)1\{x\in A^{3\delta}\},
\end{align*}
where $\epsilon = P(\|Z\|_d > \delta / \sigma )$. 
\end{lemma}
\begin{proof} The proof follows closely the proof of Lemma 18 in Section 10 of Pollard \cite{PollardUsers}, with the difference that we allow for a general metric $d(\cdot,\cdot)$. Let $\phi_\sigma$ denote the pdf of a $N(0,\sigma^2 I_k)$ random vector. We have
$$ \frac{\partial }{\partial z}\phi_\sigma(z)= -\frac{z}{\sigma^2}\phi_\sigma(z) \ \ \ \mbox{and} \ \ \ \frac{\partial^2}{\partial zz'}\phi_\sigma(z) = \left(\frac{zz'}{\sigma^4}-\frac{I_k}{\sigma^2}\right)\phi_\sigma(z).$$
Further, for fixed $x$ and $y$, consider the function $h(t):= f(x+t y)$, $0\leq t\leq 1$. Its second derivative is
$$
\ddot{h}(t) = \sigma^{-2} E[ g(x+ty+\sigma Z) \left( (y'Z)^2 - \|y\|^2 \right)].
$$
In addition, the function $g$ has the following Lipschitz property with respect to the metric $d$:
\begin{align*}
 |g(x+ty+\sigma Z) - g(x+\sigma Z)| &\leq \left|\frac{d(x+ty+\sigma Z,A^\delta) - d(x+\sigma Z, A^\delta)}{\delta} \right| \\
 &\leq d(x+ty+\sigma Z,x+\sigma Z) / \delta = td(y,0) / \delta.
\end{align*}
Therefore,
$$
|\ddot{h}(t) - \ddot{h}(0) | \leq \frac{td(y,0)}{\sigma^2\delta}E\left[ (y'Z)^2 + \|y\|^2 \right] = \frac{2t \|y\|^2d(y,0)}{\sigma^2\delta}.
$$
The first asserted claim now follows from a Taylor expansion for the function $h$, namely, we have for some $t^* \in (0,1)$ that
$$
| h(1) - h(0) - \dot{h}(0) - \frac{1}{2} \ddot{h}(0)| = \frac{1}{2}|\ddot{h}(t^*) - \ddot{h}(0)|.
$$

To establish the second asserted claim, we proceed as follows. By construction, we have for any $x \in A$ that $g(x+\sigma Z) = 1$ if $d(Z,0) \leq \delta/\sigma$ since $x+\sigma Z \in A^\delta$ in this case. Therefore,
$$ f(x) \geq E\Big[g(x+\sigma Z) \cdot 1\{d(Z,0) \leq \delta/\sigma\}\Big] \geq P( d(Z,0) \leq \delta / \sigma ) = 1 - P( d(Z,0) > \delta / \sigma ).$$

Also by construction, we have for any $x \notin A^{3\delta}$ that $g(x+\sigma Z) = 0$ if $d(Z,0) \leq \delta/\sigma$ since $d(x+\sigma Z,A^\delta) > \delta$ in this case. Therefore
$$ f(x) = E\Big[g(x+\sigma Z) \cdot 1\{d(Z,0) \leq \delta/\sigma\}\Big] + E\Big[g(x+\sigma Z) \cdot 1\{d(Z,0) > \delta/\sigma\}\Big] \leq P( d(0,Z) > \delta / \sigma ).$$
This completes the proof of the lemma.
\end{proof}

\section{A Lemma on Gaussian Approximation for the Supremum of a Weighted Bootstrap Process}\label{sec: weighted bootstrap}

\begin{lemma}\label{lemma:SupremumCoupling} (Gaussian Approximation for Supremum of Weighted Bootstrap Process)
Let $(Z_i)_{i=1}^n$ be a sequence of non-stochastic vectors in $\mathbb R^m$ and consider the empirical process $\mathbb{U}_{n}$ in $[\ell^{\infty}(\mathcal{U})]^m$, $\mathcal{U} \subseteq (0,1)$, defined by
$$
\mathbb{U}_{n}(u) = \mathbb{G}_n\(v_iZ_i \psi_i(u)\), \ \ \psi_i(u) = u-1\{U_i \leq u\},\quad u\in\mathcal U,
$$
where $(U_i,v_i)_{i=1}^n$ is an i.i.d. sequence of pairs of independent random variables where $U_i\sim\Uniform (0,1)$, $E[v_i] = 0$, $E[v_i^2] = 1$, $E[|v_i|^4]\lesssim 1$, and $\max_{1\leq i\leq n}|v_i| \lesssim_P \log n$.  Suppose that the vectors $Z_i$ are such that
$$
\sup_{\alpha\in S^{m-1}}\En\[(\alpha'Z_i)^2\] \lesssim 1 \ \text{and} \ \max_{1\leq i\leq n}\|Z_i\| \lesssim \zeta_m,
$$
where $\zeta_m$ satisfies $1/\zeta_m\lesssim 1$ and $\zeta_m^2 = o(n^{1 - \varepsilon})$ for some $\varepsilon >0 $. Also, let $\mathcal W$ be a set in $\mathbb R^d$ and let $I$ be a subset of $\mathcal U\times\mathcal W$ whose dimension $d_I$ is independent of $n$ and whose diameter is bounded uniformly over $n$. Moreover, let $\ell\colon I \to \mathbb R^m$ be a function such that
$$
\|\ell(u,w)\|\leq \xi_{\ell} \ \text{and} \ \|\ell(u,w) - \ell(\tilde u,\tilde w)\| \leq L_{\ell}\|(u, w) - (\tilde u, \tilde w)\|
$$
for all $(u,w)$ and $(\tilde u,\tilde w)$ in $I$, where $\xi_{\ell}$ and $L_{\ell}$ satisfy $\xi_{\ell} \lesssim 1$ and $\log L_{\ell}\lesssim \log n$, and define
$$
V_n = \sup_{(u,w)\in I}|\ell(u,w)'\mathbb U_n(u)|.
$$
Then for each $n$, there exists a random variable $\bar V_n$ that is (i) such that
\begin{equation}\label{eq: suprema approximation main}
|V_n - \bar V_n| = o_P(n^{-\varepsilon'})
\end{equation}
for some $\varepsilon' >0$, (ii) independent of $(U_i)_{i=1}^n$, and (iii) equal in distribution to the random variable $\sup_{(u,w)\in I} |G_n(u,w)|$, where $G_n(\cdot,\cdot)$ is a zero-mean Gaussian process with a.s. continuous sample paths and the covariance function
$$
E[\Z_n(u,w) \Z_n(\tilde u,\tilde w)] = \En[\ell(u,w)'Z_i Z_i'\ell(\tilde u,\tilde w)] ( u \wedge \tilde u -  u  \tilde u),
$$
for all   $(u,w)$ and $(\tilde u,\tilde w)$ in $I$.
\end{lemma}
\begin{proof}
Throughout the proof, we will take the process $G_n(\cdot,\cdot)$ from the statement of the lemma to be independent of $(U_i,v_i)_{i=1}^n$. In addition, we will use $\mathcal B$ to denote the class of all Borel sets in $\mathbb R$. 

To prove the asserted claim, we will show that there exists a sequence of positive numbers $(\eta_n)_{n\geq 1}$ converging to zero $\varepsilon' > 0$ such that for
$$
\mathcal U_1^n = \left\{(U_i)_{i=1}^n\colon P(V_n \in A|(U_i)_{i=1}^n) \leq P\left(\sup_{(u,w)\in I} |G_n(u,w)| \in A^{\delta_n}\right) + \eta_n\text{ for all }A\in\mathcal B \right\}
$$
we have
\begin{equation}\label{eq: strassen conditional}
P(\mathcal U_1^n) \geq 1 - o(1),
\end{equation}
where $\delta_n = n^{-\varepsilon'}$ and $A^{\delta_n} = \{x\in \mathbb R\colon \inf_{y\in A}|x - y|\leq \delta_n\}$. Then applying Strassen's theorem conditional on $(U_i)_{i=1}^n$ on the event $\mathcal U_1^n$ shows that on this event one can construct a random variable $\bar V_n$ that is conditional on $(U_i)_{i=1}^n$ equal in distribution to $\sup_{(u,w)\in I}|G_n(u,w)|$ and is such that
\begin{equation}\label{eq: conditional on U}
P\Big(|V_n - \bar V_n|> \delta_n | (U_i)_{i=1}^n\Big) \leq \eta_n \ \text{ on }\mathcal U_1^n.
\end{equation}
Outside of the event $\mathcal U_1^n$, we define $\bar V_n = \sup_{(u,w)\in I}|G_n(u,w)|$. Then combining \eqref{eq: strassen conditional} with \eqref{eq: conditional on U} gives \eqref{eq: suprema approximation main} and since the conditional distribution of $\bar V_n$ given $(U_i)_{i=1}^n$ is equal to that of $\sup_{(u,w)\in I}|G_n(u,w)|$, it follows that $\bar V_n$ is independent of $(U_i)_{i=1}^n$ and is equal in distribution to  $\sup_{(u,w)\in I}|G_n(u,w)|$. Thus, it remains to prove \eqref{eq: strassen conditional}.

To prove \eqref{eq: strassen conditional}, recall that $\log L_{\ell} \lesssim \log n$ by assumption, and so $\log L_{\ell} \leq C_L\log n$ for some constant $C_L$. Let $\gamma_n = 1/n^{C_L + 1}$ and note that since the set $I$ is such that its dimension is independent of $n$ and its diameter is bounded uniformly over $n$, there exists a sequence $(u_j,w_j)_{j=1}^{k_n}$ in $I$ with $\log k_n \lesssim \log n$ such that balls with centers at $(u_j,w_j)$ and radius $\gamma_n$ cover $I$. Then applying Lemma \ref{Lemma:Fact1linear} shows that
\begin{align*}
&\sup_{\|(u,w) - (\tilde u,\tilde w)\|\leq \gamma_n}\Big|\ell(u,w)'\mathbb U_n(u) - \ell(\tilde u,\tilde w)'\mathbb U_n(\tilde u)\Big|\\
&\quad  \lesssim_P \sqrt{(L_{\ell}^2\gamma^2_n + \gamma_n)\log n} + \sqrt{\frac{\zeta_m^2\log^4 n}{n}} = o(n^{-\varepsilon'})
\end{align*}
for some $\varepsilon' >0$, and so there exists a sequence of positive numbers $(\eta_{n,1})_{n\geq 1}$ converging to zero such that the event
$$
\left|\sup_{(u,w)\in I}|\ell(u,w)'\mathbb U_n(u)| - \max_{1\leq j\leq k_n}|\ell(u_j,w_j)\mathbb U_n(u_j)|\right| > \delta_{n,1} = n^{-\varepsilon'}
$$
holds with probability at most $\eta_{n,1}$. Hence, there exists a set of values of $(U_i)_{i=1}^n$, say $\mathcal U_{1,1}^n$, such that $P(\mathcal U_{1,1}^n) \geq 1 - \eta_{n,1}^{1/2} = 1 - o(1)$ and
\begin{equation}\label{eq: approx 1 lem 38}
P\left(\Big|\sup_{(u,w)\in I}|\ell(u,w)'\mathbb U_n(u)| - \max_{1\leq j\leq k_n}|\ell(u_j,w_j)\mathbb U_n(u_j)|\Big| > \delta_{n,1} |(U_i)_{i=1}^n\right) \leq \eta_{n,1}^{1/2} \ \text{ on }\mathcal U_{1,1}^n.
\end{equation}

Next, let $\Sigma^X$ and $\Sigma^Y$ be $k_n\times k_n$-dimensional matrices given by
\begin{align}
&\Sigma^X_{j,l} =  \En\Big[\ell(u_j,w_j)'Z_i Z_i'\ell(u_l,w_l)(u_j - 1\{U_i \leq u_j\})(u_l - 1\{U_i \leq u_l\})\Big],\label{eq: sigma x}\\
&\Sigma^Y_{j,l} = \En\Big[\ell(u_j,w_j)'Z_i Z_i'\ell(u_l,w_l)\Big](u_j\wedge u_l - u_j u_l),\label{eq: sigma y}
\end{align}
for all $j,l = 1,\dots,k_n$. Note that
$
E[\Sigma^X_{j,l}] = \Sigma^Y_{j,l}
$
for all $j,l = 1,\dots,k_n$. In addition,
\begin{align*}
\sigma^2
& = \max_{1\leq j, l\leq k_n} \sum_{i=1}^n E\Big[ (\ell(u_j,w_j)'Z_i Z_i'\ell(u_l,w_l))^2(u_j - 1\{U_i \leq u_j\})^2(u_l - 1\{U_i \leq u_l\})^2 \Big]\\
& \leq \max_{1\leq j,l\leq k_n}\sum_{i=1}^n (\ell(u_j,w_j)'Z_i Z_i'\ell(u_l,w_l))^2 \lesssim n\zeta_m^2.
\end{align*}
Moreover,
$$
M = \max_{1\leq i\leq n}\max_{1\leq j,l\leq k_n}\Big| \ell(u_j,w_j)'Z_i Z_i'\ell(u_l,w_l)(u_j - 1\{U_i \leq u_j\})(u_l - 1\{U_i \leq u_l\}) \Big| \lesssim \zeta_m^2.
$$
Hence, Lemma \ref{lem: CCK maximal inequality} gives
$$
\max_{1\leq j,l\leq k_n} |\Sigma^X_{j,l} - \Sigma^Y_{j,l}| \lesssim_P \sqrt{\frac{\zeta_m^2\log n}{n}} +  \frac{\zeta_m^2\log n}{n} = o(n^{-\varepsilon'})
$$
for some $\varepsilon' >0$. Hence, there exists a set of values of $(U_i)_{i=1}^n$, say $\mathcal U_{1,2}^n$, such that $P(\mathcal U_{1,2}^n) = 1 - o(1)$ and
$$
\max_{1\leq j,l\leq k_n}|\Sigma^X_{j,l} - \Sigma^Y_{j,l}| \leq \delta_{n,2} = n^{-\varepsilon'} \ \text{ on }\mathcal U_{1,2}^n.
$$
In the rest of the proof, we will show that
\begin{equation}\label{eq: remaining inclusion}
\mathcal U_{1,1}^n \cap \mathcal U_{1,2}^n \subset\mathcal U_1^n
\end{equation}
if $\varepsilon'$ is small enough and $\eta_n $ converges to zero slowly enough in the definition of $\mathcal U_1^n$. The desired inequality \eqref{eq: strassen conditional} then follows since $P(\mathcal U_{1,1}^n) = 1 - o(1)$ and $P(\mathcal U_{1,2}^n) = 1 - o(1)$.

To prove \eqref{eq: remaining inclusion}, we proceed in several steps.
First, we apply Lemma \ref{lem: CCK finite dimensional gaussian approximation} conditional on $(U_i)_{i=1}^n$ with vectors $X_i = (\ell(u_j,w_j)'Z_i v_i(u_j - 1\{U_i\leq u_j\}))_{j=1}^{k_n}$ and $Y_i =  (\ell(u_j,w_j)'Z_i e_i(u_j - 1\{U_i\leq u_j\}))_{j=1}^{k_n}$, where $(e_i)_{i=1}^n$ is a sequence of $n$ independent $N(0,1)$ random variables that are independent of $(U_i)_{i=1}^n$. Note that
\begin{align*}
L_n
& = \max_{1\leq j\leq k_n} \frac{1}{n}\sum_{i=1}^n E[|X_{i j}|^3|(U_i)_{i=1}^n] \\
& = \max_{1\leq j\leq k_n} \frac{1}{n}\sum_{i=1}^n E\Big[|\ell(u_j,w_j)'Z_i v_i (u_j - 1\{U_i \leq u_j\})|^3|(U_i)_{i=1}^n\Big]\\
& =  \max_{1\leq j\leq k_n} \frac{1}{n}\sum_{i=1}^n |\ell(u_j,w_j)'Z_i|^3 \leq \max_{1\leq j\leq k_n} \frac{\zeta_m}{n}\sum_{i=1}^n |\ell(u_j,w_j)'Z_i|^2 \lesssim \zeta_m.
\end{align*}
Further, uniformly over $\delta>0$,
\begin{align*}
M_{n,X}(\delta)
& = \frac{1}{n}\sum_{i=1}^n E\Big[\max_{1\leq j\leq k_n} |X_{i j}|^3\cdot 1\Big\{\max_{1\leq j\leq k_n}|X_{i j}| > \delta \sqrt n/\log k_n\Big\}|(U_i)_{i=1}^n\Big]\\
& \leq \frac{\log k_n}{\delta n^{3/2}}\sum_{i=1}^n E\Big[\max_{1\leq j\leq k_n}|X_{i j}|^4 |(U_i)_{i=1}^n\Big] \\
& \lesssim \frac{\log k_n}{\delta n^{3/2}}\sum_{i=1}^n \max_{1\leq j\leq k_n} |\ell(u_j,w_j)'Z_i|^4 \lesssim\frac{\zeta_m^4 \log n}{\delta \sqrt n}.
\end{align*}
In addition, by the same argument, uniformly over $\delta>0$,
$$
M_{n,Y}(\delta) = \frac{1}{n}\sum_{i=1}^n E\Big[\max_{1\leq j\leq k_n} |Y_{i j}|^3\cdot 1\Big\{\max_{1\leq j\leq k_n}|Y_{i j}| > \delta \sqrt n/\log k_n\Big\}|(U_i)_{i=1}^n\Big] \lesssim \frac{\zeta_m^4\log n}{\delta \sqrt n}.
$$
Hence, given that $\zeta_m^2 = o(n^{1 - \varepsilon})$ and that
$$
 \ell(u_j,w_j)'\mathbb U_n(u_j) = \frac{1}{\sqrt n}\sum_{i=1}^n X_{i j},\quad j=1,\dots,k_n,
$$
and denoting
$$
\bar G_{n,j} = \frac{1}{\sqrt n}\sum_{i=1}^n Y_{i j},\quad j=1,\dots,k_n,
$$
and $\delta_{n,3} = n^{-\varepsilon'}$ for sufficiently small $\varepsilon'>0$,
we obtain via Lemma \ref{lem: CCK finite dimensional gaussian approximation} that
\begin{equation}\label{eq: approx 2 lem 38}
P\left(\max_{1\leq j\leq k_n}|\ell(u_j,w_j)'\mathbb U_n(u_j)| \in A |(U_i)_{i=1}^n\right) \leq P\left(\max_{1\leq j\leq k_n}|\bar G_{n,j}| \in A^{\delta_{n,3}} |(U_i)_{i=1}^n\right) + \eta_{n,2}
\end{equation}
for all $A \in \mathcal B$, where $(\eta_{n,2})_{n\geq 1}$ is a sequence of positive numbers converging to zero.

Second, we apply Lemma \ref{lem: CCK gaussian comparison inequality} conditional on $(U_i)_{i=1}^n$ on $\mathcal U_1^n$ with $X = (\bar G_{n,j})_{j=1}^{k_n}$ and $Y = (G_n(u_j,w_j))_{j=1}^{k_n}$. Note that conditional covariance matrices of the vectors $X$ and $Y$ are given by
$$
E[X_{j}X_{l} | (U_i)_{i=1}^n] = \Sigma^X_{j,l},\quad E[Y_{j}Y_{l} | (U_i)_{i=1}^n] = E[Y_j Y_l] = \Sigma^Y_{j,l},\quad j,l = 1,\dots,k_n
$$
for $\Sigma^X_{j,l}$ and $\Sigma^Y_{j,l}$ defined in \eqref{eq: sigma x} and \eqref{eq: sigma y}, respectively. Since on the event $\mathcal U_{1,2}^n$ we have
$$
\max_{1\leq j,l\leq k_n} |\Sigma^X_{j,l} - \Sigma^Y_{j,l}| \leq \delta_{n,2},
$$
it follows from Lemma \ref{lem: CCK gaussian comparison inequality} that on the same event we have
\begin{equation}\label{eq: approx 3 lem 38}
P\left(\max_{1\leq j\leq k_n}|\bar G_{n,j}| \in A |(U_i)_{i=1}^n\right) \leq P\left(\max_{1\leq j\leq k_n}|G_n(u_j,w_j)| \in A^{\delta_{n,2}^{1/3}}\right) + \eta_{n,3}
\end{equation}
for all $A\in\mathcal B$, where $(\eta_{n,3})_{n\geq 1}$ is a sequence of positive numbers converging to zero.

Third, by Lemma \ref{Lemma:Fact4-b extension},
$$
\sup_{\|(u,w) - (\tilde u,\tilde w)\|\leq \gamma_n}\Big|G_n(u,w) - G_n(\tilde u,\tilde w)\Big| \lesssim_P \sqrt{L_{\ell} \gamma_n\log(1/\gamma_n)} = o(n^{-\varepsilon'})
$$
for some $\varepsilon' >0$,  and so
\begin{equation}\label{eq: approx 4 lem 38}
P\left(\Big| \sup_{(u,w)\in I} |G_n(u,w)| - \max_{1\leq j\leq k_n} |G_n(u_j,w_j)|  \Big| > \delta_{n,4} \right) \leq \eta_{n,4},
\end{equation}
where $\delta_{n,4} = n^{-\varepsilon'}$ and $(\eta_{n,4})_{n\geq 1}$ is a sequence of positive numbers converging to zero.

Finally, combining \eqref{eq: approx 1 lem 38}, \eqref{eq: approx 2 lem 38}, \eqref{eq: approx 3 lem 38}, and \eqref{eq: approx 4 lem 38} shows that for all $A\in\mathcal B$,
\begin{align*}
P\Big(\sup_{(u,w)\in I}|\ell(u,w)'\mathbb U_n(u)| \in A|(U_i)_{i=1}^n\Big)
&\leq P\Big(\sup_{(u,w)\in I}|G_n(u,w)|\in A^{\delta_{n,1} + \delta_{n,2}^{1/3} + \delta_{n,3} + \delta_{n,4}}\Big)\\
&\quad + \eta_{n,1}^{1/2} + \eta_{n,2} + \eta_{n,3} + \eta_{n,4} \ \text{ on }\mathcal U_{1,1}^n\cap \mathcal U_{1,2}^n.
\end{align*}
Thus, \eqref{eq: remaining inclusion} follows if in the definition of $\mathcal U_1^n$ we set $\eta_n = \eta_{n,1}^{1/2} + \eta_{n,2} + \eta_{n,3} + \eta_{n,4}$ and $\delta_n = n^{-\varepsilon'}$ for sufficiently small $\varepsilon' > 0$ so that $\delta_n \geq \delta_{n,1} + \delta_{n,2}^{1/3} + \delta_{n,3} + \delta_{n,4}$. This completes the proof of the lemma.
\end{proof}

\subsection{Technical Lemmas}
\begin{lemma}\label{lem: CCK maximal inequality}
Let $Z_1,\dots,Z_n$ be independent random vectors in $\mathbb R^p$ with $p\geq 2$. Define $M = \max_{1\leq i\leq n}\max_{1\leq j\leq p} |Z_{i j}|$ and $\sigma^2 = \max_{1\leq j\leq p}\sum_{i=1}^n E[Z_{i j}^2]$. Then
$$
E\left[\max_{1\leq j\leq p}\left| \sum_{i=1}^n (Z_{i j} - E[Z_{i j}]) \right|\right] \leq C\Big(\sigma\sqrt{\log p} + \sqrt{E[M^2]}\log p\Big),
$$
where $C$ is an absolute constant.
\end{lemma}
\begin{proof}
See Lemma 8 in Chernozhukov, Chetverikov and Kato \cite{CCK15}.
\end{proof}

\begin{lemma}\label{lem: CCK finite dimensional gaussian approximation}
Let $X_1,\dots,X_n$ be independent centered random vectors in $\mathbb R^p$, $p\geq 2$, with finite absolute third moments. Consider the statistic $Z = \max_{1\leq j\leq p} n^{-1/2}\sum_{i=1}^n X_{i j}$. Let $Y_1,\dots,Y_n$ be independent random vectors in $\mathbb R^p$ with $Y_i \sim N(0_p, E[X_i X_i'])$, and define $\widetilde Z = \max_{1\leq j\leq p} n^{-1/2} \sum_{i=1}^n Y_{i j}$. Then for every $\delta>0$ and every Borel subset $A$ of $\mathbb R$, we have
$$
P(Z\in A)\leq P(\widetilde Z\in A^{C_1 \delta}) + \frac{C_2\log^2 p}{\delta^3 \sqrt n}\cdot\Big(L_n + M_{n,X}(\delta) + M_{n,Y}(\delta)\Big),
$$
where $C_1$ and $C_2$ are absolute constants and
\begin{align*}
&L_n = \max_{1\leq j\leq p}\frac{1}{n}\sum_{i=1}^n E[|X_{i j}|^3],\\
&M_{n,X}(\delta) = \frac{1}{n}\sum_{i=1}^n E\Big[\max_{1\leq j\leq p}|X_{i j}|^3 \cdot 1\Big\{\max_{1\leq j\leq p}|X_{i j}| > \delta \sqrt n/\log p\Big\}\Big],\\
&M_{n,Y}(\delta) = \frac{1}{n}\sum_{i=1}^n E\Big[\max_{1\leq j\leq p}|Y_{i j}|^3 \cdot 1\Big\{\max_{1\leq j\leq p}|Y_{i j}| > \delta \sqrt n/\log p\Big\}\Big].
\end{align*}
\end{lemma}
\begin{proof}
See Theorem 3.1 in Chernozhukov, Chetverikov and Kato \cite{CCK2016}.
\end{proof}

\begin{lemma}\label{lem: CCK gaussian comparison inequality}
Let $X = (X_1,\dots,X_p)'$ and $Y = (Y_1,\dots,Y_p)'$ be random vectors in $\mathbb R^p$, $p\geq 2$, with $X\sim N(0_p,\Sigma^X)$ and $Y\sim N(0_p,\Sigma^Y)$. Let $\Delta = \max_{1\leq j,k\leq p}|\Sigma_{j k}^X - \Sigma_{j k}^Y|$, where $\Sigma_{j k}^X$ and $\Sigma_{j k}^Y$ denote the $(j,k)$-th elements of $\Sigma^X$ and $\Sigma^Y$, respectively. Define $Z = \max_{1\leq j\leq p} X_j$ and $\widetilde Z = \max_{1\leq j\leq p} Y_j$. Then for every $\delta>0$ and every Borel subset $A$ of $\mathbb R$,
$$
P(Z\in A) \leq P(\widetilde Z \in A^{\delta}) + C\delta^{-1}\sqrt{\Delta\log p},
$$
where $C$ is an absolute constant.
\end{lemma}
\begin{proof}
See Theorem 3.2 in Chernozhukov, Chetverikov and Kato \cite{CCK2016}.
\end{proof}

\begin{figure}[!hp]
\centering
 \includegraphics[width=.9\textwidth, height=.9\textheight, angle=0]{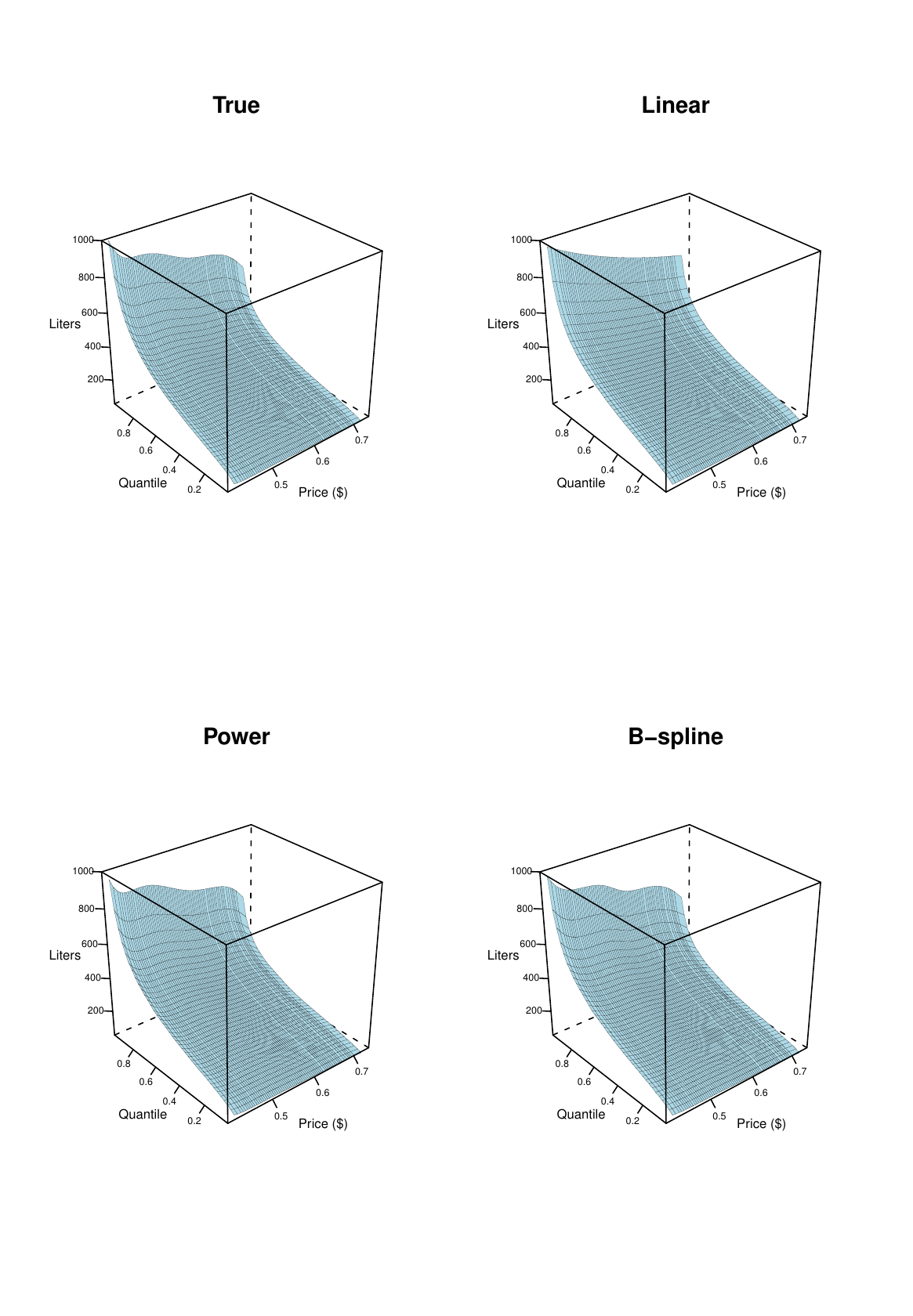}
\caption{\label{fig: estimands demand} Estimands of the quantile
demand surface. Estimands for the linear, power and B-spline series
estimators are obtained numerically using 500,100 simulations.}
\end{figure}

\begin{figure}[!hp]
\centering
 \includegraphics[width=.9\textwidth, height=.9\textheight]{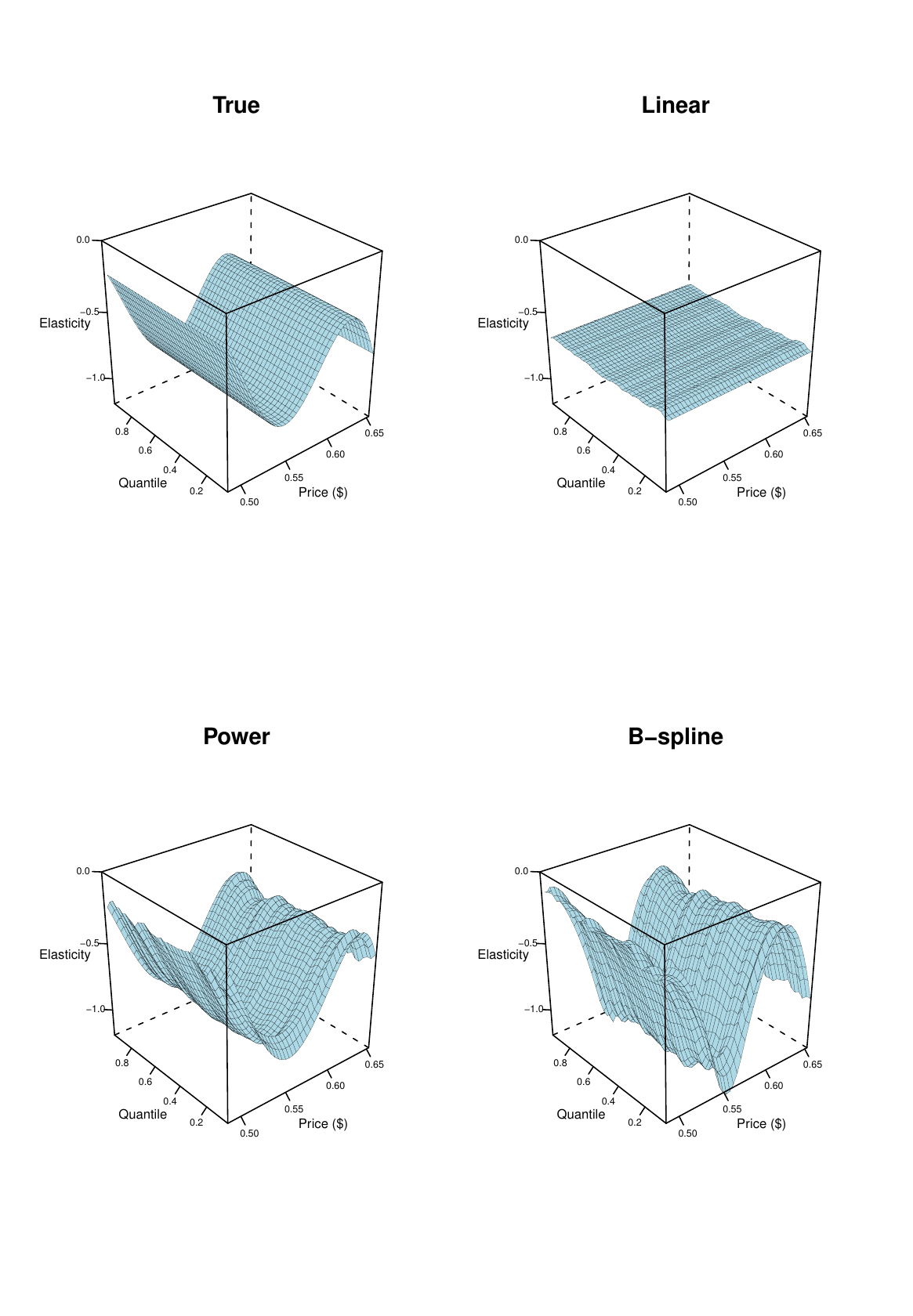}
\caption{\label{fig: estimands elasticity} Estimands of the quantile
price elasticity surface. Estimands for the linear, power and B-spline
series estimators are obtained numerically using 500,100
simulations.}
\end{figure}

\bibliographystyle{plain}

\end{document}